%% file: RR-8934.tex
\documentclass[twoside]{article}

\pdfoutput=1
\input{command}

\usepackage[a4paper]{geometry}
\usepackage[T1]{fontenc} 
\usepackage{RR}

\usepackage{a4wide}

\RRdate{July 2016}

\RRauthor{
  François Clément%
  \thanks[serena]{Équipe Serena.
    {\tt Francois.Clement@inria.fr}.}
  \and
  Vincent Martin%
  \thanks[lmac]{LMAC, UTC, BP 20529, FR-60205 Compiègne, France.
    {\tt Vincent.Martin@utc.fr}}
}
\authorhead{%
  F. Clément, \& V. Martin}

\RRtitle{%
  Le théorème de {\LM}. \\
  Une preuve détaillée en vue d'une formalisation en Coq}
\RRetitle{%
  The {\LM} Theorem. \\
  A detailed proof to be formalized in Coq}
\titlehead{%
  A detailed proof of the {\LM} Theorem to be formalized in Coq}

\RRnote{%
  This research was partly supported by
  GT ELFIC from Labex DigiCosme - Paris-Saclay.}

\RRresume{%
  Pour obtenir la plus grande confiance en la correction de programmes de
  simulation numérique implémentant la méthode des éléments finis, il faut
  formaliser les notions et résultats mathématiques qui permettent d'établir la
  justesse de la méthode.
  Le théorème de {\LM} peut être vu comme l'un de ces fondements théoriques~:
  sous des hypothèses de complétude et de coercivité, il énonce l'existence et
  l'unicité de la solution de certains problèmes aux limites posés sous forme
  faible.
  L'objectif de ce document est de fournir à la communauté preuve formelle une
  preuve papier très détaillée du théorème de {\LM}.}
\RRabstract{%
  To obtain the highest confidence on the correction of numerical simulation
  programs implementing the finite element method, one has to formalize the
  mathematical notions and results that allow to establish the soundness of the
  method.
  The {\LM} theorem may be seen as one of those theoretical cornerstones: under
  some completeness and coercivity assumptions, it states existence and
  uniqueness of the solution to the weak formulation of some boundary value
  problems.
  The purpose of this document is to provide the formal proof community with a
  very detailed pen-and-paper proof of the {\LM} theorem.}

\RRmotcle{%
  théorème de {\LM},
  méthode des éléments finis,
  preuve mathématique détaillée,
  preuve formelle en analyse réelle}
\RRkeyword{%
  {\LM} theorem,
  finite element method,
  detailed mathematical proof,
  formal proof in real analysis}

\RRprojet{Serena}

\RCParis

\begin{document}

\RRNo{8934}
\makeRR

\clearpage
\tableofcontents

\clearpage
\input{lax_milgram}

\clearpage
\nocite{*}
\bibliography{biblio}
\bibliographystyle{plain}
\phantomsection \label{references}
\addcontentsline{toc}{section}{References}

\clearpage
\appendix

\section{Lists of statements}
\label{s:lists-of-statements}

\renewcommand{\listtheoremname}{List of Definitions}
\phantomsection \label{listofdef}
\addcontentsline{toc}{subsection}{\listtheoremname}
\listoftheorems[ignoreall,show=definition]

\renewcommand{\listtheoremname}{List of Lemmas}
\phantomsection \label{listoflem}
\addcontentsline{toc}{subsection}{\listtheoremname}
\listoftheorems[ignoreall,show=lemma]

\renewcommand{\listtheoremname}{List of Theorems}
\phantomsection \label{listofthm}
\addcontentsline{toc}{subsection}{\listtheoremname}
\listoftheorems[ignoreall,show=theorem]

\clearpage
\begin{sidewaysfigure}
  \centering
  \includegraphics[width=\textheight]{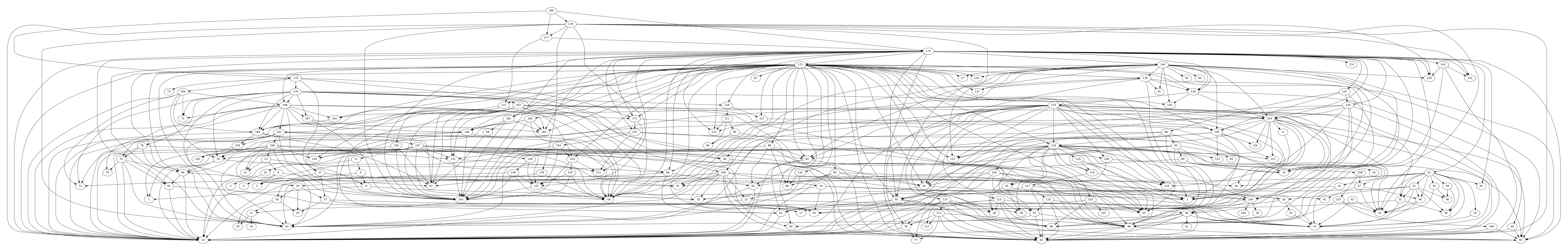}
  \\[2cm]
  \includegraphics[width=\textheight]{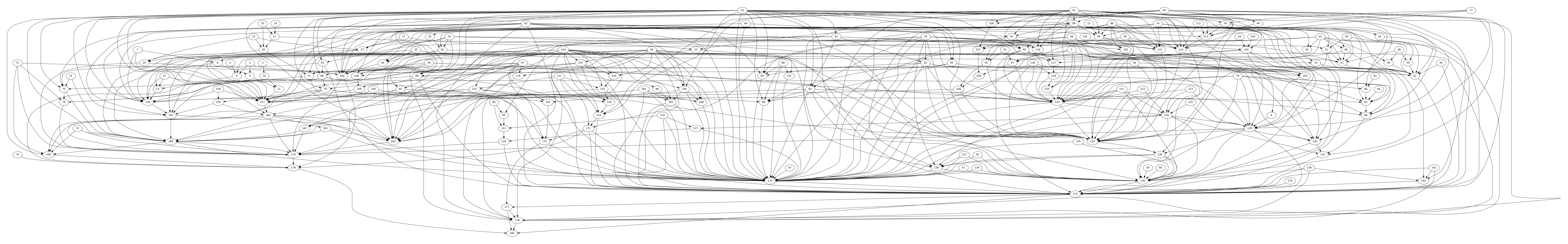}
  \caption{%
    Dependency graph (both ways).
    All dependencies are detailed in Appendices~\ref{s:depends-directly-from}
    and~\ref{s:is-a-direct-dependency-of}.}
  \label{f:depend:graph}
\end{sidewaysfigure}

\clearpage
\input{tree_lm}

\end{document}

%% file: command.tex
\usepackage[english]{babel}
\usepackage[utf8]{inputenc} 
\usepackage{amssymb}
\usepackage{amsmath}

\usepackage{titlesec}
\usepackage{color}

\definecolor{darkgreen}{rgb}{0.2,0.8,0.2}
\definecolor{lightred}{RGB}{255,225,225}
\definecolor{lightgreen}{RGB}{225,255,225}
\definecolor{lightblue}{RGB}{225,255,255}
\definecolor{lightgray}{gray}{0.8}

\newcommand{\assume}[1]{\textcolor{red}{\bf\boldmath{#1}}}
\newcommand{\defcolor}{lightgreen}
\newcommand{\lemcolor}{lightblue}
\newcommand{\thmcolor}{lightred}
\newcommand{\rmkcolor}{lightgray}

\newcommand{\coloredbox}[2]{\fcolorbox{black}{#1}{#2}}
\newcommand{\defbox}[1]{\coloredbox{\defcolor}{#1}}
\newcommand{\lembox}[1]{\coloredbox{\lemcolor}{#1}}
\newcommand{\thmbox}[1]{\coloredbox{\thmcolor}{#1}}
\newcommand{\rmkbox}[1]{\coloredbox{\rmkcolor}{#1}}

\usepackage{epic,eepic}

\newcommand{\mybox}[4]{%
  \put#1{\makebox(0,0){%
      \coloredbox{#4}{%
        \begin{minipage}{4cm}
          #2\\
          \mbox{}\hfill{\scriptsize (#3 space)}
        \end{minipage}
      }}}
}

\newcommand{\myspacebox}[4]{%
  \put#1{\makebox(0,0){%
      \coloredbox{#4}{%
        \begin{minipage}{#2}
          \centering
          \mbox{}\\#3 space\\\mbox{}
        \end{minipage}
      }}}
}

\definecolor{metriccolor}{RGB}{225,255,255}
\definecolor{completecolor}{RGB}{195,255,255}
\definecolor{vectorcolor}{RGB}{255,225,225}
\definecolor{normedcolor}{RGB}{255,195,195}
\definecolor{innercolor}{RGB}{255,135,135}
\definecolor{hilbertcolor}{RGB}{212,162,212}

\newcommand{\completebox}[2]{\mybox{#1}{#2}{complete}{completecolor}}

\newcommand{\normedbox}[2]{\mybox{#1}{#2}{normed vector}{normedcolor}}
\newcommand{\innerbox}[2]{\mybox{#1}{#2}{inner product}{innercolor}}
\newcommand{\hilbertbox}[2]{\mybox{#1}{#2}{Hilbert}{hilbertcolor}}

\usepackage{amsthm, thmtools}
\usepackage{nameref, hyperref}
\usepackage[normalem]{ulem}

\declaretheoremstyle[
  headfont=\normalfont\bfseries,
  notefont=\normalfont\bfseries\itshape,
  bodyfont=\normalfont]{mydef}
\declaretheoremstyle[
  headfont=\normalfont\bfseries,
  notefont=\normalfont\bfseries\itshape,
  bodyfont=\normalfont\itshape]{mythm}
\declaretheoremstyle[
  headfont=\normalfont\itshape,
  notefont=\normalfont\itshape,
  bodyfont=\normalfont]{myrmk}

\newcommand{\mydeclaretheoremshaded}[3]{
  \declaretheorem[
    #1,
    shaded={
      textwidth=0.99\textwidth,
      bgcolor=#2,
      rulecolor=black,
      rulewidth=0.5pt}]{#3}
}

\newcommand{\mydeclaretheoremstyle}[3]{
  \mydeclaretheoremshaded{numberlike=theorem, style=#1}{#2}{#3}
}

\newcommand{\BLM}{Babu\v{s}ka--Lax--Milgram}
\newcommand{\BNB}{Banach--Ne\v{c}as--Babu\v{s}ka}
\newcommand{\CS}{Cauchy--Schwarz}
\newcommand{\HB}{Hahn--Banach}
\newcommand{\LM}{Lax--Milgram}
\newcommand{\LMC}{Lax--Milgram--C\'ea}
\newcommand{\MP}{Milman--Pettis}
\newcommand{\RF}{Riesz--Fr\'echet}
\newcommand{\ZF}{Zermelo--Fraenkel}

\mydeclaretheoremshaded{style=mythm}{\thmcolor}{theorem}
\mydeclaretheoremstyle{mythm}{\lemcolor}{lemma}
\mydeclaretheoremstyle{mydef}{\defcolor}{definition}
\mydeclaretheoremstyle{myrmk}{\rmkcolor}{remark}
\mydeclaretheoremshaded{%
  numbered=no, name={\LM} theorem}{\thmcolor}{lmthm}
\mydeclaretheoremshaded{%
  numbered=no, name={\LMC} theorem}{\thmcolor}{lmcthm}
\mydeclaretheoremshaded{%
  numbered=no, name=Representation lemma for bounded bilinear forms}%
  {\lemcolor}{bilinlem}
\mydeclaretheoremshaded{%
  numbered=no, name={\RF} theorem}{\thmcolor}{rfthm}
\mydeclaretheoremshaded{%
  numbered=no, name=Orthogonal projection theorem for a complete subspace}%
  {\thmcolor}{projthm}
\mydeclaretheoremshaded{%
  numbered=no, name=Fixed point theorem}{\thmcolor}{fpthm}

\usepackage{rotating}
\usepackage{graphicx}

\graphicspath{{./}}

\newcommand{\thref}[1]{\ref{#1} ({\em \nameref{#1}})}
\newcommand{\threfc}[2]{\ref{#1} ({\em \nameref{#1}}, \uline{#2})}

\newcommand{\proofpar}[1]{{\bf\boldmath{#1}.}}
\newcommand{\proofparskip}[1]{\medskip\noindent\proofpar{#1}}

\newcommand{\FEM}{Finite Element Method}

\newcommand{\vectorspace}{space}
\newcommand{\Normedvectorspace}{Normed vector space}
\newcommand{\normedvectorspace}{normed vector space}

\newcommand{\software}[1]{{\sf #1}}
\newcommand{\coq}{\software{Coq}}

\newcommand{\st}{{\,|\,}}
\newcommand{\ie}{i.e.}
\newcommand{\eg}{e.g.}
\newcommand{\half}{{\frac{1}{2}}}
\newcommand{\p}{\partial}

\renewcommand{\equiv}{\Leftrightarrow}
\newcommand{\Equiv}{\quad\Longleftrightarrow\quad}
\renewcommand{\implies}{\Rightarrow}
\newcommand{\Implies}{\quad\Longrightarrow\quad}
\newcommand{\conj}{\land}
\newcommand{\Conj}{\quad\conj\quad}
\newcommand{\disj}{\lor}
\newcommand{\Disj}{\quad\disj\quad}

\newcommand{\calB}{\mathcal{B}}

\newcommand{\calF}{\mathcal{F}}
\newcommand{\calL}{\mathcal{L}}

\newcommand{\calS}{\mathcal{S}}

\newcommand{\matC}{{\mathbb C}}
\newcommand{\matK}{{\mathbb K}}
\newcommand{\matN}{{\mathbb N}}

\newcommand{\matR}{{\mathbb R}}
\newcommand{\matRbar}{\overline{\matR}}
\newcommand{\matRplus}{{\matR^+}}
\newcommand{\eps}{\varepsilon}
\newcommand{\fhi}{\varphi}

\newcommand{\Closure}[1]{\overline{#1}}

\newcommand{\pOmega}{\p\Omega}

\newcommand{\Ball}{\calB}

\newcommand{\ClosedBall}{\Ball^c}

\newcommand{\Sphere}{\calS}

\newcommand{\cball}[3]{\ClosedBall_{#1} (#2, #3)}

\newcommand{\sphere}[3]{\Sphere_{#1} (#2, #3)}

\newcommand{\unitcB}{\ClosedBall_1}

\newcommand{\unitS}{\Sphere_1}

\newcommand{\Ep}{{E^\prime}}
\newcommand{\Fp}{{F^\prime}}
\newcommand{\Hp}{{H^\prime}}
\newcommand{\Np}{{N^\prime}}

\newcommand{\ap}{{a^\prime}}
\newcommand{\lp}{{l^\prime}}
\newcommand{\np}{{n^\prime}}
\newcommand{\pp}{{p^\prime}}
\newcommand{\qp}{{q^\prime}}
\newcommand{\up}{u^\prime}
\newcommand{\upp}{{u^{\prime\prime}}}
\newcommand{\vp}{{v^\prime}}
\newcommand{\vpp}{{v^{\prime\prime}}}
\newcommand{\wpr}{w^\prime}
\newcommand{\wpp}{w^{\prime\prime}}
\newcommand{\xp}{{x^\prime}}
\newcommand{\lambdap}{{\lambda^\prime}}
\newcommand{\lambdapp}{{\lambda^{\prime\prime}}}

\newcommand{\fhip}{{\fhi^\prime}}
\newcommand{\fhipp}{{\fhi^{\prime\prime}}}
\newcommand{\epsp}{\eps^\prime}

\newcommand{\Hh}{H_h}
\newcommand{\uh}{u_h}
\newcommand{\vh}{v_h}

\newcommand{\XxX}{{X \times X}}

\newcommand{\ExE}{{E \times E}}
\newcommand{\EpxE}{{\Ep \times E}}
\newcommand{\matKxE}{{\matK \times E}}
\newcommand{\FxF}{{F \times F}}
\newcommand{\matKxF}{{\matK \times F}}
\newcommand{\ExF}{{E \times F}}
\newcommand{\ExFxExF}{{(\ExF)\times(\ExF)}}
\newcommand{\matKxExF}{{\matK\times(\ExF)}}
\newcommand{\GxG}{{G \times G}}

\newcommand{\supX}[1]{{\sup (#1 (X))}}
\newcommand{\supXf}{\supX{f}}
\newcommand{\supXmf}{\supX{(- f)}}
\newcommand{\supXlambdaf}{\supX{(\lambda f)}}
\newcommand{\maxX}[1]{{\max (#1 (X))}}
\newcommand{\maxXf}{\maxX{f}}
\newcommand{\maxXmf}{\maxX{(- f)}}
\newcommand{\infX}[1]{{\inf (#1 (X))}}
\newcommand{\infXf}{\infX{f}}

\newcommand{\minX}[1]{{\min (#1 (X))}}
\newcommand{\minXf}{\minX{f}}

\newcommand{\floor}[1]{\left\lfloor #1 \right\rfloor}
\newcommand{\ceil}[1]{\left\lceil #1 \right\rceil}

\newcommand{\Span}[1]{{{\rm span} (#1)}}
\newcommand{\Line}[1]{\Span{\{#1\}}}
\newcommand{\Ker}[1]{{\ker (#1)}}

\newcommand{\FSp}[2]{{\calF \left(#1, #2\right)}}
\newcommand{\LSp}[2]{{\calL \left(#1, #2\right)}}
\newcommand{\LcSp}[2]{{\calL_c \left(#1, #2\right)}}
\newcommand{\Ld}[1]{{\calL_2 \left(#1\right)}}

\newcommand{\FXE}{\FSp{X}{E}}
\newcommand{\FXExFXE}{{\FXE \times \FXE}}
\newcommand{\matKxFXE}{{\matK \times \FXE}}
\newcommand{\FEF}{\FSp{E}{F}}

\newcommand{\LEF}{\LSp{E}{F}}
\newcommand{\LFG}{\LSp{F}{G}}
\newcommand{\LEG}{\LSp{E}{G}}
\newcommand{\LEFxLFG}{{\LEF \times \LFG}}

\newcommand{\LEmatK}{\LSp{E}{\matK}}
\newcommand{\LExEmatK}{\LSp{\ExE}{\matK}}

\newcommand{\LcEE}{\LcSp{E}{E}}
\newcommand{\LcEF}{\LcSp{E}{F}}
\newcommand{\LcFG}{\LcSp{F}{G}}
\newcommand{\LcEG}{\LcSp{E}{G}}
\newcommand{\LcEmatK}{\LcSp{E}{\matK}}
\newcommand{\Lcp}[1]{\LcSp{#1}{#1^\prime}}
\newcommand{\LcEEp}{\Lcp{E}}
\newcommand{\LcHH}{\LcSp{H}{H}}
\newcommand{\LcHHp}{\Lcp{H}}
\newcommand{\LcHpH}{\LcSp{H^\prime}{H}}

\newcommand{\LdE}{\Ld{E}}

\newcommand{\Identity}{{\rm Id}}
\newcommand{\identity}[1]{\Identity_{#1}}

\newcommand{\idE}{\identity{E}}
\newcommand{\idH}{\identity{H}}

\newcommand{\vertiii}[1]{
  {\left\vert\kern-0.25ex\left\vert\kern-0.25ex\left\vert #1
    \right\vert\kern-0.25ex\right\vert\kern-0.25ex\right\vert}}

\newcommand{\innerprod}[3]{\left(#2, #3\right)_{#1}}
\newcommand{\bra}[1]{\left< #1 \right|}
\newcommand{\ket}[1]{\left| #1 \right>}
\newcommand{\braket}[2]{\left< #1 | #2 \right>}
\newcommand{\dualprod}[3]{\braket{#2}{#3}_{#1^\prime,#1}}
\newcommand{\norm}[2]{\left\|#2\right\|_{#1}}
\newcommand{\tnorm}[2]{\vertiii{#2}_{#1}}

\newcommand{\psG}[2]{\innerprod{G}{#1}{#2}}
\newcommand{\psGdotdot}{\psG{\cdot}{\cdot}}
\newcommand{\psH}[2]{\innerprod{H}{#1}{#2}}
\newcommand{\psHdotdot}{\psH{\cdot}{\cdot}}
\newcommand{\pdE}[2]{\dualprod{E}{#1}{#2}}
\newcommand{\pdEdotdot}{\pdE{\cdot}{\cdot}}

\newcommand{\pdH}[2]{\dualprod{H}{#1}{#2}}

\newcommand{\PrenG}{N_G}
\newcommand{\prenG}[1]{\PrenG (#1)}
\newcommand{\PrenEF}{N_{E, F}}
\newcommand{\prenEF}[1]{\PrenEF (#1)}

\newcommand{\nrm}[1]{\norm{}{#1}}
\newcommand{\nrmdot}{\nrm{\cdot}}
\newcommand{\nE}[1]{\norm{E}{#1}}
\newcommand{\nEdot}{\nE{\cdot}}
\newcommand{\nExE}[2]{\norm{\ExE}{(#1, #2)}}
\newcommand{\nExEdotdot}{\nExE{\cdot}{\cdot}}
\newcommand{\nF}[1]{\norm{F}{#1}}
\newcommand{\nFdot}{\nF{\cdot}}
\newcommand{\nExF}[2]{\norm{\ExF}{(#1, #2)}}
\newcommand{\nExFdotdot}{\nExF{\cdot}{\cdot}}
\newcommand{\nG}[1]{\norm{G}{#1}}
\newcommand{\nGdot}{\nG{\cdot}}
\newcommand{\nH}[1]{\norm{H}{#1}}
\newcommand{\nHdot}{\nH{\cdot}}

\newcommand{\tnEF}[1]{\tnorm{F, E}{#1}}
\newcommand{\tnEFdot}{\tnEF{\cdot}}
\newcommand{\tnFG}[1]{\tnorm{G, F}{#1}}
\newcommand{\tnEG}[1]{\tnorm{G, E}{#1}}

\newcommand{\tnEmatK}[1]{\tnorm{\matK, E}{#1}}
\newcommand{\tnEmatKdot}{\tnEmatK{\cdot}}
\newcommand{\tnHH}[1]{\tnorm{H, H}{#1}}
\newcommand{\tnHHdot}{\tnHH{\cdot}}
\newcommand{\tnp}[2]{\tnorm{#1^\prime, #1}{#2}}
\newcommand{\tnEEp}[1]{\tnp{E}{#1}}
\newcommand{\tnHHp}[1]{\tnp{H}{#1}}

\newcommand{\dist}[1]{d_{#1}}
\newcommand{\dst}{\dist{}}
\newcommand{\dE}{\dist{E}}
\newcommand{\dExE}{\dist{\ExE}}
\newcommand{\dEF}{\dist{F, E}}

\newcommand{\nEp}[1]{\norm{E^\prime}{#1}}
\newcommand{\nEpdot}{\nEp{\cdot}}
\newcommand{\nHp}[1]{\norm{H^\prime}{#1}}
\newcommand{\nHpdot}{\nHp{\cdot}}

\newcommand{\Blf}{{\fhi}}
\newcommand{\blfE}[2]{{\Blf (#1, #2)}}
\newcommand{\BlfH}{a}
\newcommand{\blfH}[2]{{\BlfH (#1, #2)}}
\newcommand{\LfH}{f}
\newcommand{\lfH}[1]{{\LfH (#1)}}

\newcommand{\Hsob}[3]{H_{#2}^{#1}(#3)}
\newcommand{\HsobOm}[2]{\Hsob{#1}{#2}{\Omega}}
\newcommand{\HsobOmI}{\HsobOm{1}{}}
\newcommand{\HsobOmIO}{\HsobOm{1}{0}}
\newcommand{\Lsob}[1]{L^{2}(#1)}
\newcommand{\LsobOm}{\Lsob{\Omega}}

%% file: lax_milgram.tex
\section{Introduction}

As stated and demonstrated in~\cite{bol:tcm:14}, formal proof tools are now
mature to address the verification of scientific computing programs.
One of the most thrilling aspects of the approach is that the round-off error
due to the use of IEEE-754 floating-point arithmetic can be fully taken into
account.
One of the most important issue in terms of manpower is that all the
mathematical notions and results that allow to establish the soundness of the
implemented algorithm must be formalized.

The long term purpose of this study is to formally prove programs using the
{\FEM}.
The {\FEM} is now widely used to solve partial differential equations, and
its success is partly due to its well established
mathematical foundation,
{\eg} see~\cite{cia:fem:02,qv:nap:97,eg:tpf:04,ztz:fem:13}.
It seems important now to verify the scientific
computing programs based on the {\FEM}, in
order to certify their results.
The present report is a first contribution toward this ultimate goal.

The {\LM} theorem is one of the key ingredients used to build the {\FEM}.
It is a way to establish existence and uniqueness of the solution to the weak
formulation and its discrete approximation;
it is valid for coercive linear operators set on Hilbert spaces
({\ie} complete inner product spaces over the field of real or complex
numbers).
A corollary known as the Céa's lemma provides a quantification of the error
between the computed approximation and the unknown solution.
In particular, the {\LM} theorem is sufficient to prove existence and
uniqueness of the solution to the (weak formulation of the) standard Poisson
problem defined as follows.
Knowing a function~$f$ defined over a regular and bounded domain~$\Omega$
of~$\matR^d$ with $d=1$, $2$, or~$3$, with its boundary denoted by~$\pOmega$,
\begin{equation}
  \label{e:laplace-problem-strong}
  \mbox{find~$u$ such that:}
  \left\{
    \begin{array}{rcll}
      - \Delta u & = & f & \mbox{in } \Omega, \\
      u & = & 0 & \mbox{on } \pOmega,
    \end{array}
  \right.
\end{equation}
where $\Delta=\frac{\p^2}{\p x^2}+\frac{\p^2}{\p y^2}+\frac{\p^2}{\p z^2}$ is
the Laplace operator.
Equation~\eqref{e:laplace-problem-strong} is the strong formulation of
Laplace problem, its weak formulation and the link with the {\LM} theorem is
given in Conclusion, perspectives, see
Section~\ref{s:conclusions-perspectives}.
We do not intend to limit ourselves to this particular problem, but we stress
that our work covers this standard problem that is the basis for the study of
many other physical problems.

Other mathematical tools can be used to establish existence and uniqueness of
the solution to weak problems.
For instance, the {\BNB} theorem for Banach spaces ({\ie} complete normed real
or complex vector spaces), from which one can deduce the {\LM} theorem, {\eg}
see~\cite{eg:tpf:04}, or the theory for mixed and saddle-point problems,
that is used for instance for some fluid problems, {\eg}
see~\cite{bre:eua:74,gr:fem:86}.
However, our choice is mainly guided by our limited manpower and by the
intuitionistic logic of the interactive theorem prover we intend to use: we
try to select an elementary and constructive path of proof.
This advocates to work in a first step with Hilbert spaces rather than Banach
spaces, and to try to avoid the use of {\HB} theorem whose proof is based on
Zorn's lemma (an equivalent of the axiom of choice in {\ZF} set theory).

Some other steps will be necessary for the formalization of the {\FEM}:
the measure theory is required to formalize Sobolev spaces such
as~$L^2(\Omega)$, $H^1(\Omega)$ and~$H_0^1(\Omega)$ on some reasonable
domain~$\Omega$, and establish that they are Hilbert spaces on which the
{\LM} theorem applies;
as well as the notion of distribution to set up the correct framework to deal
with weak formulations;
and finally chapters of the interpolation and approximation theories to
define the discrete finite element approximation spaces.

The purpose of this document is to provide the ``formal proof'' community
with a very detailed pen-and-paper proof of the {\LM} theorem.
The most basic notions and results such as ordered field properties
of~$\matR$ and properties of elementary functions over~$\matR$ are supposed
known and are not detailed further.
One of the key issues is to select in the literature the proof involving the
simplest notions, and in particular not to justify the result by applying a
more general statement.
Once a detailed proof of the {\LM} theorem is written, the next step is to
formalize all notions and results in a formal proof tool such as {\coq}%
\footnote{\href{http://coq.inria.fr/}{http://coq.inria.fr/}}.
At this point, which is not the subject of the present paper, it will be
necessary to take care of the specificities of the classical logic commonly
used in mathematics: in particular, determine where there is need for the law
of excluded middle, and discuss decidability issues.

\bigskip

The paper is organized as follows.
Different ways to prove variants of the {\LM} theorem collected from the
literature are first reviewed in Section~\ref{s:state-of-the-art}.
The chosen proof path is then sketched in
Section~\ref{s:statement-and-sketch-of-the-proof}, and fully detailed in
Section~\ref{s:detailed-proof}.
Finally, lists of statements and direct dependencies are gathered in the
appendix.

\section{State of the art}
\label{s:state-of-the-art}

We review some works of a few authors, mainly from the French school, that
provide some details about statements similar to the {\LM} theorem.

As usual, proofs provided in the literature are not comprehensive, and we
have to cover a series of books to collect all the details necessary for a
formalization in a formal proof tool such as~{\coq}.
Usually, Lecture Notes in undergraduate mathematics are very helpful and
we selected~\cite{gos:cms1:93,gos:cms2:93,gos:cms3:93} among many other
possible choices.

\subsection{Brézis}

In~\cite{bre:af:83}, the {\LM} theorem is stated as Corollary~V.8 (p.~84).
Its proof is obtained from Theorem~V.6 (Stampacchia, p.~83) and by means
similar to the ones used in the proof of Corollary~V.4 (p.~80) for the
characterization of the projection onto a closed subspace.

The proof of the Stampacchia theorem for a bilinear form on a closed convex
set has four main arguments:
the {\RF} representation theorem (Theorem~V.5 p.~81), the characterization of
the projection onto a closed convex set (Theorem~V.2 p.~79), the fixed point
theorem on a complete metric space (Theorem~V.7 p.~83), and the continuity of
the projection onto a closed convex set (Proposition~V.3 p.~80).

The proof of the fixed point theorem uses the notions of distance,
completeness, and sequential continuity
({\eg} see~\cite[Theorem~4.102 p.~115]{gos:cms2:93}).

The proof of the {\RF} representation theorem and the existence of the
projection onto a closed convex set share the possibility to use the notion
of reflexive space through Proposition~V.1 (p.~78) and Theorem~III.29 ({\MP},
p.~51).
The latter states that uniformly convex Banach spaces are reflexive
({\ie} isomorphic to their topological double dual), and its proof uses the
notions of weak and weak-$\star$ topologies.
In this case, the existence of the projection onto a closed convex set also
needs the notions of compactness and lower semi-continuity through
Corollary~III.20 (p.~46), and the call to {\HB} theorem which depends on
Zorn's lemma or the axiom of choice through Theorem~III.7 and Corollary~III.8
(p.~38).

More elementary proofs are also presented in~\cite{bre:af:83}.
The {\RF} theorem only needs the closed kernel lemma (for continuous linear
maps) and the already cited Corollary~V.4.
The existence of the projection onto a closed convex set can be obtained
through elementary and geometrical arguments.
Then, the uniqueness and the characterization of the projection onto a closed
convex set derives from the parallelogram identity and {\CS} inequality.

Complements about the projection onto a closed convex set subset can be found
in~\cite[Lemmas~14.30 and~14.32 pp.~225--228]{gos:cms3:93}.
See also~\cite[p.~90]{yos:fa:80} and~\cite[Theorem~A.28 p.~467]{eg:tpf:04}
for proofs of the {\RF} representation theorem.

\subsection{Ciarlet}

In~\cite{cia:fem:02}, the {\LM} theorem is stated as Theorem~1.1.3
(pp.~8--10).
The structure of the proof is similar to the one proposed
in~\cite{bre:af:83} for Stampacchia theorem, but simplified to the case of a
subspace instead of a closed convex subset
({\eg} see~\cite[Theorems~14.27 and~14.29 pp.~224--225]{gos:cms3:93}).

\subsection{Ern--Guermond}

In~\cite{eg:tpf:04}, the {\LM} theorem is stated as Lemma~2.2 (p.~83).
The first proof is obtained as a consequence of the more general {\BNB}
theorem set on a Banach space (Theorem~2.6 p.~85).
The proof is spread out in Section~A.2 through Theorem~A.43 (p.~472) for
the characterization of bijective Banach operators, Lemma~A.39 (p.~470) which
is a consequence of the closed range theorem (Theorem~A.34 p.~468, see
also~\cite[pp.~205--208]{yos:fa:80} and~\cite[p.~28]{bre:af:83}) and of the
open mapping theorem (Theorem~A.35 p.~469).

A simpler alternative proof without the use of the {\BNB} theorem is proposed
in Exercise~2.11 (p.~107) through the closed range theorem and a density
argument.
For the latter, Corollary~A.18 (p.~466) is a consequence of the {\HB} theorem
(Theorem~A.16 p.~465, see also~\cite[Theorem~5.19]{rud:rca:87}
and~\cite[p.~7]{bre:af:83}).

\subsection{Quarteroni--Valli}

In~\cite{qv:nap:97}, the {\LM} theorem is stated as Theorem~5.1.1 (p.~133).
A variant, also known as {\BLM} theorem, is stated for a bilinear form
defined over two different Hilbert spaces (Theorem~5.1.2 p.~135).
Their proofs are similar: they both use the {\RF} representation theorem and
the closed range theorem.
Note that when the bilinear form is symmetric, the {\RF} representation
theorem and a minimization argument are sufficient to build the proof
(Remark~5.1.1 p.~134).

\section{Statement and sketch of the proof}
\label{s:statement-and-sketch-of-the-proof}

Let~$H$ be a real Hilbert space.
Let~$\psHdotdot$ be its inner product, and~$\nHdot$ the associated norm.
Let~$\Hp$ be its topological dual ({\ie} the space of continuous linear forms
on~$H$).
Let~$\BlfH$ be a bilinear form on~$H$, and let $\LfH\in\Hp$ be a continuous
linear form on~$H$.
Let~$\Hh$ be a closed vector subspace of~$H$ (in practice, $\Hh$~is finite
dimensional).
The {\LM} theorem states existence and uniqueness of the solution to the
following general problems:
\begin{eqnarray}
  \label{e:general-problem}
   \mbox{find } u \in H \mbox{ such that:}
   & &
   \forall v \in H,\quad
   \blfH{u}{v} = \lfH{v}; \\
  \label{e:general-problem-discrete}
  \mbox{find } \uh \in \Hh \mbox{ such that:}
  & &
  \forall \vh \in \Hh,\quad
  \blfH{\uh}{\vh} = \lfH{\vh}.
\end{eqnarray}
The main statement is the following:
\begin{lmthm}
  Assume that~$\BlfH$ is bounded and coercive with constant $\alpha>0$.
  Then, there exists a unique $u\in H$ solution to
  Problem~\eqref{e:general-problem}.
  Moreover, $\nH{u}\leq\frac{1}{\alpha}\nHp{\LfH}$.
\end{lmthm}

\bigskip

\setlength{\unitlength}{1mm}

\begin{figure}[htb]
  \centering
  \begin{picture}(100,50)(0,5)
    \thicklines
    \innerbox{(0,50)}{Orthogonal projection\\onto complete subspace}
    \hilbertbox{(50,50)}{{\RF}\\representation of dual}
    \normedbox{(0,30)}{Representation of\\bounded bilinear form}
    \hilbertbox{(50,30)}{{\LM}}
    \completebox{(100,30)}{Fixed point theorem}
    \hilbertbox{(50,10)}{{\LMC\\on finite dim subspace}}
    \put(21.75,50){\vector(1,0){6.5}}
    \put(50,43){\vector(0,-1){8.25}}
    \put(21.75,30){\vector(1,0){6.5}}
    \put(78.25,30){\vector(-1,0){6.5}}
    \put(50,25.25){\vector(0,-1){8.25}}
  \end{picture}
  \caption{Hierarchy of results for a proof of the {\LM} theorem.}
  \label{fig:hier:lm}
\end{figure}

The ingredients for the proof were mainly collected from~\cite{bre:af:83}
and~\cite{cia:fem:02}.
The key arguments of the chosen proof path are (the hierarchy is sketched in
Figure~\ref{fig:hier:lm}):
\begin{itemize}
\item the representation lemma for bounded bilinear forms;
\item the {\RF} representation theorem;
\item the orthogonal projection theorem for a complete subspace;
\item the fixed point theorem for a contraction on a complete metric space.
\end{itemize}
Note that the same type of arguments can be used to prove the
more general Stampacchia theorem.
We give now more hints about the structure of the main steps of the proof.

\subsection{Sketch of the proof of the {\LMC} theorem}
\label{ss:sketch-of-proof-of-lmc-th}

\begin{lmcthm}
  Assume that~$\BlfH$ is bounded with continuity constant~$C\geq 0$ and
  coercive with constant $\alpha>0$.
  Then, there exist a unique $u\in H$ solution to
  Problem~\eqref{e:general-problem}, and a unique $\uh\in\Hh$ solution to
  Problem~\eqref{e:general-problem-discrete}.
  Moreover, $\nH{u}\leq\frac{1}{\alpha}\nHp{\LfH}$ and for all $\vh\in\Hh$,
  $\nH{u-\uh}\leq\frac{C}{\alpha}\nH{u-\vh}$.
\end{lmcthm}

The proof of the {\LMC} theorem goes as follows (this proof uses the notion
of finite dimensional subspace):
\begin{itemize}
\item a finite dimensional subspace is closed;
\item a closed subspace of a Hilbert space is a Hilbert space;
\item thus, the {\LM} theorem applies to~$\Hh$;
\item finally, Céa's error estimation is obtained from the Galerkin
  orthogonality property, and boundedness and coercivity of the bilinear
  form~$\BlfH$.
\end{itemize}

\subsection{Sketch of the proof of the {\LM} theorem}
\label{ss:sketch-of-proof-of-lm-th}

\begin{lmthm}
  Assume that~$\BlfH$ is bounded and coercive with constant $\alpha>0$.
  Then, there exists a unique $u\in H$ solution to
  Problem~\eqref{e:general-problem}.
  Moreover, $\nH{u}\leq\frac{1}{\alpha}\nHp{\LfH}$.
\end{lmthm}

The proof of the {\LM} theorem goes as follows (this proof uses the notions of
Lipschitz continuity, {\normedvectorspace}, bounded and coercive bilinear form,
inner product, orthogonal complement, and Hilbert space):
\begin{itemize}
\item Problem~\eqref{e:general-problem} is first shown to be equivalent to a
  fixed point problem for some contraction~$g$ on the complete metric
  space~$H$:
  \begin{itemize}
  \item the representation lemma for bounded bilinear forms exhibits a
    continuous linear form $A(u)$ such that $\blfH{u}{v}=(A(u))(v)$,
  \item then, the {\RF} representation theorem exhibits representatives
    $\tau(A(u))$ and $\tau(f)$ for both continuous linear forms, and
    Problem~\eqref{e:general-problem} is shown to be equivalent to the linear
    problem $\tau(A(u))=\tau(f)$,
  \item the affine function~$g$ is defined over~$H$ by
    $g(v)=v-\rho\tau(A(v))+\rho\tau(f)$ for some small enough number~$\rho$,
    then, Problem~\eqref{e:general-problem} is shown to be equivalent to find a
    fixed point of~$g$,
  \item finally, $g$ is shown to be a contraction;
  \end{itemize}
\item thus existence and uniqueness of the solution~$u$ are obtained from the
  fixed point theorem applied to~$g$;
\item finally, the estimation of the nonzero solution~$u$ is obtained from the
  coercivity of the bilinear form~$\BlfH$ and the continuity of the linear
  form~$\LfH$.
\end{itemize}

\subsection{Sketch of the proof of the representation lemma
  for bounded bilinear forms}
\label{ss:sketch-of-proof-of-repr-bilin-forms}

\begin{bilinlem}
  Let $(E,\nEdot)$ be a {\normedvectorspace}.
  Let~$\Blf$ be a bilinear form on~$E$.
  Assume that~$\Blf$ is bounded.
  Then, there exists a unique continuous linear map~$A$ from~$E$ to~$\Ep$
  such that for all $u,v\in E$, $\blfE{u}{v}=(A(u))(v)$.
  Moreover, for all~$C\geq 0$ continuity constant of~$\Blf$, we have
  $\tnEEp{A}\leq C$.
\end{bilinlem}

The proof of the representation lemma for bounded bilinear forms goes as
follows (this proof uses the notions of {\normedvectorspace}, continuous
linear map, topological dual, dual norm, and bounded bilinear form):
\begin{itemize}
\item existence of the representative~$A$ is obtained by construction:
  \begin{itemize}
  \item for each~$u$, the function $A_u=(v\mapsto\fhi(u,v))$ is shown to be
    a continuous linear form with $\nEp{A_u}\leq C\,\nE{u}$,
  \item then, the function $A=(u\mapsto A_u)$ is shown to be a continuous
    linear map from~$E$ to~$\Ep$ with $\tnEEp{A}\leq C$;
  \end{itemize}
\item uniqueness of the representative~$A$ follows from the fact that
  continuous linear maps between two {\normedvectorspace}s form a
  {\normedvectorspace}.
\end{itemize}

\subsection{Sketch of the proof of the {\RF} theorem}
\label{ss:sketch-of-proof-of-rf-th}

\begin{rfthm}
  Let $\fhi\in\Hp$ be a continuous linear form on~$H$.
  Then, there exists a unique vector $u\in H$ such that for all $v\in H$,
  $\fhi(v)=\psH{u}{v}$.
  Moreover, the mapping $\tau=(\fhi\mapsto u)$ is a continuous isometric
  isomorphism from~$\Hp$ onto~$H$.
\end{rfthm}

The proof of the {\RF} representation theorem goes as follows (this proof
uses the notions of kernel of a linear map, {\normedvectorspace}, operator
norm, continuous linear map, topological dual, dual norm, inner product
space, orthogonal projection onto a complete subspace, orthogonal complement,
and Hilbert space):
\begin{itemize}
\item uniqueness of the representative~$u_\fhi$ follows from the definiteness
  of the inner product;
\item existence of the representative~$u_\fhi$ is obtained by construction for
  a nonzero~$\fhi$:
  \begin{itemize}
  \item consider the orthogonal projection onto $F=\Ker{\fhi}$, which is
    closed, hence a complete subspace,
  \item a unit vector~$\xi_0$ in~$F^\bot$ such that $\fhi(\xi_0)\not=0$ is
    built from some~$u_0$ picked in the complement of~$F$, and using the
    theorem on the direct sum of a complete subspace and its orthogonal
    complement,
  \item the candidate $u=\fhi(\xi_0)\xi_0\in F^\bot$ is then shown to satisfy
    $\psH{u}{v}=\fhi(v)$;
  \end{itemize}
\item the mapping $\tau=(\fhi\mapsto u_\fhi)$ goes from the topological
  dual~$\Hp$ to the Hilbert space~$H$, its linearity follows the bilinearity of
  the inner product and of the application of linear maps;
\item injectivity of~$\tau$ is straightforward, and surjectivity comes from
  {\CS} inequality;
\item the isometric property of $\tau$ follows from the definition of the dual
  norm, and again from {\CS} inequality;
\item continuity of~$\tau$ follows from the isometric property.
\end{itemize}

\subsection{Sketch of the proof of the orthogonal projection theorem
  for a complete subspace}
\label{ss:sketch-of-proof-of-orth-proj}

\begin{projthm}
  Let $(G,\psGdotdot)$ be a real inner product space.
  Let~$F$ be a complete subspace of~$G$.
  Then, for all $u\in G$, there exists a unique $v\in F$ such that
  $\nG{u-v}=\min_{w\in F}\nG{u-w}$.
\end{projthm}

The proof of the orthogonal projection theorem for a complete subspace goes
as follows (this proof uses the notions of infimum, completeness, inner
product space, and convexity):
\begin{itemize}
\item the result is first shown for a nonempty complete convex subset~$K$:
  \begin{itemize}
  \item existence of the projection of~$u\in G$ onto~$K$ is built as the limit
    of a sequence:
    \begin{itemize}
    \item existence of a sequence $(w_n)_{n\in\matN}$ in~$K$ and a nonnegative
      number~$\delta$ such that $\nG{u-w_n}<\delta+\frac{1}{n+1}$ is first
      obtained from the fact that the function $(w\mapsto \nG{u-w})$ is bounded
      from below (by~$0$),
    \item this sequence is shown to be a Cauchy sequence using the
      parallelogram identity and the definition of convexity,
    \item hence, it is convergent in the complete subset~$K$,
    \item continuity of the norm ensures that the limit of the sequence
      realizes the minimum of the distance;
    \end{itemize}
  \item uniqueness of the projection follows again the parallelogram identity
    and the definition of convexity;
  \end{itemize}
\item a complete subspace is also a nonempty complete convex subset.
\end{itemize}

\subsection{Sketch of the proof of the fixed point theorem}
\label{ss:sketch-of-proof-of-fixed-point-th}

\begin{fpthm}
  Let $(X,d)$ be a complete metric space.
  Let~$f:X\rightarrow X$ be a contraction.
  Then, there exists a unique fixed point~$a\in X$ such that~$f(a)=a$.
  Moreover, all iterated function sequences associated with~$f$ are
  convergent with limit~$a$.
\end{fpthm}

The proof of the fixed point theorem goes as follows (this proof uses the
notions of distance, completeness and Lipschitz continuity).
\begin{itemize}
\item uniqueness of the fixed point is obtained from the properties of the
  distance;
\item existence of the fixed point is built from the sequence of iterates of
  the contraction:
  \begin{itemize}
  \item when nonstationary, the sequence is first proved to be a Cauchy
    sequence using the iterated triangle inequality and the formula for the sum
    of the first terms of a geometric series,
  \item the sequence is then convergent in a complete metric space,
  \item the limit of the sequence is finally proved to be a fixed point
    of the contraction from properties of the contraction and of the
    distance.
  \end{itemize}
\end{itemize}

\section{Detailed proof}
\label{s:detailed-proof}

\setlength{\unitlength}{1mm}

\begin{figure}[htb]
  \centering
  \begin{picture}(50,50)(0,5)
    \Thicklines
    \myspacebox{(0,50)}{3cm}{Vector}{vectorcolor}
    \myspacebox{(0,30)}{4cm}{Normed vector}{normedcolor}
    \myspacebox{(0,10)}{4cm}{Inner product}{innercolor}
    \myspacebox{(50,50)}{3cm}{Metric}{metriccolor}
    \myspacebox{(50,30)}{3cm}{Complete}{completecolor}
    \myspacebox{(50,10)}{3cm}{Hilbert}{hilbertcolor}
    \put(0,44.6){\vector(0,-1){9.25}}
    {\thinlines
      \put(0,24.6){\vector(0,-1){9.25}}
    }
    \put(0,16.6){\vector(0,-1){1.25}}
    \put(50,44.6){\vector(0,-1){9.25}}
    \put(50,24.6){\vector(0,-1){9.25}}
    \put(21.75,10){\vector(1,0){11.5}}
  \end{picture}
  \caption{%
    Hierarchy of notions used to build a Hilbert space.
    Thick arrows indicate inheritance: the target notion is built upon the
    source one, whereas the lower left thin arrow indicates that inner
    product spaces are only shown to be normed vector spaces.
  }
  \label{fig:hier:hilbert}
\end{figure}

The {\LM} theorem is stated on a Hilbert space.
The notion of Hilbert space is built from a
series of notions of spaces, see Figure~\ref{fig:hier:hilbert} for a sketch
of the hierarchy.
Thus, a large part of the present section collects standard definitions and
results from linear and bilinear algebra.
One of the main steps is the construction of the complete normed space of
continuous linear maps,
used in particular to obtain the notion of topological dual.
Then, another step is the
construction of bilinear forms, and their representation as linear maps with
values in the topological dual.
The results on finite dimensional vector spaces are avoided as much as
possible, as well as the results on Banach spaces.

A set of basic results from topology in metric spaces are necessary to
formulate the fixed point theorem, and to link completeness and closedness,
which is useful in particular to characterize finite dimensional spaces as
complete.

The last statement of this document is dedicated to the finite dimensional
case.
To prove Theorem~\ref{t:lax-milgram-cea-finite-dimensional-subspace}
({\LMC}), we use the closedness of finite dimensional subspaces, which is a
direct consequence of the closedness of the sum of a closed subspace and a
linear span.
Such a result is of course valid in any {\normedvectorspace}, but the general
proof is based on the equivalence of norms in a finite dimensional space, and
the latter needs more advanced results on continuity involving compactness.
To avoid that, we propose a much simpler proof that is only valid in inner
product spaces; which is fine here since we apply it on a Hilbert space.

\bigskip

Statements are displayed inside colored boxes.
Their nature can be identified at a glance by using the following color code:
\begin{center}
  \rmkbox{light gray is for remarks}, \qquad
  \defbox{light green for definitions},\\
  \lembox{light blue for lemmas}, \quad and \quad
  \thmbox{light red for theorems}.
\end{center}
Moreover, inside the bodies of proof for lemmas and theorems, the most
basic results are supposed to be known and are not detailed further;
they are displayed in \assume{bold red}.
This includes:
\begin{itemize}
  \item properties from propositional calculus;
  \item basic notions and results from set theory such as the complement of a
    subset, the composition of functions, injective and surjective functions;
  \item basic results from group theory;
  \item ordered field properties of~$\matR$,
    ordered set properties of~$\matRbar$;
  \item basic properties of the complete valued fields~$\matR$ and~$\matC$;
  \item definition and properties of basic functions over~$\matR$ such the
    square, square root, and exponential functions, and the discriminant of a
    quadratic polynomial;
  \item basic properties of geometric series (sum of the first terms).
\end{itemize}

\bigskip

This section is organized as follows.
Some facts about infima and suprema are first collected in
Section~\ref{ss:sup-inf}, they are useful to define the operator norm for
continuous linear maps, and orthogonal projections in inner product spaces.
Then, Section~\ref{ss:metric-space} is devoted to complements on metric
spaces, it concludes with the fixed point theorem.
Section~\ref{ss:vector-space} is for the general notion of vector spaces.
{\Normedvectorspace}s are introduced in Section~\ref{ss:normed-vector-space},
with the continuous linear map equivalency theorem.
Then, we add an inner product in Section~\ref{ss:inner-product-space} to
define inner product spaces and state a series of orthogonal projection
theorems.
Finally, Section~\ref{ss:hilbert-space} is dedicated to Hilbert spaces with
the {\RF} representation theorem, and variants of the {\LM} theorem.

\subsection{Supremum, infimum}
\label{ss:sup-inf}

\begin{remark}
  From the completeness of the set of real numbers, every nonempty subset
  of~$\matR$ that is bounded from above has a least upper bound in~$\matR$.
  On the affinely extended real number system~$\matRbar$, every nonempty
  subset has a least upper bound that may be~$+\infty$ (and a greatest lower
  bound that may be~$-\infty$).
  Thus, we have the following extended notions of supremum and infimum for a
  numerical function over a set.
\end{remark}

\begin{definition}[supremum]
  \label{d:supremum}
  Let~$X$ be a set.
  Let $f:X\rightarrow\matR$ be a function.
  The extended number~$L$ is the {\em supremum of~$f$ over~$X$}, and is
  denoted $L=\supXf$, iff
  it is the least upper bound of $f(X)=\{f(x)\st x\in X\}\subset\matR$:
  \begin{eqnarray}
    \label{e:supremum-upper-bound}
    \forall x \in X,
    & & f(x) \leq L; \\
    \label{e:supremum-least}
    \forall M \in \matRbar,
    & & (\forall x \in X,\quad f(x) \leq M)
    \Implies
    L \leq M.
  \end{eqnarray}
\end{definition}

\begin{lemma}[finite supremum]
  \label{l:finite-supremum}
  Let~$X$ be a set.
  Let $f:X\rightarrow\matR$ be a function.
  Assume that there exists a finite upper bound for~$f(X)$, {\ie} there
  exists $M\in\matR$ such that, for all $x\in X$, $f(x)\leq M$.
  Then, the supremum is finite and $L=\supXf$ iff
  \eqref{e:supremum-upper-bound}~and
  \begin{equation}
    \label{e:supremum-accum}
    \forall \eps > 0,\,
    \exists x_\eps \in X,\quad
    L - \eps < f (x_\eps).
  \end{equation}
\end{lemma}

\begin{proof}
  From
  hypothesis, and
  completeness of~$\matR$,
  $\supXf$ is finite.
  Let~$L$ be a number.
  Assume that~$L$ is an upper bound of $f(X)$,
  {\ie}~\eqref{e:supremum-upper-bound}.
  
  \proofparskip{(\ref{e:supremum-least}) implies~(\ref{e:supremum-accum})}
  Assume that~\eqref{e:supremum-least} holds.
  Let $\eps>0$.
  Suppose that for all $x\in X$, $f(x)\leq L-\eps$.
  Then, $L-\eps$ is an upper bound for~$f(X)$.
  Thus, from
  hypothesis,
  we have $L\leq L-\eps$, and from
  \assume{ordered field properties of~$\matR$},
  $\eps\leq 0$, which is impossible.
  Hence, there exists $x\in X$, such that $L-\eps<f(x)$.
  
  \proofparskip{(\ref{e:supremum-accum}) implies~(\ref{e:supremum-least})}
  Conversely, assume now that~\eqref{e:supremum-accum} holds.
  Let~$M$ be an upper bound, {\ie} for all $x\in X$, $f(x)\leq M$.
  Suppose that $M<L$.
  Let $\eps=\frac{L-M}{2}>0$.
  Then, from
  hypotheses, and
  \assume{ordered field properties of~$\matR$},
  there exists $x_\eps\in X$ such that $f(x_\eps)>L-\eps=\frac{L+M}{2}>M$,
  which is impossible.
  Hence, we have $L\leq M$.
  
  Therefore, \eqref{e:supremum-upper-bound}~implies the equivalence
  between~\eqref{e:supremum-least} and~\eqref{e:supremum-accum}.
\end{proof}

\begin{lemma}[discrete lower accumulation]
  \label{l:discrete-lower-accumulation}
  Let~$X$ be a set.
  Let $f:X\rightarrow\matR$ be a function.
  Let~$L$ be a number.
  Then, \eqref{e:supremum-accum} iff
  \begin{equation}
    \label{e:supremum-accum-discrete}
    \forall n \in \matN,\,
    \exists x_n \in X,\quad
    L - \frac{1}{n + 1} < f (x_n).
  \end{equation}
\end{lemma}

\begin{proof}
  \proofpar{(\ref{e:supremum-accum})
    implies~(\ref{e:supremum-accum-discrete})}
  Assume that~\eqref{e:supremum-accum} holds.
  Let $n\in\matN$.
  Let $\eps=\frac{1}{n+1}>0$.
  Then, from
  hypothesis,
  there exists $x_n=x_\eps\in X$ such that
  \begin{equation*}
    L - \frac{1}{n + 1} = L - \eps < f(x_\eps) = f(x_n).
  \end{equation*}
  
  \proofparskip{(\ref{e:supremum-accum-discrete})
    implies~(\ref{e:supremum-accum})}
  Conversely, assume now that~\eqref{e:supremum-accum-discrete} holds.
  Let $\eps>0$.
  From
  the \assume{Archimedean property of~$\matR$},
  there exists $n\in\matN$ such that $n>\frac{1}{\eps}-1$
  (e.g. $n=\floor{\frac{1}{\eps}}$).
  Then, from
  \assume{ordered field properties of~$\matR$}, and
  hypothesis,
  we have $\eps>\frac{1}{n+1}$, and there exists $x_\eps=x_n\in X$ such that
  \begin{equation*}
    L - \eps < L - \frac{1}{n + 1} < f(x_n) = f(x_\eps).
  \end{equation*}
\end{proof}



\begin{lemma}[supremum is positive scalar multiplicative]
  \label{l:supremum-is-positive-scalar-multiplicative}
  Let~$X$ be a set.
  let $f:X\rightarrow\matR$ be a function.
  Let $\lambda\geq 0$ be a nonnegative number.
  Then, $\supXlambdaf=\lambda\,\supXf$
  (\assume{with the convention that 0 times $+\infty$ is~0}).
\end{lemma}

\begin{proof}
  Let $L=\supXf$ and $M=\supXlambdaf$ be extended numbers of~$\matRbar$.
  Let~$x\in X$.
  
  From
  Definition~\threfc{d:supremum}{$L$~is an upper bound for $f(X)$}, and
  \assume{ordered set properties of~$\matRbar$},
  we have
  \begin{equation*}
    (\lambda f) (x) = \lambda f (x) \leq \lambda L.
  \end{equation*}
  Thus, $\lambda L$ is an upper bound of $(\lambda f)(X)$.
  Hence, from
  Definition~\threfc{d:supremum}{%
    $M$~is the least upper bound of $(\lambda f)(X)$},
  we have $M\leq\lambda L$.
  
  \proofparskip{Case $\lambda=0$}
  Then, $\lambda f$ is the zero function (for all $x\in X$,
  $(\lambda f)(x)=\lambda f(x)=0$).
  Thus, $(\lambda f)(X)=\{0\}$, and from
  Definition~\threfc{d:supremum}{0~is the least upper bound of $\{0\}$},
  $M=0$.
  Hence, from
  \assume{ordered set properties of~$\matRbar$},
  we have
  \begin{equation*}
    \lambda L = 0 \, L = 0 \leq 0 = M.
  \end{equation*}
  
  \proofparskip{Case $\lambda\not=0$}
  Then, from
  hypothesis,
  $\lambda>0$.
  From
  Definition~\threfc{d:supremum}{$M$~is an upper bound for $(\lambda f)(X)$},
  \assume{field properties of~$\matR$}, and
  \assume{ordered set properties of~$\matRbar$},
  we have
  \begin{equation*}
    f (x)
    = \frac{1}{\lambda} \lambda f (x)
    = \frac{1}{\lambda} (\lambda f) (x)
    \leq \frac{M}{\lambda}.
  \end{equation*}
  Thus, $\frac{M}{\lambda}$ is an upper bound for~$f(X)$.
  Hence, from
  Definition~\threfc{d:supremum}{$L$~is the least upper bound of $f(X)$}, and
  \assume{ordered set properties of~$\matRbar$},
  we have $L\leq\frac{M}{\lambda}$, or equivalently $\lambda L\leq M$.
  
  Therefore, $M=\lambda L$.
\end{proof}

\begin{remark}
  As a consequence, $\supXlambdaf$ is finite iff $\lambda\supXf$ is finite.
\end{remark}

\begin{definition}[maximum]
  \label{d:maximum}
  Let~$X$ be a set.
  Let $f:X\rightarrow\matR$ be a function.
  The supremum of~$f$ over~$X$ is called {\em maximum of~$f$ over~$X$}, and
  it is denoted $\maxXf$, iff
  there exists~$y\in X$ such that $f(y)=\supXf$.
\end{definition}

\begin{lemma}[finite maximum]
  \label{l:finite-maximum}
  Let~$X$ be a set.
  Let $f:X\rightarrow\matR$ be a function.
  Let $y\in X$.
  Then,
  \begin{equation}
    \label{e:finite-maximum}
    f(y) = \maxXf
    \Equiv
    \forall x \in X,\quad
    f(x) \leq f(y).
  \end{equation}
\end{lemma}

\begin{proof}
  \proofpar{``Left'' implies ``right''}
  Assume that~$y$ realizes the maximum of~$f$ over~$X$.
  Let $x\in X$.
  Then, from hypothesis,
  Definition~\thref{d:maximum}, and
  Definition~\threfc{d:supremum}{$f(y)$~is an upper bound of $f(X)$},
  we have $f(x)\leq f(y)$.
  
  \proofparskip{``Right'' implies ``left''}
  Conversely, assume now that~$f(y)$ is an upper bound of~$f(X)$.
  Let $\eps>0$.
  Let $x_\eps=y\in X$.
  Then, from
  \assume{ordered field properties of~$\matR$},
  we have $f(y)-\eps<f(y)=f(x_\eps)$.
  Hence, from
  Lemma~\thref{l:finite-supremum}, and
  Definition~\thref{d:maximum},
  we have $f(y)=\maxXf$.
\end{proof}

\begin{definition}[infimum]
  \label{d:infimum}
  Let~$X$ be a set.
  Let $f:X\rightarrow\matR$ be a function.
  The extended number~$l$ is the {\em infimum of~$f$ over~$X$}, and is
  denoted $l=\infXf$, iff
  it is the greatest lower bound of $f(X)\subset\matR$:
  \begin{eqnarray}
    \label{e:infimum-lower-bound}
    \forall x \in X,
    & & l \leq f(x); \\
    \label{e:infimum-greatest}
    \forall m \in \matRbar,
    & & (\forall x \in X,\quad m \leq f(x))
    \Implies
    m \leq l.
  \end{eqnarray}
\end{definition}

\begin{lemma}[duality infimum-supremum]
  \label{l:duality-infimum-supremum}
  Let~$X$ be a set.
  Let $f:X\rightarrow\matR$ be a function.
  Then, $\infXf=-\supXmf$.
\end{lemma}

\begin{proof}
  Let $L=\supXmf$ be an extended number in~$\matRbar$.
  Let $l=-L$.
  
  Let $x\in X$.
  From
  Definition~\threfc{d:supremum}{$L$~is an upper bound of $(-f)(X)$}, and
  \assume{ordered set properties of~$\matRbar$},
  we have $l=-L\leq f(x)$.
  Hence, $l$~is a lower bound of~$f(X)$.
  
  Let~$m\in\matRbar$ be a lower bound of~$f(X)$.
  Let $x\in X$.
  Then, from
  \assume{ordered set properties of~$\matRbar$},
  we have $-f(x)\leq-m$.
  Thus, from
  Definition~\threfc{d:supremum}{%
    $L$~is the least upper bound of $(-f)(X)$}, and
  \assume{ordered set properties of~$\matRbar$},
  we have $m\leq-L=l$.
  Hence, $l$~is the greatest lower bound of~$f(X)$.
  
  Therefore, from
  Definition~\thref{d:infimum},
  $l=\infXf$.
\end{proof}

\begin{lemma}[finite infimum]
  \label{l:finite-infimum}
  Let~$X$ be a set.
  Let $f:X\rightarrow\matR$ be a function.
  Assume that there exists a finite lower bound for~$f(X)$, {\ie} there
  exists   $m\in\matR$ such that, for all $x\in X$, $m\leq f(x)$.
  Then, the infimum is finite and $l=\infXf$ iff
  \eqref{e:infimum-lower-bound}~and
  \begin{equation}
    \label{e:infimum-accum}
    \forall \eps > 0,\,
    \exists x_\eps \in X,\quad
    f(x_\eps) < l + \eps.
  \end{equation}
\end{lemma}

\begin{proof}
  Let $x\in X$.
  Then, from
  hypothesis, and
  \assume{ordered field properties of~$\matR$},
  we have $-f(x)\leq -m$.
  Thus, from
  Lemma~\thref{l:finite-supremum},
  $\supXmf$ is finite.
  Hence, from
  Lemma~\thref{l:duality-infimum-supremum},
  $\infXf=-\supXmf$ is finite too.
  
  From
  Lemma~\thref{l:duality-infimum-supremum},
  Lemma~\thref{l:finite-supremum}, and
  \assume{ordered field properties of~$\matR$},
  we have
  \begin{eqnarray*}
    l = \infXf
    & \equiv & -l = \supXmf \\
    & \equiv &
    \left\{
      \begin{array}{l}
        \forall x \in X,\quad
        -f (x) \leq -l \\
        \forall \eps > 0,\,
        \exists x_\eps \in X,\quad
        -l - \eps < -f (x_\eps)
      \end{array}
    \right. \\
    & \equiv &
    \left\{
      \begin{array}{l}
        \forall x \in X,\quad
        l \leq f (x) \\
        \forall \eps > 0,\,
        \exists x_\eps \in X,\quad
        f (x_\eps) < l + \eps.
      \end{array}
    \right. \\
  \end{eqnarray*}
\end{proof}

\begin{lemma}[discrete upper accumulation]
  \label{l:discrete-upper-accumulation}
  Let~$X$ be a set.
  Let $f:X\rightarrow\matR$ be a function.
  Let~$l$ be a number.
  Then,
  \begin{equation}
    \label{e:discrete-upper-accumulation}
    \forall \eps > 0,\,
    \exists x_\eps \in X,\quad
    f (x_\eps) < l + \eps
    \Equiv
    \forall n \in \matN,\,
    \exists x_n \in X,\quad
    f (x_n) < l + \frac{1}{n + 1}.
  \end{equation}
\end{lemma}

\begin{proof}
  Direct consequence of
  Lemma~\threfc{l:discrete-lower-accumulation}{for~$-f$ and $L=-l$}, and
  \assume{ordered field properties of~$\matR$}.
\end{proof}

\begin{lemma}[finite infimum discrete]
  \label{l:finite-infimum-discrete}
  Let~$X$ be a set.
  Let $f:X\rightarrow\matR$ be a function.
  Assume that there exists a finite lower bound for~$f(X)$, {\ie} there
  exists   $m\in\matR$ such that, for all $x\in X$, $m\leq f(x)$.
  Then, the infimum is finite and $l=\infXf$ iff
  \eqref{e:infimum-lower-bound}~and
  \begin{equation}
    \label{e:infimum-accum-discrete}
    \forall n \in \matN,\,
    \exists x_n \in X,\quad f(x_n) < l + \frac{1}{n + 1}.
  \end{equation}
\end{lemma}

\begin{proof}
  Direct consequence of
  Lemma~\thref{l:finite-infimum}, and
  Lemma~\thref{l:discrete-upper-accumulation}.
\end{proof}



\begin{definition}[minimum]
  \label{d:minimum}
  Let~$X$ be a set.
  Let $f:X\rightarrow\matR$ be a function.
  The infimum of~$f$ over~$X$ is called {\em minimum of~$f$ over~$X$} and it
  is denoted $\minXf$, iff
  there exists~$y\in X$ such that $f(y)=\infXf$.
\end{definition}

\begin{lemma}[finite minimum]
  \label{l:finite-minimum}
  Let~$X$ be a set.
  Let $f:X\rightarrow\matR$ be a function.
  Let $y\in X$.
  Then,
  \begin{equation}
    \label{e:finite-minimum}
    f(y) = \minXf
    \Equiv
    \forall x \in X,\quad
    f(y) \leq f(x).
  \end{equation}
\end{lemma}

\begin{proof}
  From
  Definition~\thref{d:minimum},
  Lemma~\thref{l:duality-infimum-supremum},
  Lemma~\thref{l:finite-maximum}, and
  \assume{ordered field properties of~$\matR$},
  we have
  \begin{eqnarray*}
    f(y) = \minXf
    & \equiv & -f(y) = \maxXmf \\
    & \equiv &
    \forall x \in X,\quad
    -f(x) \leq -f(y) \\
    & \equiv &
    \forall x \in X,\quad
    f(y) \leq f(x).
  \end{eqnarray*}
\end{proof}

\subsection{Metric space}
\label{ss:metric-space}

\begin{definition}[distance]
  \label{d:distance}
  Let~$X$ be a nonempty set.
  An application $d:\XxX\rightarrow\matR$ is a {\em distance over~$X$} iff
  it is nonnegative,
  symmetric,
  it separates points,
  and it satisfies the triangle inequality:
  \begin{eqnarray}
    \label{e:distance-nonnegative}
    \forall x, y \in X, & & d(x,y) \geq 0; \\
    \label{e:distance-symmetric}
    \forall x, y \in X, & & d(y,x) = d(x,y); \\
    \label{e:distance-definite}
    \forall x, y \in X, & & d(x,y) = 0 \Equiv x = y; \\
    \label{e:distance-triangle}
    \forall x, y, z \in X, & & d(x, z) \leq d(x,y) + d(y,z).
  \end{eqnarray}
\end{definition}

\begin{definition}[metric space]
  \label{d:metric-space}
  $(X,d)$ is a {\em metric space} iff
  $X$~is a nonempty set and $d$~is a distance over~$X$.
\end{definition}

\begin{lemma}[iterated triangle inequality]
  \label{l:iterated-triangle-inequality}
   Let $(X,d)$ be a metric space.
   Let $(x_n)_{n\in\matN}$ be a sequence of points of~$X$.
   Then,
   \begin{equation}
     \label{e:iterated-triangle-inequality}
     \forall n, p \in \matN,\quad
     d(x_n, x_{n + p}) \leq \sum_{i=0}^{p - 1} d(x_{n + i}, x_{n + i + 1}).
   \end{equation}
\end{lemma}

\begin{proof}
  Let $n\in\matN$ be a natural number.
  For $p\in\matN$, let $P(p)$ be the property
  \begin{equation*}
    d(x_n, x_{n + p}) \leq \sum_{i=0}^{p - 1} d(x_{n + i}, x_{n + i + 1}).
  \end{equation*}
  
  \proofparskip{Induction: $P(0)$}
  From
  Definition~\threfc{d:distance}{$d$~separates points}, and
  \assume{ordered field properties of~$\matR$},
  $P(0)$ is obviously satisfied.
  
  \proofparskip{Induction: $P(p)$ implies $P(p+1)$}
  Let $p\in\matN$.
  Assume that $P(p)$ holds.
  Then, from
  Definition~\threfc{d:distance}{$d$~satisfies triangle inequality},
  we have $d(x_n,x_{n+p+1})\leq d(x_n,x_{n+p})+d(x_{n+p},x_{n+p+1})$.
  Hence, from
  hypothesis, and
  \assume{ordered field properties of~$\matR$},
  we have $P(p+1)$.
  
  Therefore, by induction on $p\in\matN$, we have, for all $p\in\matN$,
  $P(p)$.
\end{proof}

\subsubsection{Topology of balls}



\begin{definition}[closed ball]
  \label{d:closed-ball}
  Let $(X,d)$ be a metric space.
  Let $x\in X$ be a point.
  Let $r\geq 0$ be a nonnegative number.
  The {\em closed ball centered in~$x$ of radius~$r$}, denoted
  $\cball{d}{x}{r}$, is the subset of~$X$ defined by
  \begin{equation}
    \label{e:closed-ball}
    \cball{d}{x}{r} = \left\{ y \in X \st d(x, y) \leq r \right\}.
  \end{equation}
\end{definition}


\begin{definition}[sphere]
  \label{d:sphere}
  Let $(X,d)$ be a metric space.
  Let $x\in X$ be a point.
  Let $r\geq 0$ be a nonnegative number.
  The {\em sphere centered in~$x$ of radius~$r$}, denoted $\sphere{d}{x}{r}$,
  is the subset of~$X$ defined by
  \begin{equation}
    \label{e:sphere}
    \sphere{d}{x}{r} = \left\{ y \in X \st d(x, y) = r \right\}.
  \end{equation}
\end{definition}

\begin{definition}[open subset]
  \label{d:open-subset}
  Let $(X,d)$ be a metric space.
  A subset~$Y$ of~$X$ is {\em open (for distance~$d$)} iff
  \begin{equation}
    \label{e:open-subset}
    \forall x \in Y,\,
    \exists r > 0,\quad
    \cball{d}{x}{r} \subset Y.
  \end{equation}
\end{definition}


\begin{definition}[closed subset]
  \label{d:closed-subset}
  Let $(X,d)$ be a metric space.
  A subset~$Y$ of~$X$ is {\em closed (for distance~$d$)} iff
  $X\backslash Y$ is open for distance~$d$.
\end{definition}

\begin{lemma}[equivalent definition of closed subset]
  \label{l:equivalent-definition-of-closed-subset}
  Let $(X,d)$ be a metric space.
  A subset~$Y$ of~$X$ is {\em closed (for distance~$d$)} iff
  \begin{equation}
    \label{e:equivalent-definition-of-closed-subset}
     \forall x \in X \backslash Y,\,
    \exists r > 0,\quad
    \cball{d}{x}{r} \cap Y = \emptyset.
  \end{equation}
\end{lemma}

\begin{proof}
  Direct consequence of
  Definition~\thref{d:closed-subset},
  Definition~\thref{d:open-subset}, and
  \assume{the definition of the complement from set theory}.
\end{proof}

\begin{lemma}[singleton is closed]
  \label{l:singleton-is-closed}
  Let $(X,d)$ be a metric space.
  Let $x\in X$ be a point.
  Then $\{x\}$ is closed.
\end{lemma}

\begin{proof}
  Let $\xp\in X$ be a point.
  Assume that $\xp\not=x$.
  Then, from
  Definition~\threfc{d:distance}{$d$~separates points, contrapositive}, and
  \assume{ordered field properties of~$\matR$},
  $\eps=\half\,d(\xp,x)$ is positive.
  Hence, $d(\xp,x)>\eps$ and $\cball{d}{\xp}{\eps}\cap\{x\}=\emptyset$.
  
  Therefore, from
  Lemma~\thref{l:equivalent-definition-of-closed-subset},
  $\{x\}$ is closed.
\end{proof}

\begin{definition}[closure]
  \label{d:closure}
  Let $(X,d)$ be a metric space.
  Let~$Y$ be a subset of~$X$.
  The {\em closure of $Y$}, denoted $\Closure{Y}$, is the subset
  \begin{equation}
    \label{e:closure}
    \Closure{Y} =
    \left\{
      x \in X \st
      \forall \eps > 0,\,
      \cball{d}{x}{\eps} \cap Y \not= \emptyset
    \right\}.
  \end{equation}
\end{definition}



\begin{definition}[convergent sequence]
  \label{d:convergent-sequence}
  Let $(X,d)$ be a metric space.
  Let $l\in X$.
  A sequence $(x_n)_{n\in \matN}$ of~$X$ is {\em convergent with limit~$l$}
  iff
  \begin{equation}
    \label{e:convergent-sequence}
    \forall \eps > 0,\,
    \exists N \in \matN,\,
    \forall n \in \matN,\quad
    n\geq N \Implies d(x_n, l) \leq \eps.
  \end{equation}
\end{definition}

\begin{lemma}[variant of point separation]
  \label{l:variant-of-point-separation}
  Let $(X,d)$ be a metric space.
  Let $x,\xp\in X$ such that for all $\eps>0$, we have $d(x,\xp)\leq\eps$.
  Then, $x=\xp$.
\end{lemma}

\begin{proof}
  Assume that $d(x,\xp)>0$.
  Let $\eps=\frac{d(x,\xp)}{2}$.
  Then, $0<d(x,\xp)\leq\eps=\frac{d(x,\xp)}{2}$.
  Hence, from
  \assume{ordered field properties of~$\matR$ (with $d(x,\xp)>0$)},
  we have $0<1\leq\half$, which is wrong.
  Thus, from
  Definition~\threfc{d:distance}{$d$~is nonnegative},
  we have $d(x,\xp)=0$.
  
  Therefore, from
  Definition~\threfc{d:distance}{$d$~separates points},
  we have $x=\xp$.
\end{proof}

\begin{lemma}[limit is unique]
  \label{l:limit-is-unique}
  Let $(X,d)$ be a metric space.
  Let $(x_n)_{n\in \matN}$ be a convergent sequence of~$X$.
  Then, the limit of the sequence is unique.
  The limit is denoted $\lim_{n\rightarrow+\infty}x_n$.
\end{lemma}

\begin{proof}
  Let $l,\lp\in X$ be two limits of the sequence.
  Let $\eps>0$.
  Then, from
  \assume{ordered field properties of~$\matR$}, and
  Definition~\threfc{d:convergent-sequence}{with $\frac{\eps}{2}>0$},
  let $N,\Np\in\matN$ such that, for all $n,\np\in\matN$, $n\geq N$
  and $\np\geq\Np$ implies $d(x_n,l)\leq\frac{\eps}{2}$ and
  $d(x_{\np},\lp)\leq\frac{\eps}{2}$.
  Let $M=\max(N,\Np)$.
  Let $p\in\matN$.
  Assume that $p\geq M$.
  Then, from
  \assume{the definition of the maximum},
  we have $d(x_p,l)\leq\frac{\eps}{2}$ and $d(x_p,\lp)\leq\frac{\eps}{2}$.
  Hence, from
  Definition~\threfc{d:distance}{%
    $d$~is nonnegative, satisfies triangle inequality, and is symmetric}, and
  \assume{ordered field properties of~$\matR$},
  we have
  \begin{equation*}
    0 \leq d(l, \lp)
    \leq d(l, x_p) + d(x_p, \lp)
    = d(x_p, l) + d(x_p, \lp)
    \leq \frac{\eps}{2} + \frac{\eps}{2}
    = \eps.
  \end{equation*}
  
  Therefore, from
  Lemma~\thref{l:variant-of-point-separation},
  we have $l=\lp$.
\end{proof}

\begin{lemma}[closure is limit of sequences]
  \label{l:closure-is-limit-of-sequences}
  Let $(X,d)$ be a metric space.
  Let~$Y$ be a nonempty subset of~$X$.
  Let $a\in X$ be a point.
  Then,
  \begin{equation}
    \label{e:closure-is-limit-of-sequences}
    a \in \Closure{Y}
    \Equiv
    \exists (a_n)_{n \in \matN} \in Y^\matN,\quad
    a = \lim_{n \rightarrow \infty} a_n.
  \end{equation}
\end{lemma}

\begin{proof}
  \proofpar{``Left'' implies ``right''}
  Assume that~$a\in\Closure{Y}$.
  Let~$a_0$ be a point of~$Y$.
  Let $n\in\matN$.
  Assume that $n\geq 1$.
  Then, $\frac{1}{n}>0$, and from
  Definition~\thref{d:closure},
  let~$a_n$ be in the nonempty intersection
  $\cball{d}{a}{\frac{1}{n}}\cap Y$.
  From
  Definition~\thref{d:closed-ball}, and
  Definition~\threfc{d:distance}{$d$~is symmetric},
  we have $a_n\in Y$ and $d(a_n,a)\leq\frac{1}{n}$.
  Let $\eps>0$.
  Let $N=\ceil{\frac{1}{\eps}}$.
  Let $n\in\matN$.
  Assume that $n\geq N$.
  Then, from
  \assume{ordered field properties of~$\matR$}, and
  \assume{the definition of ceiling function},
  we have
  \begin{equation*}
    d(a_n, a) \leq \frac{1}{n}
    \leq \frac{1}{N}
    \leq \eps.
  \end{equation*}
  Hence, from
  Definition~\thref{d:convergent-sequence},
  the sequence $(a_n)_{n\in\matN}$ is convergent with limit~$a$.
  
  \proofparskip{``Right'' implies ``left''}
  Assume now that there exists a convergent sequence $(a_n)_{n\in\matN}$
  in~$Y$ with limit~$a$.
  Let $\eps>0$.
  Then, from
  Definition~\thref{d:convergent-sequence},
  let $N\in\matN$ such that for all $n\in\matN$, $n\geq N$ implies
  $d(a_n,a)\leq\eps$.
  Thus, $a_N$ belongs to the ball $\cball{d}{a}{\eps}$.
  Hence, from
  Definition~\thref{d:closure},
  $a$~belongs to the closure~$\Closure{Y}$.
\end{proof}

\begin{lemma}[closed equals closure]
  \label{l:closed-equals-closure}
  Let $(X,d)$ be a metric space.
  Let~$Y$ be a nonempty subset of~$X$.
  Then,
  \begin{equation}
    \label{e:closed-equals-closure}
    Y \mbox{ is closed}
    \Equiv
    Y = \Closure{Y}.
  \end{equation}
\end{lemma}

\begin{proof}
  \proofpar{``Left'' implies ``right''}
  Assume that $Y$ is closed.
  Then, from
  Definition~\thref{d:closed-subset},
  $X\backslash Y$ is open.
  Let $a\in\Closure{Y}$.
  Then, from
  Definition~\thref{d:closure},
  for all $\eps>0$, we have $\cball{d}{a}{\eps}\cap Y\not=\emptyset$.
  Assume that $a\not\in Y$.
  Then, from
  \assume{the definition of the complement from set theory}, and
  Lemma~\thref{l:equivalent-definition-of-closed-subset},
  there exists $\eps>0$ such that $\cball{d}{a}{\eps}\cap Y=\emptyset$.
  Which is impossible.
  Thus, $a$~belongs to~$Y$.
  Hence, $\Closure{Y}\subset Y$.
  Moreover, from
  Definition~\thref{d:closure},
  $Y$ is obviously a subset of~$\Closure{Y}$.
  Therefore, $Y=\Closure{Y}$.
  
  \proofparskip{``Right'' implies ``left''}
  Assume now that $Y=\Closure{Y}$.
  Let $x\in X\backslash Y$.
  From
  \assume{the definition of the complement from set theory}, and
  hypothesis,
  $x$ does not belong to $Y=\Closure{Y}$.
  Thus, from
  Definition~\thref{d:closure},
  there exists $\eps>0$ such that $\cball{d}{x}{\eps}\cap Y=\emptyset$.
  Hence, from
  Lemma~\thref{l:equivalent-definition-of-closed-subset},
  $Y$ is closed.
\end{proof}

\begin{lemma}[closed is limit of sequences]
  \label{l:closed-is-limit-of-sequences}
  Let $(X,d)$ be a metric space.
  Let~$Y$ be a nonempty subset of~$X$.
  Then,
  \begin{equation}
    \label{e:closed-is-limit-of-sequences}
    Y \mbox{ is closed}
    \Equiv
    (\forall (a_n)_{n \in \matN} \in Y^\matN,\,
    \forall a \in X,\quad
    a = \lim_{n \rightarrow \infty} a_n
    \Implies
    a \in Y).
  \end{equation}
\end{lemma}

\begin{proof}
  Direct consequence of
  Lemma~\thref{l:closed-equals-closure},
  Definition~\thref{d:closure}, and
  Lemma~\thref{l:closure-is-limit-of-sequences}.
\end{proof}



\begin{definition}[stationary sequence]
  \label{d:stationary-sequence}
  Let $(X,d)$ be a metric space.
  A sequence $(x_n)_{n\in \matN}$ of~$X$ is {\em stationary} iff
  \begin{equation}
    \label{e:stationary-sequence}
    \exists N \in \matN,\,
    \forall n \in \matN,\quad
    n\geq N \Implies x_n = x_N.
  \end{equation}
  $N$ is {\em a rank} from which the sequence is stationary and $x_N$ is
  {\em the stationary value}.
\end{definition}

\begin{lemma}[stationary sequence is convergent]
  \label{l:stationary-sequence-is-convergent}
  Let $(X,d)$ be a metric space.
  Let $(x_n)_{n\in \matN}$ be a stationary sequence of~$X$.
  Then, $(x_n)_{n\in \matN}$ is convergent with the stationary value as
  limit.
\end{lemma}

\begin{proof}
  From
  Definition~\thref{d:stationary-sequence},
  let $N\in\matN$ such that for all $n\in\matN$, $n\geq N$ implies
  $x_n=x_N$.
  Let $\eps>0$.
  Let $n\in\matN$.
  Assume that $n\geq N$.
  Then, from
  Definition~\threfc{d:distance}{$d$~separates points},
  we have $d(x_n,x_N)=d(x_N,x_N)=0\leq\eps$.
  Hence, from
  Definition~\thref{d:convergent-sequence},
  $(x_n)_{n\in \matN}$ is convergent with limit~$x_N$.
\end{proof}

\subsubsection{Completeness}



%

\bigskip

\begin{definition}[Cauchy sequence]
  \label{d:cauchy-sequence}
  Let $(X,d)$ be a metric space.
  Let $(x_n)_{n\in \matN}$ be a sequence of~$X$.
  $(x_n)_{n\in \matN}$ is a {\em Cauchy sequence} iff
  \begin{equation}
    \label{e:cauchy-sequence}
    \forall \eps > 0,\,
    \exists N \in \matN,\,
    \forall p, q \in \matN,\quad
    p \geq N \Conj q \geq N \Implies d(x_p, x_q) \leq \eps.
  \end{equation}
\end{definition}

\begin{lemma}[equivalent definition of Cauchy sequence]
  \label{l:equivalent-definition-of-cauchy-sequence}
  Let $(X,d)$ be a metric space.
  Let $(x_n)_{n\in \matN}$ be a sequence of~$X$.
  $(x_n)_{n\in \matN}$ is a Cauchy sequence iff
  \begin{equation}
  \label{e:cauchy-sequence2}
  \forall \eps >0,\,
  \exists N \in \matN,\,
  \forall p, k \in \matN,\quad
  p \geq N \Implies d(x_p, x_{p+k}) \leq \eps.
  \end{equation}
\end{lemma}

\begin{proof}
  \proofpar{(\ref{e:cauchy-sequence}) implies~(\ref{e:cauchy-sequence2})}
  Assume that $(x_n)_{n\in \matN}$ is a Cauchy sequence.
  Let $\eps>0$.
  From
  Definition~\thref{d:cauchy-sequence},
  let $N\in\matN$ such that for all $p,q\in\matN$, $p\geq N$ and $q\geq N$
  implies $d(x_p,x_q)\leq\eps$.
  Let $p,k\in\matN$.
  Assume that $p\geq N$.
  Then, we also have $q=p+k\geq N$.
  Thus, $d(x_p,x_{p+k})=d(x_p,x_q)\leq\eps$.
  
  \proofparskip{(\ref{e:cauchy-sequence2}) implies~(\ref{e:cauchy-sequence})}
  Conversely, assume now that $(x_n)_{n\in \matN}$
  satisfies~\eqref{e:cauchy-sequence2}.
  Let $\eps>0$.
  Then, let $N\in\matN$ such that for all $p,k\in\matN$, $p\geq N$ implies
  $d(x_p,x_{p+k})\leq\eps$.
  Let~$\pp,\qp\in\matN$.
  Assume that $\pp\geq N$ and $\qp\geq N$.
  Then, we also have $p=\min(\pp,\qp)\geq N$.
  Let $k=\max(\pp,\qp)-p\geq 0$.
  Then, we have $d(x_p,x_{p+k})\leq\eps$.
  Assume that $\pp\leq\qp$.
  Then, $p=\pp$ and $p+k=\qp$.
  Hence, we have $d(x_p,x_q)=d(x_p,x_{p+k})\leq\eps$.
  Conversely, assume that $\pp>\qp$.
  Then, $p=\qp$ and $p+k=\pp$.
  Hence, from
  Definition~\threfc{d:distance}{$d$~is symmetric},
  we have $d(x_p,x_q)=d(x_{p+k},x_p)=d(x_p,x_{p+k})\leq\eps$.
\end{proof}

\begin{lemma}[convergent sequence is Cauchy]
  \label{l:convergent-sequence-is-cauchy}
  Let $(X,d)$ be a metric space.
  Let $(x_n)_{n\in \matN}$ be a sequence of~$X$.
  Assume that $(x_n)_{n\in \matN}$ is a convergent sequence.
  Then, $(x_n)_{n\in \matN}$ is a Cauchy sequence.
\end{lemma}

\begin{proof}
  Let $\eps>0$.
  From
  Lemma~\thref{l:limit-is-unique},
  let~$l=\lim_{n\rightarrow+\infty}\in X$.
  From
  Definition~\thref{d:convergent-sequence},
  let $N\in\matN$ such that for all $n\in\matN$, $n\geq N$ implies
  $d(x_n,l)\leq\frac{\eps}{2}$.
  Let $p,q\geq N$.
  Then, from
  Definition~\threfc{d:distance}{%
    $d$~satisfies triangle inequality and is symmetric}, and
  \assume{field properties of~$\matR$}
  we have
  \begin{equation*}
    d (x_p, x_q)
    \leq d (x_p, l) + d (l, x_q)
    = d(x_p, l) + d (x_q, l)
    \leq \frac{\eps}{2} + \frac{\eps}{2}
    = \eps.
  \end{equation*}
  
  Therefore, from
  Definition~\thref{d:cauchy-sequence},
  $(x_n)_{n\in \matN}$ is a Cauchy sequence.
\end{proof}

\begin{definition}[complete subset]
  \label{d:complete-subset}
  Let $(X,d)$ be a metric space.
  A subset~$Y$ of~$X$ is {\em complete (for distance~$d$)} iff
  all Cauchy sequences of~$Y$ converge in~$Y$.
\end{definition}

\begin{definition}[complete metric space]
  \label{d:complete-metric-space}
  Let~$X$ be a set.
  Let~$d$ be a distance over~$X$.
  $(X,d)$ is a {\em complete metric space} iff
  $(X,d)$ is a metric space and~$X$ is complete for distance~$d$.
\end{definition}

\begin{lemma}[closed subset of complete is complete]
  \label{l:closed-subset-of-complete-is-complete}
  Let $(X,d)$ be a complete metric space.
  Let~$Y$ be a closed subset of~$X$.
  Then, $Y$~is complete.
\end{lemma}

\begin{proof}
  Let $(y_n)_{n\in\matN}$ be a sequence in~$Y$.
  Assume that $(y_n)_{n\in\matN}$ is a Cauchy sequence.
  Since~$Y$ is a subset of~$X$, $(y_n)_{n\in\matN}$ is also a Cauchy sequence
  in~$X$.
  Then, from
  Definition~\threfc{d:complete-metric-space}{$X$~is complete}, and
  Definition~\thref{d:complete-subset},
  the sequence $(y_n)_{n\in\matN}$ is convergent with limit $y\in X$.
  
  Moreover, from
  Lemma~\thref{l:closure-is-limit-of-sequences},
  the limit~$a$ belongs to the closure~$\Closure{Y}$.
  Hence, from
  Lemma~\threfc{l:closed-equals-closure}{$Y$~is closed},
  the limit~$y$ belongs to~$Y$.
  
  Therefore, from
  Definition~\thref{d:complete-subset},
  $Y$~is complete.
\end{proof}

\subsubsection{Continuity}

\begin{remark}
  The distance allows the definition of balls centered at each point of a
  metric space forming neighborhoods for these points.
  Hence, a metric space can be seen as a topological space.
\end{remark}

\begin{remark}
  The distance also allows the definition of entourages making metric spaces
  specific cases of uniform spaces.
  Let~$(X,d)$ be a metric space.
  Then, the sets
  \[
  U_r = \{ (x, \xp) \in \XxX \st d(x, \xp) \leq r \}
  \]
  for all nonnegative numbers~$r$ form a fundamental system of entourages for
  the standard uniform structure of~$X$.
  See Theorem~\ref{t:equivalent-definition-of-lipschitz-continuity} below.
\end{remark}

\begin{definition}[continuity in a point]
  \label{d:continuity-in-a-point}
  Let $(X,d_X)$ and $(Y,d_Y)$ be metric spaces.
  Let $x\in X$.
  Let $f:X\rightarrow Y$ be a mapping.
  $f$~is {\em continuous in~$x$} iff
  \begin{equation}
    \label{e:continuity-in-a-point}
    \forall \eps > 0,\,
    \exists \delta > 0,\,
    \forall  \xp \in X,\quad
    d_X(x, \xp) \leq \delta
    \Implies
    d_Y(f(x), f(\xp)) \leq \eps.
  \end{equation}
\end{definition}

\begin{definition}[pointwise continuity]
  \label{d:pointwise-continuity}
  Let $(X,d_X)$ and $(Y,d_Y)$ be metric spaces.
  Let $f:X\rightarrow Y$ be a mapping.
  $f$~is {\em (pointwise) continuous} iff
  $f$~is continuous in all points of~$X$.
\end{definition}

\begin{lemma}[compatibility of limit with continuous functions]
  \label{l:compatibility-of-limit-with-continuous-functions}
  Let $(X,d_X)$ and $(Y,d_Y)$ be metric spaces.
  Let $f:X\rightarrow Y$ be a mapping.
  Assume that $f$~is pointwise continuous.
  Then,
  for all sequence $(x_n)_{n\in\matN}$ of~$X$,
  for all $x\in X$, we have
  \begin{equation}
    \label{e:compatibility-of-limit-with-continuous-functions}
    (x_n)_{n \in \matN} \mbox{ is convergent with limit } x
    \Implies
    (f (x_n))_{n \in \matN} \mbox{ is convergent with limit } f (x).
  \end{equation}
\end{lemma}

\begin{proof}
  Let $(x_n)_{n\in\matN}$ be a sequence in~$X$.
  Let $x\in X$.
  Assume that $(x_n)_{n\in\matN}$ is convergent with limit~$x$.
  Let $\eps>0$.
  Then, from
  Definition~\threfc{d:continuity-in-a-point}{at point~$x$},
  there exists $\alpha>0$ such that,
  \begin{equation*}
    \forall n \in \matN,\quad
    d_X(x_n, x) \leq \alpha
    \Implies
    d_Y(f (x_n), f (x)) \leq \eps.
  \end{equation*}
  And from
  Definition~\threfc{d:convergent-sequence}{with $\alpha>0$},
  there exists $N\in\matN$ such that,
  \begin{equation*}
    \forall n \in \matN,\quad
    n \geq N
    \Implies
    d_X(x_n, x) \leq \alpha.
  \end{equation*}
  Thus,
  \begin{equation*}
    \forall n \in \matN,\quad
    n \geq N
    \Implies
    d_Y(f (x_n), f (x)) \leq \eps.
  \end{equation*}
  Hence, from
  Definition~\thref{d:convergent-sequence},
  the sequence $(f(x_n))_{n\in\matN}$ is convergent with limit~$f(x)$.
\end{proof}

\begin{definition}[uniform continuity]
  \label{d:uniform-continuity}
  Let $(X,d_X)$ and $(Y,d_Y)$ be metric spaces.
  Let $f:X\rightarrow Y$ be a mapping.
  $f$~is {\em uniformly continuous} iff
  \begin{equation}
    \label{e:uniform-continuity}
    \forall \eps > 0,\,
    \exists \delta > 0,\,
    \forall  x, \xp \in X,\quad
    d_X(x, \xp) \leq \delta
    \Implies
    d_Y(f(x), f(\xp)) \leq \eps.
  \end{equation}
\end{definition}

\begin{definition}[Lipschitz continuity]
  \label{d:lipschitz-continuity}
  Let $(X,d_X)$ and $(Y,d_Y)$ be metric spaces.
  Let $f:X\rightarrow Y$ be a mapping.
  Let $k\geq 0$ be a nonnegative number.
  $f$~is {\em $k$-Lipschitz continuous} iff
  \begin{equation}
    \label{e:lipschitz-continuity}
    \forall x, \xp \in X,\quad
    d_Y (f(x), f(\xp)) \leq k \, d_X (x, \xp).
  \end{equation}
  Then, $k$~is called {\em Lipschitz constant of~$f$}.
\end{definition}

\begin{theorem}[equivalent definition of Lipschitz continuity]
  \label{t:equivalent-definition-of-lipschitz-continuity}
  Let $(X,d_X)$ and $(Y,d_Y)$ be metric spaces.
  Let $f:X\rightarrow Y$ be a mapping.
  Let $k\geq 0$ be a nonnegative number.
  $f$~is {\em $k$-Lipschitz continuous} iff
  \begin{equation}
    \label{e:lipschitz-continuity2}
    \forall x, \xp \in X,\,
    \forall r \geq 0,\quad
    d_X (x, \xp) \leq r
    \Implies
    d_Y (f(x), f(\xp)) \leq kr.
  \end{equation}
\end{theorem}

\begin{proof}
  \proofpar{``Left'' implies ``right''}
  Assume that~$f$ is $k$-Lipschitz continuous.
  Let $x,\xp\in X$.
  Let $r\geq 0$.
  Assume that $d_X(x,\xp)\leq r$.
  Then, from
  Definition~\thref{d:lipschitz-continuity},
  we have
  \begin{equation*}
    d_Y (f (x), f (\xp)) \leq k \, d_X (x, \xp) \leq kr.
  \end{equation*}
  
  \proofparskip{``Right'' implies ``left''}
  Conversely, assume now that~$f$ satisfies~\eqref{e:lipschitz-continuity2}.
  Let $x,\xp\in X$.
  Let $r=d_X(x,\xp)$.
  From
  Definition~\threfc{d:distance}{$d_X$~is nonnegative},
  $r$~is also nonnegative.
  From
  \assume{ordered field properties of~$\matR$},
  we have $d_X(x,\xp)\leq r$.
  Hence, from
  hypothesis,
  we have
  \begin{equation*}
    d_Y (f (x), f (\xp)) \leq kr = k \, d_X (x, \xp).
  \end{equation*}
  Therefore, from
  Definition~\thref{d:lipschitz-continuity},
  $f$~is $k$-Lipschitz continuous.
\end{proof}

\begin{definition}[contraction]
  \label{d:contraction}
  Let $(X,d)$ be a metric space.
  Let $f:X\rightarrow Y$ be a mapping.
  Let $k\geq 0$ be a nonnegative number.
  $f$~is a {\em $k$-contraction} iff
  $f$~is $k$-Lipschitz continuous with $k<1$.
\end{definition}

\begin{lemma}[uniform continuous is continuous]
  \label{l:uniform-continuous-is-continuous}
  Let $(X,d_X)$ and $(Y,d_Y)$ be metric spaces.
  Let $f:X\rightarrow Y$ be an uniformly continuous mapping.
  Then, $f$~is continuous.
\end{lemma}

\begin{proof}
  Direct consequence of
  Definition~\thref{d:uniform-continuity},
  Definition~\thref{d:pointwise-continuity}, and
  Definition~\thref{d:continuity-in-a-point}.
\end{proof}

\begin{lemma}[zero-Lipschitz continuous is constant]
  \label{l:zero-lipschitz-continuous-is-constant}
  Let $(X,d_X)$ and $(Y,d_Y)$ be metric spaces.
  Let $f:X\rightarrow Y$ be a 0-Lipschitz continuous mapping.
  Then, $f$ is constant.
\end{lemma}

\begin{proof}
  Let $x,\xp\in X$.
  Then, from
  Definition~\thref{d:lipschitz-continuity}, and
  Definition~\threfc{d:distance}{$d_Y$~is nonnegative and separates points},
  we have $d_Y(f(x),f(\xp))=0$ and $f(x)=f(\xp)$.
\end{proof}

\begin{lemma}[Lipschitz continuous is uniform continuous]
  \label{l:lipschitz-continuous-is-uniform-continuous}
  Let $(X,d_X)$ and $(Y,d_Y)$ be metric spaces.
  Let $f:X\rightarrow Y$ be a Lipschitz continuous mapping.
  Then, $f$~is uniformly continuous.
\end{lemma}

\begin{proof}
  From
  Definition~\thref{d:lipschitz-continuity},
  let $k\geq 0$ be the Lipschitz constant of~$f$.
  Let~$\eps>0$ be a positive number.
  
  \proofparskip{Case $k=0$}
  Then, from
  Lemma~\thref{l:zero-lipschitz-continuous-is-constant},
  $f$~is a constant function.
  Let $\delta=1>0$.
  Let $x,\xp\in X$.
  Assume that $d_X(x,\xp)<\delta$.
  Then, we have $d_Y(f(x),f(\xp))=0<\eps$.
  Hence, from
  Definition~\thref{d:uniform-continuity},
  $f$~is uniformly continuous.
  
  \proofparskip{Case $k\not=0$}
  Then, $k>0$.
  From
  \assume{ordered field properties of~$\matR$},
  let $\delta=\frac{\eps}{k}>0$.
  Let $x,\xp\in X$.
  Assume that $d_X(x,\xp)\leq\delta$.
  Then, from
  \assume{ordered field properties of~$\matR$},
  we have
  \begin{equation*}
    d_Y(f(x), f(\xp))
    \leq k \, d_X(x, \xp)
    \leq k \delta
    =\eps.
  \end{equation*}
  
  Hence, from
  Definition~\thref{d:uniform-continuity},
  $f$~is uniformly continuous.
\end{proof}

\subsubsection{Fixed point theorem}

\begin{definition}[iterated function sequence]
  \label{d:iterated-function-sequence}
  Let $(X,d)$ be a metric space.
  Let $f:X\rightarrow X$ be a mapping.
  An {\em iterated function sequence associated with~$f$} is a sequence
  of~$X$ defined by
  \begin{equation}
    \label{e:iterated-function-sequence}
    x_0 \in X
    \Conj
    \forall n \in \matN,\quad
    x_{n+1} = f(x_n).
  \end{equation}
\end{definition}

\begin{lemma}[stationary iterated function sequence]
  \label{l:stationary-iterated-function-sequence}
  Let $(X,d)$ be a metric space.
  Let $f:X\rightarrow X$ be a mapping.
  Let $(x_n)_{n\in\matN}$ be an iterated function sequence associated
  with~$f$ such that
  \begin{equation}
    \label{e:stationary-iterated-function-sequence}
    \exists N \in \matN,\quad
    x_{N + 1} = f (x_N) = x_N.
  \end{equation}
  Then, the sequence $(x_n)_{n\in\matN}$ is stationary.
\end{lemma}

\begin{proof}
  For $i\in\matN$, let $P(i)$ be the property $x_{N+i+1}=x_N$.
  
  \proofparskip{Induction: $P(0)$}
  Property $P(0)$ holds by hypothesis.
  
  \proofparskip{Induction: $P(i)$ implies $P(i+1)$}
  Let $i\in\matN$.
  Assume that $P(i)$ holds.
  Then, from
  Definition~\thref{d:iterated-function-sequence}, and
  hypothesis,
  we have
  \begin{equation*}
    x_{N + i + 2}
    = f (x_{N + i + 1})
    = f (x_N)
    = x_N.
  \end{equation*}
  Hence, $P(i+1)$ holds.
  
  Therefore, by induction on $i\in\matN$, we have, for all $i\in\matN$,
  $P(i)$, and from
  Definition~\thref{d:stationary-sequence},
  the sequence $(x_n)_{n\in\matN}$ is stationary.
\end{proof}

\begin{lemma}[iterate Lipschitz continuous mapping]
  \label{l:iterate-lipschitz-continuous-mapping}
  Let $(X,d)$ be a metric space.
  Let $k\geq 0$.
  Let $f:X\rightarrow X$ be a $k$-Lipschitz continuous mapping.
  Let $(x_n)_{n\in\matN}$ be an iterated function sequence associated
  with~$f$.
  Then,
  \begin{equation}
    \label{e:contraction}
    \forall n \in \matN,\quad
    d(x_n, x_{n+1}) \leq k^n \, d(x_0, x_1).
  \end{equation}
\end{lemma}

\begin{proof}
  For $n\in\matN$, let $P(n)$ be the property
  $d(x_n,x_{n+1})\leq k^n\,d(x_0,x_1)$.
  
  \proofparskip{Induction: $P(0)$}
  Property $P(0)$ is a direct consequence of convention $0^0=1$ and
  \assume{ordered field properties of~$\matR$}.
  
  \proofparskip{Induction: $P(n)$ implies $P(n+1)$}
  Let $n\in\matN$ be a natural number.
  Assume that $P(n)$ holds.
  Then, from
  Definition~\thref{d:iterated-function-sequence},
  Definition~\thref{d:lipschitz-continuity},
  \assume{field properties of~$\matR$}, and
  hypotheses,
  we have
  \begin{eqnarray*}
    d(x_{n + 1}, x_{n + 2})
    & = & d(f(x_n), f(x_{n + 1})) \\
    & \leq & k \, d(x_n, x_{n + 1}) \\
    & \leq & k \, k^n \, d(x_0,x_1) \\
    & = & k^{n + 1} \, d(x_0,x_1).
  \end{eqnarray*}
  Hence, $P(n+1)$ holds.
  
  Therefore, by induction on $n\in\matN$, we have, for all $n\in\matN$,
  $P(n)$.
\end{proof}

\begin{lemma}[convergent iterated function sequence]
  \label{l:convergent-iterated-function-sequence}
  Let $(X,d)$ be a metric space.
  Let $f:X\rightarrow X$ be a Lipschitz continuous mapping.
  Let $(x_n)_{n\in\matN}$ be a convergent iterated function sequence
  associated with~$f$.
  Then, the limit of the sequence is a fixed point of~$f$.
\end{lemma}

\begin{proof}
  From
  Definition~\thref{d:lipschitz-continuity},
  let $k\geq 0$ be the Lipschitz constant of~$f$.
  From
  Definition~\thref{d:convergent-sequence},
  let $a=\lim_{n\rightarrow+\infty}x_n\in X$ be the limit of the sequence.
  
  \proofparskip{Case $k=0$}
  Then, from
  Lemma~\thref{l:zero-lipschitz-continuous-is-constant},
  $f$~is constant of value~$f(a)$.
  Thus, from
  Definition~\thref{d:stationary-sequence},
  the sequence $(x_n)_{n\in\matN}$ is stationary from rank~1.
  Hence, from
  Lemma~\thref{l:stationary-sequence-is-convergent},
  $(x_n)_{n\in\matN}$ is convergent with limit~$f(a)$.
  
  \proofparskip{Case $k\not=0$}
  Then, from
  Definition~\thref{d:lipschitz-continuity},
  we have $k>0$.
  Let $\eps>0$.
  From
  Definition~\thref{d:convergent-sequence},
  let $N\in\matN$ such that, for all $n\in\matN$, $n\geq N$ implies
  $d(x_n,a)\leq\frac{\eps}{k}$.
  Let $\Np=N+1$.
  Let $n\in\matN$.
  Assume that $n\geq\Np$.
  Then, $n-1\geq N$.
  Thus, from
  Definition~\thref{d:iterated-function-sequence},
  Definition~\thref{d:lipschitz-continuity}, and
  \assume{ordered field properties of~$\matR$},
  we have
  \begin{equation*}
    d(x_n, f(a))
    = d(f(x_{n - 1}), f(a))
    \leq k \, d(x_{n - 1}, a)
    \leq \eps.
  \end{equation*}
  Hence, from
  Definition~\thref{d:convergent-sequence},
  the sequence $(x_n)_{n\in\matN}$ is convergent with limit~$f(a)$.
  
  Therefore, in both cases, from
  Lemma~\thref{l:limit-is-unique},
  $f(a)=a$.
\end{proof}

\begin{theorem}[fixed point]
  \label{t:fixed-point}
  Let $(X,d)$ be a complete metric space.
  Let~$f:X\rightarrow X$ be a contraction.
  Then, there exists a unique fixed point~$a\in X$ such that~$f(a)=a$.
  Moreover, all iterated function sequences associated with~$f$ are
  convergent with limit~$a$.
\end{theorem}

\begin{proof}
  \proofpar{Uniqueness}
  Let $a,\ap\in X$ be two fixed points of~$f$.
  Then, from
  Definition~\thref{d:contraction}, and
  Definition~\thref{d:lipschitz-continuity},
  we have $d(a,\ap)=d(f(a),f(\ap))\leq k\,d(a,\ap)$.
  Thus, from
  \assume{ordered field properties of~$\matR$},
  Definition~\threfc{d:contraction}{$k<1$}, and
  Definition~\threfc{d:distance}{$d$~is nonnegative},
  we have $0\leq(1-k)\,d(a,\ap)\leq 0$.
  Therefore, from
  the zero-product property of~$\matR$,
  Definition~\threfc{d:contraction}{$k\not=1$}, and
  Definition~\threfc{d:distance}{$d$~separates points},
  we have $a=\ap$.
  
  \proofparskip{Convergence of iterated function sequences and existence}
  Let $x_0\in X$.
  Let $(x_n)_{n\in\matN}$ be an iterated function sequence associated
  with~$f$.
  Let $p,m\in\matN$.
  Then, from
  Lemma~\thref{l:iterated-triangle-inequality},
  Lemma~\thref{l:iterate-lipschitz-continuous-mapping},
  \assume{field properties of~$\matR$},
  \assume{the formula for the sum of the first terms of a geometric series},
  and Definition~\threfc{d:contraction}{$0\leq k<1$},
  we have
  \begin{eqnarray*}
    d(x_p, x_{p+m})
    & \leq & \sum_{i = 0}^{m-1} d(x_{p+i}, x_{p+i+1}) \\
    & \leq & \left( \sum_{i = 0}^{m-1} k^{p+i} \right) d(x_0, x_1) \\
    & = & k^p \, \frac{1 - k^m}{1 - k} \, d(x_0, x_1) \\
    & \leq & \frac{k^p}{1 - k} \, d(x_0, x_1).
  \end{eqnarray*}
  
  \proofparskip{Case $k=0$}
  Then, from
  Lemma~\thref{l:zero-lipschitz-continuous-is-constant},
  $f$~is constant.
  Let $a=f(x_0)$.
  Then, for all $x\in X$, $f(x)=a$.
  In particular, $f(a)=a$, and for all $n\in\matN$, $x_{n+1}=f(x_n)=a$.
  Thus, from
  Definition~\thref{d:stationary-sequence},
  the sequence $(x_n)_{n\in\matN}$ is stationary from rank~1 with stationary
  value~$a$.
  
  \proofparskip{Case $x_1=x_0$}
  Then, from
  Lemma~\thref{l:stationary-iterated-function-sequence},
  the sequence $(x_n)_{n\in\matN}$ is stationary from rank~0 with stationary
  value~$x_0=x_1=f(x_0)=a$.
  
  Hence, in both cases, from
  Lemma~\thref{l:stationary-sequence-is-convergent},
  the sequence $(x_n)_{n\in\matN}$ is convergent with limit~$a$.
  
  \proofparskip{Case $k\not=0$ and $x_1\not=x_0$}
  Then, from
  Definition~\threfc{d:contraction}{$0\leq k<1$}, and
  Definition~\threfc{d:distance}{$d$~separates points, contrapositive},
  we have $0<k<1$ and $d(x_0,x_1)\not=0$.
  Let $\eps>0$.
  Let
  \begin{equation*}
    \zeta = \frac{(1 - k) \eps}{d(x_0, x_1)} > 0,\quad
    \xi = \max \left( 0, \frac{\ln \, \zeta}{\ln \, k} \right) \geq 0,\quad
    N = \ceil{\xi} \in \matN.
  \end{equation*}
  Let $p,m\in\matN$.
  Assume that $p\geq N$.
  Then, from
  \assume{the definition of ceiling and max functions},
  we have $p\geq\xi\geq\frac{\ln\,\zeta}{\ln\,k}$.
  Thus, from
  \assume{ordered field properties of~$\matR$ ($\ln\,k$ is negative)}, and
  \assume{increase of the exponential function},
  we have $p\ln\,k\leq\ln\,\zeta$, hence $k^p\leq\zeta$, and finally
  \begin{equation*}
    \frac{k^p}{1 - k} \, d(x_0, x_1) \leq \eps.
  \end{equation*}
  Hence, from
  Lemma~\thref{l:equivalent-definition-of-cauchy-sequence},
  Definition~\thref{d:complete-metric-space}, and
  Definition~\thref{d:complete-subset},
  $(x_n)_{n\in\matN}$ is a Cauchy sequence that is convergent with
  limit~$a\in X$.
  
  Therefore, in all cases, from
  Lemma~\thref{l:convergent-iterated-function-sequence},
  $a$~is a fixed point of~$f$.
\end{proof}

\subsection{Vector space}
\label{ss:vector-space}

\begin{remark}
  Statements and proofs are presented in the case of vector spaces over the
  field of real numbers~$\matR$, but most can be generalized with minor or no
  alteration to the case of vector spaces over the field of complex
  numbers~$\matC$.
  When the same statement holds for both cases, the field is
  denoted~$\matK$.
  Note that in both cases, $\matR\subset\matK$.
\end{remark}

\subsubsection{Basic notions and notations}

\begin{definition}[vector space]
  \label{d:space}
  Let~$E$ be a set equipped with two vector operations:
  an addition ($+:\ExE\rightarrow E$) and
  a scalar multiplication ($\cdot:\matKxE\rightarrow E$).
  $(E,+,\cdot)$ is a {\em vector space over field~$\matK$}, or simply
  $E$~is a {\em space}, iff
  $(E,+)$ is an abelian group with identity element~$0_E$ (zero vector),
  and scalar multiplication is
  distributive wrt vector addition and field addition,
  compatible with field multiplication,
  and admits~$1_\matK$ as identity element (simply denoted~1):
  \begin{eqnarray}
    \label{e:distributivity-wrt-vector-addition}
    \forall \lambda \in \matK,\,
    \forall u, v \in E, & &
    \lambda \cdot (u + v) = \lambda \cdot u + \lambda \cdot v; \\
    \label{e:distributivity-wrt-field-addition}
    \forall \lambda, \mu \in \matK,\,
    \forall u \in E, & &
    (\lambda + \mu) \cdot u = \lambda \cdot u + \mu \cdot u; \\
    \label{e:scalar-multiplication-field-multiplication-compatibility}
    \forall \lambda, \mu \in \matK,\,
    \forall u \in E, & &
    \lambda \cdot (\mu \cdot u) = (\lambda \mu) \cdot u; \\
    \label{e:scalar-multiplication-has-identity-element-one}
    \forall u \in E, & &
    1 \cdot u = u.
  \end{eqnarray}
\end{definition}

\begin{remark}
  The $\cdot$ infix sign in the scalar multiplication may be omitted.
\end{remark}

\begin{remark}
  Vector spaces over~$\matR$ are called {\em real} {\vectorspace}s, and
  vector spaces over~$\matC$ are called {\em complex} {\vectorspace}s.
\end{remark}

\begin{definition}[set of mappings to space]
  \label{d:set-of-mappings-to-space}
  Let~$X$ be a set.
  Let~$E$ be a {\vectorspace}.
  The {\em set of mappings from~$X$ to~$E$} is denoted $\FXE$.
\end{definition}

\begin{definition}[linear map]
  \label{d:linear-map}
  Let $(E,+_E,\cdot_E)$ and $(F,+_F,\cdot_F)$ be {\vectorspace}s.
  A mapping $f:E\rightarrow F$ is a {\em linear map from~$E$ to~$F$} iff
  it preserves vector operations, {\ie} iff
  it is additive and homogeneous of degree~1:
  \begin{eqnarray}
    \label{e:linear-map-additive}
    \forall u, v \in E,
    & &
    f (u +_E v) = f (u) +_F f (v); \\
    \label{e:linear-map-homogeneous}
    \forall \lambda \in \matK,
    \forall u \in E,\,
    & &
    f (\lambda \cdot_E u) = \lambda \cdot_F f (u).
  \end{eqnarray}
\end{definition}

\begin{definition}[set of linear maps]
  \label{d:set-of-linear-maps}
  Let $E,F$ be {\vectorspace}s.
  The {\em set of linear maps from~$E$ to~$F$} is denoted $\LEF$.
\end{definition}

\begin{definition}[linear form]
  \label{d:linear-form}
  Let~$E$ be a vector space over field~$\matK$.
  A {\em linear form on~$E$} is a linear map from~$E$ to~$\matK$.
\end{definition}


\begin{definition}[bilinear map]
  \label{d:bilinear-map}
  Let $(E,+_E,\cdot_E)$, $(F,+_F,\cdot_F)$ and $(G,+_G,\cdot_G)$ be
  {\vectorspace}s.
  A mapping $\fhi:\ExF\rightarrow G$ is a
  {\em bilinear map from~$\ExF$ to~$G$} iff
  it is left additive, right additive, and left and right homogeneous of
  degree~1:
  \begin{eqnarray}
    \label{e:bl-left-additive}
    \forall u, v \in E,\,
    \forall w \in F, & &
    \fhi (u +_E v, w) = \fhi (u, w) +_G \fhi (v, w); \\
    \label{e:bl-right-additive}
    \forall u \in E,\,
    \forall v, w \in F, & &
    \fhi (u, v +_F w) = \fhi (u, v) +_G \fhi (u, w); \\
    \label{e:bl-homogeneous}
    \forall \lambda \in \matK,\,
    \forall u \in E,\,
    \forall v \in F, & &
    \fhi (\lambda \cdot_E u, v)
    = \lambda \cdot_G \fhi (u, v)
    = \fhi (u, \lambda \cdot_F v).
  \end{eqnarray}
\end{definition}


\begin{definition}[bilinear form]
  \label{d:bilinear-form}
  Let~$E$ be a {\vectorspace}.
  A {\em bilinear form on~$E$} is a bilinear map from~$\ExE$ to~$\matK$.
\end{definition}

\begin{definition}[set of bilinear forms]
  \label{d:set-of-bilinear-forms}
  Let $E$ be a {\vectorspace}.
  The {\em set of bilinear forms on~$E$} is denoted $\LdE=\LExEmatK$.
\end{definition}

\subsubsection{Linear algebra}

\begin{lemma}[zero times yields zero]
  \label{l:zero-times-yields-zero}
  Let $(E,+,\cdot)$ be a {\vectorspace}.
  Then,
  \begin{equation}
    \label{e:zero-times-yields-zero}
    \forall u \in E,\quad 0 \cdot u = 0_E.
  \end{equation}
\end{lemma}

\begin{proof}
  Let $u\in E$ be a vector.
  From
  Definition~\threfc{d:space}{%
    $(E,+)$ is an abelian group, scalar multiplication admits~1 as identity
    element and is distributive wrt field addition}, and
  \assume{field properties of~$\matK$},
  we have
  \begin{equation*}
    0 \cdot u
    = 0 \cdot u + u + (- u)
    = 0 \cdot u + 1 \cdot u + (- u)
    = (0 + 1) \cdot u + (- u)
    = 1 \cdot u + (- u)
    = u + (- u)
    = 0_E
  \end{equation*}
\end{proof}

\begin{lemma}[minus times yields opposite vector]
  \label{l:minus-times-yields-opposite-vector}
  Let $(E,+,\cdot)$ be a {\vectorspace}.
  Then,
  \begin{equation}
    \label{e:minus-times-yields-opposite-vector}
    \forall \lambda \in \matK,\,
    \forall u \in E,\quad
    (- \lambda) \cdot u = - (\lambda \cdot u).
  \end{equation}
\end{lemma}

\begin{proof}
  Let $\lambda\in\matK$ be a scalar.
  Let $u\in E$ be a vector.
  From
  Definition~\threfc{d:space}{%
    scalar multiplication is distributive wrt field addition},
  \assume{field properties of~$\matK$}, and
  Lemma~\thref{l:zero-times-yields-zero},
  we have
  \begin{equation*}
    \lambda \cdot u + (- \lambda) \cdot u
    = (\lambda - \lambda) \cdot u
    = 0 \cdot u
    = 0_E.
  \end{equation*}
  Therefore, from
  Definition~\threfc{d:space}{$(E,+)$ is an abelian group},
  $(-\lambda)\cdot u$ is the opposite of $\lambda\cdot u$.
\end{proof}

\begin{definition}[vector subtraction]
  \label{d:vector-subtraction}
  Let $(E,+,\cdot)$ be a {\vectorspace}.
  {\em Vector subtraction}, denoted by the infix operator~$-$, is defined by
  \begin{equation}
    \label{e:vector-subtraction}
    \forall u, v \in E,\quad
    u - v = u + (- v).
  \end{equation}
\end{definition}

\begin{definition}[scalar division]
  \label{d:scalar-division}
  Let $(E,+,\cdot)$ be a {\vectorspace}.
  {\em Scalar division}, denoted by the infix operator~$/$, is defined by
  \begin{equation}
    \label{e:scalar-division}
    \forall \lambda \in \matK^\star,\,
    \forall u \in E,\quad
    \frac{u}{\lambda} = \frac{1}{\lambda} \cdot u.
  \end{equation}
\end{definition}

\begin{lemma}[times zero yields zero]
  \label{l:times-zero-yields-zero}
  Let $(E,+,\cdot)$ be a {\vectorspace}.
  Then,
  \begin{equation}
    \label{e:times-zero-yields-zero}
    \forall \lambda \in \matK,\quad \lambda \cdot 0_E = 0_E.
  \end{equation}
\end{lemma}

\begin{proof}
  Let $\lambda\in\matK$ be a scalar.
  From
  Definition~\threfc{d:space}{%
    $(E,+)$ is an abelian group and scalar multiplication is distributive wrt
    vector addition}, and
  Definition~\thref{d:vector-subtraction},
  we have
  \begin{equation*}
    \lambda \cdot 0_E
    = \lambda \cdot (0_E - 0_E)
    = \lambda \cdot 0_E - \lambda \cdot 0_E
    = 0_E.
  \end{equation*}
\end{proof}

\begin{lemma}[zero-product property]
  \label{l:zero-product-property}
  Let $(E,+,\cdot)$ be a {\vectorspace}.
  Then,
  \begin{equation}
    \label{e:zero-product-property}
    \forall \lambda \in \matK,\,
    \forall u \in E,\quad
    \lambda \cdot u = 0_E \Equiv \lambda = 0 \Disj u = 0_E.
  \end{equation}
\end{lemma}

\begin{proof}
  Let $\lambda\in\matK$ be a scalar.
  Let $u\in E$ be a vector.
  
  \proofparskip{``Left'' implies ``right''}
  Assume that $\lambda=0$ or $u=0_E$
  Then, from
  Lemma~\thref{l:zero-times-yields-zero}, and
  Lemma~\thref{l:times-zero-yields-zero},
  we have $\lambda\cdot u=0_E$.
  
  \proofparskip{``Right'' implies ``left''}
  Assume that $\lambda\cdot u=0_E$ and $\lambda\not=0$
  Then, from
  Definition~\threfc{d:space}{%
    scalar multiplication admits~1 as identity element and is compatible with
    field multiplication},
  \assume{field properties of~$\matK$}, and
  Lemma~\thref{l:times-zero-yields-zero},
  we have
  \begin{equation*}
    u
    = 1 \cdot u
    = \left( \frac{1}{\lambda} \lambda \right) \cdot u
    = \frac{1}{\lambda} \cdot (\lambda \cdot u)
    = \frac{1}{\lambda} \cdot 0_E
    = 0_E.
  \end{equation*}
  Hence, since
  \assume{$(P\conj\neg Q\implies R)\equiv(P\implies Q\disj R)$},
  $\lambda\cdot u=0_E$ implies $\lambda=0$ or $u=0_E$.
\end{proof}

\begin{definition}[subspace]
  \label{d:subspace}
  Let $(E,+,\cdot)$ be a {\vectorspace}.
  Let $F\subset E$ be a subset of~$E$.
  Let~$+_{|F}$ be the restrictions of~$+$ to~$\FxF$.
  Let~$\cdot_{|F}$ be the restrictions of~$\cdot$ to~$\matKxF$.
  $F$~is a {\em vector subspace of~$E$}, or simply a {\em subspace of~$E$},
  iff
  $(F,+_{|F},\cdot_{|F})$~is a {\vectorspace}.
\end{definition}

\begin{remark}
  In particular, a subspace is closed under restricted vector operations.
\end{remark}

\begin{remark}
  Usually, restrictions~$+_{|F}$ and~$\cdot_{|F}$ are still denoted~$+$
  and~$\cdot$.
\end{remark}

\begin{lemma}[trivial subspaces]
  \label{l:trivial-subspaces}
  Let~$E$ be a {\vectorspace}.
  Then, $E$~and $\{0_E\}$ are subspaces of~$E$.
\end{lemma}

\begin{proof}
  $E$~and $\{0_E\}$ are trivially subsets of~$E$.
  $E$~is a {\vectorspace}.
  $\{0_E\}$ is trivially a {\vectorspace}.
  Therefore, from
  Definition~\thref{d:subspace},
  $E$~and $\{0_E\}$ are subspaces of~$E$.
\end{proof}

\begin{lemma}[closed under vector operations is subspace]
  \label{l:closed-under-vector-operations-is-subspace}
  Let~$E$ be a {\vectorspace}.
  Let~$F$ be a subset of~$E$.
  $F$~is a subspace of~$E$ iff
  $0_E\in F$
  and~$F$ is closed under vector addition and scalar multiplication:
  \begin{eqnarray}
    \label{e:subspace-vector-addition}
    \forall u, v \in F,
    & & u + v \in F; \\
    \label{e:subspace-scalar-multiplication}
    \forall \lambda \in \matK,
    \forall u \in F,
    & & \lambda u \in F.
  \end{eqnarray}
\end{lemma}

\begin{proof}
  \proofpar{``If''}
  Assume that~$F$ contains~$0_E$ and is closed under vector addition and
  scalar multiplication.
  Then, $F$ is closed under the restriction to~$F$ of vector operations.
  Let $u,v\in F$ be vectors.
  Then, from
  Lemma~\threfc{l:minus-times-yields-opposite-vector}{with $\lambda=1$},
  $-v=(-1)v$ belongs to~$F$, and $u-v=u+(-v)$ also belongs to~$F$.
  Thus, from
  \assume{group theory},
  $(F,+_{|F})$ is a subgroup of $(E,+)$.
  Hence, from
  Definition~\threfc{d:space}{$(E,+)$ is an abelian group}, and
  \assume{group theory},
  $(F,+_{|F})$ is also an abelian group and $0_F=0_E$.
  Since~$F$ is a subset of~$E$, and~$E$ is a {\vectorspace},
  properties~\eqref{e:distributivity-wrt-vector-addition}
  to~\eqref{e:scalar-multiplication-has-identity-element-one} are trivially
  satisfied over~$F$.
  Therefore, from
  Definition~\thref{d:subspace},
  $F$~is a subspace of~$E$.
  
  \proofparskip{``Only if''}
  Conversely, assume now that~$F$ is a subspace of~$E$.
  Then, from
  Definition~\threfc{d:subspace}{$F$~is a {\vectorspace}}, and
  Definition~\threfc{d:space}{$(F,+_{|F})$ is an abelian group},
  $F$~contains $0_F=0_E$ and~$F$ is closed under the restriction to~$F$ of
  vector operations.
  Therefore, $F$~is closed under vector addition and scalar multiplication.
\end{proof}

\begin{lemma}[closed under linear combination is subspace]
  \label{l:closed-under-linear-combination-is-subspace}
  Let~$E$ be a {\vectorspace}.
  Let~$F$ be a subset of~$E$.
  $F$~is a subspace of~$E$ iff
  $0_E\in F$
  and~$F$ is closed under linear combination:
  \begin{equation}
    \label{e:subspace-linear-combination}
    \forall \lambda, \mu \in \matK,\,
    \forall u, v \in F,\quad
    \lambda u + \mu v \in F.
  \end{equation}
\end{lemma}

\begin{proof}
  \proofpar{``If''}
  Assume that~$F$ contains~$0_E$ and is closed under linear combination.
  Let $u,v\in F$ be vectors.
  Let~$\lambda\in\matK$ be a scalar.
  Then, from
  Definition~\threfc{d:space}{%
    scalar multiplication in~$E$ admits~1 as identity element},
  $u+v=1u+1v$ belongs to~$F$, and from
  Lemma~\thref{l:zero-product-property}, and
  Definition~\threfc{d:space}{$(E,+)$ is an abelian group},
  $\lambda u=\lambda u+0\cdot 0_E$ belongs to~$F$.
  Thus, $F$ contains~$0_E$ and is closed under vector operations.
  Therefore, from
  Lemma~\thref{l:closed-under-vector-operations-is-subspace},
  $F$~is a subspace of~$E$.
  
  \proofparskip{``Only if''}
  Conversely, assume now that~$F$ is a subspace of~$E$.
  Let $\lambda,\mu\in\matK$ be scalars.
  Let $u,v\in F$ be vectors.
  Then, from
  Lemma~\thref{l:closed-under-vector-operations-is-subspace},
  $F$ contains~$0_E$, $F$~is closed under scalar multiplication, hence
  $\up=\lambda u$ and $\vp=\mu v$ belong to~$F$, and~$F$ is closed
  by vector addition, hence $\up+\vp=\lambda u+\mu v$ belongs
  to~$F$.
  Therefore, $F$~is closed under linear combination.
\end{proof}

\begin{definition}[linear span]
  \label{d:linear-span}
  Let~$E$ be a {\vectorspace}.
  Let~$u\in E$ be a vector.
  The {\em linear span of~$u$}, denoted $\Line{u}$, is defined by
  \begin{equation}
    \label{e:linear-span}
    \Line{u} = \{ \lambda  u \st \lambda \in \matK \}.
  \end{equation}
\end{definition}


  


\begin{definition}[sum of subspaces]
  \label{d:sum-of-subspaces}
  Let~$E$ be a {\vectorspace}.
  Let $F,\Fp$ be subspaces of~$E$.
  The {\em sum of~$F$ and~$\Fp$} is the subset of~$E$ defined by
  \begin{equation}
    \label{e:sum-of-subspaces}
    F + \Fp = \{ u + \up \st u \in F,\, \up \in \Fp \}.
  \end{equation}
\end{definition}

\begin{definition}[finite dimensional subspace]
  \label{d:finite-dimensional-subspace}
  Let~$E$ be a {\vectorspace}.
  Let~$F$ be a subspace of~$E$.
  {\em $F$~is a finite dimensional subspace} iff
  there exists $n\in\matN$, and $u_1,\ldots,u_n\in E$ such that
  \begin{equation}
    \label{e:finite-dimensional-subspace}
    \begin{array}{rcl}
      F & = &
      \Span{\{u_1, \ldots, u_n\}} =
      \Line{u_1} + \ldots + \Line{u_n} \\
      & = &
      \{ \lambda_1 u_1 + \ldots + \lambda_n u_n
      \st \lambda_1, \ldots, \lambda_n \in \matK \}.
    \end{array}
  \end{equation}
\end{definition}

\begin{definition}[direct sum of subspaces]
  \label{d:direct-sum-of-subspaces}
  Let~$E$ be a {\vectorspace}.
  Let $F,\Fp$ be subspaces of~$E$.
  The sum $F+\Fp$ is called {\em direct sum}, and it is denoted
  $F\oplus\Fp$, iff
  all vectors of the sum admit a unique decomposition:
  \begin{equation}
    \label{e:unique-decomposition}
    \forall u, v \in F,\,
    \forall \up, \vp \in \Fp,\quad
    u + \up = v + \vp \Implies u = v \Conj \up = \vp.
  \end{equation}
\end{definition}

\begin{lemma}[equivalent definitions of direct sum]
  \label{l:equivalent-definition-of-direct-sum}
  Let~$E$ be a {\vectorspace}.
  Let $F,\Fp$ be subspaces of~$E$.
  The sum $F+\Fp$ is direct iff
  one of the following equivalent properties is satisfied:
  \begin{equation}
    \label{e:zero-intersection}
    F \cap \Fp = \{ 0_E \};
  \end{equation}
  \begin{equation}
    \label{e:unique-zero-decomposition}
    \forall u \in F,\,
    \forall \up \in \Fp,\quad
    u + \up = 0_E \Implies u = \up = 0_E.
  \end{equation}
\end{lemma}

\begin{proof}
  \proofpar{(\ref{e:unique-decomposition})
    implies~(\ref{e:zero-intersection})}
  Assume that the sum $F+\Fp$ is direct.
  Let $v\in F\cap\Fp$ be a vector in the intersection.
  Then, from
  Definition~\threfc{d:space}{$(E,+)$ is an abelian group},
  $v$~admits two decompositions, $v=v+0_E=0_E+v$.
  Thus, from
  Definition~\thref{d:direct-sum-of-subspaces},
  $v=0_E$.
  
  \proofparskip{(\ref{e:zero-intersection})
    implies~(\ref{e:unique-zero-decomposition})}
  Assume now that $F\cap\Fp=\{0_E\}$.
  Let $u\in F$ and $\up\in\Fp$ be vectors.
  Assume that $u+\up=0_E$.
  Then, from
  Lemma~\threfc{l:minus-times-yields-opposite-vector}{with $\lambda=1$}, and
  Lemma~\threfc{l:closed-under-vector-operations-is-subspace}{%
    scalar multiplication},
  $u=-\up=(-1)\up$ belongs to~$\Fp$ and $\up=-u=(-1)u$ belongs to~$F$.
  Hence, $u,\up\in F\cap\Fp$, and $u=\up=0_E$.
  
  \proofparskip{(\ref{e:unique-zero-decomposition})
    implies~(\ref{e:unique-decomposition})}
  Assume finally that~$0_E$ admits a unique decomposition.
  Let $u,v\in F$ and $\up,\vp\in\Fp$ be vectors.
  Assume that $u+\up=v+\vp$.
  Then, from
  Definition~\threfc{d:space}{$(E,+)$ is an abelian group}, and
  Definition~\thref{d:vector-subtraction},
  we have $(u-v)+(\up-\vp)=0_E$.
  Thus, from
  hypothesis,
  we have $u-v=\up-\vp=0_E$, and from
  Definition~\threfc{d:space}{$(E,+)$ is an abelian group},
  we have $u=v$ and $\up=\vp$.
  
  Therefore, all three properties are equivalent.
\end{proof}



\begin{lemma}[direct sum with linear span]
  \label{l:direct-sum-with-linear-span}
  Let~$E$ be a {\vectorspace}.
  Let $F$ be a subspace of~$E$.
  Let $u\in E$ be a vector.
  Assume that $u\not\in F$.
  Then, the sum $F+\Line{u}$ is direct.
\end{lemma}

\begin{proof}
  From
  Lemma~\threfc{l:closed-under-vector-operations-is-subspace}{%
    $F$ and $\Line{u}$ are subspace},
  we have $0_E\in F$ and $0_E\in\Line{u}$.
  Let $v\in F\cap\Line{u}$.
  Assume that $v\not=0_E$.
  Then, from
  Definition~\thref{d:linear-span},
  let $\lambda\in\matK$ such that $v=\lambda u$.
  Thus, from
  hypothesis, and
  Lemma~\threfc{l:zero-product-property}{contrapositive},
  we have $\lambda\not=0_\matK$.
  Hence, from
  \assume{field properties of~$\matK$}, and
  Lemma~\thref{l:closed-under-vector-operations-is-subspace},
  we have $\frac{1}{\lambda}\,v=u\in F$.
  Which is impossible by hypothesis.
  
  Therefore, from
  Lemma~\thref{l:equivalent-definition-of-direct-sum},
  the sum $F+\Line{u}$ is direct.
\end{proof}

 
 

\begin{definition}[product vector operations]
  \label{d:product-vector-operations}
  Let $(E,+_E,\cdot_E)$ and $(F,+_F,\cdot_F)$ be {\vectorspace}s.
  The {\em product vector operations induced on $\ExF$} are the
  mappings $+_\ExF:\ExFxExF\rightarrow\ExF$ and
  $\cdot_\ExF:\matKxExF\rightarrow\ExF$ defined by
  \begin{eqnarray}
    \label{e:product-vector-addition}
    \forall (u, v), (\up, \vp) \in \ExF, & &
    (u, v) +_\ExF (\up, \vp) = (u +_E \up, v +_F \vp); \\
    \label{e:product-scalar-multiplication}
    \forall \lambda \in \matK,\,
    \forall (u, v) \in \ExF, & &
    \lambda \cdot_\ExF (u, v) = (\lambda \cdot_E u, \lambda \cdot_F v).
  \end{eqnarray}
\end{definition}

\begin{lemma}[product is space]
  \label{l:product-is-space}
  Let $(E,+_E,\cdot_E)$ and $(F,+_F,\cdot_F)$ be {\vectorspace}s.
  Let~$+_\ExF$ and~$\cdot_\ExF$ be the product vector operations induced on
  $\ExF$.
  Then, $(\ExF,+_\ExF,\cdot_\ExF)$ is a {\vectorspace}.
\end{lemma}

\begin{proof}
  From
  \assume{group theory},
  $(\ExF,+_\ExF)$ is an abelian group with identity element
  $0_\ExF=(0_E,0_F)$.
  Distributivity of the product scalar multiplication wrt product vector
  addition and field addition, compatibility of the product scalar
  multiplication with field multiplication, and 1~is the identity element for
  the product scalar multiplication are direct consequences of
  Definition~\thref{d:space}, and
  Definition~\thref{d:product-vector-operations}.
  
  Therefore, from
  Definition~\thref{d:space},
  $(\ExF,+_\ExF,\cdot_\ExF)$ is a {\vectorspace}.
\end{proof}

\begin{definition}[inherited vector operations]
  \label{d:inherited-vector-operations}
  Let~$X$ be a set.
  Let~$(E,+_E,\cdot_E)$ be a {\vectorspace}.
  The {\em vector operations inherited on $\FXE$} are the
  mappings $+_\FXE:\FXExFXE\rightarrow\FXE$
  and $\cdot_\FXE:\matKxFXE\rightarrow\FXE$ defined
  by
  \begin{eqnarray}
    \label{e:inherited-vector-addition}
    \forall f, g \in \FXE,\,
    \forall x \in X, & &
    (f +_\FXE g) (x) = f (x) +_E g (x); \\
    \label{e:inherited-scalar-multiplication}
    \forall \lambda \in \matK,\,
    \forall f \in \FXE,\,
    \forall x \in X, & &
    (\lambda \cdot_\FXE f) (x) = \lambda \cdot_E f (x).
  \end{eqnarray}
\end{definition}

\begin{remark}
  Usually, inherited vector operations are denoted the same way as the vector
  operations of the target space.
\end{remark}

\begin{lemma}[space of mappings to a space]
  \label{l:space-of-mappings-to-space}
  Let~$X$ be a set.
  Let~$(E,+_E,\cdot_E)$ be a {\vectorspace}.
  Let~$+_\FXE$ and~$\cdot_\FXE$ be the vector operations
  inherited on $\FXE$.
  Then,\\
  $(\FXE,+_\FXE,\cdot_\FXE)$ is a {\vectorspace}.
\end{lemma}

\begin{proof}
  From
  Definition~\thref{d:set-of-mappings-to-space}, and
  \assume{group theory},
  $(\FXE,+_\FXE)$ is an abelian group with identity element
  \begin{equation*}
    0_\FXE = (x \in X \mapsto 0_E).
  \end{equation*}
  Distributivity of the inherited scalar multiplication wrt inherited vector
  addition and field addition, compatibility of the inherited scalar
  multiplication with field multiplication, and 1~is the identity element for
  the inherited scalar multiplication are direct consequences of
  Definition~\thref{d:space}, and
  Definition~\thref{d:inherited-vector-operations}.
  
  Therefore, from
  Definition~\thref{d:space},
  $(\FXE,+_\FXE,\cdot_\FXE)$ is a {\vectorspace}.
\end{proof}

\begin{lemma}[linear map preserves zero]
  \label{l:linear-map-preserves-zero}
  Let~$E$ and~$F$ be {\vectorspace}s.
  Let~$f$ be a linear map from~$E$ to~$F$.
  Then, $f(0_E)=0_F$.
\end{lemma}

\begin{proof}
  From
  Lemma~\thref{l:zero-times-yields-zero}, and
  Definition~\threfc{d:linear-map}{$f$~is homogeneous of degree~1},
  we have
  \begin{equation*}
    f (0_E) = f (0_\matK \cdot_E 0_E) = 0_\matK \cdot_F f (0_E) = 0_F.
  \end{equation*}
\end{proof}

\begin{lemma}[linear map preserves linear combinations]
  \label{l:linear-map-preserves-linear-combinations}
  Let $(E,+_E,\cdot_E)$ and $(F,+_F,\cdot_F)$ be {\vectorspace}s.
  Let $f:E\rightarrow F$ be a mapping from~$E$ to~$F$.
  Then, $f$~is a linear map from~$E$ to~$F$ iff
  it preserves linear combinations:
  \begin{equation}
    \label{e:linear-map-linear-combination}
    \forall \lambda, \mu \in \matK,\,
    \forall u, v \in E,\quad
    f (\lambda \cdot_E u +_E \mu \cdot_E v) =
    \lambda \cdot_F f (u) +_F \mu \cdot_F f (v).
  \end{equation}
\end{lemma}

\begin{proof}
  \proofpar{``If''}
  Assume that~\eqref{e:linear-map-linear-combination} holds.
  Let~$u,v\in E$ be vectors.
  Then, from
  Definition~\threfc{d:space}{%
    scalar multiplications in~$E$ and~$F$ admit~1 as identity element},
  we have
  \begin{equation*}
    f (u +_E v)
    = f (1 \cdot_E u +_E 1 \cdot_E v)
    = 1 \cdot_F f (u) +_F 1 \cdot_F f (v)
    = f (u) +_F f (v).
  \end{equation*}
  Hence, $f$ is additive.
  Let~$\lambda\in\matK$ be a scalar.
  Let~$u$ be a vector.
  Then, from
  Lemma~\threfc{l:zero-times-yields-zero}{in~$E$ and~$F$}, and
  Definition~\threfc{d:space}{$(E,+_E)$ and $(F,+_F)$ are abelian groups},
  we have
  \begin{equation*}
    f (\lambda \cdot_E u)
    = f (\lambda \cdot_E u +_E 0 \cdot_E 0_E)
    = \lambda \cdot_F f (u) +_F 0 \cdot_F f (0_E)
    = \lambda \cdot_F f (u).
  \end{equation*}
  Hence, $f$~is homogeneous of degree~1.
  Therefore, from
  Definition~\thref{d:linear-map},
  $f$~is a linear map from~$E$ to~$F$.
  
  \proofparskip{``Only if''}
  Conversely, assume now that~$f$ is a linear map from~$E$ to~$F$.
  Let~$\lambda,\mu\in\matK$ be scalars.
  Let~$u,v\in E$ be vectors.
  Then, from
  Definition~\threfc{d:linear-map}{$f$~is additive},
  $f(\lambda\cdot_Eu+_E\mu\cdot_Ev)=f(\lambda\cdot_Eu)+_Ff(\mu\cdot_Ev)$,
  and ($f$~is homogeneous of degree~1)
  $f(\lambda\cdot_Eu)=\lambda\cdot_Ff(u)$ and
  $f(\mu\cdot_Ev)=\mu\cdot_Ff(v)$.
  Hence, we have
  \begin{equation*}
    f (\lambda \cdot_E u +_E \mu \cdot_E v)
    = \lambda \cdot_F f (u) +_F \mu \cdot_F f (v).
  \end{equation*}
\end{proof}

\begin{lemma}[space of linear maps]
  \label{l:space-of-linear-maps}
  Let $E$ and $(F,+_F,\cdot_F)$ be {\vectorspace}s.
  Let~$+_\FEF$ and~$\cdot_\FEF$ be the vector operations inherited on
  $\FEF$.
  Then, $\LEF$ is a subspace of\\
  $(\FEF,+_\FEF,\cdot_\FEF)$.
\end{lemma}

\begin{proof}
  From
  Definition~\thref{d:set-of-linear-maps}, and
  Lemma~\thref{l:space-of-mappings-to-space},
  $\LEF$ is a subset of $\FEF$.
  From
  Lemma~\thref{l:times-zero-yields-zero}, and
  Definition~\threfc{d:space}{$(F,+_F)$~is an abelian group},
  $0_\FEF$~trivially preserves vector operations.
  Hence, from
  Definition~\thref{d:linear-map},
  $0_\FEF$~is a linear map from~$E$ to~$F$.
  From
  Definition~\thref{d:inherited-vector-operations},
  Definition~\thref{d:linear-map}, and
  Definition~\threfc{d:space}{%
    $(E,+_E)$ and $(F,+_F)$ are abelian groups and scalar multiplications
    in~$E$ and~$F$ are compatible with field multiplication},
  $\LEF$ is trivially closed under linear combination.
  
  Therefore, from
  Lemma~\thref{l:closed-under-linear-combination-is-subspace},
  $\LEF$ is a subspace of $\FEF$.
\end{proof}

\begin{definition}[identity map]
  \label{d:identity-map}
  Let~$E$ be a {\vectorspace}.
  The {\em identity map on~$E$} is the mapping $\idE:E\rightarrow E$ defined
  by
  \begin{equation}
    \label{e:identity-map}
    \forall u \in E,\quad
    \idE (u) = u.
  \end{equation}
\end{definition}

\begin{lemma}[identity map is linear map]
  \label{l:identity-map-is-linear-map}
  Let~$E$ be a {\vectorspace}.
  Then, the identity map~$\idE$ is a linear map.
\end{lemma}

\begin{proof}
  Direct consequence of
  Definition~\thref{d:identity-map}, and
  Definition~\thref{d:linear-map}.
\end{proof}

\begin{lemma}[composition of linear maps is bilinear]
  \label{l:composition-of-linear-maps-is-bilinear}
  Let $E,F,G$ be {\vectorspace}s.
  Then, the composition of functions is a bilinear map from $\LEFxLFG$ to
  $\LEG$.
\end{lemma}

\begin{proof}
  From
  Lemma~\threfc{l:space-of-linear-maps}{%
    $\LEF$, $\LFG$ and $\LEG$ are spaces}, and
  Lemma~\thref{l:product-is-space},
  $\LEFxLFG$ is a space.
  
  Let $f\in\LEF$ and $g\in\LFG$ be linear maps.
  Let $\lambda,\mu\in\matK$ be scalars.
  Let~$u,v\in E$ be vectors.
  Then, from
  \assume{the definition of composition of functions}, and
  Lemma~\threfc{l:linear-map-preserves-linear-combinations}{for~$f$ and~$g$},
  we have
  \begin{eqnarray*}
    (g \circ f) (\lambda u + \mu v)
    & = & g (f (\lambda u + \mu v)) \\
    & = & g (\lambda f (u) + \mu f (v)) \\
    & = & \lambda g (f (u)) + \mu g (f (v)) \\
    & = & \lambda (g \circ f) (u) + \mu (g \circ f) (v).
  \end{eqnarray*}
  Hence, $g\circ f$ belongs to $\LEG$, and composition is a mapping from
  space $\LEFxLFG$ to space $\LEG$.
  
  From
  \assume{the definition of composition of functions}, and
  Definition~\thref{d:inherited-vector-operations},
  composition of linear maps is trivially left additive and left homogeneous
  of degree~1.
  From
  \assume{the definition of composition of functions},
  Definition~\thref{d:inherited-vector-operations}, and
  Definition~\threfc{d:linear-map}{%
    left argument ``$g$'' is additive and homogeneous of degree~1},
  composition of linear maps is trivially right additive and right
  homogeneous of degree~1.
  
  Therefore, from
  Definition~\thref{d:bilinear-map},
  composition of linear maps is a bilinear map from $\LEFxLFG$ to $\LEG$.
\end{proof}




\begin{definition}[isomorphism]
  \label{d:isomorphism}
  Let~$E$ and~$F$ be {\vectorspace}s.
  An {\em isomorphism from~$E$ onto~$F$} is a linear map from~$E$ to~$F$ that
  is bijective.
\end{definition}


  

\begin{definition}[kernel]
  \label{d:kernel}
  Let~$E$ and~$F$ be {\vectorspace}s.
  Let $f\in\LEF$ be a linear map from~$E$ to~$F$.
  The {\em kernel of~$f$} (or {\em null space of~$f$}), denoted $\Ker{f}$, is
  the subset of~$E$ defined by
  \begin{equation}
    \label{e:kernel}
    \Ker{f} = \{ u \in E \st f (u) = 0_F \}.
  \end{equation}
\end{definition}

\begin{lemma}[kernel is subspace]
  \label{l:kernel-is-subspace}
  Let~$E$ and~$F$ be {\vectorspace}s.
  Let $f\in\LEF$ be a linear map from~$E$ to~$F$.
  Then, $\Ker{f}$ is a subspace of~$E$.
\end{lemma}

\begin{proof}
  From
  Lemma~\thref{l:linear-map-preserves-zero}, and
  Definition~\thref{d:kernel},
  $0_E$ belongs to $\Ker{f}$.
  Let $\lambda,\mu\in\matK$ be scalars.
  Let $u,v\in\Ker{f}$ be vectors in the kernel.
  Then, from
  Lemma~\thref{l:linear-map-preserves-linear-combinations},
  Definition~\thref{d:kernel},
  Lemma~\thref{l:times-zero-yields-zero}, and
  Definition~\threfc{d:space}{$(F,+)$ is an abelian group},
  we have
  \begin{equation*}
    f (\lambda u + \mu v)
    = \lambda f(u) + \mu f(v)
    = \lambda 0_F + \mu 0_F
    = 0_F.
  \end{equation*}
  Hence, from
  Definition~\thref{d:kernel},
  $\lambda u+\mu v$ belongs to~$\Ker{f}$.
  
  Therefore, from
  Lemma~\thref{l:closed-under-linear-combination-is-subspace},
  $\Ker{f}$ is a subspace of~$E$.
\end{proof}

\begin{lemma}[injective linear map has zero kernel]
  \label{l:injective-linear-map-has-zero-kernel}
  Let~$E$ and~$F$ be {\vectorspace}s.
  Let $f\in\LEF$ be a linear map from~$E$ to~$F$.
  Then, $f$~is injective iff $\Ker{f}=\{0_E\}$.
\end{lemma}

\begin{proof}
  \proofpar{``If''}
  Assume that $\Ker{f}=\{0_E\}$.
  Let $u,v\in E$ be vectors.
  Assume that $f(u)=f(v)$.
  Then, from
  Definition~\thref{d:vector-subtraction},
  Definition~\thref{d:linear-map}, and
  Definition~\threfc{d:space}{$(F,+)$ is an abelian group},
  we have $f(u-v)=0_F$, hence $u-v$ belongs to $\Ker{f}$.
  Thus $u-v=0_E$, and from
  Definition~\thref{d:vector-subtraction}, and
  Definition~\threfc{d:space}{$(E,+)$ is an abelian group},
  we have $u=v$.
  Therefore, from
  \assume{the definition of injectivity},
  $f$ is injective.
  
  \proofparskip{``Only if''}
  Conversely, assume now that~$f$ is injective.
  Let~$u\in\Ker{f}$ be a vector in the kernel.
  Then, from
  Definition~\thref{d:kernel},
  we have $f(u)=0_F$.
  Thus, from
  Lemma~\thref{l:linear-map-preserves-zero},
  we have $f(u)=f(0_E)$.
  Therefore, from
  \assume{the definition of injectivity},
  $u=0_E$.
\end{proof}

\begin{lemma}[K is space]
  \label{l:k-is-space}
  The commutative field~$\matK$ equipped with its addition and multiplication
  is a {\vectorspace}.
\end{lemma}

\begin{proof}
  Direct consequence of the commutative field structure.
\end{proof}

\subsection{Normed vector space}
\label{ss:normed-vector-space}

\begin{definition}[norm]
  \label{d:norm}
  Let~$E$ be a {\vectorspace}.
  An application $\nrmdot:E\rightarrow\matR$ is a {\em norm over~$E$} iff
  it separates points (or it is definite),
  it is absolutely homogeneous of degree~1,
  and it satisfies the triangle inequality:
  \begin{eqnarray}
    \label{e:norm-point-separation}
    \forall u \in E, & & \nrm{u} = 0 \Implies u = 0_E; \\
    \label{e:norm-absolute-homogeneity}
    \forall \lambda \in \matK,\,
    \forall u \in E, & & \nrm{\lambda u} = | \lambda | \, \nrm{u}; \\
    \label{e:norm-triangle-inequality}
    \forall u, v \in E, & & \nrm{u + v} \leq  \nrm{u} + \nrm{v}.
  \end{eqnarray}
\end{definition}

\begin{remark}
  \label{r:absolute-value-over-k}
  The absolute value over field~$\matK$ is a function
  $|\cdot|:\matK\rightarrow\matR$ that is nonnegative, definite,
  multiplicative, and satisfies the triangle inequality.
  It is the modulus for the field of complex numbers.
\end{remark}

\begin{definition}[normed vector space]
  \label{d:normed-space}
  $(E,\nrmdot)$ is a {\em normed vector space}, or simply a
  {\em {\normedvectorspace}}, iff
  $E$~is a {\vectorspace} and $\nrmdot$~is a norm over~$E$.
\end{definition}



\begin{lemma}[K is {\normedvectorspace}]
  \label{l:k-is-normed-space}
  The commutative field~$\matK$ equipped with its absolute value is a normed
  space.
\end{lemma}

\begin{proof}
  Direct consequence of
  Definition~\thref{d:normed-space},
  Definition~\thref{d:norm}, and
  \assume{properties of the absolute value over~$\matK$}
  (see Remark~\ref{r:absolute-value-over-k}).
\end{proof}

\begin{lemma}[norm preserves zero]
  \label{l:norm-preserves-zero}
  Let $(E,\nrmdot)$ be a {\normedvectorspace}.
  Then, $\nrm{0_E}=0$.
\end{lemma}

\begin{proof}
  From
  Definition~\threfc{d:normed-space}{$E$~is a {\vectorspace}}, and
  Definition~\thref{d:space},
  $0_E$~belongs to~$E$.
  From
  Lemma~\thref{l:zero-times-yields-zero},
  Definition~\threfc{d:norm}{%
    $\nrmdot$~is absolutely homogeneous of degree~1},
  \assume{definition of the absolute value}, and
  \assume{field properties of~$\matR$},
  we have
  \begin{equation*}
    \nrm{0_E}
    = \nrm{0_\matK \cdot 0_E}
    = | 0_\matK | \, \nrm{0_E}
    = 0_\matR \nrm{0_E}
    = 0_\matR.
  \end{equation*}
\end{proof}

\begin{lemma}[norm is nonnegative]
  \label{l:norm-is-nonnegative}
  Let $(E,\nrmdot)$ be a {\normedvectorspace}.
  Then, $\nrmdot$~is nonnegative.
\end{lemma}

\begin{proof}
  From
  Definition~\thref{d:normed-space},
  $E$~is a {\vectorspace}.
  Let $u\in E$ be a vector.
  Then, from
  Definition~\threfc{d:norm}{%
    $\nrmdot$~is absolutely homogeneous of degree~1 and satisfies triangle
    inequality},
  Definition~\threfc{d:space}{$(E,+)$ is an abelian group}, and
  \assume{ordered field properties of~$\matR$},
  we have we $\nrm{-u}=\nrm{u}$ and
  \begin{equation*}
    \nrm{u}
    = \half (\nrm{u} + \nrm{- u})
    \geq \half \nrm{u - u}
    = \half \nrm{0_E}
    = 0.
  \end{equation*}
\end{proof}

\begin{lemma}[normalization by nonzero]
  \label{l:normalization-by-nonzero}
  Let $(E,\nrmdot)$ be a {\normedvectorspace}.
  Then,
  \begin{equation}
    \label{e:normalization-by-nonzero}
    \forall \lambda \in \matK,\,
    \forall u \in E,\quad
    u \neq 0 \Implies \nrm{\lambda \frac{u}{\nrm{u}}} = | \lambda |.
  \end{equation}
\end{lemma}

\begin{proof}
  Direct consequence of
  Definition~\thref{d:scalar-division},
  Definition~\threfc{d:norm}{%
    $\nrmdot$~is definite and absolutely homogeneous of degree~1},
  Lemma~\thref{l:norm-is-nonnegative}, and
  \assume{field properties of~$\matR$}.
\end{proof}

\begin{definition}[distance associated with norm]
  \label{d:distance-associated-with-norm}
  Let $(E,\nrmdot)$ be a {\normedvectorspace}.
  The {\em distance associated with norm~$\nrmdot$} is the mapping
  $\dst:\ExE\rightarrow\matR$ defined by
  \begin{equation}
    \label{e:dE}
    \forall u, v \in E,\quad \dst (u, v) = \nrm{u - v}.
  \end{equation}
\end{definition}

\begin{remark}
  The mapping $\dst$ will be proved below to be a distance; hence its name.
\end{remark}

\begin{lemma}[norm gives distance]
  \label{l:norm-gives-distance}
  Let $(E,\nrmdot)$ be a {\normedvectorspace}.
  Let~$\dst$ be the distance associated with norm~$\nrmdot$.
  Then, $(E,\dst)$ is a metric space.
\end{lemma}

\begin{proof}
  From
  Definition~\thref{d:distance-associated-with-norm},
  Lemma~\thref{l:norm-is-nonnegative},
  Definition~\thref{d:vector-subtraction}, and
  Definition~\threfc{d:norm}{%
    $\nrmdot$~is absolutely homogeneous of degree~1 with $\lambda=-1$,
    definite and satisfies triangle inequality},
  $\dst$~is nonnegative and symmetric, separates points, and satisfies the
  triangle inequality.
  Thus, from
  Definition~\thref{d:distance},
  $\dst$~is a distance over~$E$.
  Therefore, from
  Definition~\thref{d:metric-space},
  $(E,\dst)$ is a metric space.
\end{proof}

\begin{lemma}[linear span is closed]
  \label{l:linear-span-is-closed}
  Let $(E,\nrmdot)$ be a {\normedvectorspace}.
  Let~$\dst$ be the distance associated with norm~$\nrmdot$.
  Let~$u\in E$ be a vector.
  Then, $\Line{u}$ is closed for distance~$\dst$.
\end{lemma}

\begin{proof}
  Let $\lambda,\lambdap\in\matK$ be scalars.
  From
  Definition~\thref{d:distance-associated-with-norm},
  Definition~\threfc{d:space}{%
    scalar multiplication is distributive wrt field addition},
  Definition~\thref{d:vector-subtraction}, and
  Definition~\threfc{d:norm}{%
    $\nrmdot$ is absolutely homogeneous of degree~1},
  we have
  \begin{equation}
    \label{e:linear-span-is-closed-aux}
    \dst (\lambda u, \lambdap u)
    = \nrm{\lambda u - \lambdap u}
    = \nrm{(\lambda - \lambdap) u}
    = | \lambda - \lambdap | \nrm{u}.
  \end{equation}

  \proofparskip{Case $u=0_E$}
  Direct consequence of
  Definition~\thref{d:linear-span}, and
  Lemma~\threfc{l:singleton-is-closed}{$\Line{u}=\{0_E\}$}.
  
  \proofparskip{Case $u\not=0_E$}
  Then, from
  Definition~\threfc{d:norm}{$\nrmdot$ is definite, contrapositive}, and
  Lemma~\thref{l:norm-is-nonnegative},
  we have $\nrm{u}>0$.
  Let $(\lambda_n u)_{n\in\matN}$ be a sequence in $\Line{u}$.
  Assume that this sequence is convergent.
  Let $\eps>0$.
  
  From
  Lemma~\thref{l:convergent-sequence-is-cauchy},
  \assume{ordered field properties of~$\matR$ (with $\nrm{u}>0$)},
  Definition~\threfc{d:cauchy-sequence}{%
    $(\lambda_n u)_{n\in\matN}$ is a Cauchy sequence with
    $\eps\nrm{u}>0$}, and
  Equation~\eqref{e:linear-span-is-closed-aux},
  there exists $N\in\matN$ such that for all $p,q\in\matN$, $p,q\geq N$
  implies
  \begin{equation*}
    | \lambda_p - \lambda_q |
    = \frac{\dst (\lambda_p u, \lambda_q u)}{\nrm{u}}
    \leq \frac{\eps \nrm{u}}{\nrm{u}}
    = \eps.
  \end{equation*}
  Hence, from
  Definition~\thref{d:cauchy-sequence},
  Definition~\threfc{d:complete-subset}{%
    $(\lambda_n)_{n\in\matN}$ is a Cauchy sequence in~$\matK$ complete},
  Lemma~\threfc{l:limit-is-unique}{$(\matK, |\cdot|)$ is a metric space},
  let $\lambda\in\matK$ be the limit $\lim_{n\rightarrow+\infty}\lambda_n$.
  
  Then, from
  \assume{ordered field properties of~$\matR$ (with $\nrm{u}>0$)},
  Definition~\threfc{d:convergent-sequence}{%
    $(\lambda_n)_{n\in\matN}$ is convergent with limit~$\lambda$, with
    $\frac{\eps}{\nrm{u}}>0$}, and
  Equation~\eqref{e:linear-span-is-closed-aux},
  there exists $N\in\matN$ such that for all $n\in\matN$, $n\geq N$ implies
  \begin{equation*}
    \dst (\lambda_n u, \lambda u)
    = | \lambda_n - \lambda | \nrm{u}
    \leq \frac{\eps}{\nrm{u}} \nrm{u}
    = \eps.
  \end{equation*}
  Hence, from
  Definition~\thref{d:convergent-sequence},
  $(\lambda_n u)_{n\in\matN}$ has limit $\lambda u\in\Line{u}$.
  Therefore, from
  Lemma~\thref{l:closed-is-limit-of-sequences},
  $\Line{u}$ is closed (for distance~$\dst$).
\end{proof}

\begin{definition}[closed unit ball]
  \label{d:closed-unit-ball}
  Let $(E,\nrmdot)$ be a {\normedvectorspace}.
  Let~$\dst$ be the distance associated with norm~$\nrmdot$.
  The {\em closed unit ball of~$E$} is $\cball{\dst}{0_E}{1}$ in the metric
  space $(E,\dst)$.
\end{definition}

\begin{lemma}[equivalent definition of closed unit ball]
  \label{l:equivalent-definition-of-closed-unit-ball}
  Let $(E,\nrmdot)$ be a {\normedvectorspace}.
  Let~$\unitcB$ be the closed unit ball in~$E$.
  Then, $\unitcB=\{u\in E\st \nrm{u}\leq 1\}$.
\end{lemma}

\begin{proof}
  Direct consequence of
  Definition~\thref{d:closed-unit-ball},
  Definition~\thref{d:distance-associated-with-norm},
  Lemma~\thref{l:norm-gives-distance}, and
  Definition~\thref{d:closed-ball}.
\end{proof}




\begin{definition}[unit sphere]
  \label{d:unit-sphere}
  Let $(E,\nrmdot)$ be a {\normedvectorspace}.
  Let~$\dst$ be the distance associated with norm~$\nrmdot$.
  The {\em unit sphere of~$E$} is $\sphere{\dst}{0_E}{1}$ in the metric space
  $(E,\dst)$.
\end{definition}

\begin{lemma}[equivalent definition of unit sphere]
  \label{l:equivalent-definition-of-unit-sphere}
  Let $(E,\nrmdot)$ be a {\normedvectorspace}.
  Let~$\unitS$ be the unit sphere in~$E$.
  Then, $\unitS=\{u\in E\st \nrm{u}=1\}$.
\end{lemma}

\begin{proof}
  Direct consequence of
  Definition~\thref{d:unit-sphere},
  Definition~\thref{d:distance-associated-with-norm},
  Lemma~\thref{l:norm-gives-distance}, and
  Definition~\thref{d:sphere}.
\end{proof}

\begin{lemma}[zero on unit sphere is zero]
  \label{l:zero-on-unit-sphere-is-zero}
  Let $(E,\nrmdot)$ be a {\normedvectorspace}.
  Let~$\unitS$ be the unit sphere in~$E$.
  Let~$F$ be a {\vectorspace}.
  Let $f\in\LEF$ be a linear map from~$E$ to~$F$.
  Then, $f=0_\LEF$ iff $f$~is zero on~$\unitS$.
\end{lemma}

\begin{proof}
  \proofpar{``If''}
  Assume that $f$~is zero on the unit sphere.
  Let $u\in E$ be a vector.
  
  \proofparskip{Case $u=0_E$}
  Then, from
  Lemma~\thref{l:linear-map-preserves-zero},
  $f(u)=f(0_E)=0_F$.
  
  \proofparskip{Case $u\not=0_E$}
  Then, from
  Lemma~\threfc{l:normalization-by-nonzero}{with $\lambda=1$}, and
  Lemma~\thref{l:equivalent-definition-of-unit-sphere},
  $\xi=\frac{u}{\nrm{u}}$ belongs to~$\unitS$.
  Thus, from
  Definition~\threfc{d:space}{%
    scalar multiplication is compatible with field multiplication}, and
  \assume{field properties of~$\matR$},
  we have $u=\nrm{u}\,\xi$.
  Hence, from
  Definition~\threfc{d:linear-map}{homogeneity of degree~1},
  hypothesis, and
  Lemma~\thref{l:times-zero-yields-zero},
  we have
  \begin{equation*}
    f (u)
    = f (\nrm{u} \, \xi)
    = \nrm{u} \, f (\xi)
    = \nrm{u} \, 0_F
    = 0_F.
  \end{equation*}
  
  Therefore, in both cases, $f=0_\LEF$.
  
  \proofparskip{``Only if''}
  Conversely, assume now that $f=0_\LEF$.
  Then, from
  Lemma~\threfc{l:equivalent-definition-of-unit-sphere}{%
    $\unitS$~is a subset of~$E$},
  $f$~is also zero on the unit sphere.
\end{proof}

\begin{lemma}[reverse triangle inequality]
  \label{l:reverse-triangle-inequality}
  Let $(E,\nrmdot)$ be a {\normedvectorspace}.
  Then,
  \begin{equation}
    \label{e:reverse-triangle-inequality}
    \forall u, v \in E,\quad
    \left| \nrm{u} - \nrm{v} \right| \leq \nrm{u - v}.
  \end{equation}
\end{lemma}

\begin{proof}
  Let $u,v\in E$ be vectors.
  Then, from
  Definition~\threfc{d:normed-space}{$\nrmdot$~is a norm}, and
  Definition~\threfc{d:norm}{$\nrmdot$~satisfies triangle inequality},
  we have $\nrm{u}\leq \nrm{u-v}+\nrm{v}$.
  Hence, from
  \assume{ordered field properties of~$\matR$},
  we have $\nrm{u}-\nrm{v}\leq \nrm{u-v}$.
  Thus, from
  Definition~\threfc{d:norm}{%
    $\nrmdot$~is absolutely homogeneous of degree~1 with $\lambda=-1$},
  we have
  \begin{equation*}
    \nrm{v} - \nrm{u}
    \leq \nrm{v - u}
    = \nrm{u - v}.
  \end{equation*}
  Therefore, from
  \assume{properties of the absolute value in~$\matR$},
  we have $\left|\nrm{u}-\nrm{v}\right|\leq \nrm{u-v}$.
\end{proof}

\begin{lemma}[norm is one-Lipschitz continuous]
  \label{l:norm-is-one-lipschitz-continuous}
  Let $(E,\nrmdot)$ be a {\normedvectorspace}.
  Let~$\dst$ be the distance associated with norm~$\nrmdot$.
  Then, $\nrmdot$~is 1-Lipschitz continuous from $(E,\dst)$ to
  $(\matR,|\cdot|)$.
\end{lemma}

\begin{proof}
  Direct consequence of
  Lemma~\thref{l:reverse-triangle-inequality},
  Definition~\thref{d:distance-associated-with-norm}, and
  Definition~\threfc{d:lipschitz-continuity}{with $k=1$}.
\end{proof}

\begin{lemma}[norm is uniformly continuous]
  \label{l:norm-is-uniformly-continuous}
  Let $(E,\nrmdot)$ be a {\normedvectorspace}.
  Let~$\dst$ be the distance associated with norm~$\nrmdot$.
  Then, $\nrmdot$~is uniformly continuous from $(E,\dst)$ to
  $(\matR,|\cdot|)$.
\end{lemma}

\begin{proof}
  Direct consequence of
  Lemma~\thref{l:norm-is-one-lipschitz-continuous}, and
  Definition~\thref{l:lipschitz-continuous-is-uniform-continuous}.
\end{proof}

\begin{lemma}[norm is continuous]
  \label{l:norm-is-continuous}
  Let $(E,\nrmdot)$ be a {\normedvectorspace}.
  Let~$\dst$ be the distance associated with norm~$\nrmdot$.
  Then, $\nrmdot$~is continuous from $(E,\dst)$ to $(\matR,|\cdot|)$.
\end{lemma}

\begin{proof}
  Direct consequence of
  Lemma~\thref{l:norm-is-uniformly-continuous}, and
  Lemma~\thref{l:uniform-continuous-is-continuous}.
\end{proof}

\begin{definition}[linear isometry]
  \label{d:linear-isometry}
  Let $(E,\nEdot)$ and $(F,\nFdot)$ be {\normedvectorspace}s.
  Let $f\in\LEF$ be a linear map from~$E$ to~$F$.
  $f$~is a {\em linear isometry from~$E$ to~$F$} iff
  it preserves the norm:
  \begin{equation}
    \label{e:linear-isometry}
    \forall u \in E,\quad
    \nF{f (u)} = \nE{u}.
  \end{equation}
\end{definition}

\begin{lemma}[identity map is linear isometry]
  \label{l:identity-map-is-linear-isometry}
  Let $(E,\nrmdot)$ be a {\normedvectorspace}.
  Then, the identity map~$\idE$ is a linear isometry.
\end{lemma}

\begin{proof}
  Direct consequence of
  Lemma~\thref{l:identity-map-is-linear-map},
  Definition~\thref{d:identity-map}, and
  Definition~\thref{d:linear-isometry}.
\end{proof}

\begin{definition}[product norm]
  \label{d:product-norm}
  Let $(E,\nEdot)$ and $(F,\nFdot)$ be {\normedvectorspace}s.
  The {\em product norm induced over $\ExF$} is the mapping
  $\nExFdotdot:\ExF\rightarrow\matR$ defined by
  \begin{equation}
    \label{e:product-norm}
    \forall (u, v) \in \ExF,\quad
    \nExF{u}{v} = \nE{u} + \nF{v}.
  \end{equation}
\end{definition}

\begin{remark}
  The mapping~$\nExFdotdot$ will be proved below to be a norm; hence its
  name and notation.
\end{remark}

\begin{remark}
  The norm~$\nExFdotdot$ is the $L^1$-like norm over the product $\ExF$.
  $L^p$-like norms for $p\geq 1$ and~$p=+\infty$ are also possible;
  they are all equivalent norms.
\end{remark}

\begin{lemma}[product is {\normedvectorspace}]
  \label{l:product-is-normed-space}
  Let $(E,\nEdot)$ and $(F,\nFdot)$ be {\normedvectorspace}s.
  Let~$\nExFdotdot$ be the product norm induced over $\ExF$.
  Then, $(\ExF,\nExFdotdot)$ is a {\normedvectorspace}.
\end{lemma}

\begin{proof}
  From
  Lemma~\thref{l:product-is-space},
  $\ExF$, equipped with product vector operations of
  Definition~\thref{d:product-vector-operations},
  is a {\vectorspace}.
  
  Let $(u,v)\in\ExF$ be vectors.
  Assume that $\nExF{u}{v}=0$.
  Then, from
  Definition~\thref{d:product-norm},
  we have $\nE{u}+\nF{v}=0$.
  And, from
  Lemma~\threfc{l:norm-is-nonnegative}{for~$\nEdot$ and~$\nFdot$}, and
  \assume{ordered field properties of~$\matR$},
  we have $\nE{u}=\nF{v}=0$.
  Thus, from
  Definition~\threfc{d:norm}{$\nEdot$ and $\nFdot$ are definite}, and
  Lemma~\threfc{l:product-is-space}{$0_\ExF=(0_E,0_F)$},
  we have $(u,v)=0_\ExF$.
  Hence, $\nExFdotdot$~is definite.
  
  Let $\lambda\in\matK$ be a scalar.
  Let $(u,v)\in\ExF$ be vectors.
  Then, from
  Definition~\threfc{d:product-vector-operations}{scalar multiplication},
  Definition~\thref{d:product-norm},
  Definition~\threfc{d:norm}{%
    $\nEdot$ and~$\nFdot$ are absolutely homogeneous of degree~1}, and
  \assume{field properties of~$\matR$},
  we have
  \begin{multline*}
    \norm{\ExF}{\lambda (u, v)}
    = \nExF{\lambda u}{\lambda v}
    = \nE{\lambda u} + \nF{\lambda v} \\
    = | \lambda | \, \nE{u} + | \lambda | \, \nF{v}
    = | \lambda | (\nE{u} + \nF{v})
    = | \lambda | \, \nExF{u}{v}.
  \end{multline*}
  Hence, $\nExFdotdot$~is absolutely homogeneous of degree~1.
  
  Let $(u,v),(\up,\vp)\in\ExF$ be vectors.
  Then, from
  Definition~\threfc{d:product-vector-operations}{vector addition},
  Definition~\thref{d:product-norm},
  Definition~\threfc{d:norm}{%
    $\nEdot$ and~$\nFdot$ satisfy triangle inequality}, and
  \assume{field properties of~$\matR$},
  we have
  \begin{multline*}
    \norm{\ExF}{(u, v) + (\up, \vp)}
    = \nExF{u + \up}{v + \vp}
    = \nE{u + \up} + \nF{v + \vp} \\
    \leq (\nE{u} + \nE{\up}) + (\nF{v} + \nF{\vp})
    = (\nE{u} + \nF{v}) + (\nE{\up} + \nF{\vp}) \\
    = \nExF{u}{v} + \nExF{\up}{\vp}.
  \end{multline*}
  Hence, $\nExFdotdot$~satisfies triangle inequality.
  
  Therefore, from
  Definition~\thref{d:norm},
  $\nExFdotdot$~is a norm over $\ExF$,
  hence, from
  Definition~\thref{d:normed-space},
  $(\ExF,\nExFdotdot)$ is a {\normedvectorspace}.
\end{proof}

\begin{lemma}[vector addition is continuous]
  \label{l:vector-addition-is-continuous}
  Let $(E,\nEdot)$ be a {\normedvectorspace}.
  From
  Lemma~\thref{l:product-is-normed-space},
  let~$\nExEdotdot$ be the product norm induced over~$\ExE$.
  Let~$\dE$ and~$\dExE$ be the distances associated with norms~$\nEdot$
  and~$\nExEdotdot$.
  Then, the vector addition is continuous from $(\ExE,\dExE)$ to $(E,\dE)$.
\end{lemma}

\begin{proof}
  Let $u,v,\up,\vp\in E$.
  From
  Definition~\thref{d:distance-associated-with-norm},
  Definition~\threfc{d:space}{$(E,+)$ is an abelian group},
  Definition~\threfc{d:norm}{$\nEdot$ satisfies triangle inequality},
  Definition~\thref{d:product-norm}, and
  Definition~\thref{d:product-vector-operations},
  we have
  \begin{multline*}
    \dE (u + v, \up + \vp)
    = \nE{(u + v) - (\up + \vp)}
    = \nE{u - \up + v - \vp} \\
    \leq \nE{u - \up} + \nE{v - \vp}
    = \nExE{u - \up}{v - \vp}
    = \norm{\ExE}{(u, v) - (\up, \vp)} \\
    = \dExE ((u, v), (\up, \vp)).
  \end{multline*}
  
  Let $u,v\in E$ and $\eps>0$.
  Set $\delta=\eps$, then for all $\up,\vp\in E$, we have
  \begin{equation*}
    d_\ExE ((u, v), (\up, \vp)) \leq \delta
    \Implies
    \dE (u + v, \up + \vp) \leq \delta = \eps.
  \end{equation*}
  Therefore, from
  Definition~\thref{d:continuity-in-a-point}, and
  Definition~\thref{d:pointwise-continuity},
  the vector addition is (pointwise) continuous from $(\ExE,\dExE)$ to
  $(E,\dE)$.
\end{proof}

\begin{lemma}[scalar multiplication is continuous]
  \label{l:scalar-multiplication-is-continuous}
  Let $(E,\nrmdot)$ be a {\normedvectorspace}.
  Let~$\dst$ be the distance associated with norm~$\nrmdot$.
  Let $\lambda\in\matK$ be a scalar.
  Then, the scalar multiplication by~$\lambda$ is continuous from $(E,\dst)$
  to itself.
\end{lemma}

\begin{proof}
  \proofparskip{Case $\lambda=0_\matK$}
  Let $u\in E$ and $\eps>0$.
  Set $\delta=1$.
  Then, for all $\up\in E$, from
  Lemma~\thref{l:zero-times-yields-zero}, and
  Definition~\threfc{d:distance}{$\dst$ separates points},
  we have
  \begin{equation*}
    \dst (u, \up) \leq \delta = 1
    \Implies
    \dst (\lambda u, \lambda \up) = \dst (0_E, 0_E) = 0 \leq \eps.
  \end{equation*}
  Therefore, from
  Definition~\thref{d:continuity-in-a-point}, and
  Definition~\thref{d:pointwise-continuity},
  the scalar multiplication by~$0_E$ is (pointwise) continuous from
  $(E,\dst)$ to itself.
  
  \proofparskip{Case $\lambda\not=0_\matK$}
  Let $u,\up\in E$.
  From
  Definition~\thref{d:distance-associated-with-norm},
  Definition~\threfc{d:space}{%
    scalar multiplication is distributive wrt vector addition}, and
  Definition~\threfc{d:norm}{%
    $\nrmdot$ is absolutely homogeneous of degree~1},
  we have
  \begin{equation*}
    \dst (\lambda u, \lambda \up)
    = \nrm{\lambda u - \lambda \up}
    = | \lambda | \nrm{u - \up}
    = | \lambda | \dst (u, \up).
  \end{equation*}
  
  Let $u\in E$ and $\eps>0$.
  From
  \assume{properties of the absolute value over~$\matK$},
  $|\lambda|\not=0$ and we can set $\delta=\frac{\eps}{|\lambda|}$.
  Then for all $\up\in E$, we have
  \begin{equation*}
    \dst (u, \up) \leq \delta
    \Implies
    \dst (\lambda u, \lambda \up) \leq | \lambda | \delta = \eps.
  \end{equation*}
  Therefore, from
  Definition~\thref{d:continuity-in-a-point}, and
  Definition~\thref{d:pointwise-continuity},
  the scalar multiplication is (pointwise) continuous from $(E,\dst)$ to
  itself.
\end{proof}

\subsubsection{Topology}

\begin{remark}
  Since a distance can be defined from a norm, {\normedvectorspace}s can be
  seen as metric spaces, hence as topological spaces too.
  Therefore, the important notions of {\em continuous} linear map and of
  {\em closed} subspace.
\end{remark}

\begin{remark}
  There exists a purely algebraic notion of dual of a {\vectorspace}~$E$: the
  {\vectorspace} of linear forms over~$E$, usually
  denoted~$E^\star=\LEmatK$.
  We focus here on the notion of {\em topological} dual of a
  {\normedvectorspace}~$E$: the {\vectorspace} of {\em continuous} linear
  forms over~$E$, usually denoted~$\Ep=\LcEmatK$.
\end{remark}

\begin{remark}
  When~$W$ is a subset of the set~$X$, and~$f$ a mapping from~$X$ to~$Y$,
  the notation $f(W)$ denotes the subset of~$Y$ made of the images of
  elements of~$W$.
  Applied to a norm on a vector space, when~$X$ is a subset of a normed
  space~$(E,\nrmdot)$, the notation $\nrm{X}$ denotes the subset of~$\matR$
  of values taken by norm~$\nrmdot$ on vectors of~$X$:
  \begin{equation*}
    \nrm{X} = \left\{ \nrm{u} \st u \in X \right\}.
    \end{equation*}
\end{remark}

\paragraph{Continuous linear map}

\begin{lemma}[norm of image of unit vector]
  \label{l:norm-of-image-of-unit-vector}
  Let $(E,\nEdot)$ and $(F,\nFdot)$ be {\normedvectorspace}s.
  Let~$\unitS$ be the unit sphere in~$E$.
  Let $f\in\LEF$ be a linear map from~$E$ to~$F$.
  Let $u\in E$ be a vector.
  Assume that $u\not=0_E$.
  Then, $\frac{u}{\nE{u}}$ belongs to~$\unitS$ and
  \begin{equation}
    \label{e:norm-of-image-of-unit-vector}
    \nF{f \left(\frac{u}{\nE{u}}\right)} = \frac{\nF{f (u)}}{\nE{u}}.
  \end{equation}
\end{lemma}

\begin{proof}
  From
  Definition~\threfc{d:norm}{$\nEdot$ is definite, contrapositive},
  we have $\nE{u}\not=0$.
  Thus, from
  Lemma~\thref{l:norm-is-nonnegative}, and
  \assume{field properties of~$\matR$},
  $\frac{1}{\nE{u}}\geq 0$.
  Let $\xi=\frac{u}{\nE{u}}$.
  Then, from
  Lemma~\threfc{l:normalization-by-nonzero}{with $\lambda=1$},
  we have $\nE{\xi}=1$.
  Hence, $\xi$ belongs to~$\unitS$.
  
  From
  Definition~\threfc{d:linear-map}{homogeneity of degree~1},
  Definition~\threfc{d:norm}{%
    $\nFdot$ is absolutely homogeneous of degree~1}, and
  Lemma~\thref{l:norm-is-nonnegative},
  we have
  \begin{equation*}
    \nF{f (\xi)}
    = \nF{f \left(\frac{u}{\nE{u}}\right)}
    = \nF{\frac{f (u)}{\nE{u}}}
    = \frac{\nF{f (u)}}{\nE{u}}.
  \end{equation*}
\end{proof}

\begin{lemma}[norm of image of unit sphere]
  \label{l:norm-of-image-of-unit-sphere}
  Let $(E,\nEdot)$ and $(F,\nFdot)$ be {\normedvectorspace}s.
  Let~$\unitS$ be the unit sphere in~$E$.
  Let $f\in\LEF$ be a linear map from~$E$ to~$F$.
  Then,
  \begin{equation}
    \label{e:}
    \nF{f (\unitS)} =
    \left\{
      \left. \frac{\nF{f (u)}}{\nE{u}} \,\right|\, u \in E,\, u \not= 0_E
    \right\}.
  \end{equation}
\end{lemma}

\begin{proof}
  From
  Definition~\threfc{d:norm}{$\nEdot$ is definite, contrapositive}, and
  \assume{field properties of~$\matR$},
  let $g:E\rightarrow\matR$ be the mapping defined by $g(0_E)=0$, and for all
  $u\in E\backslash\{0_E\}$, $g(u)=\frac{\nF{f(u)}}{\nE{u}}$.
  
  Let $\xi\in\unitS$ be a unit vector.
  From
  Lemma~\thref{l:equivalent-definition-of-unit-sphere},
  $\nE{\xi}=1$.
  Then, from
  Lemma~\threfc{l:norm-preserves-zero}{contrapositive},
  $\xi\not=0_E$.
  Thus, from
  \assume{field properties of~$\matR$},
  we have
  \begin{equation*}
    \nF{f (\xi)}
    = \frac{\nF{f (\xi)}}{\nE{\xi}}
    = g (\xi)
    \in g (E\backslash\{0_E\}).
  \end{equation*}
  Hence, $\nF{f(\unitS)}\subset g(E\backslash\{0_E\})$.
  
  Let $u\in E$ be a vector.
  Assume that $u\not=0_E$.
  Then, from
  Lemma~\thref{l:norm-of-image-of-unit-vector},
  $\xi=\frac{u}{\nE{u}}$ belongs to~$\unitS$ and
  \begin{equation*}
    g(u)
    = \frac{\nF{f (u)}}{\nE{u}}
    = \nF{f (\xi)}
    \in f (\unitS).
  \end{equation*}
  Hence, $g(E\backslash\{0_E\})\subset\nF{f(\unitS)}$.
  
  Therefore, $\nF{f(\unitS)}=g(E\backslash\{0_E\})$.
\end{proof}

\begin{definition}[operator norm]
  \label{d:operator-norm}
  Let $(E,\nEdot)$ and $(F,\nFdot)$ be {\normedvectorspace}s.
  Let $f\in\LEF$ be a linear map from~$E$ to~$F$.
  The {\em operator norm on $\LEF$ induced by norms on~$E$ and~$F$}
  is the mapping $\PrenEF:\LEF\rightarrow\matRbar$ defined by
  \begin{equation}
    \label{e:operator-norm}
    \prenEF{f} =
    \sup
    \left\{
      \left. \frac{\nF{f (u)}}{\nE{u}} \,\right|\, u \in E,\, u \not= 0_E
    \right\}.
  \end{equation}
\end{definition}

\begin{remark}
  When restricted to continuous linear maps, the mapping~$\PrenEF$ will be
  proved below to be a norm; hence its name.
\end{remark}

\begin{lemma}[equivalent definition of operator norm]
  \label{l:equivalent-definition-of-operator-norm}
  Let $(E,\nEdot)$ and $(F,\nFdot)$ be {\normedvectorspace}s.
  Let~$\unitS$ be the unit sphere in~$E$.
  Let $f\in\LEF$ be a linear map from~$E$ to~$F$.
  Then,
  \begin{equation}
    \label{e:equivalent-definition-of-operator-norm}
    \prenEF{f} = \sup(\nF{f(\unitS)}).
  \end{equation}
\end{lemma}

\begin{proof}
  Direct consequence of
  Definition~\thref{d:operator-norm}, and
  Lemma~\thref{l:norm-of-image-of-unit-sphere}.
\end{proof}

\begin{lemma}[operator norm is nonnegative]
  \label{l:operator-norm-is-nonnegative}
  Let $(E,\nEdot)$ and $(F,\nFdot)$ be {\normedvectorspace}s.
  Then, $\PrenEF$ is nonnegative.
\end{lemma}

\begin{proof}
  Let $f\in\LEF$ be a linear map from~$E$ to~$F$.
  Then, from
  Lemma~\thref{l:equivalent-definition-of-operator-norm},
  we have $\prenEF{f}=\sup(\nF{f(\unitS)})$.
  Let $\xi\in\unitS$ be a unit vector.
  Then, from
  Lemma~\thref{l:norm-is-nonnegative},
  $\nF{f(\xi)}$ is nonnegative.
  Therefore, from
  Definition~\threfc{d:supremum}{%
    $\prenEF{f}$~is an upper bound for $\nF{f(\unitS)}$},
  $\prenEF{f}$ is nonnegative too.
\end{proof}

\begin{definition}[bounded linear map]
  \label{d:bounded-linear-map}
  Let $(E,\nEdot)$ and $(F,\nFdot)$ be {\normedvectorspace}s.
  A linear map~$f$ from~$E$ to~$F$ is {\em bounded} iff
  \begin{equation}
    \label{e:bounded-linear-map}
    \exists C \geq 0,\,
    \forall u \in E,\quad
    \nF{f (u)} \leq C \, \nE{u}.
  \end{equation}
  Then, $C$~is called {\em continuity constant of~$f$}.
\end{definition}

\begin{definition}[linear map bounded on unit ball]
  \label{d:linear-map-bounded-on-unit-ball}
  Let $(E,\nEdot)$ and $(F,\nFdot)$ be {\normedvectorspace}s.
  Let~$\unitcB$ be the closed unit ball in~$E$.
  A linear map~$f$ from~$E$ to~$F$ is
  {\em bounded on the closed unit ball} iff
  there exists an upper bound for $\nF{f(\unitcB)}$, {\ie}
  \begin{equation}
    \label{e:linear-map-bounded-on-unit-ball}
    \exists C \geq 0,\,
    \forall \xi \in \unitcB,\quad
    \nF{f (\xi)} \leq C.
  \end{equation}
\end{definition}

\begin{definition}[linear map bounded on unit sphere]
  \label{d:linear-map-bounded-on-unit-sphere}
  Let $(E,\nEdot)$ and $(F,\nFdot)$ be {\normedvectorspace}s.
  Let~$\unitS$ be the unit sphere in~$E$.
  A linear map~$f$ from~$E$ to~$F$ is
  {\em bounded on the unit sphere} iff
  there exists an upper bound for $\nF{f(\unitS)}$, {\ie}
  \begin{equation}
    \label{e:linear-map-bounded-on-unit-sphere}
    \exists C \geq 0,\,
    \forall \xi \in \unitS,\quad
    \nF{f (\xi)} \leq C.
  \end{equation}
\end{definition}

\begin{theorem}[continuous linear map]
  \label{t:continuous-linear-map}
  Let $(E,\nEdot)$ and $(F,\nFdot)$ be {\normedvectorspace}s.
  Let~$f\in\LEF$ be a linear map from~$E$ to~$F$.
  Then, the following propositions are equivalent:
  \begin{enumerate}
  \item
    \label{i:cont-zero}
    $f$~is continuous in~$0_E$;
  \item
    \label{i:cont}
    $f$~is continuous;
  \item
    \label{i:unif-cont}
    $f$~is uniformly continuous;
  \item
    \label{i:lip-cont}
    $f$~is Lipschitz continuous;
  \item
    \label{i:bounded}
    $f$~is bounded;
  \item
    \label{i:finite-norm}
    $\prenEF{f}$ is finite;
  \item
    \label{i:bounded-unit-sphere}
    $f$~is bounded on the unit sphere.
  \item
    \label{i:bounded-unit-ball}
    $f$~is bounded on the closed unit ball.
  \end{enumerate}
\end{theorem}

\begin{proof}
  Let~$\unitS$ be the unit sphere in~$E$.
  Let~$\unitcB$ be the closed unit ball in~$E$.
  
  \proofparskip{\ref{i:bounded} implies~\ref{i:lip-cont}}
  Assume that~$f$ is bounded.
  From
  Definition~\thref{d:bounded-linear-map},
  let $C\geq 0$ such that, for all $u\in E$, we have $\nF{f(u)}\leq
  C\,\nE{u}$.
  Let $u,v\in E$ be vectors.
  Then, from
  Definition~\thref{d:vector-subtraction},
  Definition~\thref{d:linear-map}, and
  hypothesis, we have
  \begin{equation*}
    \nF{f(u) - f(v)} = \nF{f(u - v)} \leq C \, \nE{u - v}.
  \end{equation*}
  Hence, from
  Definition~\thref{d:lipschitz-continuity},
  $f$~is $C$-Lipschitz continuous.
  
  \proofparskip{\ref{i:lip-cont} implies~\ref{i:unif-cont}}
  Assume that~$f$ is Lipschitz continuous.
  Then, from
  Lemma~\thref{l:lipschitz-continuous-is-uniform-continuous},
  $f$~is uniformly continuous.
  
  \proofparskip{\ref{i:unif-cont} implies~\ref{i:cont}}
  Assume that~$f$ is uniformly continuous.
  Then, from
  Lemma~\thref{l:uniform-continuous-is-continuous},
  $f$~is (pointwise) continuous.
  
  \proofparskip{\ref{i:cont} implies~\ref{i:cont-zero}}
  Assume that~$f$ is (pointwise) continuous.
  Then, from
  Definition~\thref{d:pointwise-continuity},
  $f$~is continuous in~$0_E$.
  
  \proofparskip{\ref{i:cont-zero} implies~\ref{i:bounded-unit-ball}}
  Assume now that~$f$ is continuous in~$0_E$.
  Let $\eps=1>0$.
  Then, from
  Definition~\thref{d:continuity-in-a-point}, and
  Lemma~\thref{l:linear-map-preserves-zero},
  let $\delta>0$ such that, for all $u\in E$, $\nE{u-0_E}=\nE{u}\leq\delta$
  implies $\nF{f(u)-f(0_E)}=\nF{f(u)}\leq 1$.
  Let $C=\frac{1}{\delta}>0\geq 0$.
  Let $\xi\in\unitcB$ be a vector in the unit ball.
  From
  Lemma~\thref{l:equivalent-definition-of-closed-unit-ball},
  $\nE{\xi}\leq 1$.
  Then, from
  Definition~\threfc{d:norm}{%
    $\nEdot$ is absolutely homogeneous of degree~1}, and
  \assume{ordered field properties of~$\matR$},
  we have $\nE{\delta\xi}\leq\delta\,\nE{\xi}\leq\delta$.
  Thus, from
  Definition~\threfc{d:norm}{$\nFdot$ is absolutely homogeneous of degree~1},
  Definition~\threfc{d:linear-map}{homogeneity of degree~1},
  \assume{ordered field properties of~$\matR$}, and
  hypothesis, we have
  $\nF{f(\xi)}=\frac{1}{\delta}\,\nF{f(\delta\xi)}\leq\frac{1}{\delta}=C$.
  Hence, from
  Definition~\thref{d:linear-map-bounded-on-unit-ball},
  $f$~is bounded on the unit ball.
  
  \proofparskip{\ref{i:bounded-unit-ball} implies~\ref{i:bounded-unit-sphere}}
  Assume now that~$f$ is bounded on the unit ball.
  From
  Definition~\thref{d:linear-map-bounded-on-unit-ball}, and
  Lemma~\thref{l:equivalent-definition-of-closed-unit-ball},
  let $C\geq 0$ such that for all $\xi\in\unitcB$, $\nF{f(\xi)}\leq C$.
  Let $\xi\in\unitS$ be a unit vector.
  Then, from
  Lemma~\thref{l:equivalent-definition-of-unit-sphere}, and
  Lemma~\thref{l:equivalent-definition-of-closed-unit-ball},
  we also have $\xi\in\unitcB$.
  Thus, from
  hypothesis,
  $\nF{f(\xi)}\leq C$.
  Hence, from
  Definition~\threfc{d:linear-map-bounded-on-unit-sphere}{%
    with same constant~$C$},
  $f$~is bounded on the unit sphere.
  
  \proofparskip{\ref{i:bounded-unit-sphere} implies~\ref{i:finite-norm}}
  Assume then that~$f$ is bounded on the unit sphere.
  Then, from
  Definition~\thref{d:linear-map-bounded-on-unit-sphere},
  there exists a finite upper bound~$C\geq 0$ for~$\nF{f(\unitS)}$.
  Hence, from
  Lemma~\thref{l:finite-supremum},
  $\sup(\nF{f(\unitS)})$ is finite, and from
  Lemma~\thref{l:equivalent-definition-of-operator-norm},
  $\prenEF{f}$ is finite.
  
  \proofparskip{\ref{i:finite-norm} implies~\ref{i:bounded}}
  Assume finally that $\prenEF{f}$ is finite.
  Let $C=\prenEF{f}$.
  Then, from
  Lemma~\thref{l:operator-norm-is-nonnegative},
  $C$~is nonnegative.
  Let $u\in E$ be a vector.
  
  \proofparskip{Case $u=0_E$}
  Then, from
  Lemma~\thref{l:linear-map-preserves-zero},
  $f(u)=f(0_E)=0_F$.
  Hence, from
  Lemma~\threfc{l:norm-preserves-zero}{for $\nFdot$ and $\nEdot$}, and
  \assume{ordered field properties of~$\matR$},
  we have
  \begin{equation*}
    \nF{f (u)}
    = 0
    \leq 0
    = C \, 0
    = C \, \nE{u}.
  \end{equation*}
  
  \proofparskip{Case $u\not=0_E$}
  Then, from
  Definition\threfc{d:norm}{$\nEdot$ is definite, contrapo\-si\-ti\-ve},
  $\nE{u}\not=0$.
  Thus, from
  \assume{field properties of~$\matR$},
  Definition~\thref{d:operator-norm}, and
  Definition~\threfc{d:supremum}{%
    $\prenEF{f}$ is an upper bound for
    $\left\{\left.\frac{\nF{f(u)}}{\nE{u}}\,\right|\,
      u\in E,\,u\not=0_E\right\}$},
  we have
  \begin{equation*}
    \nF{f (u)}
    = \frac{\nF{f (u)}}{\nE{u}} \, \nE{u}
    \leq \prenEF{f} \, \nE{u}
    = C \, \nE{u}.
  \end{equation*}
  Hence, from
  Definition~\thref{d:bounded-linear-map},
  $f$~is bounded.
  
  Therefore, we have
  $\ref{i:bounded}\implies\ref{i:lip-cont}\implies\ref{i:unif-cont}
  \implies\ref{i:cont}\implies\ref{i:cont-zero}
  \implies\ref{i:bounded-unit-ball}\implies\ref{i:bounded-unit-sphere}
  \implies\ref{i:finite-norm}\implies\ref{i:bounded}$,
  hence all properties are equivalent.
\end{proof}

\begin{definition}[set of continuous linear maps]
  \label{d:set-of-continuous-linear-maps}
  Let $(E,\nEdot)$ and $(F,\nFdot)$ be {\normedvectorspace}s.
  The {\em set of continuous linear maps from~$E$ to~$F$} is denoted
  $\LcEF$.
\end{definition}

\begin{lemma}[finite operator norm is continuous]
  \label{l:finite-operator-norm-is-continuous}
  Let $(E,\nEdot)$ and $(F,\nFdot)$ be {\normedvectorspace}s.
  Let $f\in\LEF$ be a linear map from~$E$ to~$F$.
  Then, $f$~belongs to $\LcEF$ ({\ie} $f$~is continuous) iff
  $\prenEF{f}$ is finite.
  Moreover, let~$\unitcB$ and~$\unitS$ be the closed unit ball and the unit
  sphere in~$E$, and let $C\geq 0$, then we have the following equivalences:
  \begin{eqnarray}
    \label{e:finite-operator-norm-is-continuous}
    \prenEF{f} \leq C
    & \equiv &
    C \mbox{ is an upper bound for }
    \left\{
      \left. \frac{\nF{f (u)}}{\nE{u}} \, \right | \, u \in E,\, u \not= 0_E
    \right\} \\
    \nonumber
    & \equiv &
    C \mbox{ is a continuity constant for } f \\
    \nonumber
    & \equiv &
    C \mbox{ is an upper bound for } \nF{f (\unitcB)} \\
    \nonumber
    & \equiv &
    C \mbox{ is an upper bound for } \nF{f (\unitS)}.
  \end{eqnarray}
\end{lemma}

\begin{proof}
  Direct consequences of
  Definition~\thref{d:set-of-continuous-linear-maps},
  Theorem~\threfc{t:continuous-linear-map}{%
    $\ref{i:cont}\implies\ref{i:finite-norm}$}.
  Definition~\threfc{d:supremum}{%
    $\prenEF{f}$ is the least upper bound of\\
    $\left\{\left.\frac{\nF{f(u)}}{\nE{u}}\,\right|\,
      u\in E,\,u\not=0_E\right\}$},
  Definition~\thref{d:bounded-linear-map},
  Definition~\thref{d:linear-map-bounded-on-unit-ball}, and
  Definition~\thref{d:linear-map-bounded-on-unit-sphere}.
\end{proof}

\begin{lemma}[linear isometry is continuous]
  \label{l:linear-isometry-is-continuous}
  Let $(E,\nEdot)$ and $(F,\nFdot)$ be {\normedvectorspace}s.
  Let $f\in\LEF$ be a linear map from~$E$ to~$F$.
  Assume that~$f$ is a linear isometry from~$E$ to~$F$.
  Then, $f$~belongs to $\LcEF$ ({\ie} $f$~is continuous).
\end{lemma}

\begin{proof}
  Let~$\unitS$ be the unit sphere in~$E$.
  Let $\xi\in\unitS$ be a unit vector.
  Then, from
  Definition~\thref{d:linear-isometry},
  we have $\nF{f(\xi)}=\nE{\xi}=1\leq 1$.
  Hence, from
  Lemma~\threfc{l:finite-operator-norm-is-continuous}{%
    1~is an upper bound for $\nF{f(\unitS)}$},
  $f$~belongs to $\LcEF$.
\end{proof}

\begin{lemma}[identity map is continuous]
  \label{l:identity-map-is-continuous}
  Let $(E,\nEdot)$ be a {\normedvectorspace}.
  Then, the identity map~$\idE$ belongs to $\LcEE$ ({\ie} $\idE$~is
  continuous).
\end{lemma}

\begin{proof}
  Direct consequence of
  Lemma~\thref{l:identity-map-is-linear-isometry}, and
  Lemma~\thref{l:linear-isometry-is-continuous}.
\end{proof}

\begin{theorem}[{\normedvectorspace} of continuous linear maps]
  \label{t:normed-space-of-continuous-linear-maps}
  Let $(E,\nEdot)$ and $(F,\nFdot)$ be {\normedvectorspace}s.
  Let $\tnEFdot$ be the restriction of~$\PrenEF$ to continuous linear maps.
  Then, $(\LcEF,\tnEFdot)$ is a {\normedvectorspace}.
\end{theorem}

\begin{proof}
  Let~$\unitS$ be the unit sphere in~$E$.
  From
  Definition~\thref{d:set-of-continuous-linear-maps},
  $\LcEF$ is obviously a subset of $\LEF$.
  
  Let $f\in\LcEF$ be a continuous linear map from~$E$ to~$F$.
  Then, from
  Lemma~\thref{l:finite-operator-norm-is-continuous},
  $\tnEF{f}$ is finite.
  Hence, $\tnEFdot$ is a mapping from $\LcEF$ to~$\matR$.
  
  From
  Definition~\threfc{d:linear-map-bounded-on-unit-sphere}{with $C=0$}, and
  Lemma~\threfc{l:finite-operator-norm-is-continuous}{%
    upper bound for $\nF{0_\LEF(\unitS)}$},
  $0_\LEF$ belongs to $\LcEF$.
  
  Let $f\in\LcEF$ be a continuous linear map from~$E$ to~$F$.
  Assume that $\tnEF{f}=0$.
  Let $\xi\in\unitS$ be a unit vector.
  Then, from
  Lemma~\thref{l:norm-is-nonnegative},
  Lemma~\thref{l:equivalent-definition-of-operator-norm}, and
  Definition~\threfc{d:supremum}{%
    $\tnEF{f}$ is an upper bound for $\nF{f(\unitS)}$},
  we have
  \begin{equation*}
    0
    \leq \nF{f (\xi)}
    \leq \tnEF{f}
    = 0.
  \end{equation*}
  Thus, $\nF{f(\xi)}=0$, and from
  Definition~\threfc{d:norm}{$\nFdot$ is definite},
  $f(\xi)=0_F$.
  Hence, from
  Lemma~\thref{l:zero-on-unit-sphere-is-zero},
  $f=0_\LEF$, and $\tnEFdot$ is definite.
  
  Let $\lambda\in\matK$ be a scalar.
  Let $f\in\LcEF$ be a continuous linear map from~$E$ to~$F$.
  Let $\xi\in\unitS$ be a unit vector.
  Then, from
  Definition~\threfc{d:inherited-vector-operations}{%
    scalar multiplication}, and
  Definition~\threfc{d:norm}{$\nFdot$ is absolutely homogeneous of degree~1},
  we have
  \begin{equation*}
    \nF{(\lambda f) (\xi)}
    = \nF{\lambda f (\xi)}
    = | \lambda | \nF{f (\xi)}.
  \end{equation*}
  Thus, from
  Lemma~\thref{l:equivalent-definition-of-operator-norm},
  Lemma~\thref{l:supremum-is-positive-scalar-multiplicative}, and
  nonnegativeness of absolute value, we have
  \begin{equation*}
    \prenEF{\lambda f}
    = \sup (\nF{(\lambda f) (\unitS)})
    = | \lambda | \sup (\nF{f (\unitS)})
    = | \lambda | \, \tnEF{f}.
  \end{equation*}
  Then, from
  Lemma~\threfc{l:finite-operator-norm-is-continuous}{%
    $\prenEF{\lambda f}$ is finite},
  $\lambda f$ belongs to $\LcEF$.
  Hence, $\tnEFdot$ is absolutely homogeneous of degree~1, and $\LcEF$ is
  closed under scalar multiplication.
  
  Let $f,g\in\LcEF$ be continuous linear maps from~$E$ to~$F$.
  Let $\xi\in\unitS$ be a unit vector.
  Then, from
  Definition~\threfc{d:inherited-vector-operations}{vector addition},
  Definition~\threfc{d:norm}{$\nFdot$ satisfies triangle inequality},
  Lemma~\thref{l:equivalent-definition-of-operator-norm},
  Definition~\threfc{d:supremum}{%
    $\tnEF{f}$, resp. $\tnEF{g}$, is an upper bounds for $\nF{f(\unitS)}$,
    resp. $\nF{g(\unitS)}$}, and
  \assume{field properties of~$\matR$},
  we have
  \begin{equation*}
   \nF{(f + g) (\xi)}
    = \nF{f (\xi) + g (\xi)}
    \leq \nF{f (\xi)} + \nF{g (\xi)}
    \leq \tnEF{f} + \tnEF{g}.
  \end{equation*}
  Thus, from
  Lemma~\threfc{l:finite-operator-norm-is-continuous}{%
    $\tnEF{f}+\tnEF{g}$ is a finite upper bound for $\nF{(f+g)(\unitS)}$},
  $f+g$ belongs to $\LcEF$ and
  \begin{equation*}
    \tnEF{f + g} \leq \tnEF{f} + \tnEF{g}.
  \end{equation*}
  Hence, $\LcEF$ is closed under vector addition and $\tnEFdot$ satisfies
  triangle inequality.
  
  Therefore, from
  Lemma~\thref{l:closed-under-vector-operations-is-subspace},
  Definition~\thref{d:norm}, and
  Definition~\thref{d:normed-space},
  $\LcEF$ is a subspace of $\LEF$, $\tnEFdot$ is a norm over $\LEF$, and
  $(\LcEF,\tnEFdot)$ is a {\normedvectorspace}.
\end{proof}

\begin{lemma}[operator norm estimation]
  \label{l:operator-norm-estimation}
  Let $(E,\nEdot)$ and $(F,\nFdot)$ be {\normedvectorspace}s.
  Then,
  \begin{equation}
    \label{e:operator-norm-estimation}
    \forall f \in \LcEF,\,
    \forall u \in E,\quad
    \nF{f (u)} \leq \tnEF{f} \, \nE{u}.
  \end{equation}
\end{lemma}

\begin{proof}
  Let $f\in\LcEF$ be a continuous linear map from~$E$ to~$F$.
  Then, from
  Theorem~\thref{t:normed-space-of-continuous-linear-maps},
  $\tnEF{f}$ is finite.
  Let $u\in E$ be a vector.
  
  \proofparskip{Case $u=0_E$}
  Then, from
  Lemma~\thref{l:linear-map-preserves-zero},
  Lemma~\threfc{l:norm-preserves-zero}{for $\nFdot$ and $\nEdot$}, and
  \assume{ordered field properties of~$\matR$},
  we have
  \begin{equation*}
    \nF{f (u)}
    = \nF{f (0_E)}
    = \nF{0_F}
    = 0
    \leq 0
    = \tnEF{f} \, 0
    = \tnEF{f} \, \nE{0_E}
    = \tnEF{f} \, \nE{u}.
  \end{equation*}
  
  \proofparskip{Case $u\not=0_E$}
  Then, from
  Definition~\thref{d:operator-norm}, and
  Definition~\threfc{d:supremum}{%
    $\tnEF{f}$ is an upper bound for
    $\left\{\left.\frac{\nF{f(u)}}{\nE{u}}\,\right|\,
      u\in E,\,u\not=0_E\right\}$},
  we have $\frac{\nF{f(u)}}{\nE{u}}\leq\tnEF{f}$.
  From
  Definition~\threfc{d:norm}{$\nEdot$ is definite, contrapositive}, and
  Lemma~\threfc{l:norm-is-nonnegative}{for $\nEdot$},
  $\nE{u}>0$.
  Hence, from
  \assume{ordered field properties of~$\matR$},
  $\nF{f(u)}\leq\tnEF{f}\,\nE{u}$.
\end{proof}

\begin{lemma}[continuous linear maps have closed kernel]
  \label{l:continuous-linear-maps-have-closed-kernel}
  Let $(E,\nEdot)$ and $(F,\nFdot)$ be {\normedvectorspace}s.
  Let $f\in\LcEF$ be a continuous linear map from~$E$ to~$F$.
  Then, $\Ker{f}$ is closed in~$E$.
\end{lemma}

\begin{proof}
  Direct consequence of
  Definition~\threfc{d:kernel}{$\Ker{f}=f^{-1}\{0_F\}$},
  Lemma~\threfc{l:singleton-is-closed}{$\{0_F\}$ is closed}, and
  \assume{preimages of closed subsets by continuous mappings are closed}.
\end{proof}

\begin{lemma}[compatibility of composition with continuity]
  \label{l:compatibility-of-composition-with-continuity}
  Let $(E,\nEdot)$, $(F,\nFdot)$ and $(G,\nGdot)$ be {\normedvectorspace}s.
  Then,
  \begin{equation}
    \label{e:compatibility-of-composition-with-continuity}
    \forall f \in \LcEF,\,
    \forall g \in \LcFG,\quad
    g \circ f \in \LcEG
    \Conj
    \tnEG{g \circ f} \leq \tnFG{g} \, \tnEF{f}.
  \end{equation}
\end{lemma}

\begin{proof}
  Let~$\unitS$ be the unit sphere in~$E$.
  Let $f\in\LcEF$ and $g\in\LcFG$ be continuous linear maps.
  Then, from
  Lemma~\thref{l:composition-of-linear-maps-is-bilinear},
  $g\circ f$ belongs to $\LEF$.
  Let $\xi\in\unitS$ be a unit vector.
  Then, from
  \assume{the definition of composition of functions},
  Lemma~\threfc{l:operator-norm-estimation}{for~$g$ and~$f$},
  Lemma~\threfc{l:equivalent-definition-of-unit-sphere}{$\nE{\xi}=1$}, and
  \assume{field properties of~$\matR$},
  we have
  \begin{equation*}
    \nF{(g \circ f) (\xi)}
    = \nF{g (f (\xi))}
    \leq \tnFG{g} \, \nF{f (\xi)}
    \leq \tnFG{g} \, \tnEF{f} \, \nE{\xi}
    = \tnFG{g} \, \tnEF{f}.
  \end{equation*}
  Hence, from
  Lemma~\threfc{l:finite-operator-norm-is-continuous}{%
    $\tnFG{g}\,\tnEF{f}$ is a finite upper bound for
    $\nF{(g\circ f)(\unitS)}$},
  $g\circ f$ belongs to $\LcEF$ and
  \begin{equation*}
    \tnEG{g \circ f} \leq \tnFG{g} \, \tnEF{f}.
  \end{equation*}
\end{proof}



\begin{lemma}[complete {\normedvectorspace} of continuous linear maps]
  \label{l:complete-normed-space-of-continuous-linear-maps}
  Let $(E,\nEdot)$ and $(F,\nFdot)$ be {\normedvectorspace}s.
  If $(F,\nFdot)$ is complete, then the {\normedvectorspace} $\LcEF$ is also
  complete ({\ie} they are both Banach {\vectorspace}s).
\end{lemma}

\begin{proof}
  \proofparskip{Case $E=\{0_E\}$}
  Then, from
  Lemma~\thref{l:linear-map-preserves-zero},
  $\LcEF$ is also the singleton $\{0_{\LEF}\}$.
  From
  Definition~\thref{d:complete-metric-space}, and
  Lemma~\thref{l:stationary-sequence-is-convergent},
  singletons are trivially complete metric spaces since they possess only one
  sequence which is constant, hence stationary, hence convergent.
  Therefore, $\LcEF$ is complete.
  
  \proofparskip{Case $E\not=\{0_E\}$}
 
  \proofparskip{Pointwise limit}
  Let~$\dEF$ be the distance associated with norm~$\tnEFdot$.
  From Lemma~\thref{l:norm-gives-distance},
  $(\LcEF,\dEF)$ is a metric space.
  Let $(f_n)_{n\in\matN}$ be a Cauchy sequence in $(\LcEF,\dEF)$.
  Then, from
  Definition~\thref{d:cauchy-sequence},
  Definition~\thref{d:distance-associated-with-norm},
  Theorem~\threfc{t:normed-space-of-continuous-linear-maps}{%
    definition of~$\tnEFdot$},
  Definition~\thref{d:operator-norm}, and
  Lemma~\threfc{l:finite-operator-norm-is-continuous}{%
    $\tnEF{f_p-f_q}$ is lower than or equal to continuity constants},
  we have
  \begin{equation}
    \label{e:cauchy-sequence-of-continuous-linear-maps}
    \forall \eps > 0,\,
    \exists N \in \matN,\,
    \forall p, q \in \matN,\quad
    p, q \geq N \Implies
    \forall u \in E,\quad
    \nF{f_p (u) - f_q (u)} \leq \eps \nE{u}.
  \end{equation}
  Let $u\in E$.
  \proofparskip{Case $u\not=0_E$}
  Let $\epsp>0$.
  From
  Definition~\threfc{d:norm}{$\nEdot$ is definite, contrapositive},
  $\nE{u}\not=0$ and from
  Equation~\eqref{e:cauchy-sequence-of-continuous-linear-maps} with
  $\eps=\frac{\epsp}{\nE{u}}$, we have
  \begin{equation*}
    \exists N \in \matN,\,
    \forall p, q \in \matN,\quad
    p, q \geq N \Implies
    \nF{f_p (u) - f_q (u)} \leq \epsp.
  \end{equation*}
  Thus, from
  Definition~\threfc{d:cauchy-sequence}{%
    $(f_n(u))_{n\in\matN}$ is a Cauchy sequence}, and
  Definition~\threfc{d:complete-subset}{$F$ is complete},
  let $f(u)=\lim_{n\rightarrow+\infty}f_n(u)$ be the limit in~$F$.
  
  \noindent
  \proofparskip{Case $u=0_E$}
  Since from
  Lemma~\thref{l:linear-map-preserves-zero},
  we have for all $n\in\matN$, $f_n(0_E)=0_F$, let
  $f(0_E)=0_F=\lim_{n\rightarrow+\infty}f_n(0_E)$.
  
  \proofparskip{Linearity}
  Let $u,v\in E$ and $\lambda,\mu\in\matK$.
  From
  Definition~\threfc{d:linear-map}{%
    for all $n\in\matN$, $f_n$ is a linear map},
  Lemma~\thref{l:vector-addition-is-continuous},
  Lemma~\thref{l:scalar-multiplication-is-continuous}, and
  Lemma~\thref{l:compatibility-of-limit-with-continuous-functions},
  we have
  \begin{multline*}
    f (\lambda u + \mu v)
    = \lim_{n \rightarrow +\infty} f_n (\lambda u + \mu v)
    = \lim_{n \rightarrow +\infty} (\lambda f_n (u) + \mu f_n (v)) \\
    = \lim_{n \rightarrow +\infty} (\lambda f_n (u))
    + \lim_{n \rightarrow +\infty} (\mu f_n (v))
    = \lambda \lim_{n \rightarrow +\infty} f_n (u)
    + \mu \lim_{n \rightarrow +\infty} f_n (v)
    = \lambda f (u) + \mu f (v).
  \end{multline*}
  Hence, from
  Lemma~\thref{l:linear-map-preserves-linear-combinations},
  $f$ belongs to~$\LEF$.
  
  \proofparskip{Continuity}
  In Equation~\eqref{e:cauchy-sequence-of-continuous-linear-maps}, we
  consider a fixed $u\in E$ and we take the limit when~$q$ goes
  to~$+\infty$.
  Thus, from
  Lemma~\thref{l:norm-is-continuous},
  Lemma~\thref{l:compatibility-of-limit-with-continuous-functions}, and
  Definition~\threfc{d:inherited-vector-operations}{on~$\LEF$},
  we have
  \begin{equation}
    \label{e:fp-minus-f-is-bounded}
    \forall \eps > 0,\,
    \exists N \in \matN,\,
    \forall p \in \matN,\quad
    p \geq N \Implies
    \forall u \in E,\quad
    \nF{(f_p - f) (u)} \leq \eps \nE{u}.
  \end{equation}
  Hence, from
  Definition~\threfc{d:bounded-linear-map}{$f_p-f$ is bounded},
  Theorem~\threfc{t:continuous-linear-map}{$f_p-f$ is continuous},
  Theorem~\threfc{t:normed-space-of-continuous-linear-maps}{%
    $\LcEF$ is a {\vectorspace}}, and
  Definition~\threfc{d:space}{$(\LcEF,+)$ is an abelian group},
  we have $f=f_p-(f_p-f)$ belongs to~$\LcEF$.
  
  \proofparskip{Limit for $\tnEFdot$}
  From
  Lemma~\threfc{l:finite-operator-norm-is-continuous}{%
    $\eps$ is a continuity constant for $f_p-f$},
  Equation~\eqref{e:fp-minus-f-is-bounded} becomes
  \begin{equation*}
    \forall \eps > 0,\,
    \exists N \in \matN,\,
    \forall p \in \matN,\quad
    p \geq N \Implies
    \tnEF{f_p - f} \leq \eps.
  \end{equation*}
  Hence, from
  Definition~\threfc{d:distance-associated-with-norm}{%
    for norm~$\tnEFdot$}, and
  Definition~\thref{d:convergent-sequence},
  the sequence $(f_n)_{n\in\matN}$ is convergent in~$\LcEF$ for the
  distance~$\dEF$.
  
  Therefore, from
  Definition~\thref{d:complete-subset},
  the {\normedvectorspace} $\LcEF$ is complete.
\end{proof}

\begin{definition}[topological dual]
  \label{d:topological-dual}
  Let $(E,\nEdot)$ be a {\normedvectorspace}.
  The set of continuous linear forms on~$E$, denoted $\Ep=\LcEmatK$,
  is called the {\em topological dual of $E$}.
\end{definition}

\begin{definition}[dual norm]
  \label{d:dual-norm}
  Let $(E,\nEdot)$ be a {\normedvectorspace}.
  The {\em dual norm associated with~$\nEdot$}, denoted~$\nEpdot$, is the
  operator norm $\tnEmatKdot$ on $\Ep=\LcEmatK$ induced by norms $\nEdot$ and
  $|\cdot|$ (absolute value over~$\matK$).
\end{definition}

\begin{lemma}[topological dual is complete {\normedvectorspace}]
  \label{l:topological-dual-is-complete-normed-space}
  Let $(E,\nEdot)$ be a {\normedvectorspace}.
  Let $\Ep$ be the topological dual of~$E$.
  Let $\nEpdot$ be the associated dual norm.
  Then, $(\Ep,\nEpdot)$ is a complete {\normedvectorspace}.
\end{lemma}

\begin{proof}
  Direct consequence of
  Definition~\thref{d:dual-norm},
  Theorem~\thref{t:normed-space-of-continuous-linear-maps},
  Lemma~\thref{l:complete-normed-space-of-continuous-linear-maps},
  Lemma~\thref{l:k-is-normed-space}, and
  the \assume{completeness of~$\matK$}.
\end{proof}

\begin{definition}[bra-ket notation]
  \label{d:bra-ket-notation}
  Let $(E,\nEdot)$ be a {\normedvectorspace}.
  A continuous linear form $\fhi\in\Ep$ is a {\em bra},
  denoted~$\bra{\fhi}$.
  A vector $u\in E$ is a {\em ket}, denoted~$\ket{u}$.
  In {\em bra-ket notation} (or {\em Dirac notation}, or
  {\em duality pairing}), the application $\fhi(u)$ is denoted
  $\pdE{\fhi}{u}$.
\end{definition}

\begin{lemma}[bra-ket is bilinear map]
  \label{l:bra-ket-is-bilinear-map}
  Let $(E,\nEdot)$ be a {\normedvectorspace}.
  Then, $\pdEdotdot$ is a bilinear map from $\EpxE$ to~$\matK$.
\end{lemma}

\begin{proof}
  From
  Lemma~\thref{l:product-is-space}, and
  Lemma~\thref{l:topological-dual-is-complete-normed-space},
  $\EpxE$ is a {\vectorspace}.
  From
  Lemma~\thref{l:k-is-space},
  $\matK$~is a {\vectorspace}.
  From
  Definition~\thref{d:bra-ket-notation},
  $\pdEdotdot$ is a mapping from $\EpxE$ to~$\matK$.
  
  Let $\lambda,\mu\in\matR$ be scalars.
  Let $\fhi,\psi\in\Ep$ be continuous linear forms on~$E$ ({\ie} bras).
  Let $u,v\in E$ be vectors ({\ie} kets).
  Then, from
  Definition~\thref{d:bra-ket-notation}, and
  Definition~\threfc{d:inherited-vector-operations}{on~$\Ep$},
  we have
  \begin{equation*}
    \pdE{\lambda \fhi + \mu \psi}{u}
    = (\lambda \fhi + \mu \psi) (u)
    = \lambda \fhi (u) + \mu \psi (u)
    = \lambda \pdE{\fhi}{u} + \mu \pdE{\psi}{u}.
  \end{equation*}
  Moreover, from
  Definition~\thref{d:bra-ket-notation}, and
  Definition~\threfc{d:linear-map}{$\fhi$ is linear},
  we have
  \begin{equation*}
    \pdE{\fhi}{\lambda u + \mu v}
    = \fhi (\lambda u + \mu v)
    = \lambda \fhi (u) + \mu \fhi (v)
    = \lambda \pdE{\fhi}{u} + \mu \pdE{\fhi}{v}.
  \end{equation*}
  
  Therefore, from
  Definition~\thref{d:bilinear-map},
  $\pdEdotdot$ is left and right linear, hence bilinear.
\end{proof}

\paragraph{Bounded bilinear form}

\begin{definition}[bounded bilinear form]
  \label{d:bounded-bilinear-form}
  Let $(E,\nEdot)$ be a {\normedvectorspace}.
  A bilinear form $\fhi\in\LdE$ is {\em bounded} iff
  \begin{equation}
    \label{e:bounded-bilinear-form}
    \exists C \geq 0,\,
    \forall u, v \in E,\quad
    | \fhi (u, v) | \leq C \, \nE{u} \, \nE{v}.
  \end{equation}
  Then, $C$~is called {\em continuity constant of~$\fhi$}.
\end{definition}

\begin{lemma}[representation for bounded bilinear form]
  \label{l:representation-for-bounded-bilinear-form}
  Let $(E,\nEdot)$ be a {\normedvectorspace}.
  Let $\Blf\in\LdE$ be a bilinear form on~$E$.
  Assume that~$\Blf$ is bounded.
  Then, there exists a unique continuous linear map $A\in\LcEEp$ such that
  \begin{equation}
    \label{e:link-cont-bilinapp}
    \forall u, v \in E,\quad
    \blfE{u}{v} = \pdE{A (u)}{v} = A (u) (v).
  \end{equation}
  Moreover, for all~$C$ continuity constant of~$\Blf$, we have
  \begin{equation}
    \label{e:link-cont-bilinapp-norm}
    \tnEEp{A} \leq C.
  \end{equation}
\end{lemma}

\begin{proof}
  From
  Definition~\thref{d:set-of-bilinear-forms}, and
  Definition~\thref{d:bilinear-form},
  $\Blf$~is a bilinear map.
  
  \proofparskip{Existence}
  Let $u\in E$ be a vector.
  Let~$A_u:E\rightarrow\matR$ be the mapping defined by
  \begin{equation*}
    \forall v \in E,\quad
    A_u(v) = \blfE{u}{v}.
  \end{equation*}
  Let $\lambda,\lambdap\in\matR$ be scalars.
  Let $v,\vp\in E$ be vectors.
  Then, from
  Definition~\threfc{d:bilinear-map}{$\Blf$~is right linear},
  we have
  \begin{equation*}
    A_u (\lambda v + \lambdap \vp)
    = \blfE{u}{\lambda v + \lambdap \vp}
    = \lambda \blfE{u}{v} + \lambdap \blfE{u}{\vp}
    = \lambda A_u(v) + \lambdap A_u(\vp).
  \end{equation*}
  Hence, from
  Lemma~\thref{l:linear-map-preserves-linear-combinations}, and
  Definition~\thref{d:linear-form},
  $A_u$~is a linear form on~$E$.
  
  Let $v\in E$ be a vector.
  From
  Definition~\threfc{d:bounded-bilinear-form}{for~$\Blf$},
  let $C\geq 0$ such that, for all $\up,\vp\in E$, we have
  $|\Blf(\up,\vp)|\leq C\,\nE{\up}\,\nE{\vp}$.
  Let $C_u=C\,\nE{u}$.
  Then, from
  Lemma~\thref{l:norm-is-nonnegative}, and
  \assume{ordered field properties of~$\matR$},
  we have $C_u\geq 0$ and
  \begin{equation*}
    | A_u (v) |
    = | \Blf (u, v) |
    \leq C \, \nE{u} \, \nE{v}
    = C_u \, \nE{v}.
  \end{equation*}
  Hence, from
  Definition~\threfc{d:bounded-linear-map}{$A_u$~is bounded},
  Definition~\thref{d:topological-dual},
  Definition~\thref{d:dual-norm}, and
  Lemma~\threfc{l:finite-operator-norm-is-continuous}{$A_u\in\Ep$},
  we have
  \begin{equation*}
    \nEp{A_u}
    \leq C_u
    = C \, \nE{u}.
  \end{equation*}
  
  Let $A:E\rightarrow\Ep$ be the mapping defined by, for all $u\in E$,
  $A(u)=A_u$, {\ie}
  \begin{equation*}
    \forall u, v \in E,\quad
    \pdE{A (u)}{v} = A (u) (v) = A_u (v) = \Blf (u,v).
  \end{equation*}
  Let $\lambda,\lambdap\in\matR$ be scalars.
  Let $u,\up,v\in E$ be vectors.
  Then, from
  Definition~\threfc{d:bilinear-map}{$\Blf$~is left linear},
  Definition~\threfc{d:inherited-vector-operations}{on~$\Ep$}, and
  Lemma~\threfc{l:topological-dual-is-complete-normed-space}{$\Ep$~is space},
  we have
  \begin{eqnarray*}
    A (\lambda u + \lambdap \up) (v)
    & = & A_{(\lambda u + \lambdap \up)} (v) \\
    & = & \Blf (\lambda u + \lambdap \up, v) \\
    & = & \lambda \Blf (u, v) + \lambdap \Blf (\up, v) \\
    & = & \lambda A_u (v) + \lambdap A_\up (v) \\
    & = & \lambda A (u) (v) + \lambdap A (\up) (v) \\
    & = & (\lambda A (u) + \lambdap A (\up)) (v).
  \end{eqnarray*}
  Hence, from
  Lemma~\thref{l:linear-map-preserves-linear-combinations},
  $A$~is a linear map from~$E$ to~$\Ep$.
  
  Let~$\unitS$ be the unit sphere of~$E$.
  Let $\xi\in\unitS$ be a unit vector.
  Then, from
  Lemma~\threfc{l:equivalent-definition-of-unit-sphere}{$\nE{\xi}=1$},
  we have
  \begin{equation*}
    \nEp{A (\xi)}
    = \nEp{A_\xi}
    \leq C \, \nE{\xi}
    = C.
  \end{equation*}
  Hence, from
  Lemma~\threfc{l:finite-operator-norm-is-continuous}{%
    $C$~is a finite upper bound for $\nEp{A(\unitS)}$},
  $A$~belongs to $\LcEEp$ and $\tnEEp{A}\leq C$.
  
  \proofparskip{Uniqueness}
  Let $A,A^\prime\in\LcEEp$ be continuous linear maps such that
  \begin{equation*}
    \forall u, v \in E,\quad
    \blfE{u}{v} = \pdE{A (u)}{v} = \pdE{A^\prime (u)}{v}.
  \end{equation*}
  From
  Theorem~\thref{t:normed-space-of-continuous-linear-maps},
  Definition~\threfc{d:normed-space}{$\LcEEp$ is a {\vectorspace}}, and
  Definition~\thref{d:vector-subtraction},
  let $B=A-A^\prime\in\LcEEp$.
  Let $u,v\in E$ be vectors.
  Then, from
  Lemma~\thref{l:bra-ket-is-bilinear-map},
  Lemma~\thref{l:k-is-space}, and
  Definition~\threfc{d:space}{$(\matK,+)$ is an abelian group},
  we have
  \begin{equation*}
    \pdE{B (u)}{v}
    = \pdE{A (u)}{v} - \pdE{A^\prime (u)}{v}
    = \Blf(u, v) - \Blf(u, v)
    = 0.
  \end{equation*}
  Thus, from
  Definition~\thref{d:bra-ket-notation},
  $B(u)=0_\Ep$, and then $B=0_{\LcEEp}$.
  Hence, from
  Definition~\threfc{d:space}{$(\LcEEp,+)$ is an abelian group},
  $A=A^\prime$.
\end{proof}

\begin{definition}[coercive bilinear form]
  \label{d:coercive-bilinear-form}
  Let $(E,\nEdot)$ be a real {\normedvectorspace}.
  A bilinear form $\Blf\in\LdE$ is {\em coercive} (or {\em elliptic}) iff
  \begin{equation}
    \label{e:bilinear-form-coercive}
    \exists \alpha > 0,\,
    \forall u \in E,\quad
    \blfE{u}{u} \geq \alpha \nE{u}^2.
  \end{equation}
  Then, $\alpha$~is called {\em coercivity constant of~$\fhi$}.
\end{definition}

\begin{lemma}[coercivity constant is less than continuity constant]
  \label{l:coercivity-constant-is-less-than-continuity-constant}
  Let $(E,\nEdot)$ be a real {\normedvectorspace}.
  Let $\Blf\in\LdE$ be a bilinear form on~$E$.
  Assume that~$\Blf$ is continuous with constant $C\geq 0$, and coercive with
  constant~$\alpha>0$.
  Then, $\alpha\leq C$.
\end{lemma}

\begin{proof}
  Let $u\in E$ be a vector.
  Assume that $u\not=0_E$.
  Then, from
  Definition~\threfc{d:norm}{$\nEdot$ is definite, contrapositive}, and
  Lemma~\threfc{l:norm-is-nonnegative}{for $\nEdot$},
  we have $\nE{u}>0$.
  From
  Definition~\thref{d:coercive-bilinear-form},
  \assume{properties of the absolute value on~$\matR$}, and
  Definition~\threfc{d:bounded-bilinear-form}{with $v=u$},
  we have
  \begin{equation*}
    \alpha \nE{u}^2 \leq \fhi(u, u) \leq | \fhi(u, u) | \leq C \nE{u}^2.
  \end{equation*}
  Hence, from
  \assume{ordered field properties of~$\matR$},
  $\alpha\leq C$.
\end{proof}

\subsection{Inner product space}
\label{ss:inner-product-space}

\begin{definition}[inner product]
  \label{d:inner-product}
  Let~$G$ be a real {\vectorspace}.
  A mapping $\psGdotdot:\GxG\rightarrow\matR$ is an
  {\em inner product on~$G$} iff
  it is a bilinear form on~$G$ that is symmetric, nonnegative, and definite:
  \begin{eqnarray}
    \label{e:ip-symmetry}
    \forall u, v \in G, & & \psG{u}{v} = \psG{v}{u}; \\
    \label{e:ip-nonnegative}
    \forall u \in G, & & \psG{u}{u} \geq 0; \\
    \label{e:ip-definite}
    \forall u \in G, & & \psG{u}{u} = 0 \Implies u = 0_G.
  \end{eqnarray}
\end{definition}

\begin{remark}
  Note that the symmetry property~\eqref{e:ip-symmetry} implies the
  equivalence between left additivity~\eqref{e:bl-left-additive} and right
  additivity~\eqref{e:bl-right-additive} in the definition of a bilinear
  map.
\end{remark}

\begin{remark}
  Most results below are valid on a semi-inner space in which the definite
  property~\eqref{e:ip-definite} is dropped.
  The associated norm is then a semi-norm (the separation property is
  dropped).
\end{remark}

\begin{remark}
  In the case of a complex {\vectorspace}, the symmetry property becomes a
  conjugate symmetry property.
  In the sequel, we specify that the {\vectorspace} is real only in the case
  where the very same statement does not hold in a complex {\vectorspace}.
  When proofs differ, they are only given in the real case.
\end{remark}

\begin{definition}[inner product space]
  \label{d:inner-product-space}
  $(G,\psGdotdot)$ is an {\em inner product space} (or
  {\em pre-Hilbert space}) iff
  $G$~is a {\vectorspace} and $\psGdotdot$ is an inner product on~$G$.
\end{definition}

\begin{lemma}[inner product subspace]
  \label{l:inner-product-subspace}
  Let $(G,\psGdotdot)$ be an inner product space.
  Let~$F$ be a subspace of~$G$.
  Then, $F$~equipped with the restriction to~$F$ of the inner product
  $\psGdotdot$ is an inner product space.
\end{lemma}

\begin{proof}
  Direct consequence of
  Definition~\threfc{d:subspace}{$F$~is a subset of~$G$ and~$F$ is a space},
  Definition~\threfc{d:inner-product}{%
    the restriction of $\psGdotdot$ to~$F$ is trivially an inner product
    on~$F$}, and
  Definition~\thref{d:inner-product-space}.
\end{proof}

\begin{lemma}[inner product with zero is zero]
  \label{l:inner-product-with-zero-is-zero}
  Let $(G,\psGdotdot)$ be an inner product space.
  Then,
  \begin{equation}
    \label{e:inner-product-with-zero-is-zero}
    \forall u \in G,\quad
    \psG{0}{u} = \psG{u}{0} = 0.
  \end{equation}
\end{lemma}

\begin{proof}
  Let $u\in G$.
  From
  Definition~\threfc{d:inner-product}{%
    $\psGdotdot$ is symmetric and a bilinear map},
  Definition~\threfc{d:space}{$(G,+)$ is an  abelian group},
  Definition~\thref{d:vector-subtraction},
  Definition~\threfc{d:bilinear-map}{$\psGdotdot$ is right linear}, and
  \assume{field properties of~$\matR$},
  we have
  \begin{equation*}
    \psG{0_G}{u}
    = \psG{u}{0_G}
    = \psG{u}{0_G-0_G}
    = \psG{u}{0_G} - \psG{u}{0_G}
    = 0.
  \end{equation*}
\end{proof}

\begin{lemma}[square expansion plus]
  \label{l:square-expansion-plus}
  Let $(G,\psGdotdot)$ be a real inner product space.
  Then,
  \begin{equation}
    \label{e:square-expansion-plus}
    \forall u, v \in G,\quad
    \psG{u + v}{u + v} = \psG{u}{u} + 2 \psG{u}{v} + \psG{v}{v}.
  \end{equation}
\end{lemma}

\begin{proof}
  Let $u,v\in G$ be vectors.
  From
  Definition~\threfc{d:inner-product}{%
    $\psGdotdot$ is a bilinear map and symmetric},
  Definition~\thref{d:bilinear-map}, and
  \assume{field properties of~$\matR$},
  we have
  \begin{equation*}
    \psG{u + v}{u + v}
    = \psG{u}{u} + \psG{u}{v} + \psG{v}{u} + \psG{v}{v}
    = \psG{u}{u} + 2 \psG{u}{v} + \psG{v}{v}.
  \end{equation*}
\end{proof}

\begin{lemma}[square expansion minus]
  \label{l:square-expansion-minus}
  Let $(G,\psGdotdot)$ be a real inner product space.
  Then,
  \begin{equation}
    \label{e:square-expansion-minus}
    \forall u, v \in G,\quad
    \psG{u - v}{u - v} = \psG{u}{u} - 2 \psG{u}{v} + \psG{v}{v}.
  \end{equation}
\end{lemma}

\begin{proof}
  Let $u,v\in G$ be vectors.
  From
  Definition~\thref{d:vector-subtraction},
  Lemma~\thref{l:square-expansion-plus},
  Definition~\threfc{d:inner-product}{$\psGdotdot$ is a bilinear map},
  Definition~\thref{d:bilinear-map}, and
  \assume{field properties of~$\matR$},
  we have
  \begin{eqnarray*}
    \psG{u - v}{u - v}
    & = & \psG{u + (-v)}{u + (-v)} \\
    & = & \psG{u}{u} + 2 \psG{u}{-v} + \psG{-v}{-v} \\
    & = & \psG{u}{u} - 2 \psG{u}{v} + \psG{v}{v}.
  \end{eqnarray*}
\end{proof}

\begin{lemma}[parallelogram identity]
  \label{l:parallelogram-identity}
  Let $(G,\psGdotdot)$ be an inner product space.
  Then,
  \begin{equation}
    \label{e:parallelogram-identity}
    \forall u, v \in G,\quad
    \psG{u + v}{u + v} + \psG{u - v}{u - v}
    = 2 \left( \psG{u}{u} + \psG{v}{v} \right).
  \end{equation}
\end{lemma}

\begin{proof}
  Let $u,v\in G$ be vectors.
  From
  Lemma~\thref{l:square-expansion-plus},
  Lemma~\thref{l:square-expansion-minus}, and
  \assume{field properties of~$\matR$},
  we have
  \begin{eqnarray*}
    \psG{u + v}{u + v} + \psG{u - v}{u - v}
    & = & \psG{u}{u} + 2 \psG{u}{v} + \psG{v}{v}
    + \psG{u}{u} - 2 \psG{u}{v} + \psG{v}{v} \\
    & = & 2 \left( \psG{u}{u} + \psG{v}{v} \right).
  \end{eqnarray*}
\end{proof}

\begin{lemma}[{\CS} inequality]
  \label{l:cauchy-schwarz-inequality}
  Let $(G,\psGdotdot)$ be a real inner product space.
  Then,
  \begin{equation}
    \label{e:cauchy-schwarz-inequality}
    \forall u, v \in G,\quad
    \left( \psG{u}{v} \right)^2 \leq \psG{u}{u} \, \psG{v}{v}.
  \end{equation}
\end{lemma}

\begin{proof}
  Let $u,v\in G$ be vectors.
  Let~$\lambda\in\matR$ be a scalar.
  From
  Lemma~\thref{l:square-expansion-plus},
  Definition~\threfc{d:inner-product}{$\psGdotdot$ is a bilinear map},
  Definition~\thref{d:bilinear-map}, and
  \assume{field properties of~$\matR$},
  we have
  \begin{equation*}
    \psG{u + \lambda v}{u + \lambda v}
    = \lambda^2 \psG{v}{v} + 2 \lambda \psG{u}{v} + \psG{u}{u}.
  \end{equation*}
  Let $P(X)=\psG{v}{v}X^2+2\psG{u}{v}X+\psG{u}{u}$.
  It is a quadratic polynomial with real coefficients.
  From
  Definition~\threfc{d:inner-product}{$\psGdotdot$ is nonnegative},
  the associated polynomial function~$P$ is nonnegative.
  Hence, since
  \assume{a quadratic polynomial function has a constant sign iff its
    discriminant is nonpositive},
  we have
  \begin{equation*}
    4 (\psG{u}{v})^2 - 4 \psG{v}{v} \psG{u}{u} \leq 0.
  \end{equation*}
  Therefore, from
  \assume{ordered field properties of~$\matR$},
  we have
  \begin{equation*}
    (\psG{u}{v})^2 \leq \psG{u}{u} \psG{v}{v}.
  \end{equation*}
\end{proof}

\begin{definition}[square root of inner square]
  \label{d:square-root-of-inner-square}
  Let $(G,\psGdotdot)$ be an inner product space.
  The associated {\em square root of inner square} is the mapping
  $\PrenG:G\rightarrow\matR$ defined by
  \begin{equation}
    \label{e:square-root-of-inner-square}
    \forall u \in G,\quad
    \prenG{u} = \sqrt{\psG{u}{u}}.
  \end{equation}
\end{definition}

\begin{remark}
  Mapping~$\PrenG$ is well defined thanks to the nonnegativeness of the
  inner product.
  It will be proved below to be a norm.
\end{remark}

\begin{lemma}[squared norm]
  \label{l:squared-norm}
  Let $(G,\psGdotdot)$ be an inner product space.
  Then,
  \begin{equation}
    \label{e:squared-norm}
    \forall u \in G,\quad
    \prenG{u}^2 = \psG{u}{u}.
  \end{equation}
\end{lemma}

\begin{proof}
  Direct consequence of
  Definition~\thref{d:square-root-of-inner-square},
  Definition~\threfc{d:inner-product}{$\psGdotdot$ is nonnegative}, and
  \assume{properties of square and square root functions in~$\matRplus$}.
\end{proof}

\begin{lemma}[{\CS} inequality with norms]
  \label{l:cauchy-schwarz-inequality-with-norms}
  Let $(G,\psGdotdot)$ be an inner product space.
  Then,
  \begin{equation}
    \label{e:cauchy-schwarz-inequality-with-norms}
    \forall u, v \in G,\quad
    | \psG{u}{v} | \leq \prenG{u} \, \prenG{v}.
  \end{equation}
\end{lemma}

\begin{proof}
  Direct consequence of
  Lemma~\thref{l:cauchy-schwarz-inequality},
  Definition~\thref{d:square-root-of-inner-square},
  Definition~\threfc{d:inner-product}{$\psGdotdot$ is nonnegative}, and
  \assume{compatibility of the square root function with comparison
    in~$\matRplus$}.
\end{proof}

\begin{lemma}[triangle inequality]
  \label{l:triangle-inequality}
  Let $(G,\psGdotdot)$ be an inner product space.
  Then,
  \begin{equation}
    \label{e:triangle-inequality}
    \forall u, v \in G,\quad
    \prenG{u + v} \leq \prenG{u} + \prenG{v}.
  \end{equation}
\end{lemma}

\begin{proof}
  Let $u,v\in G$ be vectors.
  From
  Lemma~\thref{l:squared-norm},
  Lemma~\thref{l:square-expansion-plus},
  Lemma~\thref{l:cauchy-schwarz-inequality-with-norms}, and
  \assume{field properties of~$\matR$},
  we have
  \begin{eqnarray*}
    \left( \prenG{u + v} \right)^2 & = & \psG{u + v}{u + v} \\
    & = & \psG{u}{u} + 2 \psG{u}{v} + \psG{v}{v}  \\
    & \leq &
    \left( \prenG{u} \right)^2
    + 2 \prenG{u} \, \prenG{v}
    + \left( \prenG{v} \right)^2 \\
    & = & \left( \prenG{u} + \prenG{v} \right)^2.
  \end{eqnarray*}
  Therefore, from
  Definition~\threfc{d:inner-product}{$\psGdotdot$ is nonnegative}, and
  \assume{compatibility of the square function with comparison
    in~$\matRplus$},
  we have
  \begin{equation*}
    \prenG{u + v} \leq \prenG{u} + \prenG{v}.
  \end{equation*}
\end{proof}

\begin{lemma}[inner product gives norm]
  \label{l:inner-product-gives-norm}
  Let $(G,\psGdotdot)$ be an inner product space.
  Let $\nGdot$ be the associated square root of inner square.
  Then, $(G,\nGdot)$ is a {\normedvectorspace}.
\end{lemma}

\begin{proof}
  From
  Definition~\thref{d:square-root-of-inner-square}, and
  \assume{nonnegativeness of the square root function in~$\matRplus$},
  $\nGdot$ is nonnegative.
  
  From
  Definition~\thref{d:square-root-of-inner-square},
  \assume{definiteness of the square root function in~$\matRplus$}, and
  Definition~\threfc{d:inner-product}{$\psGdotdot$ is definite},
  $\nGdot$ is definite.
  
  Let~$\lambda\in\matK$ be a scalar.
  Let $u\in G$ be a vector.
  From
  Definition~\thref{d:square-root-of-inner-square},
  Definition~\threfc{d:inner-product}{$\psGdotdot$ is a bilinear map},
  \assume{multiplicativity of the square root function in~$\matRplus$}, and
  since \assume{for all $x\in\matR$, $\sqrt{x^2}=|x|$}, we have
  \begin{equation*}
    \nG{\lambda u}
    = \sqrt{\psG{\lambda u}{\lambda u}}
    = \sqrt{\lambda^2 \psG{u}{u}}
    = \sqrt{\lambda^2} \sqrt{\psG{u}{u}}
    = | \lambda | \, \nG{u}.
  \end{equation*}
  Thus, $\nGdot$ is absolutely homogeneous of degree~1.
  
  From
  Lemma~\thref{l:triangle-inequality},
  $\nGdot$ satisfies triangle inequality.
  
  Therefore, from
  Definition~\thref{d:norm},
  $\nGdot$ is a norm over~$G$, and from
  Definition~\thref{d:normed-space},
  $(G,\nGdot)$ is a {\normedvectorspace}.
\end{proof}

\begin{remark}
  Norm~$\nGdot$ is called
  {\em norm associated with inner product~$\psGdotdot$}.
\end{remark}

\subsubsection{Orthogonal projection}

\begin{definition}[convex subset]
  \label{d:convex-subset}
  Let~$E$ be a real {\vectorspace}.
  Let~$K\subset E$.
  $K$ is {\em convex} iff
  \begin{equation}
    \label{e:convex-subset}
    \forall u, v \in K,\,
    \forall \theta \in [0, 1],\quad
    \theta u + (1 - \theta) v \in K.
  \end{equation}
\end{definition}

\begin{theorem}[orthogonal projection onto nonempty complete convex]
  \label{t:orth-proj-onto-complete-convex}
  Let~$(G,\psGdotdot)$ be a real inner product space.
  Let~$\nGdot$ be the norm associated with inner product~$\psGdotdot$.
  Let~$d_G$ be the distance associated with norm~$\nGdot$.
  Let~$K\subset G$ be a nonempty convex subset which is complete for
  distance~$d_G$.
  Then, for all $u\in G$, there exists a unique $v\in K$ such that
  \begin{equation}
    \label{e:orth-proj-convex-min-distance}
    \nG{u - v} = \min_{w \in K} \nG{u - w}.
  \end{equation}
\end{theorem}

\begin{proof}
  Let $u\in G$.
  From
  Lemma~\threfc{l:norm-is-nonnegative}{for $\nGdot$},
  function $w\mapsto\nG{u-w}$ from~$K$ to~$\matR$ admits~0 as finite lower
  bound.
  Thus, from
  Lemma~\thref{l:finite-infimum-discrete},
  $\delta=\inf_{w\in K}\{\nG{u-w}\}$ is finite and there exists a sequence
  $(w_n)_{n\in\matN}$ in~$K$ such that for all $n\in\matN$,
  $\nG{u-w_n}<\delta+\frac{1}{n+1}$.
  
  From
  Definition~\thref{d:inner-product-space},
  $G$~is a {\vectorspace}.
  
  \proofparskip{Existence}
  Let $p,q\in\matN$.
  Let $a=u-w_q$ and $b=u-w_p$.
  From
  Definition~\threfc{d:norm}{%
    $\nGdot$ is absolutely homogeneous of degree~1},
  Definition~\thref{d:scalar-division},
  Definition~\threfc{d:space}{$(G,+)$ is an abelian group},
  Lemma~\thref{l:squared-norm}, and
  Lemma~\threfc{l:parallelogram-identity}{for $\nGdot$},
  we have
  \begin{eqnarray*}
    4 \nG{u - \frac{w_q + w_p}{2}}^2 + \nG{w_p - w_q}^2
    & = & \nG{a + b}^2 + \nG{a - b}^2 \\
    & = & 2 \nG{a}^2 + 2 \nG{b}^2 \\
    & = & 2 \nG{u - w_q}^2 + 2 \nG{u - w_p}^2.
  \end{eqnarray*}
  From
  Definition~\thref{d:convex-subset},
  $\frac{w_q+w_p}{2}$ belongs to~$K$.
  Thus, from
  Definition~\threfc{d:infimum}{%
    $\delta$~is a lower bound for $\{\nG{u-w}\st w\in K\}$}, and
  \assume{field properties of~$\matR$},
  we have
  \begin{eqnarray*}
    \nG{w_p - w_q}^2
    & = &
    -4 \nG{u - \frac{w_q + w_p}{2}}^2
    + 2 \nG{u - w_q}^2
    + 2 \nG{u - w_p}^2 \\
    & < &
    - 4 \delta^2
    + 2 \left(\delta + \frac{1}{q+1}\right)^2
    + 2 \left(\delta + \frac{1}{p+1}\right)^2 \\
    & = &
    \frac{4\delta}{q+1} + \frac{2}{(q+1)^2}
    + \frac{4\delta}{p+1} + \frac{2}{(p+1)^2}.
  \end{eqnarray*}
  Let $\eps>0$.
  Let
  $\eta=\max\left(\frac{16\delta}{\eps^2},\frac{2\sqrt{2}}{\eps}\right)$.
  From
  \assume{the definition of the max function}, and
  \assume{ordered field properties of~$\matR$},
  $\eta>0$.
  Let $N=\ceil{\eta}-1$.
  From
  \assume{the definition of the ceiling function},
  $N\geq 0$ and $N\geq\eta-1$.
  Assume that $p,q\geq N$.
  Then, from
  \assume{ordered field properties of~$\matR$},
  we have $p,q\geq\eta-1$ and
  $\frac{4\delta}{q+1},\frac{2}{(q+1)^2},\frac{4\delta}{p+1},\frac{2}{(p+1)^2}
  \leq\frac{\eps^2}{4}$.
  Thus, from
  \assume{field properties of~$\matR$},
  Lemma~\threfc{l:norm-is-nonnegative}{for $\nGdot$}, and
  \assume{compatibility of the square root function with comparison
    in~$\matRplus$},
  we have $\nG{w_p-w_q}\leq\eps$.
  Hence, from
  Definition~\thref{d:cauchy-sequence},
  $(w_n)_{n\in\matN}$ is a Cauchy sequence in~$K$.
  
  From
  hypothesis, and
  Definition~\threfc{d:complete-subset}{$K$ is complete},
  the sequence $(w_n)_{n\in\matN}$ is convergent in~$K$.
  Let $v\in K$ be its limit.
  From
  Lemma~\threfc{l:norm-is-continuous}{for $\nGdot$},
  Lemma~\thref{l:compatibility-of-limit-with-continuous-functions}, and
  Definition~\thref{d:minimum},
  we have
  \begin{equation*}
    \nG{u - v}
    = \lim_{n \rightarrow +\infty} \nG{u - w_n}
    = \delta
    = \min_{w \in K} \nG{u - w}.
  \end{equation*}
  
  \proofparskip{Uniqueness}
  Let $v,\vp\in K$ such that $\nG{u-v}=\nG{u-\vp}=\delta$.
  Let $a=u-\vp$, $b=u-v$, and $\vpp=\frac{\vp + v}{2}$.
  Then, from
  Definition~\threfc{d:norm}{%
    $\nGdot$ is absolutely homogeneous of degree~1},
  Definition~\thref{d:scalar-division},
  Definition~\threfc{d:space}{$(G,+)$ is an abelian group}, and
  Lemma~\threfc{l:parallelogram-identity}{for $\nGdot$},
  we have
  \begin{eqnarray*}
    4 \nG{u - \vpp}^2 + \nG{v - \vp}^2
    & = & \nG{a + b}^2 + \nG{a - b}^2 \\
    & = & 2 \nG{a}^2 + 2 \nG{b}^2 \\
    & = & 2 \nG{u - \vp}^2 + 2 \nG{u - v}^2 \\
    & = & 4 \delta^2.
  \end{eqnarray*}
  From
  Definition~\thref{d:convex-subset},
  $\vpp$ belongs to~$K$.
  Thus, from
  Definition~\threfc{d:infimum}{%
    $\delta$~is a lower bound for $\{\nG{u-w}\st w\in K\}$}, and
  \assume{field properties of~$\matR$},
  we have
  \begin{equation*}
    0 \leq \nG{v - \vp}^2
    = -4 \nG{u - \vpp}^2 + 4 \delta^2
    \leq -4 \delta^2 + 4 \delta^2 = 0.
  \end{equation*}
  Hence, from
  Lemma~\threfc{l:norm-is-nonnegative}{for $\nGdot$}, and
  \assume{compatibility of square root function with comparison
    in~$\matRplus$},
  $\nG{v-\vp}=0$.
  Therefore, from
  Definition~\threfc{d:norm}{$\nGdot$ is definite},
  Definition~\thref{d:vector-subtraction}, and
  Definition~\threfc{d:space}{$(G,+)$ is an abelian group},
  we have $v-\vp=0_G$ and $v=\vp$.
\end{proof}

\begin{lemma}[characterization of orthogonal projection onto convex]
  \label{l:characterization-of-orth-proj-onto-convex}
  Let~$(G,\psGdotdot)$ be a real inner product space.
  Let~$\nGdot$ be the norm associated with inner product~$\psGdotdot$.
  Let~$K\subset G$ be a nonempty convex subset.
  Then, for all $u\in G$, for all $v\in K$,
  \begin{equation}
    \label{e:characterization-of-orth-proj-onto-convex}
    \nG{u - v} = \inf_{w \in K} \nG{u - w}
    \Equiv
    \forall w \in K,\quad
    \psG{u - v}{w - v} \leq 0.
  \end{equation}
\end{lemma}

\begin{proof}
  Let $u\in G$ and $v\in K$ be vectors.
  
  \proofparskip{``Left'' implies ``right''}
  Assume that $\nG{u-v}=\inf_{w\in K}\nG{u-w}$.
  Let $w\in K$.
  Let $\theta\in(0,1]$.
  From
  Definition~\thref{d:convex-subset},
  $\theta w+(1-\theta)v$ belongs to~$K$.
  Thus, from
  Definition~\threfc{d:infimum}{%
    $\nG{u-v}$ is a lower bound for $\{\nG{u-w}\st w\in K\}$},
  \assume{compatibility of the square function with comparison
    in~$\matRplus$},
  Definition~\threfc{d:inner-product-space}{$G$~is a {\vectorspace}},
  Definition~\threfc{d:space}{%
    $(G,+)$ is an abelian group and scalar multiplication is compatible with
    scalar addition}, \thref{d:vector-subtraction},
  Lemma~\threfc{l:squared-norm}{for $\psGdotdot$},
  Lemma~\threfc{l:square-expansion-plus}{for $\psGdotdot$},
  Definition~\threfc{d:inner-product}{$\psGdotdot$ is a bilinear map},
  Definition~\thref{d:bilinear-map}, and
  Lemma~\thref{l:square-expansion-plus},
  we have
  \begin{eqnarray*}
    \nG{u - v}^2
    & \leq & \nG{u - (\theta w + (1 - \theta) v)}^2 \\
    & = & \nG{(u - v) + \theta (v - w)}^2 \\
    & = &
    \nG{u - v}^2 - 2 \theta \psG{u - v}{w - v} + \theta^2 \nG{v - w}^2.
  \end{eqnarray*}
  Let $a=\psG{u-v}{w-v}$ and $b=\nG{v-w}^2$.
  Then, from
  \assume{ordered field properties of~$\matR$ (with $\theta>0$)},
  we have
  \begin{equation*}
    \forall \theta \in (0, 1],\quad
    2 a \leq \theta b.
  \end{equation*}
  Assume that $b=0$.
  Then, from
  \assume{ordered field properties of~$\matR$},
  we have $\psG{u-v}{w-v}=a\leq 0$.
  Conversely, assume now that $b\not=0$.
  Then, from
  \assume{nonnegativeness of the square function},
  $b>0$.
  Assume that $a>0$.
  Let $\theta=\min(1,\frac{a}{b})$.
  From
  \assume{the definition of the min function}, and
  \assume{ordered field properties of~$\matR$},
  we have $\theta\leq\frac{a}{b}$ and $0<\theta\leq 1$.
  Thus, $2a\leq\theta b\leq a$.
  Hence, from
  \assume{ordered field properties of~$\matR$},
  $a\leq 0$.
  Which is impossible.
  Therefore, we have $\psG{u-v}{w-v}=a\leq 0$.
  
  \proofparskip{``Right'' implies ``left''}
  Conversely, assume now that, for all $w\in K$, $\psG{u-v}{w-v}\leq 0$.
  Let~$w\in K$.
  Then, from
  Definition~\threfc{d:inner-product-space}{$G$~is a {\vectorspace}},
  Definition~\threfc{d:space}{$(G,+)$ is an abelian group},
  Definition~\thref{d:vector-subtraction},
  Lemma~\threfc{l:squared-norm}{for $\psGdotdot$},
  Lemma~\threfc{l:square-expansion-plus}{for $\psGdotdot$},
  Definition~\threfc{d:inner-product}{$\psGdotdot$ is a bilinear map},
  Definition~\threfc{d:bilinear-map}{$\psGdotdot$ is right linear}, and
  \assume{nonnegativeness of the square function in~$\matRplus$},
  we have
  \begin{eqnarray*}
    \nG{u - w}^2
    & = & \nG{(u - v) + (v - w)}^2 \\
    & = & \nG{u - v}^2 + 2 \psG{u - v}{v - w} + \nG{v - w}^2 \\
    & \geq & \nG{u - v}^2 - 2 \psG{u - v}{w - v} \\
    & \geq & \nG{u - v}^2.
  \end{eqnarray*}
  Hence, from
  Lemma~\threfc{l:norm-is-nonnegative}{for $\nGdot$}, and
  \assume{compatibility of the square root function with comparison
    in~$\matRplus$},
  we have $\nG{u-v}\leq\nG{u-w}$.
  Therefore, from
  Lemma~\thref{l:finite-minimum}, and
  Definition~\thref{d:minimum},
  we have
  \begin{equation*}
    \nG{u - v}
    = \min_{w \in K} \nG{u - w}
    = \inf_{w \in K} \nG{u - w}.
    \end{equation*}
\end{proof}

\begin{lemma}[subspace is convex]
  \label{l:subspace-is-convex}
  Let~$E$ be a real {\vectorspace}.
  Let~$F$ be a subspace of~$E$.
  Then, $F$~is a convex subset of~$E$.
\end{lemma}

\begin{proof}
  Let $u,v\in F$ be vectors in the subspace.
  Let $\theta\in[0,1]$.
  Then, from
  Lemma~\thref{l:closed-under-linear-combination-is-subspace},
  the linear combination $w=\theta u+(1-\theta)v$ belongs to~$F$.
  Therefore, from
  Definition~\thref{d:convex-subset},
  $F$~is a convex subset of~$E$.
\end{proof}

\begin{theorem}[orthogonal projection onto complete subspace]
  \label{t:orth-proj-onto-complete-subspace}
  Let $(G,\psGdotdot)$ be a real inner product space.
  Let~$\nGdot$ be the norm associated with inner product~$\psGdotdot$.
  Let~$d_G$ be the distance associated with norm~$\nGdot$.
  Let~$F$ be a subspace of~$G$ which is complete for distance~$d_G$.
  Then, for all $u\in G$, there exists a unique $v\in F$ such that
  \begin{equation}
    \label{e:orth-proj-subspace-min-distance}
    \nG{u - v} = \min_{w \in F} \nG{u - w}.
  \end{equation}
\end{theorem}

\begin{proof}
  Direct consequence of
  Definition~\threfc{d:subspace}{$F$~is vector space},
  Definition~\threfc{d:space}{$F\ni 0_G$~is nonempty},
  Lemma~\thref{l:subspace-is-convex}, and
  Theorem~\threfc{t:orth-proj-onto-complete-convex}{%
    $F$~is a nonempty convex subset of~$G$ which is complete for
    distance~$d_G$}.
\end{proof}

\begin{definition}[orthogonal projection onto complete subspace]
  \label{d:orth-proj-onto-complete-subspace}
  Assume hypotheses of
  Theorem~\thref{t:orth-proj-onto-complete-subspace}.
  The mapping $P_F:G\rightarrow F$ associating to any vector of~$G$ the
  unique vector of~$F$ satisfying~\eqref{e:orth-proj-subspace-min-distance}
  is called {\em orthogonal projection onto~$F$}.
\end{definition}




\begin{lemma}[characterization of orthogonal projection onto subspace]
  \label{l:characterization-of-orth-proj-onto-subspace}
  Let~$(G,\psGdotdot)$ be a real inner product space.
  Let~$\nGdot$ be the norm associated with inner product~$\psGdotdot$.
  Let~$F$ be a subspace of~$G$.
  Then, for all $u\in G$, for all $v\in F$,
  \begin{equation}
    \label{e:characterization-of-orth-proj-onto-subspace}
    \nG{u - v} = \inf_{w \in F} \nG{u - w}
    \Equiv
    \forall w \in F,\quad
    \psG{v}{w} = \psG{u}{w}.
  \end{equation}
\end{lemma}

\begin{proof}
  Let $u\in G$ and $v\in F$ be vectors.
  
  \proofparskip{``Left'' implies ``right''}
  Assume that $\nG{u-v}=\inf_{w\in F}\nG{u-w}$.
  Then, from
  Definition~\threfc{d:subspace}{$F$~is vector space},
  Definition~\threfc{d:space}{$F\ni 0_G$~is nonempty},
  Lemma~\thref{l:subspace-is-convex}, and
  Lemma~\threfc{l:characterization-of-orth-proj-onto-convex}{%
    $F$~is a nonempty convex subset},
  we have for all $w\in F$, $\psG{u-v}{w-v}\leq 0$.
  Let $w\in F$.
  Let $\wpr=w+v$.
  Then, from
  Definition~\threfc{d:subspace}{$F$~is a {\vectorspace}},
  Definition~\threfc{d:space}{$(F,+)$ is an abelian group}, and
  Definition~\thref{d:vector-subtraction},
  $\wpr$ belongs to~$F$ and $w=\wpr-v$.
  Thus, we have $\psG{u-v}{w}=\psG{u-v}{\wpr-v}\leq 0$.
  Similarly, $\wpp=-w+v$ belongs to~$F$ and
  $\psG{u-v}{-w}=\psG{u-v}{\wpp-v}\leq 0$.
  Hence, from
  Definition~\threfc{d:inner-product}{$\psGdotdot$ is a bilinear map}
  Definition~\threfc{d:bilinear-map}{$\psGdotdot$ is right linear}, and
  \assume{ordered field properties of~$\matR$},
  we have $\psG{u-v}{w}=0$.
  Therefore, from
  Definition~\thref{d:vector-subtraction}, and
  Definition~\threfc{d:bilinear-map}{$\psGdotdot$ is left linear},
  we have $\psG{v}{w}=\psG{u}{w}$.
  
  \proofparskip{``Right'' implies ``left''}
  Conversely, assume now that for all $w\in F$, $\psG{v}{w}=\psG{u}{w}$.
  Let $w\in F$.
  Let $\wpr=w-v$.
  Then, from
  Definition~\threfc{d:subspace}{$F$~is a {\vectorspace}}, and
  Definition~\threfc{d:space}{$(F,+)$ is an abelian group},
  $\wpr$ belongs to~$F$.
  Hence, from
  Definition~\threfc{d:inner-product}{$\psGdotdot$ is a bilinear map},
  Definition~\threfc{d:bilinear-map}{$\psGdotdot$ is left linear},
  hypothesis, and
  \assume{ordered field properties of~$\matR$},
  we have
  \begin{equation*}
    \psG{u - v}{w - v}
    = \psG{u - v}{\wpr}
    = \psG{u}{\wpr} - \psG{v}{\wpr}
    = 0
    \leq 0.
  \end{equation*}
  Therefore, from
  Definition~\threfc{d:subspace}{$F$~is vector space},
  Definition~\threfc{d:space}{$F\ni 0_G$~is nonempty},
  Lemma~\thref{l:subspace-is-convex}, and
  Lemma~\threfc{l:characterization-of-orth-proj-onto-convex}{%
    $F$~is a nonempty convex subset},
  we have $\nG{u-v}=\inf_{w\in F}\nG{u-w}$.
\end{proof}

\begin{lemma}[orthogonal projection is continuous linear map]
  \label{l:orth-proj-is-continuous-linear-map}
  Assume hypotheses of
  Theorem~\thref{t:orth-proj-onto-complete-subspace}.
  Then, the orthogonal projection~$P_F$ is a 1-Lipschitz continuous linear
  map from~$G$ to~$F$.
\end{lemma}

\begin{proof}
  From
  Definition~\thref{d:orth-proj-onto-complete-subspace}, and
  Theorem~\thref{t:orth-proj-onto-complete-subspace},
  $P_F$ effectively defines a mapping from~$G$ to~$F$.
  
  \proofparskip{Linearity}
  Let $\up,\upp\in G$.
  Let $\lambdap,\lambdapp\in\matR$.
  From
  Definition~\threfc{d:subspace}{$F$~is a vector space}, and
  Definition~\threfc{d:space}{%
    $G$ and~$F$ are closed under vector operations},
  $\lambdap\up+\lambdapp\upp$ belongs to~$G$ and
  $\lambdap P_F(\up)+\lambdapp P_F(\upp)$ belongs to~$F$.
  Let $w\in F$.
  From
  Definition~\threfc{d:inner-product}{$\psGdotdot$ is a bilinear map},
  Definition~\threfc{d:bilinear-map}{$\psGdotdot$ is left linear}, and
  Lemma~\thref{l:characterization-of-orth-proj-onto-subspace},
  we have
  \begin{eqnarray*}
    \psG{\lambdap P_F (\up) + \lambdapp P_F (\upp)}{w}
    & = & \lambdap \psG{P_F (\up)}{w} + \lambdapp \psG{P_F (\upp)}{w} \\
    & = & \lambdap \psG{\up}{w} + \lambdapp \psG{\upp}{w} \\
    & = & \psG{\lambdap \up + \lambdapp \upp}{w}.
  \end{eqnarray*}
  Hence, from
  Lemma~\thref{l:characterization-of-orth-proj-onto-subspace}, and
  Theorem~\threfc{t:orth-proj-onto-complete-subspace}{%
    orthogonal projection is unique},
  we have
  \begin{equation*}
    P_F (\lambdap \up + \lambdapp \upp)
    = \lambdap P_F (\up) + \lambdapp P_F (\upp).
  \end{equation*}
  Therefore, from
  Lemma~\thref{l:linear-map-preserves-linear-combinations},
  $P_F$~is a linear map.
  
  \proofparskip{Continuity}
  Let $u\in G$.
  
  \proofparskip{Case $P_F(u)=0_G$}
  Then, from
  Lemma~\threfc{l:norm-is-nonnegative}{for $\nGdot$},
  we have $\nG{P_F(u)}=0\leq\nG{u}$.
  
  \proofparskip{Case $P_F(u)\not=0_G$}
  Then, from
  Lemma~\thref{l:squared-norm},
  Lemma~\threfc{l:characterization-of-orth-proj-onto-subspace}{%
    with $w=P_F(u)\in F$}, and
  Lemma~\thref{l:cauchy-schwarz-inequality-with-norms},
  we have
  \begin{equation*}
    \nG{P_F (u)}^2
    = \psG{P_F (u)}{P_F (u)}
    = \psG{u}{P_F (u)}
    \leq \nG{u} \, \nG{P_F (u)}.
  \end{equation*}
  Hence, from
  Definition~\threfc{d:norm}{$\nGdot$ is definite}, and
  \assume{ordered field properties of~$\matR$},
  we have $\nG{P_F(u)}\leq\nG{u}$.
  
  Therefore, from
  Definition~\thref{d:lipschitz-continuity},
  $P_F$~is 1-Lipschitz continuous.
\end{proof}

\begin{definition}[orthogonal complement]
  \label{d:orth-compl}
  Let~$(G,\psGdotdot)$ be an inner product space.
  Let~$F$ be a subspace of~$G$.
  The {\em orthogonal complement of~$F$ in~$G$}, denoted~$F^\perp$, is
  defined by
  \begin{equation}
    \label{e:orth-compl}
    F^\perp = \{ u \in G \st \forall v \in F,\, \psG{u}{v} = 0 \}.
  \end{equation}
\end{definition}

\begin{lemma}[trivial orthogonal complements]
  \label{l:trivial-orth-compls}
  Let~$(G,\psGdotdot)$ be an inner product space.
  Then, $G^\perp=\{0_G\}$ and $\{0_G\}^\perp=G$.
\end{lemma}

\begin{proof}
  From
  Lemma~\thref{l:trivial-subspaces},
  $G$~and $\{0_G\}$ are subspaces of~$G$.
  From
  Definition~\thref{d:orth-compl},
  $G^\perp$~and $\{0_G\}^\perp$ are subsets of~$G$.
  
  Let $u\in G$ be a vector.
  Then, from
  Lemma~\threfc{l:inner-product-with-zero-is-zero}{for $\psGdotdot$},
  we have $\psG{0_G}{u}=\psG{u}{0_G}=0$.
  Hence, from
  Definition~\thref{d:orth-compl},
  $\{0_G\}$ is a subset of~$G^\perp$ and~$G$ is a subset of~$\{0_G\}^\perp$.
  
  Let $u\in G^\perp$ be a vector in the orthogonal.
  Let $v=u\in G$.
  Then, from
  Definition~\thref{d:orth-compl},
  we have $\psG{u}{v}=\psG{u}{u}=0$.
  Thus, from
  Definition~\threfc{d:inner-product}{$\psGdotdot$ is definite},
  we have $u=0_G$.
  Hence, $G^\perp$~is a subset of $\{0_G\}$.
  
  Therefore, $G^\perp=\{0_G\}$ and $\{0_G\}^\perp=G$.
\end{proof}

\begin{lemma}[orthogonal complement is subspace]
  \label{l:orth-compl-is-subspace}
  Let~$(G,\psGdotdot)$ be an inner product space.
  Let~$F$ be a subspace of~$G$.
  Then, $F^\perp$~is a subspace of~$G$.
\end{lemma}

\begin{proof}
  Let $v\in F$.
  Then, from
  Lemma~\threfc{l:inner-product-with-zero-is-zero}{for $\psGdotdot$},
  we have $\psG{0_G}{v}=0$.
  Hence, from
  Definition~\thref{d:orth-compl},
  $0_G$ belongs to~$F^\perp$.
  
  Let $\lambda,\lambdap\in\matR$.
  Let $u,\up\in F^\perp$.
  Let $v\in F$.
  From
  Definition~\threfc{d:inner-product}{$\psGdotdot$ is a bilinear map},
  Definition~\threfc{d:bilinear-map}{$\psGdotdot$ is left linear}, and
  \assume{field properties of~$\matR$},
  we have
  \begin{equation*}
    \psG{\lambda u + \lambdap \up}{v}
    = \lambda \psG{u}{v} + \lambdap \psG{\up}{v}
    = \lambda 0 + \lambdap 0
    = 0.
  \end{equation*}
  Thus, from
  Definition~\thref{d:orth-compl},
  $\lambda u+\lambdap \up$ belongs to~$F^\perp$.
  Hence, $F^\perp$ is closed under linear combination.
  Therefore, from
  Lemma~\thref{l:closed-under-linear-combination-is-subspace},
  $F^\perp$~is a subspace of~$G$.
\end{proof}

\begin{lemma}[zero intersection with orthogonal complement]
  \label{l:zero-intersection-with-orth-compl}
  Let~$(G,\psGdotdot)$ be an inner product space.
  Let~$F$ be a subspace of~$G$.
  Then,
  \begin{equation}
    \label{e:zero-intersection-with-orth-compl}
    F \cap F^\perp = \{ 0_G \}.
  \end{equation}
\end{lemma}

\begin{proof}
  Let~$\nGdot$ be the norm associated with inner product~$\psGdotdot$.
  Let $u\in F\cap F^\perp$.
  Then, $u\in F$ and $v=u\in F^\perp$.
  Thus, from
  Definition~\thref{d:orth-compl},
  we have
  \begin{equation*}
    \psG{u}{u} = \psG{u}{v} = 0.
  \end{equation*}
  Therefore, from
  Definition~\threfc{d:inner-product}{$\psGdotdot$ is definite},
  $u=0_G$.
\end{proof}

\begin{theorem}[direct sum with orthogonal complement when complete]
  \label{t:direct-sum-with-orth-compl-when-complete}
  Assume hypotheses of
  Theorem~\thref{t:orth-proj-onto-complete-subspace}.
  Then,
  \begin{equation}
    \label{e:direct-sum-with-orth-compl-when-complete}
    G = F \oplus F^\perp.
  \end{equation}
  Moreover, for all $u\in G$, the (unique) decomposition onto
  $F\oplus F^\perp$ is
  \begin{equation}
    \label{e:decomposition-with-orth-proj-when-complete}
    u = P_F (u) + (u - P_F (u))
  \end{equation}
  and we have the following characterizations of the orthogonal complements:
  \begin{eqnarray}
    \label{e:characterization-orth-compl-f}
    u \in F & \Equiv & P_F (u) = u; \\
    \label{e:characterization-orth-compl-f-perp}
    u \in F^\perp & \Equiv & P_F (u) = 0_G.
  \end{eqnarray}
\end{theorem}

\begin{proof}
  Let $u\in G$.
  
  Then, from
  Definition~\thref{d:orth-proj-onto-complete-subspace}, and
  Lemma~\thref{l:characterization-of-orth-proj-onto-subspace},
  there exists a unique $P_F(u)\in F$ characterized by, for all $w\in F$,
  $\psG{P_F(u)}{w}=\psG{u}{w}$.
  Thus, from
  Definition~\threfc{d:inner-product}{$\psGdotdot$ is a bilinear map},
  Definition~\threfc{d:bilinear-map}{$\psGdotdot$ is left linear}, and
  Definition~\thref{d:orth-compl},
  $u-P_F(u)$ belongs to~$F^\perp$.
  From
  Definition~\threfc{d:space}{$(G,+)$ is an abelian group},
  we have
  \begin{equation*}
    u = P_F (u) + (u - P_F (u)).
  \end{equation*}
  Hence, from
  Definition~\thref{d:sum-of-subspaces},
  $G=F+F^\perp$.
  Therefore, from
  Lemma~\threfc{l:zero-intersection-with-orth-compl}{for~$F$}, and
  Lemma~\thref{l:equivalent-definition-of-direct-sum},
  we have $G=F\oplus F^\perp$.
  From
  Definition~\thref{d:direct-sum-of-subspaces},
  the decomposition~\eqref{e:decomposition-with-orth-proj-when-complete} with
  $P_F(u)\in F$ and $u-P_F(u)\in F^\perp$ is unique.
  
  From
  Lemma~\thref{l:zero-intersection-with-orth-compl},
  $0_G$~belongs to both~$F$ and~$F^\perp$.
  
  \proofparskip{(\ref{e:characterization-orth-compl-f}):
    ``left'' implies ``right''}
  Assume that $u\in F$.
  Then, from
  Definition~\threfc{d:space}{$(G,+)$ is an abelian group},
  $u=u+0_G$ is a decomposition over $F\oplus F^\perp$.
  From
  uniqueness of the decomposition,
  we have $P_F(u)=u$.
  
  \proofparskip{(\ref{e:characterization-orth-compl-f}):
    ``right'' implies ``left''}
  Conversely, assume now that $P_F(u)=u$.
  Then, from
  Definition~\threfc{d:orth-proj-onto-complete-subspace}{%
    $P_F$~is a mapping to~$F$},
  $u=P_F(u)$ belongs to~$F$.
  
  \proofparskip{(\ref{e:characterization-orth-compl-f-perp}):
    ``left'' implies ``right''}
  Assume that $u\in F^\perp$.
  Then, from
  Definition~\threfc{d:space}{$(G,+)$ is an abelian group},
  $u=0_G+u$ is a decomposition over $F\oplus F^\perp$.
  From
  uniqueness of the decomposition,
  we have $P_F(u)=0_G$.
  
  \proofparskip{(\ref{e:characterization-orth-compl-f-perp}):
    ``right'' implies ``left''}
  Conversely, assume now that $P_F(u)=0_G$.
  Let $v\in F$.
  Then, from
  Lemma~\thref{l:characterization-of-orth-proj-onto-subspace}, and
  Lemma~\thref{l:inner-product-with-zero-is-zero},
  we have
  \begin{equation*}
    \psG{u}{v}
    = \psG{P_F (u)}{v}
    = \psG{0_G}{v}
    = 0.
  \end{equation*}
  Hence, from
  Definition~\thref{d:orth-compl},
  $u$~belongs to~$F^\perp$.
\end{proof}

\begin{lemma}[sum is orthogonal sum]
  \label{l:sum-is-orth-sum}
  Assume hypotheses of
  Theorem~\thref{t:orth-proj-onto-complete-subspace}.
  Let~$u\in G$ be a vector.
  Then, there exists $\up\in F^\perp$ such that $F+\Line{u}=F+\Line{\up}$.
\end{lemma}

\begin{proof}
  Let $\up=u-P_F(u)$.
  Then, from
  Theorem~\thref{t:direct-sum-with-orth-compl-when-complete},
  $P_F(u)$ belongs to~$F$ and $\up$ belongs to~$F^\perp$.
  
  Let $w\in F+\Line{u}$.
  Then, from
  Definition~\thref{d:sum-of-subspaces}, and
  Definition~\thref{d:linear-span},
  there exists $v\in F$ and $\lambda\in\matR$ such that $w=v+\lambda u$.
  From
  Lemma~\threfc{l:closed-under-linear-combination-is-subspace}{%
    with~$1$ and~$\lambda$},
  we have $\vp=v+\lambda P_F(u)\in F$, and thus, from
  Definition~\threfc{d:space}{$(G,+)$ is an abelian group},
  we have
  \begin{equation*}
    w = v + \lambda u
    = v + \lambda P_F (u) + \lambda \up
    = \vp + \lambda \up
  \end{equation*}
  with $\vp\in F$.
  Hence, $w$ belongs to $F+\Line{\up}$, and thus
  $F+\Line{u}\subset F+\Line{\up}$.
  
  Let $w\in F+\Line{\up}$.
  Similarly, from
  Definition~\thref{d:sum-of-subspaces}, and
  Definition~\thref{d:linear-span},
  there exists $v\in F$ and $\lambda\in\matR$ such that $w=v+\lambda\up$;
  from
  Lemma~\threfc{l:closed-under-linear-combination-is-subspace}{%
    with~$1$ and~$-\lambda$},
  we have $\vp=v-\lambda P_F(u)\in F$, and thus, from
  Definition~\threfc{d:space}{$(G,+)$ is an abelian group},
  we have
  \begin{equation*}
    w = v + \lambda \up
    = v - \lambda P_F (u) + \lambda u
    = \vp + \lambda u
  \end{equation*}
  with $\vp\in F$.
  Hence, $w$ belongs to $F+\Line{u}$, and thus
  $F+\Line{\up}\subset F+\Line{u}$.
  
  Therefore, $F+\Line{u}=F+\Line{\up}$.
\end{proof}

\begin{lemma}[sum of complete subspace and linear span is closed]
  \label{l:sum-of-complete-subspace-and-linear-span-is-closed}
  Assume hypotheses of
  Theorem~\thref{t:orth-proj-onto-complete-subspace}.
  Let~$u$ be a nonzero vector in the orthogonal of~$F$.
  Then, $F\oplus\Line{u}$ is closed for distance~$d_G$.
\end{lemma}

\begin{proof}
  From
  Lemma~\thref{l:zero-intersection-with-orth-compl},
  $u$ does not belong to~$F$, thus from
  Lemma~\thref{l:direct-sum-with-linear-span},
  the sum $F+\Line{u}$ is direct.
  
  From
  Lemma~\threfc{l:orth-proj-is-continuous-linear-map}{%
    $F$ is complete for distance~$d_F$},
  $P_F$~is a continuous linear map.
  Then, from
  Lemma~\thref{l:identity-map-is-continuous},
  Theorem~\thref{t:normed-space-of-continuous-linear-maps}, and
  Lemma~\thref{l:closed-under-linear-combination-is-subspace},
  $Id-P_F$ is also a continuous linear map.
  
  Let $(w_n)_{n\in\matN}$ be a sequence in $F\oplus\Line{u}$.
  Assume that this sequence is convergent with limit $w\in G$.
  From
  Definition~\thref{d:sum-of-subspaces}, and
  Definition~\thref{d:linear-span},
  for all $n\in\matN$, there exists $v_n\in F$ and $\lambda_n\in\matR$ such
  that $w_n=v_n+\lambda_n u$.
  Then, from
  Lemma~\threfc{l:compatibility-of-limit-with-continuous-functions}{%
    $P_F$ and $Id-P_F$ are continuous},
  the sequences $(\wpr_n)_{n\in\matN}=P_F((w_n)_{n\in\matN})$ and
  $(\wpp_n)_{n\in\matN}=(Id-P_F)((w_n)_{n\in\matN})$ are also
  convergent, respectively with limits $\wpr=P_F(w)$ and
  $\wpp=(Id-P_F)(w)=w-\wpr$.
  From
  Theorem~\thref{t:direct-sum-with-orth-compl-when-complete},
  we have, $\wpr\in F$ and $\wpp\in F^\perp$, and for all $n\in\matN$,
  \begin{equation*}
    \wpp_n
    = (Id - P_F) (w_n)
    = (Id - P_F) (v_n + \lambda_n u)
    = v_n + \lambda_n u - v_n
    = \lambda_n u.
  \end{equation*}
  Thus, $(\wpp_n)_{n\in\matN}$ is also a sequence of $\Line{u}$.
  Then, from,
  Lemma~\threfc{l:linear-span-is-closed}{$\Line{u}$ is closed}, and
  Lemma~\thref{l:closed-is-limit-of-sequences},
  the limit~$\wpp$ actually belongs to~$\Line{u}$.
  Hence, from
  Definition~\thref{d:linear-span},
  there exists $\lambda\in\matR$ such that $\wpp=\lambda u$.
  And we have
  \begin{equation*}
    w = \wpr + \wpp = \wpr + \lambda u \in F \oplus \Line{u}.
  \end{equation*}
  
  Therefore, from
  Lemma~\thref{l:closed-is-limit-of-sequences},
  $F\oplus\Line{u}$ is closed for distance~$d_G$.
\end{proof}

\subsection{Hilbert space}
\label{ss:hilbert-space}

\begin{definition}[Hilbert space]
  \label{d:hilbert-space}
  Let $(H,\psHdotdot)$ be an inner product space.
  Let~$\nHdot$ be the norm associated with inner product~$\psHdotdot$ through
  Definition~\thref{d:square-root-of-inner-square}, and
  Lemma~\thref{l:inner-product-gives-norm}.
  Let~$d_H$ be the distance associated with norm~$\nHdot$ through
  Lemma~\thref{l:norm-gives-distance}.
  $(H,\psHdotdot)$ is an {\em Hilbert space} iff
  $(H,d_H)$ is a complete metric space.
\end{definition}

\begin{lemma}[closed Hilbert subspace]
  \label{l:closed-hilbert-subspace}
  Let~$(H,\psHdotdot)$ be a Hilbert space.
  Let~$\Hh$ be a closed subspace of~$H$.
  Then, $\Hh$~equipped with the restriction to~$\Hh$ of the inner product
  $\psHdotdot$ is a Hilbert space.
\end{lemma}

\begin{proof}
  Direct consequence of
  Lemma~\threfc{l:inner-product-subspace}{$\Hh$~is a subspace of~$H$},
  Definition~\threfc{d:subspace}{$\Hh$~is a subset of~$H$},
  Lemma~\threfc{d:hilbert-space}{$H$~is complete},
  Lemma~\threfc{l:closed-subset-of-complete-is-complete}{$F$~is closed}, and
  Definition~\thref{d:hilbert-space}.
\end{proof}

\begin{theorem}[{\RF}]
  \label{t:riesz-frechet}
  Let~$(H,\psHdotdot)$ be a Hilbert space.
  Let~$\nHdot$ be the norm associated with inner product $\psHdotdot$.
  Let $\fhi\in\Hp$ be a continuous linear form on~$H$.
  Then, there exists a unique vector $u\in H$ such that
  \begin{equation}
    \label{e:riesz-frechet-representation}
    \forall v \in H,\quad
    \pdH{\fhi}{v} = \psH{u}{v}.
  \end{equation}
  Moreover, the mapping $\tau:\Hp\rightarrow H$ defined by
  \begin{equation}
    \label{e:riesz-frechet-isometry}
    \forall \fhi \in \Hp,\quad
    \tau (\fhi) = u,
  \end{equation}
  where~$u$ is characterized by~\eqref{e:riesz-frechet-representation}, is a
  continuous isometric isomorphism from~$\Hp$ onto~$H$.
\end{theorem}

\begin{proof}
  From
  Definition~\threfc{d:hilbert-space}{%
    $(H,\psHdotdot)$~is an inner product space}, and
  Definition~\thref{d:inner-product-space},
  $H$~is a space.
  
  \proofparskip{Uniqueness}
  Let $u,\up\in H$ be two vectors such that
  \begin{equation*}
    \forall v \in H,\quad
    \pdH{\fhi}{v} = \fhi (v) = \psH{u}{v} = \psH{\up}{v}.
  \end{equation*}
  Let $v\in H$ be a vector.
  Then, from
  Definition~\threfc{d:inner-product}{$\psHdotdot$ is a bilinear map},
  Definition~\threfc{d:bilinear-map}{$\psHdotdot$ is left linear}, and
  Definition~\thref{d:vector-subtraction},
  we have $\psH{u-\up}{v}=0$.
  Thus, from
  Definition~\thref{d:orth-compl}, and
  Lemma~\thref{l:trivial-orth-compls},
  $u-\up$ belongs to $H^\perp=\{0_H\}$.
  Hence, from
  Definition~\threfc{d:space}{$(H,+)$ is an abelian group},
  $u=\up$.
  
  \proofparskip{Existence}
  
  \proofparskip{Case $\fhi=0_\Hp$}
  Then, from
  Definition~\threfc{d:space}{$0_H$~belongs to~$H$},
  let $u=0_H$ be the zero vector.
  Let $v\in H$ be a vector.
  Then, from
  Lemma~\threfc{l:trivial-orth-compls}{$H^\perp=\{0_H\}$},
  we have
  \begin{equation*}
    \pdH{\fhi}{v}
    = \fhi (v)
    = 0_\Hp (v)
    = 0
    = \psH{0_H}{v}
    = \psH{u}{v}.
  \end{equation*}
  
  \proofparskip{Case $\fhi\not=0_\Hp$}
  Then, let $u_0\in H$ such that $\fhi(u_0)\not=0$.
  Let~$F$ be the kernel of~$\fhi$.
  Then, from
  Definition~\thref{d:kernel},
  $u_0\not\in F$.
  Moreover, from
  Lemma~\threfc{l:continuous-linear-maps-have-closed-kernel}{for~$\fhi$}, and
  Lemma~\thref{l:kernel-is-subspace},
  $F$~is a closed subspace of~$H$.
  Thus, from
  Lemma~\thref{l:closed-hilbert-subspace},
  $F$~is a complete subspace of~$H$.
  Hence, from
  Theorem~\thref{t:orth-proj-onto-complete-subspace}, and
  Definition~\thref{d:orth-proj-onto-complete-subspace},
  let~$P_F$ be the orthogonal projection onto~$F$.
  Then, from
  Theorem~\threfc{t:direct-sum-with-orth-compl-when-complete}{%
    decomposition and contrapositive
    of~\eqref{e:characterization-orth-compl-f}},
  we have
  \begin{equation*}
    P_F (u_0) \in F,\quad
    u_0 - P_F (u_0) \in F^\perp
    \quad \mbox{and} \quad
    P_F (u_0) \not= u_0.
  \end{equation*}
  Thus, from
  Definition~\threfc{d:kernel}{$F=\Ker{\fhi}$},
  Definition~\threfc{d:linear-form}{$\fhi$~is a linear map},
  Definition~\threfc{d:linear-map}{$\fhi$~is additive}, and
  Definition~\thref{d:vector-subtraction},
  we have
  \begin{equation*}
    \fhi (P_F (u_0)) = 0
    \quad \mbox{and} \quad
    \fhi (u_0 - P_F (u_0)) = \fhi (u_0).
  \end{equation*}
  
  Let $v_0=u_0-P_F(u_0)$.
  Then,
  \begin{equation*}
    v_0 \in F^\perp
    \quad \mbox{and} \quad
    \fhi (v_0) = \fhi (u_0) \not= 0.
  \end{equation*}
  Moreover, from
  Definition~\threfc{d:space}{$(H,+)$ is an abelian group}, and
  Definition~\thref{d:vector-subtraction},
  we have $v_0\not=0_H$.
  Thus, from
  Definition~\threfc{d:norm}{$\nHdot$ is definite, contrapositive},
  we have $\nH{v_0}\not=0$.
  
  Let $\xi_0=\frac{v_0}{\nH{v_0}}$.
  Then, from
  Lemma~\threfc{l:orth-compl-is-subspace}{$F^\perp$~is subspace},
  Lemma~\threfc{l:closed-under-vector-operations-is-subspace}{%
    $F^\perp$~is closed under scalar multiplication},
  Definition~\thref{d:scalar-division},
  Definition~\threfc{d:linear-form}{$\fhi$~is a linear map},
  Definition~\threfc{d:linear-map}{$\fhi$~is homogeneous of degree~1},
  Definition~\thref{d:scalar-division},
  Lemma~\threfc{l:zero-product-property}{contrapositive}, and
  \assume{field properties of~$\matR$},
  we have
  \begin{equation*}
    \xi_0 \in F^\perp,\quad
    \fhi (\xi_0)
    = \frac{\fhi (v_0)}{\nH{v_0}}
    \not= 0
    \quad \mbox{and} \quad
    \xi_0 \not= 0_H.
  \end{equation*}
  Moreover, from
  Lemma~\threfc{l:normalization-by-nonzero}{with $\lambda=1$}, and
  \assume{field properties of~$\matR$},
  we have $\nH{\xi_0}^2=1$.
  
  Let $u=\fhi(\xi_0)\xi_0$.
  Then, from
  Lemma~\threfc{l:orth-compl-is-subspace}{%
    $F^\perp$~is subspace}, and
  Lemma~\threfc{l:closed-under-vector-operations-is-subspace}{%
    $F^\perp$~is closed under scalar multiplication},
  $u\in F^\perp$.
  
  Let $v\in H$ be a vector.
  Since $\fhi(\xi_0)\not=0$, let $\lambda=\frac{\fhi(v)}{\fhi(\xi_0)}$ and
  $w=v-\lambda\xi_0$.
  Then, from
  Definition~\threfc{d:linear-form}{$\fhi$~is a linear map},
  Lemma~\thref{l:linear-map-preserves-linear-combinations},
  Definition~\thref{d:vector-subtraction}, and
  Definition~\thref{d:scalar-division},
  we have
  \begin{equation*}
    \fhi (w)
    = \fhi (v) - \lambda \fhi (\xi_0)
    = \fhi (v) - \frac{\fhi (v)}{\fhi (\xi_0)} \fhi (\xi_0)
    = 0.
  \end{equation*}
  Thus, from
  Definition~\threfc{d:kernel}{$F=\Ker{\fhi}$},
  $w$~belongs to~$F$.
  Hence, from
  \assume{field properties of~$\matR$ (with $\fhi(\xi_0)\not=0$)},
  Lemma~\threfc{l:squared-norm}{$\nH{\xi_0}^2=1$},
  Definition~\threfc{d:inner-product}{$\psHdotdot$ is a bilinear map},
  Definition~\threfc{d:bilinear-map}{$\psHdotdot$ is left linear}, and
  Definition~\threfc{d:orth-compl}{%
    $u\in F^\perp$ and $w\in F$},
  we have
  \begin{eqnarray*}
    \psH{u}{v} - \fhi (v)
    & = &
    \psH{u}{v}
    - \fhi (v) \frac{\fhi (\xi_0)}{\fhi (\xi_0)} \psH{\xi_0}{\xi_0} \\
    & = & \psH{u}{v} - \lambda \psH{u}{\xi_0} \\
    & = & \psH{u}{v - \lambda \xi_0} \\
    & = & \psH{u}{w} \\
    & = & 0.
  \end{eqnarray*}
  Hence, from
  Definition~\thref{d:bra-ket-notation}, and
  \assume{field properties of~$\matR$},
  we have
  \begin{equation*}
    \pdH{\fhi}{v} = \fhi (v) = \psH{u}{v}.
  \end{equation*}
  
  \proofparskip{Linearity}
  From
  Lemma~\threfc{l:topological-dual-is-complete-normed-space}{for~$H$}, and
  Definition~\thref{d:normed-space},
  $\Hp$~is a space.
  
  Let $\tau:\Hp\rightarrow H$ be the mapping defined by, for all
  $\fhi\in\Hp$, $\tau(\fhi)=u$ where~$u$ is uniquely characterized by
  \begin{equation}
    \label{e:characterization-tau-fhi}
    \forall v \in H,\quad
    \pdH{\fhi}{v} = \fhi (v) = \psH{u}{v}.
  \end{equation}
  
  Let $\lambdap,\lambdapp\in\matK$ be scalars.
  Let $\fhip,\fhipp\in\Hp$ be continuous linear forms on~$H$.
  Then, $\tau(\fhip)$ and $\tau(\fhipp)$ belong to~$H$.
  Thus, from
  Definition~\threfc{d:space}{%
    $\Hp$ and~$H$ are closed under vector operations},
  $\fhi=\lambdap\fhip+\lambdapp\fhipp$ belongs to~$\Hp$ and
  $u=\lambdap\tau(\fhip)+\lambdapp\tau(\fhipp)$ belongs to~$H$.
  Let $v\in H$ be a vector.
  Then, from
  Lemma~\thref{l:bra-ket-is-bilinear-map},
  Definition~\threfc{d:inner-product}{$\psHdotdot$ is a bilinear map},
  Definition~\threfc{d:bilinear-map}{%
    bra-ket and $\psHdotdot$ are left linear}, and
  characterization~\eqref{e:characterization-tau-fhi}, we have
  \begin{eqnarray*}
    \pdH{\fhi}{v}
    & = & \pdH{\lambdap \fhip + \lambdapp \fhipp}{v} \\
    & = & \lambdap \pdH{\fhip}{v} + \lambdapp \pdH{\fhipp}{v} \\
    & = & \lambdap \psH{\tau (\fhip)}{v} + \lambdapp \psH{\tau (\fhipp)}{v} \\
    & = & \psH{\lambdap \tau (\fhip) + \lambdapp \tau (\fhipp)}{v} \\
    & = & \psH{u}{v}.
  \end{eqnarray*}
  Hence, from
  unique characterization~\eqref{e:characterization-tau-fhi},
  we have
  \begin{equation*}
    \tau (\lambdap \fhip + \lambdapp \fhipp)
    = \tau (\fhi)
    = u
    = \lambdap \tau (\fhip) + \lambdapp \tau (\fhipp).
  \end{equation*}
  
  Therefore, from
  Definition~\thref{d:linear-map},
  $\tau$~is a linear map from~$\Hp$ to~$H$.
  
  \proofparskip{Isomorphism}
  Let $\fhi\in\Hp$ be a continuous linear form on~$H$.
  Assume that $\tau(\fhi)=0_H$.
  Let $v\in H$ be a vector.
  Then, from
  characterization~\eqref{e:characterization-tau-fhi}, and
  Lemma~\thref{l:inner-product-with-zero-is-zero},
  we have
  \begin{equation*}
    \fhi (v)
    = \psH{\tau (\fhi)}{v}
    = \psH{0_H}{v}
    = 0.
  \end{equation*}
  Thus, $\fhi=0_\Hp$ is the zero linear form.
  Hence, from
  Definition~\threfc{d:kernel}{$\Ker{\tau}=\{0_\Hp\}$}, and
  Lemma~\thref{l:injective-linear-map-has-zero-kernel},
  $\tau$~is injective.
  
  Let $u\in H$ be a vector.
  Let $\fhi:H\rightarrow\matK$ be the mapping defined by, for all $v\in H$,
  $\fhi(v)=\psH{u}{v}$.
  Then, from
  Definition~\threfc{d:inner-product}{$\psHdotdot$ is a bilinear map},
  Definition~\threfc{d:bilinear-map}{$\psHdotdot$ is right linear}, and
  Definition~\thref{d:linear-form},
  $\fhi$~is a linear form on~$H$.
  Let $v\in H$ be a vector.
  Then, from
  Lemma~\thref{l:cauchy-schwarz-inequality-with-norms},
  we have
  \begin{equation*}
    | \fhi (v) |
    = | \psH{u}{v} |
    \leq \nH{u} \, \nH{v}.
  \end{equation*}
  Thus, from
  Definition~\threfc{d:bounded-linear-map}{with $C=\nH{u}\geq 0$}, and
  Theorem~\threfc{t:continuous-linear-map}{%
    $\ref{i:bounded}\implies \ref{i:cont}$},
  $\fhi$ is continuous.
  Hence, from
  Definition~\thref{d:topological-dual},
  $\fhi$~belongs to~$\Hp$.
  Moreover, from
  characterization~\eqref{e:characterization-tau-fhi},
  we have $\tau(\fhi)=u$.
  Hence, from
  \assume{the definition of a surjective function},
  $\tau$~is surjective.
  
  Therefore, from
  \assume{the definition of a bijective function},
  $\tau$~is bijective, and from
  Definition~\thref{d:isomorphism},
  $\tau$~is an isomorphism from~$\Hp$ onto~$H$.
  
  \proofparskip{Isometry}
  Let $\fhi\in\Hp$ be a continuous linear form on~$H$.
  Let $u=\tau(\fhi)\in H$.
  
  \proofparskip{Case $\fhi=0_\Hp$}
  Then, from
  Lemma~\threfc{l:linear-map-preserves-zero}{$\tau$~is a linear map},
  we have $u=0_H$.
  Hence, from
  Lemma~\threfc{l:norm-preserves-zero}{for $\nHpdot$ and $\nHdot$},
  we have
  \begin{equation*}
    \nH{\tau (\fhi)} = \nH{u} = 0 = \nHp{\fhi}.
  \end{equation*}
  
  \proofparskip{Case $\fhi\not=0_\Hp$}
  Then, from
  Definition~\threfc{d:kernel}{$\Ker{\fhi}=\{0_\Hp\}$},
  we have $u\not=0_H$.
  Thus, from
  Definition~\threfc{d:norm}{$\nHdot$~is definite, contrapositive},
  $\nH{u}\not=0$.
  Hence, from
  characterization~\eqref{e:characterization-tau-fhi} (with $v=u$),
  Lemma~\threfc{l:squared-norm}{for $\nHdot$},
  \assume{nonnegativeness of the square function in~$\matR$}, and
  \assume{field properties of~$\matR$ (with $\nH{u}\not=0$)},
  we have
  \begin{equation*}
    \frac{| \fhi (u) |}{\nH{u}}
    = \frac{| \psH{u}{u} |}{\nH{u}}
    = \frac{\nH{u}^2}{\nH{u}}
    = \nH{u}.
  \end{equation*}
  Hence, from
  Definition~\thref{d:dual-norm},
  Definition~\thref{d:operator-norm}, and
  Definition~\threfc{d:supremum}{%
    $\nHp{\fhi}$ is an upper bound for
    $\left\{\left.\frac{|\fhi(v)|}{\nH{v}}\,\right|\,
      v\in H,\,v\not=0_H\right\}$},
  we have
  \begin{equation*}
    \nH{u} \leq \nHp{\fhi}.
  \end{equation*}
  Finally, let $v\in H$ be a vector.
  Assume that $v\not=0_H$.
  Then, from
  Definition~\threfc{d:norm}{$\nHdot$ is definite, contrapositive},
  Lemma~\threfc{l:norm-is-nonnegative}{for $\nHdot$},
  Lemma~\thref{l:cauchy-schwarz-inequality-with-norms}, and
  \assume{ordered field properties of~$\matR$ (with $\nH{v}>0$)},
  we have
  \begin{equation*}
    \frac{| \fhi (v) |}{\nH{v}}
    = \frac{| \psH{u}{v} |}{\nH{v}}
    \leq \nH{u}.
  \end{equation*}
  Thus, from
  Lemma~\threfc{l:finite-operator-norm-is-continuous}{%
    $\nH{u}$ is an upper bound for the subset
    $\left\{\left.\frac{|\fhi(v)|}{\nH{v}}\,\right|\,
      v\in H,\,v\not=0_H\right\}$}, and
  Definition~\thref{d:dual-norm},
  we have
  \begin{equation*}
    \nHp {\fhi} \leq \nH{u}.
  \end{equation*}
  Hence, $\nH{\tau(\fhi)}=\nH{u}=\nHp{\fhi}$.
  
  Therefore, from
  Definition~\thref{d:linear-isometry},
  $\tau$~is a linear isometry from~$\Hp$ to~$H$.
  
  \proofparskip{Continuity}
  From
  Lemma~\threfc{l:linear-isometry-is-continuous}{%
    $\tau$~is a linear isometry from~$\Hp$ to~$H$},
  $\tau$~belongs to $\LcHpH$.
\end{proof}

\begin{lemma}[compatible $\rho$ for {\LM}]
  \label{l:compatible-rho-for-lax-milgram}
  Let $\alpha,C\in\matR$.
  Assume that $0<\alpha\leq C$.
  Then,
  \begin{equation}
    \label{e:compatible-rho-for-lax-milgram}
    \forall \rho \in \matR,\quad
    0 < \rho < \frac{2\alpha}{C^2}
    \Implies
    0 \leq \sqrt{1 - 2 \rho \alpha + \rho^2 C^2} < 1.
  \end{equation}
\end{lemma}

\begin{proof}
  From
  hypothesis ($0<\alpha\leq C$),
  \assume{ordered field properties of~$\matR$}, and
  \assume{increase of the square function over~$\matRplus$},
  we have
  $0<\frac{\alpha^2}{C^2}\leq 1$.
  Let $\rho\in\matR$.
  Then, from
  \assume{field properties of~$\matR$},
  we have
  \begin{equation*}
    1 - 2 \rho \alpha + \rho^2 C^2
    = \left( \rho C - \frac{\alpha}{C} \right)^2 + 1 - \frac{\alpha^2}{C^2}
    \geq 0.
  \end{equation*}
  
  Assume that $0<\rho<\frac{2\alpha}{C^2}$.
  Then, from
  \assume{ordered field properties of~$\matR$ (with $C>0$ and $\rho>0$)},
  we successively have $ \rho C^2<2\alpha$, $\rho^2C^2<2\rho\alpha$
  and $1-2\rho\alpha+\rho^2C^2<1$.
  Hence, from
  \assume{compatibility of the square root with comparison in~$\matRplus$},
  we have
  \begin{equation*}
    0 = \sqrt{0}
    \leq \sqrt{1 - 2 \rho \alpha + \rho^2 C^2}
    < \sqrt{1} = 1.
  \end{equation*}
\end{proof}

\begin{theorem}[{\LM}]
  \label{t:lax-milgram}
  Let $(H,\psHdotdot)$ be a real Hilbert space.
  Let $\nHdot$ be the norm associated with inner product $\psHdotdot$.
  Let~$\Hp$ be the topological dual of~$H$.
  Let $\nHpdot$ be the dual norm associated with~$\nHdot$.
  Let~$\BlfH$ be a bounded bilinear form on~$H$.
  Let $\LfH\in\Hp$ be a continuous linear form on~$H$.
  Assume that~$\BlfH$ is coercive with constant $\alpha>0$.
  Then, there exists a unique $u\in H$ solution to
  Problem~\eqref{e:general-problem}.
  Moreover,
  \begin{equation}
    \label{e:lax-milgram-estimation}
    \nH{u} \leq \frac{1}{\alpha} \nHp{\LfH}.
  \end{equation}
\end{theorem}

\begin{proof}
  Let~$d_H$ be the distance associated with norm $\nHdot$.
  Then, from
  Definition~\thref{d:hilbert-space},
  $(H,\psHdotdot)$~is an inner product space and $(H,d_H)$ is a complete
  metric space.
  Thus, from
  Lemma~\thref{l:inner-product-gives-norm},
  $(H,\nHdot)$ is a {\normedvectorspace}.
  Moreover, from
  Lemma~\thref{l:topological-dual-is-complete-normed-space},
  $(\Hp,\nHpdot)$ is a also {\normedvectorspace}.
  Hence, from
  Definition~\thref{d:normed-space},
  $H$ and~$\Hp$ are both spaces.
  
  \proofparskip{Existence and uniqueness}
  From
  Lemma~\threfc{l:representation-for-bounded-bilinear-form}{for~$a$},
  let $A\in\LcHHp$ be the (unique) continuous linear map from~$H$ to~$\Hp$
  such that
  \begin{equation*}
    \forall u, v \in H,\quad
    \blfH{u}{v} = \pdH{A (u)}{v}.
  \end{equation*}
  Then, from
  Definition~\threfc{d:coercive-bilinear-form}{for~$\BlfH$},
  we have
  \begin{equation}
    \label{e:coercivity-A}
    \forall u \in H,\quad
    \pdH{A (u)}{u} = \blfH{u}{u} \geq \alpha \nH{u}^2.
  \end{equation}
  From
  Definition~\thref{d:bounded-bilinear-form},
  let $C\geq 0$ be a continuity constant of~$\BlfH$.
  Then, from
  Lemma~\threfc{l:operator-norm-estimation}{in $\LcHHp$}, and
  Lemma~\threfc{l:representation-for-bounded-bilinear-form}{for~$\BlfH$},
  we have
  \begin{equation}
    \label{e:continuity-A}
    \forall u \in H,\quad
    \nHp{A (u)} \leq \tnHHp{A} \, \nH{u} \leq C \nH{u}.
  \end{equation}
  
  Let $u\in H$ be a vector.
  Then, from
  Theorem~\threfc{t:riesz-frechet}{for $\fhi=A(u)$ and $\fhi=\LfH$},
  $\tau(A(u)),\tau(\LfH)\in H$ are the (unique) vectors such that
  \begin{eqnarray*}
    \forall v \in H, & &
    \blfH{u}{v} = \pdH{A (u)}{v} = \psH{\tau (A (u))}{v}; \\
    \forall v \in H, & &
    \lfH{v} = \pdH{\LfH}{v} = \psH{\tau (\LfH)}{v}.
  \end{eqnarray*}
  Moreover, from
  \eqref{e:coercivity-A},
  \assume{ordered field properties of~$\matR$},
  \eqref{e:continuity-A}, and
  Definition~\threfc{d:linear-isometry}{$\tau$ is a linear isometry},
  we have
  \begin{eqnarray}
    \label{e:coercivity-tau-A}
    \forall u \in H, & &
    - \psH{\tau (A (u))}{u}
    = - \pdH{A (u)}{u}
    \leq - \alpha \nH{u}^2; \\
    \label{e:continuity-tau-A}
    \forall u \in H, & &
    \nH{\tau (A (u))} = \nHp{A (u)} \leq C \nH{u}.
  \end{eqnarray}
  
  Let $u,v\in H$ be vectors.
  Then, from
  Definition~\threfc{d:inner-product}{$\psHdotdot$ is a bilinear map}, and
  Definition~\threfc{d:bilinear-map}{$\psHdotdot$ is left linear},
  we have the equivalences
  \begin{eqnarray*}
    \blfH{u}{v} = \lfH{v}
    & \equiv & \psH{\tau (A (u))}{v} = \psH{\tau (\LfH)}{v} \\
    & \equiv & \psH{\tau (A (u)) - \tau (\LfH)}{v} = 0.
  \end{eqnarray*}
  Hence, from
  Definition~\threfc{d:orth-compl}{%
    $\tau(A(u))-\tau(f)$ belongs to~$H^\perp$},
  Lemma~\threfc{l:trivial-orth-compls}{$H^\perp=\{0_H\}$}, and
  Definition~\threfc{d:space}{$(H,+)$ is an abelian group},
  we have the equivalence
  \begin{equation}
    \label{e:equivalent-problem-tau-A}
    Problem~\eqref{e:general-problem}
    \Equiv
    \mbox{find } u \in H \mbox{ such that:}\quad
    \tau (A (u)) = \tau (f).
  \end{equation}
  
  From
  Lemma~\threfc{l:compatibility-of-composition-with-continuity}{%
    $\tau$~belongs to $\LcHpH$},
  $\tau\circ A$ belongs to $\LcHH$.
  From
  Lemma~\thref{l:coercivity-constant-is-less-than-continuity-constant},
  we have $0<\alpha\leq C$, hence $\frac{2\alpha}{C^2}>0$.
  Let $\rho\in\matR$ be a number.
  Assume that $0<\rho<\frac{2\alpha}{C^2}$.
  Then, from
  Theorem~\threfc{t:normed-space-of-continuous-linear-maps}{%
    $(\LcHH,\tnHHdot)$ is a {\normedvectorspace}},
  Definition~\threfc{d:normed-space}{$\LcHH$~is a space},
  Definition~\threfc{d:space}{$\LcHH$ is closed under vector operations},
  Definition~\thref{d:vector-subtraction}, and
  Lemma~\thref{l:identity-map-is-continuous},
  $g_0=\idH-\rho\tau\circ A$ belongs to~$\LcHH$.
  
  From
  Definition~\threfc{d:space}{%
    $H$~is closed under vector operations and $\tau(\LfH)\in H$},
  let $g:H\rightarrow H$ be the mapping defined by
  \begin{equation*}
    \forall v \in H,\quad
    g (v) = g_0 (v) + \rho \tau (\LfH).
  \end{equation*}
  Let $u\in H$ be a vector.
  Then, from
  the definition of mappings~$g$ and~$g_0$,
  Definition~\threfc{d:space}{%
    $(H,+)$ is an abelian group and scalar multiplication is distributive wrt
    vector addition}, and
  Lemma~\threfc{l:zero-product-property}{with $\lambda=\rho\not=0$},
  we have
  \begin{eqnarray*}
    g (u) = u
    & \equiv & g_0 (u) + \rho \tau (\LfH) = u \\
    & \equiv & u - \rho \tau (A (u)) + \rho \tau (\LfH) = u \\
    & \equiv & \rho (\tau (A (u)) - \tau (\LfH)) = 0_H \\
    & \equiv & \tau (A (u)) = \tau (\LfH).
  \end{eqnarray*}
  Hence, from
  \eqref{e:equivalent-problem-tau-A},
  we have the equivalence
  \begin{equation}
    \label{e:equivalent-problem-g}
    Problem~\eqref{e:general-problem}
    \Equiv
    \mbox{find } u \in H \mbox{ such that:}\quad
    g (u) = u.
  \end{equation}
  
  Let $v,\vp\in H$ be vectors.
  Then, from
  Definition~\threfc{d:space}{$(H,+)$ is an abelian group}, and
  Definition~\thref{d:vector-subtraction},
  let $z=v-\vp\in H$.
  Then, from
  Definition~\thref{d:vector-subtraction},
  Lemma~\threfc{l:minus-times-yields-opposite-vector}{with $\lambda=1$}, and
  Definition~\threfc{d:space}{%
    $(H,+)$ is an abelian group and scalar multiplication is distributive wrt
    vector addition},
  we have
  \begin{equation*}
    g (v) - g (\vp)
    = g_0 (v) + \rho \tau (\LfH) - (g_0 (\vp) + \rho \tau (\LfH))
    = g_0 (v - \vp)
    = g_0 (z).
  \end{equation*}
  Thus, from
  Lemma~\threfc{l:square-expansion-plus}{for $\nHdot$},
  Definition~\threfc{d:inner-product}{%
    $\psHdotdot$ is a symmetric bilinear map},
  Definition~\threfc{d:bilinear-map}{$\psHdotdot$ is right linear},
  Definition~\threfc{d:norm}{$\nHdot$ is absolutely homogeneous of degree~1},
  \assume{ordered field properties of~$\matR$},
  \eqref{e:coercivity-tau-A}, and
  \eqref{e:continuity-tau-A},
  we have
  \begin{eqnarray*}
    \nH{g(v) - g(\vp)}^2
    & = & \nH{g_0 (z)}^2 \\
    & = & \nH{z - \rho \tau (A (z))}^2 \\
    & = &
    \nH{z}^2 - 2 \rho \psH{\tau (A (z))}{z} + \rho^2 \nH{\tau (A (z))}^2 \\
    & \leq & \nH{z}^2 - 2 \rho \alpha \nH{z}^2 + \rho^2 C^2 \nH{z}^2 \\
    & = & (1 - 2 \rho \alpha + \rho^2 C^2) \nH{v - \vp}^2.
  \end{eqnarray*}
  Hence, from
  \assume{compatibility of the square root function with comparison
    in~$\matRplus$},
  Definition~\threfc{d:lipschitz-continuity}{%
    with $k=\sqrt{1-2\rho\alpha+\rho^2C^2}$},
  Lemma~\threfc{l:compatible-rho-for-lax-milgram}{%
    since $0<\alpha\leq C$ and $0<\rho<\frac{2\alpha}{C^2}$}, and
  Definition~\threfc{d:contraction}{%
    since $0\leq\sqrt{1-2\rho\alpha+\rho^2C^2}<1$},
  $g$~is a contraction.
  Then, from
  Theorem~\threfc{t:fixed-point}{for~$g$ contraction in $(H,d_H)$ complete},
  there exists a unique fixed point $u\in H$ such that $g(u)=u$.
  Hence, from
  \eqref{e:equivalent-problem-g},
  there exists a unique solution to Problem~\eqref{e:general-problem}.
  
  \proofparskip{Estimation}
  Let $u\in H$ be the solution to Problem~\eqref{e:general-problem}.
  
  \proofparskip{Case $u=0_H$}
  Then, from
  Lemma~\threfc{l:norm-preserves-zero}{for $\nHdot$},
  Lemma~\threfc{l:norm-is-nonnegative}{for $\nHpdot$}, and
  \assume{ordered field properties of~$\matR$ (with $\alpha>0$)},
  we have
  \begin{equation*}
    \nH{u} = 0 \leq \frac{1}{\alpha} \nHp{\LfH}.
  \end{equation*}
  
  \proofparskip{Case $u\not=0_H$}
  Then, from
  Definition~\threfc{d:norm}{$\nHdot$ is definite, contrapositive}, and
  Lemma~\threfc{l:norm-is-nonnegative}{for $\nHdot$},
  we have $\nH{u}>0$.
  Moreover, from
  Definition~\threfc{d:coercive-bilinear-form}{for~$\BlfH$},
  \assume{properties of the absolute value on~$\matR$},
  \eqref{e:general-problem} with $v=u$, and
  Lemma~\threfc{l:operator-norm-estimation}{for~$\LfH\in\Hp$},
  we have
  \begin{equation*}
    \alpha \nH{u}^2
    \leq \blfH{u}{u}
    \leq | \blfH{u}{u} |
    = | \lfH{u} |
    \leq \nHp{f} \, \nH{u}.
  \end{equation*}
  Hence, from
  \assume{ordered field properties of~$\matR$ (with $\nH{u},\alpha>0$)},
  we have the estimation
  \begin{equation*}
    \nH{u} \leq \frac{1}{\alpha} \nHp{\LfH}.
  \end{equation*}
\end{proof}

\begin{lemma}[Galerkin orthogonality]
  \label{l:galerkin-orthogonality}
  Let~$(H,\psHdotdot)$ be a real Hilbert space.
  Let $\BlfH$ be a bounded bilinear form on~$H$.
  Let $\LfH\in\Hp$ be a continuous linear form on~$H$.
  Let $\Hh$ be a subspace of~$H$.
  Let $u\in H$ be a solution to Problem~\eqref{e:general-problem}.
  Let $\uh\in\Hh$ be a solution to
  Problem~\eqref{e:general-problem-discrete}.
  Then,
  \begin{equation}
    \label{e:galerkin-orthogonality}
    \forall \vh \in \Hh,\quad
    \blfH{u - \uh}{\vh} = 0.
  \end{equation}
\end{lemma}

\begin{proof}
  Let $\vh\in\Hh$ be a vector.
  Then, from
  Definition~\threfc{d:subspace}{$\Hh$~is a subset of~$H$},
  $\vh$ also belongs to~$H$.
  Hence, from
  \eqref{e:general-problem} with $v=\vh$,
  \eqref{e:general-problem-discrete},
  Definition~\threfc{d:bilinear-map}{$\BlfH$~is left linear}, and
  \assume{field properties of~$\matR$},
  we have
  \begin{equation*}
    \blfH{u - \uh}{\vh}
    = \blfH{u}{\vh} - \blfH{\uh}{\vh}
    = \lfH{\vh} - \lfH{\vh}
    = 0.
  \end{equation*}
\end{proof}

\begin{theorem}[{\LM}, closed subspace]
  \label{t:lax-milgram-closed-subspace}
  Assume hypotheses of
  Theorem~\thref{t:lax-milgram}.
  Let~$\Hh$ be a closed subspace of~$H$.
  Then, there exists a unique $\uh\in\Hh$ solution to
  Problem~\eqref{e:general-problem-discrete}.
  Moreover,
  \begin{equation}
    \label{e:lax-milgram-closed-subspace-estimation}
    \nH{\uh} \leq \frac{1}{\alpha} \nHp{\LfH}.
  \end{equation}
\end{theorem}

\begin{proof}
  Direct consequence of
  Lemma~\threfc{l:closed-hilbert-subspace}{%
    $\Hh$~is a closed subspace of~$H$}, and
  Theorem~\threfc{t:lax-milgram}{$(\Hh,\psHdotdot)$ is a Hilbert space}
  where the restriction to~$\Hh$ of the norm associated to $\psHdotdot$ is
  still denoted $\nHdot$.
\end{proof}

\begin{lemma}[Céa]
  \label{l:cea}
  Assume hypotheses of
  Theorem~\thref{t:lax-milgram-closed-subspace}.
  Let $C\geq 0$ be a continuity constant of the bounded bilinear
  form~$\BlfH$.
  Let $u\in H$ be the unique solution to Problem~\eqref{e:general-problem}.
  Let $\uh\in\Hh$ be the unique solution to
  Problem~\eqref{e:general-problem-discrete}.
  Then,
  \begin{equation}
    \label{e:cea-closed-subspace-error}
    \forall \vh \in \Hh,\quad
    \nH{u - \uh} \leq \frac{C}{\alpha} \nH{u - \vh}.
  \end{equation}
\end{lemma}

\begin{proof}
  Let $\vh\in\Hh$ be a vector in the subspace.
  
  \proofparskip{Case $u=\uh$}
  Then, from
  Definition~\threfc{d:space}{$(H,+)$ is an abelian group},
  Lemma~\threfc{l:norm-preserves-zero}{$u-\uh=0_H$},
  Lemma~\threfc{l:norm-is-nonnegative}{for $\nHdot$}, and
  \assume{ordered field properties of~$\matR$ with $\alpha>0$ and $C\geq 0$},
  we have
  \begin{equation*}
    \nH{u - \uh} = 0 \leq \frac{C}{\alpha} \nH{u - \vh}.
  \end{equation*}
  
  \proofparskip{Case $u\not=\uh$}
  Then, from
  Definition~\threfc{d:space}{$(H,+)$ is an abelian group}, and
  Definition~\threfc{d:norm}{$\nHdot$~is definite, contrapositive},
  we have $\nH{u-\uh}\not=0$.
  Moreover, from
  Definition~\thref{d:vector-subtraction},
  Definition~\threfc{d:inner-product}{$\psHdotdot$~is a bilinear map},
  Definition~\threfc{d:bilinear-map}{$\psHdotdot$~is right linear}, and
  Lemma~\thref{l:galerkin-orthogonality},
  we have
  \begin{equation}
    \label{e:galerkin-orthogonality-bis}
    \blfH{u - \uh}{u - \vh}
    = \blfH{u - \uh}{u} - \blfH{u - \uh}{\vh}
    = \blfH{u - \uh}{u}.
  \end{equation}
  Thus, from
  Definition~\threfc{d:coercive-bilinear-form}{for~$\BlfH$ with $u=u-\uh$},
  \assume{properties of the absolute value on~$\matR$},
  \eqref{e:galerkin-orthogonality-bis} with~$\uh$ and~$\vh$ in~$\Hh$,
  compatibility of the absolute value with comparison in~$\matR$, and
  Definition~\thref{d:bounded-bilinear-form},
  we have
  \begin{eqnarray*}
    \alpha \nH{u - \uh}^2
    & \leq & \blfH{u - \uh}{u - \uh} \\
    & \leq & | \blfH{u - \uh}{u - \uh} | \\
    & = & | \blfH{u - \uh}{u} | \\
    & = & | \blfH{u - \uh}{u - \vh} | \\
    & \leq & C \nH{u - \uh} \, \nH{u - \vh}.
  \end{eqnarray*}
  Hence, from
  \assume{ordered field properties of~$\matR$ with $\alpha,\nH{u-\uh}>0$},
  we have
  \begin{equation*}
    \nH{u - \uh} \leq \frac{C}{\alpha} \nH{u - \vh}.
  \end{equation*}
\end{proof}

\begin{lemma}[finite dimensional subspace in Hilbert space is closed]
  \label{l:finite-dimensional-subspace-in-hilbert-space-is-closed}
  Let $(H,\psGdotdot)$ be a real Hilbert space.
  Let~$\nHdot$ be the norm associated with inner product~$\psHdotdot$.
  Let~$d_H$ be the distance associated with norm~$\nHdot$.
  Let~$F$ be a subspace of~$H$.
  Assume that~$F$ is a finite dimensional subspace.
  Then, $F$~is closed for distance~$d_H$.
\end{lemma}

\begin{proof}
  From
  Definition~\thref{d:finite-dimensional-subspace},
  let $n\in\matN$, and let $u_1,\ldots,u_n\in H$ such that
  $F=\Span{\{u_1,\ldots,u_n\}}=\Line{u_1}+\ldots+\Line{u_n}$.
  For $i\in\matN$ with $1\leq i\leq n$, let $F_i=\Span{\{u_1,\ldots,u_i\}}$.
  Then, for $2\leq i\leq n$, we have $F_i=F_{i-1}+\Line{u_i}$.
  Let $P(i)$ be the property ``$F_i$~is closed for distance~$d_H$''.
  
  \proofparskip{Induction: $P(1)$}
  From
  Lemma~\thref{l:linear-span-is-closed},
  $F_1=\Line{u_1}$ is closed for distance~$d_H$.
  
  \proofparskip{Induction: $P(i-1)$ implies $P(i)$}
  Assume that $2\leq i\leq n$.
  Assume that $P(i-1)$ holds.
  Then, from
  Lemma~\thref{l:sum-is-orth-sum},
  there exists $\up_i\in F_{i-1}^\perp$ such that
  \begin{equation*}
    F_i = F_{i - 1} + \Line{u_i} = F_{i - 1} + \Line{\up_i}.
  \end{equation*}
  \proofpar{Case $\up_i=0_G$}
  Then, $F_i=F_{i-1}$ is closed for distance~$d_G$.
  \proofpar{Case $\up_i\not=0_G$}
  Then, from
  Definition~\threfc{d:hilbert-space}{$H$ is complete for distance~$d_H$}, and
  Lemma~\thref{l:closed-subset-of-complete-is-complete},
  $F_{i-1}$ is complete for distance~$d_H$.
  Thus, from
  Lemma~\thref{l:sum-of-complete-subspace-and-linear-span-is-closed},
  $F_i=F_{i-1}+\Line{\up_i}$ is closed for distance~$d_G$.
  
  Hence, by (finite) induction on $i\in\matN$ with $1\leq i\leq n$, we have
  $P(n)$.
  Therefore, $F=F_n$ is closed for distance~$d_G$.
\end{proof}

\begin{theorem}[{\LMC}, finite dimensional subspace]
  \label{t:lax-milgram-cea-finite-dimensional-subspace}
  Assume hypotheses of
  Theorem~\thref{t:lax-milgram}.
  Let $C\geq 0$ be a continuity constant of the bounded bilinear
  form~$\BlfH$.
  Let $u\in H$ be the unique solution to Problem~\eqref{e:general-problem}.
  Let~$\Hh$ be a finite dimensional subspace of~$H$.
  Then, there exists a unique $\uh\in\Hh$ solution to
  Problem~\eqref{e:general-problem-discrete}.
  Moreover,
  \begin{eqnarray}
    \label{e:lax-milgram-finite-dim-subspace-estimation}
    & & \nH{\uh} \leq \frac{1}{\alpha} \nHp{\LfH}; \\
    \label{e:lax-milgram-finite-dim-subspace-error}
    \forall \vh \in \Hh, & &
    \nH{u - \uh} \leq \frac{C}{\alpha} \nH{u - \vh}.
  \end{eqnarray}
\end{theorem}

\begin{proof}
  Direct consequence of
  Lemma~\thref{l:finite-dimensional-subspace-in-hilbert-space-is-closed},
  Theorem~\thref{t:lax-milgram-closed-subspace}, and
  Lemma~\thref{l:cea}.
\end{proof}

\section{Conclusions, perspectives}
\label{s:conclusions-perspectives}

We have presented a very detailed proof of the {\LM} theorem for the
resolution on a Hilbert space of linear (partial differential) equations set
under their weak form.
Among the various proofs available in the literature, we have chosen a path
using basic notions.
In particular, we have avoided to obtain the result from a more general one,
{\eg} set on a Banach space.
The proof uses the following main arguments:
the representation lemma for bounded bilinear forms,
the {\RF} representation theorem,
the orthogonal projection theorem for a complete subspace,
and the fixed point theorem for a contraction on a complete metric space.

The short-term purpose of this work was to help the formalization of such a
result in the {\coq} formal proof assistant.
This was recently achieved~\cite{bol:cfp:16}.
One of the key issues for the computer scientists that formalize the
pen-and-paper proof was to deal with the embedded algebraic structures:
group,
vector space (an external operation is added),
normed vector space (a norm is added),
inner vector space (an inner product is added),
Hilbert space (completeness is added).
New {\coq} structures should be extensions of the previous ones: the addition
operation in the Hilbert space should be the very same addition operation from
the initial group structure.

The long-term purpose of these studies is the formal proof of programs
implementing the {\FEM}.
For instance, considering the standard Laplace
equation~\eqref{e:laplace-problem-strong}, one proves that it can be written
in weak formulation as
\begin{equation}
  \label{e:laplace-problem-weak}
  \mbox{find $u\in\HsobOmIO$ such that:} \quad
  \forall v \in \HsobOmIO,\;
  \int_{\Omega} \nabla u \cdot \nabla v = \int_{\Omega} f v,
\end{equation}
where $\HsobOmIO$ is a Sobolev space.
Problem~\eqref{e:laplace-problem-weak} takes the form of
Problem~\eqref{e:general-problem}, with the following notations:
$H=\HsobOmIO$, the bilinear form is defined by
$\blfH{v}{w}=\int_{\Omega}\nabla v\cdot\nabla w$, and the linear form by
$\lfH{v}=\int_{\Omega}qv$.
To apply the {\LM} theorem, one needs to prove in particular that $\HsobOmIO$
is a Hilbert space.

As a consequence, we will have to write very detailed pen-and-paper
proofs for the following notions and results:
large parts of the integration and distribution theories,
define Sobolev spaces (at least $\LsobOm$, $\HsobOmI$ and $\HsobOmIO$ for some
bounded domain~$\Omega$ of $\matR^d$ with $d=1$, 2, or 3),
and prove that they are Hilbert spaces.
And finally, many results of the interpolation and approximation theory to
define the {\FEM} itself.

%% file: tree_lm.tex
\section{Depends directly from\ldots}
\label{s:depends-directly-from}

\begin{description}

\item[Definition~\thref{d:supremum}] has no direct dependency.

\item[Lemma~\thref{l:finite-supremum}] has no direct dependency.

\item[Lemma~\thref{l:discrete-lower-accumulation}] has no direct dependency.

\item[Lemma~\thref{l:supremum-is-positive-scalar-multiplicative}] depends directly from:\\
  Definition~\thref{d:supremum}.

\item[Definition~\thref{d:maximum}] has no direct dependency.

\item[Lemma~\thref{l:finite-maximum}] depends directly from:\\
  Definition~\thref{d:supremum},\\
  Lemma~\thref{l:finite-supremum},\\
  Definition~\thref{d:maximum}.

\item[Definition~\thref{d:infimum}] has no direct dependency.

\item[Lemma~\thref{l:duality-infimum-supremum}] depends directly from:\\
  Definition~\thref{d:supremum},\\
  Definition~\thref{d:infimum}.

\item[Lemma~\thref{l:finite-infimum}] depends directly from:\\
  Lemma~\thref{l:finite-supremum},\\
  Lemma~\thref{l:duality-infimum-supremum}.

\item[Lemma~\thref{l:discrete-upper-accumulation}] depends directly from:\\
  Lemma~\thref{l:discrete-lower-accumulation}.

\item[Lemma~\thref{l:finite-infimum-discrete}] depends directly from:\\
  Lemma~\thref{l:finite-infimum},\\
  Lemma~\thref{l:discrete-upper-accumulation}.

\item[Definition~\thref{d:minimum}] has no direct dependency.

\item[Lemma~\thref{l:finite-minimum}] depends directly from:\\
  Lemma~\thref{l:finite-maximum},\\
  Lemma~\thref{l:duality-infimum-supremum},\\
  Definition~\thref{d:minimum}.

\item[Definition~\thref{d:distance}] has no direct dependency.

\item[Definition~\thref{d:metric-space}] has no direct dependency.

\item[Lemma~\thref{l:iterated-triangle-inequality}] depends directly from:\\
  Definition~\thref{d:distance}.

\item[Definition~\thref{d:closed-ball}] has no direct dependency.

\item[Definition~\thref{d:sphere}] has no direct dependency.

\item[Definition~\thref{d:open-subset}] has no direct dependency.

\item[Definition~\thref{d:closed-subset}] has no direct dependency.

\item[Lemma~\thref{l:equivalent-definition-of-closed-subset}] depends directly from:\\
  Definition~\thref{d:open-subset},\\
  Definition~\thref{d:closed-subset}.

\item[Lemma~\thref{l:singleton-is-closed}] depends directly from:\\
  Definition~\thref{d:distance},\\
  Lemma~\thref{l:equivalent-definition-of-closed-subset}.

\item[Definition~\thref{d:closure}] has no direct dependency.

\item[Definition~\thref{d:convergent-sequence}] has no direct dependency.

\item[Lemma~\thref{l:variant-of-point-separation}] depends directly from:\\
  Definition~\thref{d:distance}.

\item[Lemma~\thref{l:limit-is-unique}] depends directly from:\\
  Definition~\thref{d:distance},\\
  Definition~\thref{d:convergent-sequence},\\
  Lemma~\thref{l:variant-of-point-separation}.

\item[Lemma~\thref{l:closure-is-limit-of-sequences}] depends directly from:\\
  Definition~\thref{d:distance},\\
  Definition~\thref{d:closed-ball},\\
  Definition~\thref{d:closure},\\
  Definition~\thref{d:convergent-sequence}.

\item[Lemma~\thref{l:closed-equals-closure}] depends directly from:\\
  Definition~\thref{d:closed-subset},\\
  Lemma~\thref{l:equivalent-definition-of-closed-subset},\\
  Definition~\thref{d:closure}.

\item[Lemma~\thref{l:closed-is-limit-of-sequences}] depends directly from:\\
  Definition~\thref{d:closure},\\
  Lemma~\thref{l:closure-is-limit-of-sequences},\\
  Lemma~\thref{l:closed-equals-closure}.

\item[Definition~\thref{d:stationary-sequence}] has no direct dependency.

\item[Lemma~\thref{l:stationary-sequence-is-convergent}] depends directly from:\\
  Definition~\thref{d:distance},\\
  Definition~\thref{d:convergent-sequence},\\
  Definition~\thref{d:stationary-sequence}.

\item[Definition~\thref{d:cauchy-sequence}] has no direct dependency.

\item[Lemma~\thref{l:equivalent-definition-of-cauchy-sequence}] depends directly from:\\
  Definition~\thref{d:distance},\\
  Definition~\thref{d:cauchy-sequence}.

\item[Lemma~\thref{l:convergent-sequence-is-cauchy}] depends directly from:\\
  Definition~\thref{d:distance},\\
  Definition~\thref{d:convergent-sequence},\\
  Lemma~\thref{l:limit-is-unique},\\
  Definition~\thref{d:cauchy-sequence}.

\item[Definition~\thref{d:complete-subset}] has no direct dependency.

\item[Definition~\thref{d:complete-metric-space}] has no direct dependency.

\item[Lemma~\thref{l:closed-subset-of-complete-is-complete}] depends directly from:\\
  Lemma~\thref{l:closure-is-limit-of-sequences},\\
  Lemma~\thref{l:closed-equals-closure},\\
  Definition~\thref{d:complete-subset},\\
  Definition~\thref{d:complete-metric-space}.

\item[Definition~\thref{d:continuity-in-a-point}] has no direct dependency.

\item[Definition~\thref{d:pointwise-continuity}] has no direct dependency.

\item[Lemma~\thref{l:compatibility-of-limit-with-continuous-functions}] depends directly from:\\
  Definition~\thref{d:convergent-sequence},\\
  Definition~\thref{d:continuity-in-a-point}.

\item[Definition~\thref{d:uniform-continuity}] has no direct dependency.

\item[Definition~\thref{d:lipschitz-continuity}] has no direct dependency.

\item[Theorem~\thref{t:equivalent-definition-of-lipschitz-continuity}] depends directly from:\\
  Definition~\thref{d:distance},\\
  Definition~\thref{d:lipschitz-continuity}.

\item[Definition~\thref{d:contraction}] has no direct dependency.

\item[Lemma~\thref{l:uniform-continuous-is-continuous}] depends directly from:\\
  Definition~\thref{d:continuity-in-a-point},\\
  Definition~\thref{d:pointwise-continuity},\\
  Definition~\thref{d:uniform-continuity}.

\item[Lemma~\thref{l:zero-lipschitz-continuous-is-constant}] depends directly from:\\
  Definition~\thref{d:distance},\\
  Definition~\thref{d:lipschitz-continuity}.

\item[Lemma~\thref{l:lipschitz-continuous-is-uniform-continuous}] depends directly from:\\
  Definition~\thref{d:uniform-continuity},\\
  Definition~\thref{d:lipschitz-continuity},\\
  Lemma~\thref{l:zero-lipschitz-continuous-is-constant}.

\item[Definition~\thref{d:iterated-function-sequence}] has no direct dependency.

\item[Lemma~\thref{l:stationary-iterated-function-sequence}] depends directly from:\\
  Definition~\thref{d:stationary-sequence},\\
  Definition~\thref{d:iterated-function-sequence}.

\item[Lemma~\thref{l:iterate-lipschitz-continuous-mapping}] depends directly from:\\
  Definition~\thref{d:lipschitz-continuity},\\
  Definition~\thref{d:iterated-function-sequence}.

\item[Lemma~\thref{l:convergent-iterated-function-sequence}] depends directly from:\\
  Definition~\thref{d:convergent-sequence},\\
  Lemma~\thref{l:limit-is-unique},\\
  Definition~\thref{d:stationary-sequence},\\
  Lemma~\thref{l:stationary-sequence-is-convergent},\\
  Definition~\thref{d:lipschitz-continuity},\\
  Lemma~\thref{l:zero-lipschitz-continuous-is-constant},\\
  Definition~\thref{d:iterated-function-sequence}.

\item[Theorem~\thref{t:fixed-point}] depends directly from:\\
  Definition~\thref{d:distance},\\
  Lemma~\thref{l:iterated-triangle-inequality},\\
  Definition~\thref{d:stationary-sequence},\\
  Lemma~\thref{l:stationary-sequence-is-convergent},\\
  Lemma~\thref{l:equivalent-definition-of-cauchy-sequence},\\
  Definition~\thref{d:complete-subset},\\
  Definition~\thref{d:complete-metric-space},\\
  Definition~\thref{d:lipschitz-continuity},\\
  Definition~\thref{d:contraction},\\
  Lemma~\thref{l:zero-lipschitz-continuous-is-constant},\\
  Lemma~\thref{l:stationary-iterated-function-sequence},\\
  Lemma~\thref{l:iterate-lipschitz-continuous-mapping},\\
  Lemma~\thref{l:convergent-iterated-function-sequence}.

\item[Definition~\thref{d:space}] has no direct dependency.

\item[Definition~\thref{d:set-of-mappings-to-space}] has no direct dependency.

\item[Definition~\thref{d:linear-map}] has no direct dependency.

\item[Definition~\thref{d:set-of-linear-maps}] has no direct dependency.

\item[Definition~\thref{d:linear-form}] has no direct dependency.

\item[Definition~\thref{d:bilinear-map}] has no direct dependency.

\item[Definition~\thref{d:bilinear-form}] has no direct dependency.

\item[Definition~\thref{d:set-of-bilinear-forms}] has no direct dependency.

\item[Lemma~\thref{l:zero-times-yields-zero}] depends directly from:\\
  Definition~\thref{d:space}.

\item[Lemma~\thref{l:minus-times-yields-opposite-vector}] depends directly from:\\
  Definition~\thref{d:space},\\
  Lemma~\thref{l:zero-times-yields-zero}.

\item[Definition~\thref{d:vector-subtraction}] has no direct dependency.

\item[Definition~\thref{d:scalar-division}] has no direct dependency.

\item[Lemma~\thref{l:times-zero-yields-zero}] depends directly from:\\
  Definition~\thref{d:space},\\
  Definition~\thref{d:vector-subtraction}.

\item[Lemma~\thref{l:zero-product-property}] depends directly from:\\
  Definition~\thref{d:space},\\
  Lemma~\thref{l:zero-times-yields-zero},\\
  Lemma~\thref{l:times-zero-yields-zero}.

\item[Definition~\thref{d:subspace}] has no direct dependency.

\item[Lemma~\thref{l:trivial-subspaces}] depends directly from:\\
  Definition~\thref{d:subspace}.

\item[Lemma~\thref{l:closed-under-vector-operations-is-subspace}] depends directly from:\\
  Definition~\thref{d:space},\\
  Lemma~\thref{l:minus-times-yields-opposite-vector},\\
  Definition~\thref{d:subspace}.

\item[Lemma~\thref{l:closed-under-linear-combination-is-subspace}] depends directly from:\\
  Definition~\thref{d:space},\\
  Lemma~\thref{l:zero-product-property},\\
  Lemma~\thref{l:closed-under-vector-operations-is-subspace}.

\item[Definition~\thref{d:linear-span}] has no direct dependency.

\item[Definition~\thref{d:sum-of-subspaces}] has no direct dependency.

\item[Definition~\thref{d:finite-dimensional-subspace}] has no direct dependency.

\item[Definition~\thref{d:direct-sum-of-subspaces}] has no direct dependency.

\item[Lemma~\thref{l:equivalent-definition-of-direct-sum}] depends directly from:\\
  Definition~\thref{d:space},\\
  Lemma~\thref{l:minus-times-yields-opposite-vector},\\
  Definition~\thref{d:vector-subtraction},\\
  Lemma~\thref{l:closed-under-vector-operations-is-subspace},\\
  Definition~\thref{d:direct-sum-of-subspaces}.

\item[Lemma~\thref{l:direct-sum-with-linear-span}] depends directly from:\\
  Lemma~\thref{l:zero-product-property},\\
  Lemma~\thref{l:closed-under-vector-operations-is-subspace},\\
  Definition~\thref{d:linear-span},\\
  Lemma~\thref{l:equivalent-definition-of-direct-sum}.

\item[Definition~\thref{d:product-vector-operations}] has no direct dependency.

\item[Lemma~\thref{l:product-is-space}] depends directly from:\\
  Definition~\thref{d:space},\\
  Definition~\thref{d:product-vector-operations}.

\item[Definition~\thref{d:inherited-vector-operations}] has no direct dependency.

\item[Lemma~\thref{l:space-of-mappings-to-space}] depends directly from:\\
  Definition~\thref{d:space},\\
  Definition~\thref{d:set-of-mappings-to-space},\\
  Definition~\thref{d:inherited-vector-operations}.

\item[Lemma~\thref{l:linear-map-preserves-zero}] depends directly from:\\
  Definition~\thref{d:linear-map},\\
  Lemma~\thref{l:zero-times-yields-zero}.

\item[Lemma~\thref{l:linear-map-preserves-linear-combinations}] depends directly from:\\
  Definition~\thref{d:space},\\
  Definition~\thref{d:linear-map},\\
  Lemma~\thref{l:zero-times-yields-zero}.

\item[Lemma~\thref{l:space-of-linear-maps}] depends directly from:\\
  Definition~\thref{d:space},\\
  Definition~\thref{d:linear-map},\\
  Definition~\thref{d:set-of-linear-maps},\\
  Lemma~\thref{l:times-zero-yields-zero},\\
  Lemma~\thref{l:closed-under-linear-combination-is-subspace},\\
  Definition~\thref{d:inherited-vector-operations},\\
  Lemma~\thref{l:space-of-mappings-to-space}.

\item[Definition~\thref{d:identity-map}] has no direct dependency.

\item[Lemma~\thref{l:identity-map-is-linear-map}] depends directly from:\\
  Definition~\thref{d:linear-map},\\
  Definition~\thref{d:identity-map}.

\item[Lemma~\thref{l:composition-of-linear-maps-is-bilinear}] depends directly from:\\
  Definition~\thref{d:linear-map},\\
  Definition~\thref{d:bilinear-map},\\
  Lemma~\thref{l:product-is-space},\\
  Definition~\thref{d:inherited-vector-operations},\\
  Lemma~\thref{l:linear-map-preserves-linear-combinations},\\
  Lemma~\thref{l:space-of-linear-maps}.

\item[Definition~\thref{d:isomorphism}] has no direct dependency.

\item[Definition~\thref{d:kernel}] has no direct dependency.

\item[Lemma~\thref{l:kernel-is-subspace}] depends directly from:\\
  Definition~\thref{d:space},\\
  Lemma~\thref{l:times-zero-yields-zero},\\
  Lemma~\thref{l:closed-under-linear-combination-is-subspace},\\
  Lemma~\thref{l:linear-map-preserves-zero},\\
  Lemma~\thref{l:linear-map-preserves-linear-combinations},\\
  Definition~\thref{d:kernel}.

\item[Lemma~\thref{l:injective-linear-map-has-zero-kernel}] depends directly from:\\
  Definition~\thref{d:space},\\
  Definition~\thref{d:linear-map},\\
  Definition~\thref{d:vector-subtraction},\\
  Lemma~\thref{l:linear-map-preserves-zero},\\
  Definition~\thref{d:kernel}.

\item[Lemma~\thref{l:k-is-space}] has no direct dependency.

\item[Definition~\thref{d:norm}] has no direct dependency.

\item[Definition~\thref{d:normed-space}] has no direct dependency.

\item[Lemma~\thref{l:k-is-normed-space}] depends directly from:\\
  Definition~\thref{d:norm},\\
  Definition~\thref{d:normed-space}.

\item[Lemma~\thref{l:norm-preserves-zero}] depends directly from:\\
  Definition~\thref{d:space},\\
  Lemma~\thref{l:zero-times-yields-zero},\\
  Definition~\thref{d:norm},\\
  Definition~\thref{d:normed-space}.

\item[Lemma~\thref{l:norm-is-nonnegative}] depends directly from:\\
  Definition~\thref{d:space},\\
  Definition~\thref{d:norm},\\
  Definition~\thref{d:normed-space}.

\item[Lemma~\thref{l:normalization-by-nonzero}] depends directly from:\\
  Definition~\thref{d:scalar-division},\\
  Definition~\thref{d:norm},\\
  Lemma~\thref{l:norm-is-nonnegative}.

\item[Definition~\thref{d:distance-associated-with-norm}] has no direct dependency.

\item[Lemma~\thref{l:norm-gives-distance}] depends directly from:\\
  Definition~\thref{d:distance},\\
  Definition~\thref{d:metric-space},\\
  Definition~\thref{d:vector-subtraction},\\
  Definition~\thref{d:norm},\\
  Lemma~\thref{l:norm-is-nonnegative},\\
  Definition~\thref{d:distance-associated-with-norm}.

\item[Lemma~\thref{l:linear-span-is-closed}] depends directly from:\\
  Lemma~\thref{l:singleton-is-closed},\\
  Definition~\thref{d:convergent-sequence},\\
  Lemma~\thref{l:limit-is-unique},\\
  Lemma~\thref{l:closed-is-limit-of-sequences},\\
  Definition~\thref{d:cauchy-sequence},\\
  Lemma~\thref{l:convergent-sequence-is-cauchy},\\
  Definition~\thref{d:complete-subset},\\
  Definition~\thref{d:space},\\
  Definition~\thref{d:vector-subtraction},\\
  Definition~\thref{d:linear-span},\\
  Definition~\thref{d:norm},\\
  Lemma~\thref{l:norm-is-nonnegative},\\
  Definition~\thref{d:distance-associated-with-norm}.

\item[Definition~\thref{d:closed-unit-ball}] has no direct dependency.

\item[Lemma~\thref{l:equivalent-definition-of-closed-unit-ball}] depends directly from:\\
  Definition~\thref{d:closed-ball},\\
  Definition~\thref{d:distance-associated-with-norm},\\
  Lemma~\thref{l:norm-gives-distance},\\
  Definition~\thref{d:closed-unit-ball}.

\item[Definition~\thref{d:unit-sphere}] has no direct dependency.

\item[Lemma~\thref{l:equivalent-definition-of-unit-sphere}] depends directly from:\\
  Definition~\thref{d:sphere},\\
  Definition~\thref{d:distance-associated-with-norm},\\
  Lemma~\thref{l:norm-gives-distance},\\
  Definition~\thref{d:unit-sphere}.

\item[Lemma~\thref{l:zero-on-unit-sphere-is-zero}] depends directly from:\\
  Definition~\thref{d:space},\\
  Definition~\thref{d:linear-map},\\
  Lemma~\thref{l:times-zero-yields-zero},\\
  Lemma~\thref{l:linear-map-preserves-zero},\\
  Lemma~\thref{l:normalization-by-nonzero},\\
  Lemma~\thref{l:equivalent-definition-of-unit-sphere}.

\item[Lemma~\thref{l:reverse-triangle-inequality}] depends directly from:\\
  Definition~\thref{d:norm},\\
  Definition~\thref{d:normed-space}.

\item[Lemma~\thref{l:norm-is-one-lipschitz-continuous}] depends directly from:\\
  Definition~\thref{d:lipschitz-continuity},\\
  Definition~\thref{d:distance-associated-with-norm},\\
  Lemma~\thref{l:reverse-triangle-inequality}.

\item[Lemma~\thref{l:norm-is-uniformly-continuous}] depends directly from:\\
  Lemma~\thref{l:lipschitz-continuous-is-uniform-continuous},\\
  Lemma~\thref{l:norm-is-one-lipschitz-continuous}.

\item[Lemma~\thref{l:norm-is-continuous}] depends directly from:\\
  Lemma~\thref{l:uniform-continuous-is-continuous},\\
  Lemma~\thref{l:norm-is-uniformly-continuous}.

\item[Definition~\thref{d:linear-isometry}] has no direct dependency.

\item[Lemma~\thref{l:identity-map-is-linear-isometry}] depends directly from:\\
  Definition~\thref{d:identity-map},\\
  Lemma~\thref{l:identity-map-is-linear-map},\\
  Definition~\thref{d:linear-isometry}.

\item[Definition~\thref{d:product-norm}] has no direct dependency.

\item[Lemma~\thref{l:product-is-normed-space}] depends directly from:\\
  Definition~\thref{d:product-vector-operations},\\
  Lemma~\thref{l:product-is-space},\\
  Definition~\thref{d:norm},\\
  Definition~\thref{d:normed-space},\\
  Lemma~\thref{l:norm-is-nonnegative},\\
  Definition~\thref{d:product-norm}.

\item[Lemma~\thref{l:vector-addition-is-continuous}] depends directly from:\\
  Definition~\thref{d:continuity-in-a-point},\\
  Definition~\thref{d:pointwise-continuity},\\
  Definition~\thref{d:space},\\
  Definition~\thref{d:product-vector-operations},\\
  Definition~\thref{d:norm},\\
  Definition~\thref{d:distance-associated-with-norm},\\
  Definition~\thref{d:product-norm},\\
  Lemma~\thref{l:product-is-normed-space}.

\item[Lemma~\thref{l:scalar-multiplication-is-continuous}] depends directly from:\\
  Definition~\thref{d:distance},\\
  Definition~\thref{d:continuity-in-a-point},\\
  Definition~\thref{d:pointwise-continuity},\\
  Definition~\thref{d:space},\\
  Lemma~\thref{l:zero-times-yields-zero},\\
  Definition~\thref{d:norm},\\
  Definition~\thref{d:distance-associated-with-norm}.

\item[Lemma~\thref{l:norm-of-image-of-unit-vector}] depends directly from:\\
  Definition~\thref{d:linear-map},\\
  Definition~\thref{d:norm},\\
  Lemma~\thref{l:norm-is-nonnegative},\\
  Lemma~\thref{l:normalization-by-nonzero}.

\item[Lemma~\thref{l:norm-of-image-of-unit-sphere}] depends directly from:\\
  Definition~\thref{d:norm},\\
  Lemma~\thref{l:norm-preserves-zero},\\
  Lemma~\thref{l:equivalent-definition-of-unit-sphere},\\
  Lemma~\thref{l:norm-of-image-of-unit-vector}.

\item[Definition~\thref{d:operator-norm}] has no direct dependency.

\item[Lemma~\thref{l:equivalent-definition-of-operator-norm}] depends directly from:\\
  Lemma~\thref{l:norm-of-image-of-unit-sphere},\\
  Definition~\thref{d:operator-norm}.

\item[Lemma~\thref{l:operator-norm-is-nonnegative}] depends directly from:\\
  Definition~\thref{d:supremum},\\
  Lemma~\thref{l:norm-is-nonnegative},\\
  Lemma~\thref{l:equivalent-definition-of-operator-norm}.

\item[Definition~\thref{d:bounded-linear-map}] has no direct dependency.

\item[Definition~\thref{d:linear-map-bounded-on-unit-ball}] has no direct dependency.

\item[Definition~\thref{d:linear-map-bounded-on-unit-sphere}] has no direct dependency.

\item[Theorem~\thref{t:continuous-linear-map}] depends directly from:\\
  Definition~\thref{d:supremum},\\
  Lemma~\thref{l:finite-supremum},\\
  Definition~\thref{d:continuity-in-a-point},\\
  Definition~\thref{d:pointwise-continuity},\\
  Definition~\thref{d:lipschitz-continuity},\\
  Lemma~\thref{l:uniform-continuous-is-continuous},\\
  Lemma~\thref{l:lipschitz-continuous-is-uniform-continuous},\\
  Definition~\thref{d:linear-map},\\
  Definition~\thref{d:vector-subtraction},\\
  Lemma~\thref{l:linear-map-preserves-zero},\\
  Definition~\thref{d:norm},\\
  Lemma~\thref{l:norm-preserves-zero},\\
  Lemma~\thref{l:equivalent-definition-of-closed-unit-ball},\\
  Lemma~\thref{l:equivalent-definition-of-unit-sphere},\\
  Definition~\thref{d:operator-norm},\\
  Lemma~\thref{l:equivalent-definition-of-operator-norm},\\
  Lemma~\thref{l:operator-norm-is-nonnegative},\\
  Definition~\thref{d:bounded-linear-map},\\
  Definition~\thref{d:linear-map-bounded-on-unit-ball},\\
  Definition~\thref{d:linear-map-bounded-on-unit-sphere}.

\item[Definition~\thref{d:set-of-continuous-linear-maps}] has no direct dependency.

\item[Lemma~\thref{l:finite-operator-norm-is-continuous}] depends directly from:\\
  Definition~\thref{d:supremum},\\
  Definition~\thref{d:bounded-linear-map},\\
  Definition~\thref{d:linear-map-bounded-on-unit-ball},\\
  Definition~\thref{d:linear-map-bounded-on-unit-sphere},\\
  Theorem~\thref{t:continuous-linear-map},\\
  Definition~\thref{d:set-of-continuous-linear-maps}.

\item[Lemma~\thref{l:linear-isometry-is-continuous}] depends directly from:\\
  Definition~\thref{d:linear-isometry},\\
  Lemma~\thref{l:finite-operator-norm-is-continuous}.

\item[Lemma~\thref{l:identity-map-is-continuous}] depends directly from:\\
  Lemma~\thref{l:identity-map-is-linear-isometry},\\
  Lemma~\thref{l:linear-isometry-is-continuous}.

\item[Theorem~\thref{t:normed-space-of-continuous-linear-maps}] depends directly from:\\
  Definition~\thref{d:supremum},\\
  Lemma~\thref{l:supremum-is-positive-scalar-multiplicative},\\
  Lemma~\thref{l:closed-under-vector-operations-is-subspace},\\
  Definition~\thref{d:inherited-vector-operations},\\
  Definition~\thref{d:norm},\\
  Definition~\thref{d:normed-space},\\
  Lemma~\thref{l:norm-is-nonnegative},\\
  Lemma~\thref{l:zero-on-unit-sphere-is-zero},\\
  Lemma~\thref{l:equivalent-definition-of-operator-norm},\\
  Definition~\thref{d:linear-map-bounded-on-unit-sphere},\\
  Definition~\thref{d:set-of-continuous-linear-maps},\\
  Lemma~\thref{l:finite-operator-norm-is-continuous}.

\item[Lemma~\thref{l:operator-norm-estimation}] depends directly from:\\
  Definition~\thref{d:supremum},\\
  Lemma~\thref{l:linear-map-preserves-zero},\\
  Definition~\thref{d:norm},\\
  Lemma~\thref{l:norm-preserves-zero},\\
  Lemma~\thref{l:norm-is-nonnegative},\\
  Definition~\thref{d:operator-norm},\\
  Theorem~\thref{t:normed-space-of-continuous-linear-maps}.

\item[Lemma~\thref{l:continuous-linear-maps-have-closed-kernel}] depends directly from:\\
  Lemma~\thref{l:singleton-is-closed},\\
  Definition~\thref{d:kernel}.

\item[Lemma~\thref{l:compatibility-of-composition-with-continuity}] depends directly from:\\
  Lemma~\thref{l:composition-of-linear-maps-is-bilinear},\\
  Lemma~\thref{l:equivalent-definition-of-unit-sphere},\\
  Lemma~\thref{l:finite-operator-norm-is-continuous},\\
  Lemma~\thref{l:operator-norm-estimation}.

\item[Lemma~\thref{l:complete-normed-space-of-continuous-linear-maps}] depends directly from:\\
  Definition~\thref{d:convergent-sequence},\\
  Lemma~\thref{l:stationary-sequence-is-convergent},\\
  Definition~\thref{d:cauchy-sequence},\\
  Definition~\thref{d:complete-subset},\\
  Definition~\thref{d:complete-metric-space},\\
  Lemma~\thref{l:compatibility-of-limit-with-continuous-functions},\\
  Definition~\thref{d:space},\\
  Definition~\thref{d:linear-map},\\
  Definition~\thref{d:inherited-vector-operations},\\
  Lemma~\thref{l:linear-map-preserves-zero},\\
  Lemma~\thref{l:linear-map-preserves-linear-combinations},\\
  Definition~\thref{d:norm},\\
  Definition~\thref{d:distance-associated-with-norm},\\
  Lemma~\thref{l:norm-gives-distance},\\
  Lemma~\thref{l:norm-is-continuous},\\
  Lemma~\thref{l:vector-addition-is-continuous},\\
  Lemma~\thref{l:scalar-multiplication-is-continuous},\\
  Definition~\thref{d:operator-norm},\\
  Definition~\thref{d:bounded-linear-map},\\
  Theorem~\thref{t:continuous-linear-map},\\
  Lemma~\thref{l:finite-operator-norm-is-continuous},\\
  Theorem~\thref{t:normed-space-of-continuous-linear-maps}.

\item[Definition~\thref{d:topological-dual}] has no direct dependency.

\item[Definition~\thref{d:dual-norm}] has no direct dependency.

\item[Lemma~\thref{l:topological-dual-is-complete-normed-space}] depends directly from:\\
  Lemma~\thref{l:k-is-normed-space},\\
  Theorem~\thref{t:normed-space-of-continuous-linear-maps},\\
  Lemma~\thref{l:complete-normed-space-of-continuous-linear-maps},\\
  Definition~\thref{d:dual-norm}.

\item[Definition~\thref{d:bra-ket-notation}] has no direct dependency.

\item[Lemma~\thref{l:bra-ket-is-bilinear-map}] depends directly from:\\
  Definition~\thref{d:linear-map},\\
  Definition~\thref{d:bilinear-map},\\
  Lemma~\thref{l:product-is-space},\\
  Definition~\thref{d:inherited-vector-operations},\\
  Lemma~\thref{l:k-is-space},\\
  Lemma~\thref{l:topological-dual-is-complete-normed-space},\\
  Definition~\thref{d:bra-ket-notation}.

\item[Definition~\thref{d:bounded-bilinear-form}] has no direct dependency.

\item[Lemma~\thref{l:representation-for-bounded-bilinear-form}] depends directly from:\\
  Definition~\thref{d:space},\\
  Definition~\thref{d:linear-form},\\
  Definition~\thref{d:bilinear-map},\\
  Definition~\thref{d:bilinear-form},\\
  Definition~\thref{d:set-of-bilinear-forms},\\
  Definition~\thref{d:vector-subtraction},\\
  Definition~\thref{d:inherited-vector-operations},\\
  Lemma~\thref{l:linear-map-preserves-linear-combinations},\\
  Lemma~\thref{l:k-is-space},\\
  Definition~\thref{d:normed-space},\\
  Lemma~\thref{l:norm-is-nonnegative},\\
  Lemma~\thref{l:equivalent-definition-of-unit-sphere},\\
  Definition~\thref{d:bounded-linear-map},\\
  Lemma~\thref{l:finite-operator-norm-is-continuous},\\
  Theorem~\thref{t:normed-space-of-continuous-linear-maps},\\
  Definition~\thref{d:topological-dual},\\
  Definition~\thref{d:dual-norm},\\
  Lemma~\thref{l:topological-dual-is-complete-normed-space},\\
  Definition~\thref{d:bra-ket-notation},\\
  Lemma~\thref{l:bra-ket-is-bilinear-map},\\
  Definition~\thref{d:bounded-bilinear-form}.

\item[Definition~\thref{d:coercive-bilinear-form}] has no direct dependency.

\item[Lemma~\thref{l:coercivity-constant-is-less-than-continuity-constant}] depends directly from:\\
  Definition~\thref{d:norm},\\
  Lemma~\thref{l:norm-is-nonnegative},\\
  Definition~\thref{d:bounded-bilinear-form},\\
  Definition~\thref{d:coercive-bilinear-form}.

\item[Definition~\thref{d:inner-product}] has no direct dependency.

\item[Definition~\thref{d:inner-product-space}] has no direct dependency.

\item[Lemma~\thref{l:inner-product-subspace}] depends directly from:\\
  Definition~\thref{d:subspace},\\
  Definition~\thref{d:inner-product},\\
  Definition~\thref{d:inner-product-space}.

\item[Lemma~\thref{l:inner-product-with-zero-is-zero}] depends directly from:\\
  Definition~\thref{d:space},\\
  Definition~\thref{d:bilinear-map},\\
  Definition~\thref{d:vector-subtraction},\\
  Definition~\thref{d:inner-product}.

\item[Lemma~\thref{l:square-expansion-plus}] depends directly from:\\
  Definition~\thref{d:bilinear-map},\\
  Definition~\thref{d:inner-product}.

\item[Lemma~\thref{l:square-expansion-minus}] depends directly from:\\
  Definition~\thref{d:bilinear-map},\\
  Definition~\thref{d:vector-subtraction},\\
  Definition~\thref{d:inner-product},\\
  Lemma~\thref{l:square-expansion-plus}.

\item[Lemma~\thref{l:parallelogram-identity}] depends directly from:\\
  Lemma~\thref{l:square-expansion-plus},\\
  Lemma~\thref{l:square-expansion-minus}.

\item[Lemma~\thref{l:cauchy-schwarz-inequality}] depends directly from:\\
  Definition~\thref{d:bilinear-map},\\
  Definition~\thref{d:inner-product},\\
  Lemma~\thref{l:square-expansion-plus}.

\item[Definition~\thref{d:square-root-of-inner-square}] has no direct dependency.

\item[Lemma~\thref{l:squared-norm}] depends directly from:\\
  Definition~\thref{d:inner-product},\\
  Definition~\thref{d:square-root-of-inner-square}.

\item[Lemma~\thref{l:cauchy-schwarz-inequality-with-norms}] depends directly from:\\
  Definition~\thref{d:inner-product},\\
  Lemma~\thref{l:cauchy-schwarz-inequality},\\
  Definition~\thref{d:square-root-of-inner-square}.

\item[Lemma~\thref{l:triangle-inequality}] depends directly from:\\
  Definition~\thref{d:inner-product},\\
  Lemma~\thref{l:square-expansion-plus},\\
  Lemma~\thref{l:squared-norm},\\
  Lemma~\thref{l:cauchy-schwarz-inequality-with-norms}.

\item[Lemma~\thref{l:inner-product-gives-norm}] depends directly from:\\
  Definition~\thref{d:norm},\\
  Definition~\thref{d:normed-space},\\
  Definition~\thref{d:inner-product},\\
  Definition~\thref{d:square-root-of-inner-square},\\
  Lemma~\thref{l:triangle-inequality}.

\item[Definition~\thref{d:convex-subset}] has no direct dependency.

\item[Theorem~\thref{t:orth-proj-onto-complete-convex}] depends directly from:\\
  Definition~\thref{d:infimum},\\
  Lemma~\thref{l:finite-infimum-discrete},\\
  Definition~\thref{d:minimum},\\
  Definition~\thref{d:cauchy-sequence},\\
  Definition~\thref{d:complete-subset},\\
  Lemma~\thref{l:compatibility-of-limit-with-continuous-functions},\\
  Definition~\thref{d:space},\\
  Definition~\thref{d:vector-subtraction},\\
  Definition~\thref{d:scalar-division},\\
  Definition~\thref{d:norm},\\
  Lemma~\thref{l:norm-is-nonnegative},\\
  Lemma~\thref{l:norm-is-continuous},\\
  Definition~\thref{d:inner-product-space},\\
  Lemma~\thref{l:parallelogram-identity},\\
  Lemma~\thref{l:squared-norm},\\
  Definition~\thref{d:convex-subset}.

\item[Lemma~\thref{l:characterization-of-orth-proj-onto-convex}] depends directly from:\\
  Definition~\thref{d:infimum},\\
  Definition~\thref{d:minimum},\\
  Lemma~\thref{l:finite-minimum},\\
  Definition~\thref{d:space},\\
  Definition~\thref{d:bilinear-map},\\
  Definition~\thref{d:vector-subtraction},\\
  Lemma~\thref{l:norm-is-nonnegative},\\
  Definition~\thref{d:inner-product},\\
  Definition~\thref{d:inner-product-space},\\
  Lemma~\thref{l:square-expansion-plus},\\
  Lemma~\thref{l:squared-norm},\\
  Definition~\thref{d:convex-subset}.

\item[Lemma~\thref{l:subspace-is-convex}] depends directly from:\\
  Lemma~\thref{l:closed-under-linear-combination-is-subspace},\\
  Definition~\thref{d:convex-subset}.

\item[Theorem~\thref{t:orth-proj-onto-complete-subspace}] depends directly from:\\
  Definition~\thref{d:space},\\
  Definition~\thref{d:subspace},\\
  Theorem~\thref{t:orth-proj-onto-complete-convex},\\
  Lemma~\thref{l:subspace-is-convex}.

\item[Definition~\thref{d:orth-proj-onto-complete-subspace}] depends directly from:\\
  Theorem~\thref{t:orth-proj-onto-complete-subspace}.

\item[Lemma~\thref{l:characterization-of-orth-proj-onto-subspace}] depends directly from:\\
  Definition~\thref{d:space},\\
  Definition~\thref{d:bilinear-map},\\
  Definition~\thref{d:vector-subtraction},\\
  Definition~\thref{d:subspace},\\
  Definition~\thref{d:inner-product},\\
  Lemma~\thref{l:characterization-of-orth-proj-onto-convex},\\
  Lemma~\thref{l:subspace-is-convex}.

\item[Lemma~\thref{l:orth-proj-is-continuous-linear-map}] depends directly from:\\
  Definition~\thref{d:lipschitz-continuity},\\
  Definition~\thref{d:space},\\
  Definition~\thref{d:bilinear-map},\\
  Definition~\thref{d:subspace},\\
  Lemma~\thref{l:linear-map-preserves-linear-combinations},\\
  Definition~\thref{d:norm},\\
  Lemma~\thref{l:norm-is-nonnegative},\\
  Definition~\thref{d:inner-product},\\
  Lemma~\thref{l:squared-norm},\\
  Lemma~\thref{l:cauchy-schwarz-inequality-with-norms},\\
  Theorem~\thref{t:orth-proj-onto-complete-subspace},\\
  Definition~\thref{d:orth-proj-onto-complete-subspace},\\
  Lemma~\thref{l:characterization-of-orth-proj-onto-subspace}.

\item[Definition~\thref{d:orth-compl}] has no direct dependency.

\item[Lemma~\thref{l:trivial-orth-compls}] depends directly from:\\
  Lemma~\thref{l:trivial-subspaces},\\
  Definition~\thref{d:inner-product},\\
  Lemma~\thref{l:inner-product-with-zero-is-zero},\\
  Definition~\thref{d:orth-compl}.

\item[Lemma~\thref{l:orth-compl-is-subspace}] depends directly from:\\
  Definition~\thref{d:bilinear-map},\\
  Lemma~\thref{l:closed-under-linear-combination-is-subspace},\\
  Definition~\thref{d:inner-product},\\
  Lemma~\thref{l:inner-product-with-zero-is-zero},\\
  Definition~\thref{d:orth-compl}.

\item[Lemma~\thref{l:zero-intersection-with-orth-compl}] depends directly from:\\
  Definition~\thref{d:inner-product},\\
  Definition~\thref{d:orth-compl}.

\item[Theorem~\thref{t:direct-sum-with-orth-compl-when-complete}] depends directly from:\\
  Definition~\thref{d:space},\\
  Definition~\thref{d:bilinear-map},\\
  Definition~\thref{d:sum-of-subspaces},\\
  Definition~\thref{d:direct-sum-of-subspaces},\\
  Lemma~\thref{l:equivalent-definition-of-direct-sum},\\
  Definition~\thref{d:inner-product},\\
  Lemma~\thref{l:inner-product-with-zero-is-zero},\\
  Theorem~\thref{t:orth-proj-onto-complete-subspace},\\
  Definition~\thref{d:orth-proj-onto-complete-subspace},\\
  Lemma~\thref{l:characterization-of-orth-proj-onto-subspace},\\
  Definition~\thref{d:orth-compl},\\
  Lemma~\thref{l:zero-intersection-with-orth-compl}.

\item[Lemma~\thref{l:sum-is-orth-sum}] depends directly from:\\
  Definition~\thref{d:space},\\
  Lemma~\thref{l:closed-under-linear-combination-is-subspace},\\
  Definition~\thref{d:linear-span},\\
  Definition~\thref{d:sum-of-subspaces},\\
  Theorem~\thref{t:orth-proj-onto-complete-subspace},\\
  Theorem~\thref{t:direct-sum-with-orth-compl-when-complete}.

\item[Lemma~\thref{l:sum-of-complete-subspace-and-linear-span-is-closed}] depends directly from:\\
  Lemma~\thref{l:closed-is-limit-of-sequences},\\
  Lemma~\thref{l:compatibility-of-limit-with-continuous-functions},\\
  Lemma~\thref{l:closed-under-linear-combination-is-subspace},\\
  Definition~\thref{d:linear-span},\\
  Definition~\thref{d:sum-of-subspaces},\\
  Lemma~\thref{l:direct-sum-with-linear-span},\\
  Lemma~\thref{l:linear-span-is-closed},\\
  Lemma~\thref{l:identity-map-is-continuous},\\
  Theorem~\thref{t:normed-space-of-continuous-linear-maps},\\
  Theorem~\thref{t:orth-proj-onto-complete-subspace},\\
  Lemma~\thref{l:orth-proj-is-continuous-linear-map},\\
  Lemma~\thref{l:zero-intersection-with-orth-compl},\\
  Theorem~\thref{t:direct-sum-with-orth-compl-when-complete}.

\item[Definition~\thref{d:hilbert-space}] depends directly from:\\
  Lemma~\thref{l:norm-gives-distance},\\
  Definition~\thref{d:square-root-of-inner-square},\\
  Lemma~\thref{l:inner-product-gives-norm}.

\item[Lemma~\thref{l:closed-hilbert-subspace}] depends directly from:\\
  Lemma~\thref{l:closed-subset-of-complete-is-complete},\\
  Definition~\thref{d:subspace},\\
  Lemma~\thref{l:inner-product-subspace},\\
  Definition~\thref{d:hilbert-space}.

\item[Theorem~\thref{t:riesz-frechet}] depends directly from:\\
  Definition~\thref{d:supremum},\\
  Definition~\thref{d:space},\\
  Definition~\thref{d:linear-map},\\
  Definition~\thref{d:linear-form},\\
  Definition~\thref{d:bilinear-map},\\
  Definition~\thref{d:vector-subtraction},\\
  Definition~\thref{d:scalar-division},\\
  Lemma~\thref{l:zero-product-property},\\
  Lemma~\thref{l:closed-under-vector-operations-is-subspace},\\
  Lemma~\thref{l:linear-map-preserves-zero},\\
  Lemma~\thref{l:linear-map-preserves-linear-combinations},\\
  Definition~\thref{d:isomorphism},\\
  Definition~\thref{d:kernel},\\
  Lemma~\thref{l:kernel-is-subspace},\\
  Lemma~\thref{l:injective-linear-map-has-zero-kernel},\\
  Definition~\thref{d:norm},\\
  Definition~\thref{d:normed-space},\\
  Lemma~\thref{l:norm-preserves-zero},\\
  Lemma~\thref{l:norm-is-nonnegative},\\
  Lemma~\thref{l:normalization-by-nonzero},\\
  Definition~\thref{d:linear-isometry},\\
  Definition~\thref{d:operator-norm},\\
  Definition~\thref{d:bounded-linear-map},\\
  Theorem~\thref{t:continuous-linear-map},\\
  Lemma~\thref{l:finite-operator-norm-is-continuous},\\
  Lemma~\thref{l:linear-isometry-is-continuous},\\
  Lemma~\thref{l:continuous-linear-maps-have-closed-kernel},\\
  Definition~\thref{d:topological-dual},\\
  Definition~\thref{d:dual-norm},\\
  Lemma~\thref{l:topological-dual-is-complete-normed-space},\\
  Definition~\thref{d:bra-ket-notation},\\
  Lemma~\thref{l:bra-ket-is-bilinear-map},\\
  Definition~\thref{d:inner-product},\\
  Definition~\thref{d:inner-product-space},\\
  Lemma~\thref{l:inner-product-with-zero-is-zero},\\
  Lemma~\thref{l:squared-norm},\\
  Lemma~\thref{l:cauchy-schwarz-inequality-with-norms},\\
  Theorem~\thref{t:orth-proj-onto-complete-subspace},\\
  Definition~\thref{d:orth-proj-onto-complete-subspace},\\
  Definition~\thref{d:orth-compl},\\
  Lemma~\thref{l:trivial-orth-compls},\\
  Lemma~\thref{l:orth-compl-is-subspace},\\
  Theorem~\thref{t:direct-sum-with-orth-compl-when-complete},\\
  Definition~\thref{d:hilbert-space},\\
  Lemma~\thref{l:closed-hilbert-subspace}.

\item[Lemma~\thref{l:compatible-rho-for-lax-milgram}] has no direct dependency.

\item[Theorem~\thref{t:lax-milgram}] depends directly from:\\
  Definition~\thref{d:lipschitz-continuity},\\
  Definition~\thref{d:contraction},\\
  Theorem~\thref{t:fixed-point},\\
  Definition~\thref{d:space},\\
  Definition~\thref{d:bilinear-map},\\
  Lemma~\thref{l:minus-times-yields-opposite-vector},\\
  Definition~\thref{d:vector-subtraction},\\
  Lemma~\thref{l:zero-product-property},\\
  Definition~\thref{d:norm},\\
  Definition~\thref{d:normed-space},\\
  Lemma~\thref{l:norm-preserves-zero},\\
  Lemma~\thref{l:norm-is-nonnegative},\\
  Definition~\thref{d:linear-isometry},\\
  Lemma~\thref{l:identity-map-is-continuous},\\
  Theorem~\thref{t:normed-space-of-continuous-linear-maps},\\
  Lemma~\thref{l:operator-norm-estimation},\\
  Lemma~\thref{l:compatibility-of-composition-with-continuity},\\
  Lemma~\thref{l:topological-dual-is-complete-normed-space},\\
  Definition~\thref{d:bounded-bilinear-form},\\
  Lemma~\thref{l:representation-for-bounded-bilinear-form},\\
  Definition~\thref{d:coercive-bilinear-form},\\
  Lemma~\thref{l:coercivity-constant-is-less-than-continuity-constant},\\
  Definition~\thref{d:inner-product},\\
  Lemma~\thref{l:square-expansion-plus},\\
  Lemma~\thref{l:inner-product-gives-norm},\\
  Definition~\thref{d:orth-compl},\\
  Lemma~\thref{l:trivial-orth-compls},\\
  Definition~\thref{d:hilbert-space},\\
  Theorem~\thref{t:riesz-frechet},\\
  Lemma~\thref{l:compatible-rho-for-lax-milgram}.

\item[Lemma~\thref{l:galerkin-orthogonality}] depends directly from:\\
  Definition~\thref{d:bilinear-map},\\
  Definition~\thref{d:subspace}.

\item[Theorem~\thref{t:lax-milgram-closed-subspace}] depends directly from:\\
  Lemma~\thref{l:closed-hilbert-subspace},\\
  Theorem~\thref{t:lax-milgram}.

\item[Lemma~\thref{l:cea}] depends directly from:\\
  Definition~\thref{d:space},\\
  Definition~\thref{d:bilinear-map},\\
  Definition~\thref{d:vector-subtraction},\\
  Definition~\thref{d:norm},\\
  Lemma~\thref{l:norm-preserves-zero},\\
  Lemma~\thref{l:norm-is-nonnegative},\\
  Definition~\thref{d:bounded-bilinear-form},\\
  Definition~\thref{d:coercive-bilinear-form},\\
  Definition~\thref{d:inner-product},\\
  Lemma~\thref{l:galerkin-orthogonality},\\
  Theorem~\thref{t:lax-milgram-closed-subspace}.

\item[Lemma~\thref{l:finite-dimensional-subspace-in-hilbert-space-is-closed}] depends directly from:\\
  Lemma~\thref{l:closed-subset-of-complete-is-complete},\\
  Definition~\thref{d:finite-dimensional-subspace},\\
  Lemma~\thref{l:linear-span-is-closed},\\
  Lemma~\thref{l:sum-is-orth-sum},\\
  Lemma~\thref{l:sum-of-complete-subspace-and-linear-span-is-closed},\\
  Definition~\thref{d:hilbert-space}.

\item[Theorem~\thref{t:lax-milgram-cea-finite-dimensional-subspace}] depends directly from:\\
  Theorem~\thref{t:lax-milgram},\\
  Theorem~\thref{t:lax-milgram-closed-subspace},\\
  Lemma~\thref{l:cea},\\
  Lemma~\thref{l:finite-dimensional-subspace-in-hilbert-space-is-closed}.

\end{description}

\clearpage
\section{Is a direct dependency of\ldots}
\label{s:is-a-direct-dependency-of}

\begin{description}

\item[Definition~\thref{d:supremum}]is a direct dependency of:\\
  Lemma~\thref{l:supremum-is-positive-scalar-multiplicative},\\
  Lemma~\thref{l:finite-maximum},\\
  Lemma~\thref{l:duality-infimum-supremum},\\
  Lemma~\thref{l:operator-norm-is-nonnegative},\\
  Theorem~\thref{t:continuous-linear-map},\\
  Lemma~\thref{l:finite-operator-norm-is-continuous},\\
  Theorem~\thref{t:normed-space-of-continuous-linear-maps},\\
  Lemma~\thref{l:operator-norm-estimation},\\
  Theorem~\thref{t:riesz-frechet}.

\item[Lemma~\thref{l:finite-supremum}]is a direct dependency of:\\
  Lemma~\thref{l:finite-maximum},\\
  Lemma~\thref{l:finite-infimum},\\
  Theorem~\thref{t:continuous-linear-map}.

\item[Lemma~\thref{l:discrete-lower-accumulation}]is a direct dependency of:\\
  Lemma~\thref{l:discrete-upper-accumulation}.

\item[Lemma~\thref{l:supremum-is-positive-scalar-multiplicative}]is a direct dependency of:\\
  Theorem~\thref{t:normed-space-of-continuous-linear-maps}.

\item[Definition~\thref{d:maximum}]is a direct dependency of:\\
  Lemma~\thref{l:finite-maximum}.

\item[Lemma~\thref{l:finite-maximum}]is a direct dependency of:\\
  Lemma~\thref{l:finite-minimum}.

\item[Definition~\thref{d:infimum}]is a direct dependency of:\\
  Lemma~\thref{l:duality-infimum-supremum},\\
  Theorem~\thref{t:orth-proj-onto-complete-convex},\\
  Lemma~\thref{l:characterization-of-orth-proj-onto-convex}.

\item[Lemma~\thref{l:duality-infimum-supremum}]is a direct dependency of:\\
  Lemma~\thref{l:finite-infimum},\\
  Lemma~\thref{l:finite-minimum}.

\item[Lemma~\thref{l:finite-infimum}]is a direct dependency of:\\
  Lemma~\thref{l:finite-infimum-discrete}.

\item[Lemma~\thref{l:discrete-upper-accumulation}]is a direct dependency of:\\
  Lemma~\thref{l:finite-infimum-discrete}.

\item[Lemma~\thref{l:finite-infimum-discrete}]is a direct dependency of:\\
  Theorem~\thref{t:orth-proj-onto-complete-convex}.

\item[Definition~\thref{d:minimum}]is a direct dependency of:\\
  Lemma~\thref{l:finite-minimum},\\
  Theorem~\thref{t:orth-proj-onto-complete-convex},\\
  Lemma~\thref{l:characterization-of-orth-proj-onto-convex}.

\item[Lemma~\thref{l:finite-minimum}]is a direct dependency of:\\
  Lemma~\thref{l:characterization-of-orth-proj-onto-convex}.

\item[Definition~\thref{d:distance}]is a direct dependency of:\\
  Lemma~\thref{l:iterated-triangle-inequality},\\
  Lemma~\thref{l:singleton-is-closed},\\
  Lemma~\thref{l:variant-of-point-separation},\\
  Lemma~\thref{l:limit-is-unique},\\
  Lemma~\thref{l:closure-is-limit-of-sequences},\\
  Lemma~\thref{l:stationary-sequence-is-convergent},\\
  Lemma~\thref{l:equivalent-definition-of-cauchy-sequence},\\
  Lemma~\thref{l:convergent-sequence-is-cauchy},\\
  Theorem~\thref{t:equivalent-definition-of-lipschitz-continuity},\\
  Lemma~\thref{l:zero-lipschitz-continuous-is-constant},\\
  Theorem~\thref{t:fixed-point},\\
  Lemma~\thref{l:norm-gives-distance},\\
  Lemma~\thref{l:scalar-multiplication-is-continuous}.

\item[Definition~\thref{d:metric-space}]is a direct dependency of:\\
  Lemma~\thref{l:norm-gives-distance}.

\item[Lemma~\thref{l:iterated-triangle-inequality}]is a direct dependency of:\\
  Theorem~\thref{t:fixed-point}.

\item[Definition~\thref{d:closed-ball}]is a direct dependency of:\\
  Lemma~\thref{l:closure-is-limit-of-sequences},\\
  Lemma~\thref{l:equivalent-definition-of-closed-unit-ball}.

\item[Definition~\thref{d:sphere}]is a direct dependency of:\\
  Lemma~\thref{l:equivalent-definition-of-unit-sphere}.

\item[Definition~\thref{d:open-subset}]is a direct dependency of:\\
  Lemma~\thref{l:equivalent-definition-of-closed-subset}.

\item[Definition~\thref{d:closed-subset}]is a direct dependency of:\\
  Lemma~\thref{l:equivalent-definition-of-closed-subset},\\
  Lemma~\thref{l:closed-equals-closure}.

\item[Lemma~\thref{l:equivalent-definition-of-closed-subset}]is a direct dependency of:\\
  Lemma~\thref{l:singleton-is-closed},\\
  Lemma~\thref{l:closed-equals-closure}.

\item[Lemma~\thref{l:singleton-is-closed}]is a direct dependency of:\\
  Lemma~\thref{l:linear-span-is-closed},\\
  Lemma~\thref{l:continuous-linear-maps-have-closed-kernel}.

\item[Definition~\thref{d:closure}]is a direct dependency of:\\
  Lemma~\thref{l:closure-is-limit-of-sequences},\\
  Lemma~\thref{l:closed-equals-closure},\\
  Lemma~\thref{l:closed-is-limit-of-sequences}.

\item[Definition~\thref{d:convergent-sequence}]is a direct dependency of:\\
  Lemma~\thref{l:limit-is-unique},\\
  Lemma~\thref{l:closure-is-limit-of-sequences},\\
  Lemma~\thref{l:stationary-sequence-is-convergent},\\
  Lemma~\thref{l:convergent-sequence-is-cauchy},\\
  Lemma~\thref{l:compatibility-of-limit-with-continuous-functions},\\
  Lemma~\thref{l:convergent-iterated-function-sequence},\\
  Lemma~\thref{l:linear-span-is-closed},\\
  Lemma~\thref{l:complete-normed-space-of-continuous-linear-maps}.

\item[Lemma~\thref{l:variant-of-point-separation}]is a direct dependency of:\\
  Lemma~\thref{l:limit-is-unique}.

\item[Lemma~\thref{l:limit-is-unique}]is a direct dependency of:\\
  Lemma~\thref{l:convergent-sequence-is-cauchy},\\
  Lemma~\thref{l:convergent-iterated-function-sequence},\\
  Lemma~\thref{l:linear-span-is-closed}.

\item[Lemma~\thref{l:closure-is-limit-of-sequences}]is a direct dependency of:\\
  Lemma~\thref{l:closed-is-limit-of-sequences},\\
  Lemma~\thref{l:closed-subset-of-complete-is-complete}.

\item[Lemma~\thref{l:closed-equals-closure}]is a direct dependency of:\\
  Lemma~\thref{l:closed-is-limit-of-sequences},\\
  Lemma~\thref{l:closed-subset-of-complete-is-complete}.

\item[Lemma~\thref{l:closed-is-limit-of-sequences}]is a direct dependency of:\\
  Lemma~\thref{l:linear-span-is-closed},\\
  Lemma~\thref{l:sum-of-complete-subspace-and-linear-span-is-closed}.

\item[Definition~\thref{d:stationary-sequence}]is a direct dependency of:\\
  Lemma~\thref{l:stationary-sequence-is-convergent},\\
  Lemma~\thref{l:stationary-iterated-function-sequence},\\
  Lemma~\thref{l:convergent-iterated-function-sequence},\\
  Theorem~\thref{t:fixed-point}.

\item[Lemma~\thref{l:stationary-sequence-is-convergent}]is a direct dependency of:\\
  Lemma~\thref{l:convergent-iterated-function-sequence},\\
  Theorem~\thref{t:fixed-point},\\
  Lemma~\thref{l:complete-normed-space-of-continuous-linear-maps}.

\item[Definition~\thref{d:cauchy-sequence}]is a direct dependency of:\\
  Lemma~\thref{l:equivalent-definition-of-cauchy-sequence},\\
  Lemma~\thref{l:convergent-sequence-is-cauchy},\\
  Lemma~\thref{l:linear-span-is-closed},\\
  Lemma~\thref{l:complete-normed-space-of-continuous-linear-maps},\\
  Theorem~\thref{t:orth-proj-onto-complete-convex}.

\item[Lemma~\thref{l:equivalent-definition-of-cauchy-sequence}]is a direct dependency of:\\
  Theorem~\thref{t:fixed-point}.

\item[Lemma~\thref{l:convergent-sequence-is-cauchy}]is a direct dependency of:\\
  Lemma~\thref{l:linear-span-is-closed}.

\item[Definition~\thref{d:complete-subset}]is a direct dependency of:\\
  Lemma~\thref{l:closed-subset-of-complete-is-complete},\\
  Theorem~\thref{t:fixed-point},\\
  Lemma~\thref{l:linear-span-is-closed},\\
  Lemma~\thref{l:complete-normed-space-of-continuous-linear-maps},\\
  Theorem~\thref{t:orth-proj-onto-complete-convex}.

\item[Definition~\thref{d:complete-metric-space}]is a direct dependency of:\\
  Lemma~\thref{l:closed-subset-of-complete-is-complete},\\
  Theorem~\thref{t:fixed-point},\\
  Lemma~\thref{l:complete-normed-space-of-continuous-linear-maps}.

\item[Lemma~\thref{l:closed-subset-of-complete-is-complete}]is a direct dependency of:\\
  Lemma~\thref{l:closed-hilbert-subspace},\\
  Lemma~\thref{l:finite-dimensional-subspace-in-hilbert-space-is-closed}.

\item[Definition~\thref{d:continuity-in-a-point}]is a direct dependency of:\\
  Lemma~\thref{l:compatibility-of-limit-with-continuous-functions},\\
  Lemma~\thref{l:uniform-continuous-is-continuous},\\
  Lemma~\thref{l:vector-addition-is-continuous},\\
  Lemma~\thref{l:scalar-multiplication-is-continuous},\\
  Theorem~\thref{t:continuous-linear-map}.

\item[Definition~\thref{d:pointwise-continuity}]is a direct dependency of:\\
  Lemma~\thref{l:uniform-continuous-is-continuous},\\
  Lemma~\thref{l:vector-addition-is-continuous},\\
  Lemma~\thref{l:scalar-multiplication-is-continuous},\\
  Theorem~\thref{t:continuous-linear-map}.

\item[Lemma~\thref{l:compatibility-of-limit-with-continuous-functions}]is a direct dependency of:\\
  Lemma~\thref{l:complete-normed-space-of-continuous-linear-maps},\\
  Theorem~\thref{t:orth-proj-onto-complete-convex},\\
  Lemma~\thref{l:sum-of-complete-subspace-and-linear-span-is-closed}.

\item[Definition~\thref{d:uniform-continuity}]is a direct dependency of:\\
  Lemma~\thref{l:uniform-continuous-is-continuous},\\
  Lemma~\thref{l:lipschitz-continuous-is-uniform-continuous}.

\item[Definition~\thref{d:lipschitz-continuity}]is a direct dependency of:\\
  Theorem~\thref{t:equivalent-definition-of-lipschitz-continuity},\\
  Lemma~\thref{l:zero-lipschitz-continuous-is-constant},\\
  Lemma~\thref{l:lipschitz-continuous-is-uniform-continuous},\\
  Lemma~\thref{l:iterate-lipschitz-continuous-mapping},\\
  Lemma~\thref{l:convergent-iterated-function-sequence},\\
  Theorem~\thref{t:fixed-point},\\
  Lemma~\thref{l:norm-is-one-lipschitz-continuous},\\
  Theorem~\thref{t:continuous-linear-map},\\
  Lemma~\thref{l:orth-proj-is-continuous-linear-map},\\
  Theorem~\thref{t:lax-milgram}.

\item[Theorem~\thref{t:equivalent-definition-of-lipschitz-continuity}]

\item[Definition~\thref{d:contraction}]is a direct dependency of:\\
  Theorem~\thref{t:fixed-point},\\
  Theorem~\thref{t:lax-milgram}.

\item[Lemma~\thref{l:uniform-continuous-is-continuous}]is a direct dependency of:\\
  Lemma~\thref{l:norm-is-continuous},\\
  Theorem~\thref{t:continuous-linear-map}.

\item[Lemma~\thref{l:zero-lipschitz-continuous-is-constant}]is a direct dependency of:\\
  Lemma~\thref{l:lipschitz-continuous-is-uniform-continuous},\\
  Lemma~\thref{l:convergent-iterated-function-sequence},\\
  Theorem~\thref{t:fixed-point}.

\item[Lemma~\thref{l:lipschitz-continuous-is-uniform-continuous}]is a direct dependency of:\\
  Lemma~\thref{l:norm-is-uniformly-continuous},\\
  Theorem~\thref{t:continuous-linear-map}.

\item[Definition~\thref{d:iterated-function-sequence}]is a direct dependency of:\\
  Lemma~\thref{l:stationary-iterated-function-sequence},\\
  Lemma~\thref{l:iterate-lipschitz-continuous-mapping},\\
  Lemma~\thref{l:convergent-iterated-function-sequence}.

\item[Lemma~\thref{l:stationary-iterated-function-sequence}]is a direct dependency of:\\
  Theorem~\thref{t:fixed-point}.

\item[Lemma~\thref{l:iterate-lipschitz-continuous-mapping}]is a direct dependency of:\\
  Theorem~\thref{t:fixed-point}.

\item[Lemma~\thref{l:convergent-iterated-function-sequence}]is a direct dependency of:\\
  Theorem~\thref{t:fixed-point}.

\item[Theorem~\thref{t:fixed-point}]is a direct dependency of:\\
  Theorem~\thref{t:lax-milgram}.

\item[Definition~\thref{d:space}]is a direct dependency of:\\
  Lemma~\thref{l:zero-times-yields-zero},\\
  Lemma~\thref{l:minus-times-yields-opposite-vector},\\
  Lemma~\thref{l:times-zero-yields-zero},\\
  Lemma~\thref{l:zero-product-property},\\
  Lemma~\thref{l:closed-under-vector-operations-is-subspace},\\
  Lemma~\thref{l:closed-under-linear-combination-is-subspace},\\
  Lemma~\thref{l:equivalent-definition-of-direct-sum},\\
  Lemma~\thref{l:product-is-space},\\
  Lemma~\thref{l:space-of-mappings-to-space},\\
  Lemma~\thref{l:linear-map-preserves-linear-combinations},\\
  Lemma~\thref{l:space-of-linear-maps},\\
  Lemma~\thref{l:kernel-is-subspace},\\
  Lemma~\thref{l:injective-linear-map-has-zero-kernel},\\
  Lemma~\thref{l:norm-preserves-zero},\\
  Lemma~\thref{l:norm-is-nonnegative},\\
  Lemma~\thref{l:linear-span-is-closed},\\
  Lemma~\thref{l:zero-on-unit-sphere-is-zero},\\
  Lemma~\thref{l:vector-addition-is-continuous},\\
  Lemma~\thref{l:scalar-multiplication-is-continuous},\\
  Lemma~\thref{l:complete-normed-space-of-continuous-linear-maps},\\
  Lemma~\thref{l:representation-for-bounded-bilinear-form},\\
  Lemma~\thref{l:inner-product-with-zero-is-zero},\\
  Theorem~\thref{t:orth-proj-onto-complete-convex},\\
  Lemma~\thref{l:characterization-of-orth-proj-onto-convex},\\
  Theorem~\thref{t:orth-proj-onto-complete-subspace},\\
  Lemma~\thref{l:characterization-of-orth-proj-onto-subspace},\\
  Lemma~\thref{l:orth-proj-is-continuous-linear-map},\\
  Theorem~\thref{t:direct-sum-with-orth-compl-when-complete},\\
  Lemma~\thref{l:sum-is-orth-sum},\\
  Theorem~\thref{t:riesz-frechet},\\
  Theorem~\thref{t:lax-milgram},\\
  Lemma~\thref{l:cea}.

\item[Definition~\thref{d:set-of-mappings-to-space}]is a direct dependency of:\\
  Lemma~\thref{l:space-of-mappings-to-space}.

\item[Definition~\thref{d:linear-map}]is a direct dependency of:\\
  Lemma~\thref{l:linear-map-preserves-zero},\\
  Lemma~\thref{l:linear-map-preserves-linear-combinations},\\
  Lemma~\thref{l:space-of-linear-maps},\\
  Lemma~\thref{l:identity-map-is-linear-map},\\
  Lemma~\thref{l:composition-of-linear-maps-is-bilinear},\\
  Lemma~\thref{l:injective-linear-map-has-zero-kernel},\\
  Lemma~\thref{l:zero-on-unit-sphere-is-zero},\\
  Lemma~\thref{l:norm-of-image-of-unit-vector},\\
  Theorem~\thref{t:continuous-linear-map},\\
  Lemma~\thref{l:complete-normed-space-of-continuous-linear-maps},\\
  Lemma~\thref{l:bra-ket-is-bilinear-map},\\
  Theorem~\thref{t:riesz-frechet}.

\item[Definition~\thref{d:set-of-linear-maps}]is a direct dependency of:\\
  Lemma~\thref{l:space-of-linear-maps}.

\item[Definition~\thref{d:linear-form}]is a direct dependency of:\\
  Lemma~\thref{l:representation-for-bounded-bilinear-form},\\
  Theorem~\thref{t:riesz-frechet}.

\item[Definition~\thref{d:bilinear-map}]is a direct dependency of:\\
  Lemma~\thref{l:composition-of-linear-maps-is-bilinear},\\
  Lemma~\thref{l:bra-ket-is-bilinear-map},\\
  Lemma~\thref{l:representation-for-bounded-bilinear-form},\\
  Lemma~\thref{l:inner-product-with-zero-is-zero},\\
  Lemma~\thref{l:square-expansion-plus},\\
  Lemma~\thref{l:square-expansion-minus},\\
  Lemma~\thref{l:cauchy-schwarz-inequality},\\
  Lemma~\thref{l:characterization-of-orth-proj-onto-convex},\\
  Lemma~\thref{l:characterization-of-orth-proj-onto-subspace},\\
  Lemma~\thref{l:orth-proj-is-continuous-linear-map},\\
  Lemma~\thref{l:orth-compl-is-subspace},\\
  Theorem~\thref{t:direct-sum-with-orth-compl-when-complete},\\
  Theorem~\thref{t:riesz-frechet},\\
  Theorem~\thref{t:lax-milgram},\\
  Lemma~\thref{l:galerkin-orthogonality},\\
  Lemma~\thref{l:cea}.

\item[Definition~\thref{d:bilinear-form}]is a direct dependency of:\\
  Lemma~\thref{l:representation-for-bounded-bilinear-form}.

\item[Definition~\thref{d:set-of-bilinear-forms}]is a direct dependency of:\\
  Lemma~\thref{l:representation-for-bounded-bilinear-form}.

\item[Lemma~\thref{l:zero-times-yields-zero}]is a direct dependency of:\\
  Lemma~\thref{l:minus-times-yields-opposite-vector},\\
  Lemma~\thref{l:zero-product-property},\\
  Lemma~\thref{l:linear-map-preserves-zero},\\
  Lemma~\thref{l:linear-map-preserves-linear-combinations},\\
  Lemma~\thref{l:norm-preserves-zero},\\
  Lemma~\thref{l:scalar-multiplication-is-continuous}.

\item[Lemma~\thref{l:minus-times-yields-opposite-vector}]is a direct dependency of:\\
  Lemma~\thref{l:closed-under-vector-operations-is-subspace},\\
  Lemma~\thref{l:equivalent-definition-of-direct-sum},\\
  Theorem~\thref{t:lax-milgram}.

\item[Definition~\thref{d:vector-subtraction}]is a direct dependency of:\\
  Lemma~\thref{l:times-zero-yields-zero},\\
  Lemma~\thref{l:equivalent-definition-of-direct-sum},\\
  Lemma~\thref{l:injective-linear-map-has-zero-kernel},\\
  Lemma~\thref{l:norm-gives-distance},\\
  Lemma~\thref{l:linear-span-is-closed},\\
  Theorem~\thref{t:continuous-linear-map},\\
  Lemma~\thref{l:representation-for-bounded-bilinear-form},\\
  Lemma~\thref{l:inner-product-with-zero-is-zero},\\
  Lemma~\thref{l:square-expansion-minus},\\
  Theorem~\thref{t:orth-proj-onto-complete-convex},\\
  Lemma~\thref{l:characterization-of-orth-proj-onto-convex},\\
  Lemma~\thref{l:characterization-of-orth-proj-onto-subspace},\\
  Theorem~\thref{t:riesz-frechet},\\
  Theorem~\thref{t:lax-milgram},\\
  Lemma~\thref{l:cea}.

\item[Definition~\thref{d:scalar-division}]is a direct dependency of:\\
  Lemma~\thref{l:normalization-by-nonzero},\\
  Theorem~\thref{t:orth-proj-onto-complete-convex},\\
  Theorem~\thref{t:riesz-frechet}.

\item[Lemma~\thref{l:times-zero-yields-zero}]is a direct dependency of:\\
  Lemma~\thref{l:zero-product-property},\\
  Lemma~\thref{l:space-of-linear-maps},\\
  Lemma~\thref{l:kernel-is-subspace},\\
  Lemma~\thref{l:zero-on-unit-sphere-is-zero}.

\item[Lemma~\thref{l:zero-product-property}]is a direct dependency of:\\
  Lemma~\thref{l:closed-under-linear-combination-is-subspace},\\
  Lemma~\thref{l:direct-sum-with-linear-span},\\
  Theorem~\thref{t:riesz-frechet},\\
  Theorem~\thref{t:lax-milgram}.

\item[Definition~\thref{d:subspace}]is a direct dependency of:\\
  Lemma~\thref{l:trivial-subspaces},\\
  Lemma~\thref{l:closed-under-vector-operations-is-subspace},\\
  Lemma~\thref{l:inner-product-subspace},\\
  Theorem~\thref{t:orth-proj-onto-complete-subspace},\\
  Lemma~\thref{l:characterization-of-orth-proj-onto-subspace},\\
  Lemma~\thref{l:orth-proj-is-continuous-linear-map},\\
  Lemma~\thref{l:closed-hilbert-subspace},\\
  Lemma~\thref{l:galerkin-orthogonality}.

\item[Lemma~\thref{l:trivial-subspaces}]is a direct dependency of:\\
  Lemma~\thref{l:trivial-orth-compls}.

\item[Lemma~\thref{l:closed-under-vector-operations-is-subspace}]is a direct dependency of:\\
  Lemma~\thref{l:closed-under-linear-combination-is-subspace},\\
  Lemma~\thref{l:equivalent-definition-of-direct-sum},\\
  Lemma~\thref{l:direct-sum-with-linear-span},\\
  Theorem~\thref{t:normed-space-of-continuous-linear-maps},\\
  Theorem~\thref{t:riesz-frechet}.

\item[Lemma~\thref{l:closed-under-linear-combination-is-subspace}]is a direct dependency of:\\
  Lemma~\thref{l:space-of-linear-maps},\\
  Lemma~\thref{l:kernel-is-subspace},\\
  Lemma~\thref{l:subspace-is-convex},\\
  Lemma~\thref{l:orth-compl-is-subspace},\\
  Lemma~\thref{l:sum-is-orth-sum},\\
  Lemma~\thref{l:sum-of-complete-subspace-and-linear-span-is-closed}.

\item[Definition~\thref{d:linear-span}]is a direct dependency of:\\
  Lemma~\thref{l:direct-sum-with-linear-span},\\
  Lemma~\thref{l:linear-span-is-closed},\\
  Lemma~\thref{l:sum-is-orth-sum},\\
  Lemma~\thref{l:sum-of-complete-subspace-and-linear-span-is-closed}.

\item[Definition~\thref{d:sum-of-subspaces}]is a direct dependency of:\\
  Theorem~\thref{t:direct-sum-with-orth-compl-when-complete},\\
  Lemma~\thref{l:sum-is-orth-sum},\\
  Lemma~\thref{l:sum-of-complete-subspace-and-linear-span-is-closed}.

\item[Definition~\thref{d:finite-dimensional-subspace}]is a direct dependency of:\\
  Lemma~\thref{l:finite-dimensional-subspace-in-hilbert-space-is-closed}.

\item[Definition~\thref{d:direct-sum-of-subspaces}]is a direct dependency of:\\
  Lemma~\thref{l:equivalent-definition-of-direct-sum},\\
  Theorem~\thref{t:direct-sum-with-orth-compl-when-complete}.

\item[Lemma~\thref{l:equivalent-definition-of-direct-sum}]is a direct dependency of:\\
  Lemma~\thref{l:direct-sum-with-linear-span},\\
  Theorem~\thref{t:direct-sum-with-orth-compl-when-complete}.

\item[Lemma~\thref{l:direct-sum-with-linear-span}]is a direct dependency of:\\
  Lemma~\thref{l:sum-of-complete-subspace-and-linear-span-is-closed}.

\item[Definition~\thref{d:product-vector-operations}]is a direct dependency of:\\
  Lemma~\thref{l:product-is-space},\\
  Lemma~\thref{l:product-is-normed-space},\\
  Lemma~\thref{l:vector-addition-is-continuous}.

\item[Lemma~\thref{l:product-is-space}]is a direct dependency of:\\
  Lemma~\thref{l:composition-of-linear-maps-is-bilinear},\\
  Lemma~\thref{l:product-is-normed-space},\\
  Lemma~\thref{l:bra-ket-is-bilinear-map}.

\item[Definition~\thref{d:inherited-vector-operations}]is a direct dependency of:\\
  Lemma~\thref{l:space-of-mappings-to-space},\\
  Lemma~\thref{l:space-of-linear-maps},\\
  Lemma~\thref{l:composition-of-linear-maps-is-bilinear},\\
  Theorem~\thref{t:normed-space-of-continuous-linear-maps},\\
  Lemma~\thref{l:complete-normed-space-of-continuous-linear-maps},\\
  Lemma~\thref{l:bra-ket-is-bilinear-map},\\
  Lemma~\thref{l:representation-for-bounded-bilinear-form}.

\item[Lemma~\thref{l:space-of-mappings-to-space}]is a direct dependency of:\\
  Lemma~\thref{l:space-of-linear-maps}.

\item[Lemma~\thref{l:linear-map-preserves-zero}]is a direct dependency of:\\
  Lemma~\thref{l:kernel-is-subspace},\\
  Lemma~\thref{l:injective-linear-map-has-zero-kernel},\\
  Lemma~\thref{l:zero-on-unit-sphere-is-zero},\\
  Theorem~\thref{t:continuous-linear-map},\\
  Lemma~\thref{l:operator-norm-estimation},\\
  Lemma~\thref{l:complete-normed-space-of-continuous-linear-maps},\\
  Theorem~\thref{t:riesz-frechet}.

\item[Lemma~\thref{l:linear-map-preserves-linear-combinations}]is a direct dependency of:\\
  Lemma~\thref{l:composition-of-linear-maps-is-bilinear},\\
  Lemma~\thref{l:kernel-is-subspace},\\
  Lemma~\thref{l:complete-normed-space-of-continuous-linear-maps},\\
  Lemma~\thref{l:representation-for-bounded-bilinear-form},\\
  Lemma~\thref{l:orth-proj-is-continuous-linear-map},\\
  Theorem~\thref{t:riesz-frechet}.

\item[Lemma~\thref{l:space-of-linear-maps}]is a direct dependency of:\\
  Lemma~\thref{l:composition-of-linear-maps-is-bilinear}.

\item[Definition~\thref{d:identity-map}]is a direct dependency of:\\
  Lemma~\thref{l:identity-map-is-linear-map},\\
  Lemma~\thref{l:identity-map-is-linear-isometry}.

\item[Lemma~\thref{l:identity-map-is-linear-map}]is a direct dependency of:\\
  Lemma~\thref{l:identity-map-is-linear-isometry}.

\item[Lemma~\thref{l:composition-of-linear-maps-is-bilinear}]is a direct dependency of:\\
  Lemma~\thref{l:compatibility-of-composition-with-continuity}.

\item[Definition~\thref{d:isomorphism}]is a direct dependency of:\\
  Theorem~\thref{t:riesz-frechet}.

\item[Definition~\thref{d:kernel}]is a direct dependency of:\\
  Lemma~\thref{l:kernel-is-subspace},\\
  Lemma~\thref{l:injective-linear-map-has-zero-kernel},\\
  Lemma~\thref{l:continuous-linear-maps-have-closed-kernel},\\
  Theorem~\thref{t:riesz-frechet}.

\item[Lemma~\thref{l:kernel-is-subspace}]is a direct dependency of:\\
  Theorem~\thref{t:riesz-frechet}.

\item[Lemma~\thref{l:injective-linear-map-has-zero-kernel}]is a direct dependency of:\\
  Theorem~\thref{t:riesz-frechet}.

\item[Lemma~\thref{l:k-is-space}]is a direct dependency of:\\
  Lemma~\thref{l:bra-ket-is-bilinear-map},\\
  Lemma~\thref{l:representation-for-bounded-bilinear-form}.

\item[Definition~\thref{d:norm}]is a direct dependency of:\\
  Lemma~\thref{l:k-is-normed-space},\\
  Lemma~\thref{l:norm-preserves-zero},\\
  Lemma~\thref{l:norm-is-nonnegative},\\
  Lemma~\thref{l:normalization-by-nonzero},\\
  Lemma~\thref{l:norm-gives-distance},\\
  Lemma~\thref{l:linear-span-is-closed},\\
  Lemma~\thref{l:reverse-triangle-inequality},\\
  Lemma~\thref{l:product-is-normed-space},\\
  Lemma~\thref{l:vector-addition-is-continuous},\\
  Lemma~\thref{l:scalar-multiplication-is-continuous},\\
  Lemma~\thref{l:norm-of-image-of-unit-vector},\\
  Lemma~\thref{l:norm-of-image-of-unit-sphere},\\
  Theorem~\thref{t:continuous-linear-map},\\
  Theorem~\thref{t:normed-space-of-continuous-linear-maps},\\
  Lemma~\thref{l:operator-norm-estimation},\\
  Lemma~\thref{l:complete-normed-space-of-continuous-linear-maps},\\
  Lemma~\thref{l:coercivity-constant-is-less-than-continuity-constant},\\
  Lemma~\thref{l:inner-product-gives-norm},\\
  Theorem~\thref{t:orth-proj-onto-complete-convex},\\
  Lemma~\thref{l:orth-proj-is-continuous-linear-map},\\
  Theorem~\thref{t:riesz-frechet},\\
  Theorem~\thref{t:lax-milgram},\\
  Lemma~\thref{l:cea}.

\item[Definition~\thref{d:normed-space}]is a direct dependency of:\\
  Lemma~\thref{l:k-is-normed-space},\\
  Lemma~\thref{l:norm-preserves-zero},\\
  Lemma~\thref{l:norm-is-nonnegative},\\
  Lemma~\thref{l:reverse-triangle-inequality},\\
  Lemma~\thref{l:product-is-normed-space},\\
  Theorem~\thref{t:normed-space-of-continuous-linear-maps},\\
  Lemma~\thref{l:representation-for-bounded-bilinear-form},\\
  Lemma~\thref{l:inner-product-gives-norm},\\
  Theorem~\thref{t:riesz-frechet},\\
  Theorem~\thref{t:lax-milgram}.

\item[Lemma~\thref{l:k-is-normed-space}]is a direct dependency of:\\
  Lemma~\thref{l:topological-dual-is-complete-normed-space}.

\item[Lemma~\thref{l:norm-preserves-zero}]is a direct dependency of:\\
  Lemma~\thref{l:norm-of-image-of-unit-sphere},\\
  Theorem~\thref{t:continuous-linear-map},\\
  Lemma~\thref{l:operator-norm-estimation},\\
  Theorem~\thref{t:riesz-frechet},\\
  Theorem~\thref{t:lax-milgram},\\
  Lemma~\thref{l:cea}.

\item[Lemma~\thref{l:norm-is-nonnegative}]is a direct dependency of:\\
  Lemma~\thref{l:normalization-by-nonzero},\\
  Lemma~\thref{l:norm-gives-distance},\\
  Lemma~\thref{l:linear-span-is-closed},\\
  Lemma~\thref{l:product-is-normed-space},\\
  Lemma~\thref{l:norm-of-image-of-unit-vector},\\
  Lemma~\thref{l:operator-norm-is-nonnegative},\\
  Theorem~\thref{t:normed-space-of-continuous-linear-maps},\\
  Lemma~\thref{l:operator-norm-estimation},\\
  Lemma~\thref{l:representation-for-bounded-bilinear-form},\\
  Lemma~\thref{l:coercivity-constant-is-less-than-continuity-constant},\\
  Theorem~\thref{t:orth-proj-onto-complete-convex},\\
  Lemma~\thref{l:characterization-of-orth-proj-onto-convex},\\
  Lemma~\thref{l:orth-proj-is-continuous-linear-map},\\
  Theorem~\thref{t:riesz-frechet},\\
  Theorem~\thref{t:lax-milgram},\\
  Lemma~\thref{l:cea}.

\item[Lemma~\thref{l:normalization-by-nonzero}]is a direct dependency of:\\
  Lemma~\thref{l:zero-on-unit-sphere-is-zero},\\
  Lemma~\thref{l:norm-of-image-of-unit-vector},\\
  Theorem~\thref{t:riesz-frechet}.

\item[Definition~\thref{d:distance-associated-with-norm}]is a direct dependency of:\\
  Lemma~\thref{l:norm-gives-distance},\\
  Lemma~\thref{l:linear-span-is-closed},\\
  Lemma~\thref{l:equivalent-definition-of-closed-unit-ball},\\
  Lemma~\thref{l:equivalent-definition-of-unit-sphere},\\
  Lemma~\thref{l:norm-is-one-lipschitz-continuous},\\
  Lemma~\thref{l:vector-addition-is-continuous},\\
  Lemma~\thref{l:scalar-multiplication-is-continuous},\\
  Lemma~\thref{l:complete-normed-space-of-continuous-linear-maps}.

\item[Lemma~\thref{l:norm-gives-distance}]is a direct dependency of:\\
  Lemma~\thref{l:equivalent-definition-of-closed-unit-ball},\\
  Lemma~\thref{l:equivalent-definition-of-unit-sphere},\\
  Lemma~\thref{l:complete-normed-space-of-continuous-linear-maps},\\
  Definition~\thref{d:hilbert-space}.

\item[Lemma~\thref{l:linear-span-is-closed}]is a direct dependency of:\\
  Lemma~\thref{l:sum-of-complete-subspace-and-linear-span-is-closed},\\
  Lemma~\thref{l:finite-dimensional-subspace-in-hilbert-space-is-closed}.

\item[Definition~\thref{d:closed-unit-ball}]is a direct dependency of:\\
  Lemma~\thref{l:equivalent-definition-of-closed-unit-ball}.

\item[Lemma~\thref{l:equivalent-definition-of-closed-unit-ball}]is a direct dependency of:\\
  Theorem~\thref{t:continuous-linear-map}.

\item[Definition~\thref{d:unit-sphere}]is a direct dependency of:\\
  Lemma~\thref{l:equivalent-definition-of-unit-sphere}.

\item[Lemma~\thref{l:equivalent-definition-of-unit-sphere}]is a direct dependency of:\\
  Lemma~\thref{l:zero-on-unit-sphere-is-zero},\\
  Lemma~\thref{l:norm-of-image-of-unit-sphere},\\
  Theorem~\thref{t:continuous-linear-map},\\
  Lemma~\thref{l:compatibility-of-composition-with-continuity},\\
  Lemma~\thref{l:representation-for-bounded-bilinear-form}.

\item[Lemma~\thref{l:zero-on-unit-sphere-is-zero}]is a direct dependency of:\\
  Theorem~\thref{t:normed-space-of-continuous-linear-maps}.

\item[Lemma~\thref{l:reverse-triangle-inequality}]is a direct dependency of:\\
  Lemma~\thref{l:norm-is-one-lipschitz-continuous}.

\item[Lemma~\thref{l:norm-is-one-lipschitz-continuous}]is a direct dependency of:\\
  Lemma~\thref{l:norm-is-uniformly-continuous}.

\item[Lemma~\thref{l:norm-is-uniformly-continuous}]is a direct dependency of:\\
  Lemma~\thref{l:norm-is-continuous}.

\item[Lemma~\thref{l:norm-is-continuous}]is a direct dependency of:\\
  Lemma~\thref{l:complete-normed-space-of-continuous-linear-maps},\\
  Theorem~\thref{t:orth-proj-onto-complete-convex}.

\item[Definition~\thref{d:linear-isometry}]is a direct dependency of:\\
  Lemma~\thref{l:identity-map-is-linear-isometry},\\
  Lemma~\thref{l:linear-isometry-is-continuous},\\
  Theorem~\thref{t:riesz-frechet},\\
  Theorem~\thref{t:lax-milgram}.

\item[Lemma~\thref{l:identity-map-is-linear-isometry}]is a direct dependency of:\\
  Lemma~\thref{l:identity-map-is-continuous}.

\item[Definition~\thref{d:product-norm}]is a direct dependency of:\\
  Lemma~\thref{l:product-is-normed-space},\\
  Lemma~\thref{l:vector-addition-is-continuous}.

\item[Lemma~\thref{l:product-is-normed-space}]is a direct dependency of:\\
  Lemma~\thref{l:vector-addition-is-continuous}.

\item[Lemma~\thref{l:vector-addition-is-continuous}]is a direct dependency of:\\
  Lemma~\thref{l:complete-normed-space-of-continuous-linear-maps}.

\item[Lemma~\thref{l:scalar-multiplication-is-continuous}]is a direct dependency of:\\
  Lemma~\thref{l:complete-normed-space-of-continuous-linear-maps}.

\item[Lemma~\thref{l:norm-of-image-of-unit-vector}]is a direct dependency of:\\
  Lemma~\thref{l:norm-of-image-of-unit-sphere}.

\item[Lemma~\thref{l:norm-of-image-of-unit-sphere}]is a direct dependency of:\\
  Lemma~\thref{l:equivalent-definition-of-operator-norm}.

\item[Definition~\thref{d:operator-norm}]is a direct dependency of:\\
  Lemma~\thref{l:equivalent-definition-of-operator-norm},\\
  Theorem~\thref{t:continuous-linear-map},\\
  Lemma~\thref{l:operator-norm-estimation},\\
  Lemma~\thref{l:complete-normed-space-of-continuous-linear-maps},\\
  Theorem~\thref{t:riesz-frechet}.

\item[Lemma~\thref{l:equivalent-definition-of-operator-norm}]is a direct dependency of:\\
  Lemma~\thref{l:operator-norm-is-nonnegative},\\
  Theorem~\thref{t:continuous-linear-map},\\
  Theorem~\thref{t:normed-space-of-continuous-linear-maps}.

\item[Lemma~\thref{l:operator-norm-is-nonnegative}]is a direct dependency of:\\
  Theorem~\thref{t:continuous-linear-map}.

\item[Definition~\thref{d:bounded-linear-map}]is a direct dependency of:\\
  Theorem~\thref{t:continuous-linear-map},\\
  Lemma~\thref{l:finite-operator-norm-is-continuous},\\
  Lemma~\thref{l:complete-normed-space-of-continuous-linear-maps},\\
  Lemma~\thref{l:representation-for-bounded-bilinear-form},\\
  Theorem~\thref{t:riesz-frechet}.

\item[Definition~\thref{d:linear-map-bounded-on-unit-ball}]is a direct dependency of:\\
  Theorem~\thref{t:continuous-linear-map},\\
  Lemma~\thref{l:finite-operator-norm-is-continuous}.

\item[Definition~\thref{d:linear-map-bounded-on-unit-sphere}]is a direct dependency of:\\
  Theorem~\thref{t:continuous-linear-map},\\
  Lemma~\thref{l:finite-operator-norm-is-continuous},\\
  Theorem~\thref{t:normed-space-of-continuous-linear-maps}.

\item[Theorem~\thref{t:continuous-linear-map}]is a direct dependency of:\\
  Lemma~\thref{l:finite-operator-norm-is-continuous},\\
  Lemma~\thref{l:complete-normed-space-of-continuous-linear-maps},\\
  Theorem~\thref{t:riesz-frechet}.

\item[Definition~\thref{d:set-of-continuous-linear-maps}]is a direct dependency of:\\
  Lemma~\thref{l:finite-operator-norm-is-continuous},\\
  Theorem~\thref{t:normed-space-of-continuous-linear-maps}.

\item[Lemma~\thref{l:finite-operator-norm-is-continuous}]is a direct dependency of:\\
  Lemma~\thref{l:linear-isometry-is-continuous},\\
  Theorem~\thref{t:normed-space-of-continuous-linear-maps},\\
  Lemma~\thref{l:compatibility-of-composition-with-continuity},\\
  Lemma~\thref{l:complete-normed-space-of-continuous-linear-maps},\\
  Lemma~\thref{l:representation-for-bounded-bilinear-form},\\
  Theorem~\thref{t:riesz-frechet}.

\item[Lemma~\thref{l:linear-isometry-is-continuous}]is a direct dependency of:\\
  Lemma~\thref{l:identity-map-is-continuous},\\
  Theorem~\thref{t:riesz-frechet}.

\item[Lemma~\thref{l:identity-map-is-continuous}]is a direct dependency of:\\
  Lemma~\thref{l:sum-of-complete-subspace-and-linear-span-is-closed},\\
  Theorem~\thref{t:lax-milgram}.

\item[Theorem~\thref{t:normed-space-of-continuous-linear-maps}]is a direct dependency of:\\
  Lemma~\thref{l:operator-norm-estimation},\\
  Lemma~\thref{l:complete-normed-space-of-continuous-linear-maps},\\
  Lemma~\thref{l:topological-dual-is-complete-normed-space},\\
  Lemma~\thref{l:representation-for-bounded-bilinear-form},\\
  Lemma~\thref{l:sum-of-complete-subspace-and-linear-span-is-closed},\\
  Theorem~\thref{t:lax-milgram}.

\item[Lemma~\thref{l:operator-norm-estimation}]is a direct dependency of:\\
  Lemma~\thref{l:compatibility-of-composition-with-continuity},\\
  Theorem~\thref{t:lax-milgram}.

\item[Lemma~\thref{l:continuous-linear-maps-have-closed-kernel}]is a direct dependency of:\\
  Theorem~\thref{t:riesz-frechet}.

\item[Lemma~\thref{l:compatibility-of-composition-with-continuity}]is a direct dependency of:\\
  Theorem~\thref{t:lax-milgram}.

\item[Lemma~\thref{l:complete-normed-space-of-continuous-linear-maps}]is a direct dependency of:\\
  Lemma~\thref{l:topological-dual-is-complete-normed-space}.

\item[Definition~\thref{d:topological-dual}]is a direct dependency of:\\
  Lemma~\thref{l:representation-for-bounded-bilinear-form},\\
  Theorem~\thref{t:riesz-frechet}.

\item[Definition~\thref{d:dual-norm}]is a direct dependency of:\\
  Lemma~\thref{l:topological-dual-is-complete-normed-space},\\
  Lemma~\thref{l:representation-for-bounded-bilinear-form},\\
  Theorem~\thref{t:riesz-frechet}.

\item[Lemma~\thref{l:topological-dual-is-complete-normed-space}]is a direct dependency of:\\
  Lemma~\thref{l:bra-ket-is-bilinear-map},\\
  Lemma~\thref{l:representation-for-bounded-bilinear-form},\\
  Theorem~\thref{t:riesz-frechet},\\
  Theorem~\thref{t:lax-milgram}.

\item[Definition~\thref{d:bra-ket-notation}]is a direct dependency of:\\
  Lemma~\thref{l:bra-ket-is-bilinear-map},\\
  Lemma~\thref{l:representation-for-bounded-bilinear-form},\\
  Theorem~\thref{t:riesz-frechet}.

\item[Lemma~\thref{l:bra-ket-is-bilinear-map}]is a direct dependency of:\\
  Lemma~\thref{l:representation-for-bounded-bilinear-form},\\
  Theorem~\thref{t:riesz-frechet}.

\item[Definition~\thref{d:bounded-bilinear-form}]is a direct dependency of:\\
  Lemma~\thref{l:representation-for-bounded-bilinear-form},\\
  Lemma~\thref{l:coercivity-constant-is-less-than-continuity-constant},\\
  Theorem~\thref{t:lax-milgram},\\
  Lemma~\thref{l:cea}.

\item[Lemma~\thref{l:representation-for-bounded-bilinear-form}]is a direct dependency of:\\
  Theorem~\thref{t:lax-milgram}.

\item[Definition~\thref{d:coercive-bilinear-form}]is a direct dependency of:\\
  Lemma~\thref{l:coercivity-constant-is-less-than-continuity-constant},\\
  Theorem~\thref{t:lax-milgram},\\
  Lemma~\thref{l:cea}.

\item[Lemma~\thref{l:coercivity-constant-is-less-than-continuity-constant}]is a direct dependency of:\\
  Theorem~\thref{t:lax-milgram}.

\item[Definition~\thref{d:inner-product}]is a direct dependency of:\\
  Lemma~\thref{l:inner-product-subspace},\\
  Lemma~\thref{l:inner-product-with-zero-is-zero},\\
  Lemma~\thref{l:square-expansion-plus},\\
  Lemma~\thref{l:square-expansion-minus},\\
  Lemma~\thref{l:cauchy-schwarz-inequality},\\
  Lemma~\thref{l:squared-norm},\\
  Lemma~\thref{l:cauchy-schwarz-inequality-with-norms},\\
  Lemma~\thref{l:triangle-inequality},\\
  Lemma~\thref{l:inner-product-gives-norm},\\
  Lemma~\thref{l:characterization-of-orth-proj-onto-convex},\\
  Lemma~\thref{l:characterization-of-orth-proj-onto-subspace},\\
  Lemma~\thref{l:orth-proj-is-continuous-linear-map},\\
  Lemma~\thref{l:trivial-orth-compls},\\
  Lemma~\thref{l:orth-compl-is-subspace},\\
  Lemma~\thref{l:zero-intersection-with-orth-compl},\\
  Theorem~\thref{t:direct-sum-with-orth-compl-when-complete},\\
  Theorem~\thref{t:riesz-frechet},\\
  Theorem~\thref{t:lax-milgram},\\
  Lemma~\thref{l:cea}.

\item[Definition~\thref{d:inner-product-space}]is a direct dependency of:\\
  Lemma~\thref{l:inner-product-subspace},\\
  Theorem~\thref{t:orth-proj-onto-complete-convex},\\
  Lemma~\thref{l:characterization-of-orth-proj-onto-convex},\\
  Theorem~\thref{t:riesz-frechet}.

\item[Lemma~\thref{l:inner-product-subspace}]is a direct dependency of:\\
  Lemma~\thref{l:closed-hilbert-subspace}.

\item[Lemma~\thref{l:inner-product-with-zero-is-zero}]is a direct dependency of:\\
  Lemma~\thref{l:trivial-orth-compls},\\
  Lemma~\thref{l:orth-compl-is-subspace},\\
  Theorem~\thref{t:direct-sum-with-orth-compl-when-complete},\\
  Theorem~\thref{t:riesz-frechet}.

\item[Lemma~\thref{l:square-expansion-plus}]is a direct dependency of:\\
  Lemma~\thref{l:square-expansion-minus},\\
  Lemma~\thref{l:parallelogram-identity},\\
  Lemma~\thref{l:cauchy-schwarz-inequality},\\
  Lemma~\thref{l:triangle-inequality},\\
  Lemma~\thref{l:characterization-of-orth-proj-onto-convex},\\
  Theorem~\thref{t:lax-milgram}.

\item[Lemma~\thref{l:square-expansion-minus}]is a direct dependency of:\\
  Lemma~\thref{l:parallelogram-identity}.

\item[Lemma~\thref{l:parallelogram-identity}]is a direct dependency of:\\
  Theorem~\thref{t:orth-proj-onto-complete-convex}.

\item[Lemma~\thref{l:cauchy-schwarz-inequality}]is a direct dependency of:\\
  Lemma~\thref{l:cauchy-schwarz-inequality-with-norms}.

\item[Definition~\thref{d:square-root-of-inner-square}]is a direct dependency of:\\
  Lemma~\thref{l:squared-norm},\\
  Lemma~\thref{l:cauchy-schwarz-inequality-with-norms},\\
  Lemma~\thref{l:inner-product-gives-norm},\\
  Definition~\thref{d:hilbert-space}.

\item[Lemma~\thref{l:squared-norm}]is a direct dependency of:\\
  Lemma~\thref{l:triangle-inequality},\\
  Theorem~\thref{t:orth-proj-onto-complete-convex},\\
  Lemma~\thref{l:characterization-of-orth-proj-onto-convex},\\
  Lemma~\thref{l:orth-proj-is-continuous-linear-map},\\
  Theorem~\thref{t:riesz-frechet}.

\item[Lemma~\thref{l:cauchy-schwarz-inequality-with-norms}]is a direct dependency of:\\
  Lemma~\thref{l:triangle-inequality},\\
  Lemma~\thref{l:orth-proj-is-continuous-linear-map},\\
  Theorem~\thref{t:riesz-frechet}.

\item[Lemma~\thref{l:triangle-inequality}]is a direct dependency of:\\
  Lemma~\thref{l:inner-product-gives-norm}.

\item[Lemma~\thref{l:inner-product-gives-norm}]is a direct dependency of:\\
  Definition~\thref{d:hilbert-space},\\
  Theorem~\thref{t:lax-milgram}.

\item[Definition~\thref{d:convex-subset}]is a direct dependency of:\\
  Theorem~\thref{t:orth-proj-onto-complete-convex},\\
  Lemma~\thref{l:characterization-of-orth-proj-onto-convex},\\
  Lemma~\thref{l:subspace-is-convex}.

\item[Theorem~\thref{t:orth-proj-onto-complete-convex}]is a direct dependency of:\\
  Theorem~\thref{t:orth-proj-onto-complete-subspace}.

\item[Lemma~\thref{l:characterization-of-orth-proj-onto-convex}]is a direct dependency of:\\
  Lemma~\thref{l:characterization-of-orth-proj-onto-subspace}.

\item[Lemma~\thref{l:subspace-is-convex}]is a direct dependency of:\\
  Theorem~\thref{t:orth-proj-onto-complete-subspace},\\
  Lemma~\thref{l:characterization-of-orth-proj-onto-subspace}.

\item[Theorem~\thref{t:orth-proj-onto-complete-subspace}]is a direct dependency of:\\
  Definition~\thref{d:orth-proj-onto-complete-subspace},\\
  Lemma~\thref{l:orth-proj-is-continuous-linear-map},\\
  Theorem~\thref{t:direct-sum-with-orth-compl-when-complete},\\
  Lemma~\thref{l:sum-is-orth-sum},\\
  Lemma~\thref{l:sum-of-complete-subspace-and-linear-span-is-closed},\\
  Theorem~\thref{t:riesz-frechet}.

\item[Definition~\thref{d:orth-proj-onto-complete-subspace}]is a direct dependency of:\\
  Lemma~\thref{l:orth-proj-is-continuous-linear-map},\\
  Theorem~\thref{t:direct-sum-with-orth-compl-when-complete},\\
  Theorem~\thref{t:riesz-frechet}.

\item[Lemma~\thref{l:characterization-of-orth-proj-onto-subspace}]is a direct dependency of:\\
  Lemma~\thref{l:orth-proj-is-continuous-linear-map},\\
  Theorem~\thref{t:direct-sum-with-orth-compl-when-complete}.

\item[Lemma~\thref{l:orth-proj-is-continuous-linear-map}]is a direct dependency of:\\
  Lemma~\thref{l:sum-of-complete-subspace-and-linear-span-is-closed}.

\item[Definition~\thref{d:orth-compl}]is a direct dependency of:\\
  Lemma~\thref{l:trivial-orth-compls},\\
  Lemma~\thref{l:orth-compl-is-subspace},\\
  Lemma~\thref{l:zero-intersection-with-orth-compl},\\
  Theorem~\thref{t:direct-sum-with-orth-compl-when-complete},\\
  Theorem~\thref{t:riesz-frechet},\\
  Theorem~\thref{t:lax-milgram}.

\item[Lemma~\thref{l:trivial-orth-compls}]is a direct dependency of:\\
  Theorem~\thref{t:riesz-frechet},\\
  Theorem~\thref{t:lax-milgram}.

\item[Lemma~\thref{l:orth-compl-is-subspace}]is a direct dependency of:\\
  Theorem~\thref{t:riesz-frechet}.

\item[Lemma~\thref{l:zero-intersection-with-orth-compl}]is a direct dependency of:\\
  Theorem~\thref{t:direct-sum-with-orth-compl-when-complete},\\
  Lemma~\thref{l:sum-of-complete-subspace-and-linear-span-is-closed}.

\item[Theorem~\thref{t:direct-sum-with-orth-compl-when-complete}]is a direct dependency of:\\
  Lemma~\thref{l:sum-is-orth-sum},\\
  Lemma~\thref{l:sum-of-complete-subspace-and-linear-span-is-closed},\\
  Theorem~\thref{t:riesz-frechet}.

\item[Lemma~\thref{l:sum-is-orth-sum}]is a direct dependency of:\\
  Lemma~\thref{l:finite-dimensional-subspace-in-hilbert-space-is-closed}.

\item[Lemma~\thref{l:sum-of-complete-subspace-and-linear-span-is-closed}]is a direct dependency of:\\
  Lemma~\thref{l:finite-dimensional-subspace-in-hilbert-space-is-closed}.

\item[Definition~\thref{d:hilbert-space}]is a direct dependency of:\\
  Lemma~\thref{l:closed-hilbert-subspace},\\
  Theorem~\thref{t:riesz-frechet},\\
  Theorem~\thref{t:lax-milgram},\\
  Lemma~\thref{l:finite-dimensional-subspace-in-hilbert-space-is-closed}.

\item[Lemma~\thref{l:closed-hilbert-subspace}]is a direct dependency of:\\
  Theorem~\thref{t:riesz-frechet},\\
  Theorem~\thref{t:lax-milgram-closed-subspace}.

\item[Theorem~\thref{t:riesz-frechet}]is a direct dependency of:\\
  Theorem~\thref{t:lax-milgram}.

\item[Lemma~\thref{l:compatible-rho-for-lax-milgram}]is a direct dependency of:\\
  Theorem~\thref{t:lax-milgram}.

\item[Theorem~\thref{t:lax-milgram}]is a direct dependency of:\\
  Theorem~\thref{t:lax-milgram-closed-subspace},\\
  Theorem~\thref{t:lax-milgram-cea-finite-dimensional-subspace}.

\item[Lemma~\thref{l:galerkin-orthogonality}]is a direct dependency of:\\
  Lemma~\thref{l:cea}.

\item[Theorem~\thref{t:lax-milgram-closed-subspace}]is a direct dependency of:\\
  Lemma~\thref{l:cea},\\
  Theorem~\thref{t:lax-milgram-cea-finite-dimensional-subspace}.

\item[Lemma~\thref{l:cea}]is a direct dependency of:\\
  Theorem~\thref{t:lax-milgram-cea-finite-dimensional-subspace}.

\item[Lemma~\thref{l:finite-dimensional-subspace-in-hilbert-space-is-closed}]is a direct dependency of:\\
  Theorem~\thref{t:lax-milgram-cea-finite-dimensional-subspace}.

\item[Theorem~\thref{t:lax-milgram-cea-finite-dimensional-subspace}]

\end{description}